\documentclass{hwthesis}
 
\usepackage{latexsym,amsmath,amssymb,amscd,eucal,appendix,amsthm}
\usepackage{mathrsfs}
\usepackage[all]{xy}
\usepackage{fancyhdr}
\usepackage{graphicx}
\newcommand{\Dslash}{\ensuremath \hspace{0.25cm}\raisebox{0.025cm}{\slash}\hspace{-0.32cm} D}

\newcommand{\dslash}{\not{\hbox{\kern-2pt $\partial$}}}
\newcommand{\pslash}{\not{\hbox{\kern-2.3pt $p$}}}

 \newtoks\nslashfraction
 \nslashfraction={.13}
 \newcommand{\nslash}[1]{\setbox0\hbox{$ #1 $}
   \setbox0\hbox to \the\nslashfraction\wd0{\hss \box0}/\box0 }

\def\appendix#1{\addtocounter{section}{1}\setcounter{equation}{0}
\renewcommand{\thesection}{\Alph{section}}
\section*{Appendix \thesection\protect\indent \parbox[t]{11.715cm} {#1}}
\addcontentsline{toc}{section}{Appendix \thesection\ \ \ #1} }
\newcommand{\eq}{\begin{equation}}
\newcommand{\eqend}{\end{equation}}

\newenvironment{remark}{\\[2mm]
\textbf{Remark }}
{\\[2mm]}

\newenvironment{romanlist}{%

        \begin{enumerate}
        }{%
        \end{enumerate}}
\newbox\ncintdbox \newbox\ncinttbox
\setbox0=\hbox{$-$} \setbox2=\hbox{$\displaystyle\int$}
\setbox\ncintdbox=\hbox{\rlap{\hbox to
\wd2{\hskip-.125em\box2\relax \hfil}}\box0\kern.1em}
\setbox0=\hbox{$\vcenter{\hrule width 4pt}$}
\setbox2=\hbox{$\textstyle\int$}
\setbox\ncinttbox=\hbox{\rlap{\hbox to
\wd2{\hskip-.175em\box2\relax \hfil}}\box0\kern.1em}

\def\Dirac{{D\!\!\!\!/\,}} 
\def\={\ =\ }

\def\spinc{spin$^c$~}

\newcommand{\complex}{{\mathbb C}} 
\newcommand{\zed}{{\mathbb Z}} 
\newcommand{\nat}{{\mathbb N}} 
\newcommand{\real}{{\mathbb R}} 
\newcommand{\reals}{{\mathbb R}} 
\newcommand{\rat}{{\mathbb Q}} 
\newcommand{\mat}{{\mathbb M}} 
\newcommand{\id}{{1\!\!1}} 
\def\slash{{\!\!\!/\,}} 
\def\alg{{\mathcal A}}
\def\balg{{\mathcal B}}
\def\hil{{\mathcal H}}
\def\hilR{{\mathcal H}_\real}
\def\bun{{\mathcal E}}
\def\lin{{\mathcal L}}
\def\comp{{\mathcal K}}

\def\pt{{\rm pt}}
\def\ev{{\rm ev}}
\def\CS{{\rm WZ}}
\def\cpt{{\rm cpt}}
\def\K{{\rm K}}
\def\H{{\rm H}}
\def\C{{\rm C}}
\def\E{{\rm E}}

\def\B{{\mathbb{B}}}

\def\Ltwo{{\rm L}^2}
\def\S{{\mathbb{S}}}
\def\P{{\mathbb{P}}}

\def\U{{\rm U}}
\def\im{{\rm im}}

\def\Tor{{\rm Tor}}
\def\Ext{{\rm Ext}}
\def\Hom{{\rm Hom}}

\def\End{{\rm End}}

\def\Cl{{{\rm C}\ell}}

\def\ch{{\rm ch}}

\def\Id{{\rm id}}
\def\pr{{\rm pr}}

\def\Index{{\rm Index}}

\def\Todd{{\rm Todd}}
\def\ind{{\rm ind}}
\def\Fred{{\rm Fred}}
\def\Cliff{{\rm Cliff}}

\def\Thom{{\rm Thom}}

\newcommand{\Tr}[1]{\:{\rm Tr}\,#1}

\def\e{{\,\rm e}\,}

\def\be{\begin{equation}}
\def\ee{\end{equation}}
\def\bea{\begin{eqnarray}}
\def\eea{\end{eqnarray}}
\def\bd{\begin{displaymath}}
\def\ed{\end{displaymath}}

\def\dd{{\rm d}}

\def\ii{{\,{\rm i}\,}}

\def\tor{{\rm tor}}
\def\Vect{{\rm Vect}}

\makeatletter
\newdimen\normalarrayskip              
\newdimen\minarrayskip                 
\normalarrayskip\baselineskip
\minarrayskip\jot
\newif\ifold             \oldtrue            
\def\arraymode{\ifold\relax\else\displaystyle\fi} 
\def\@arrayskip{\ifold\baselineskip\z@\lineskip\z@
     \else
     \baselineskip\minarrayskip\lineskip2\minarrayskip\fi}
\def\@arrayclassz{\ifcase \@lastchclass \@acolampacol \or
\@ampacol \or \or \or \@addamp \or
   \@acolampacol \or \@firstampfalse \@acol \fi
\edef\@preamble{\@preamble
  \ifcase \@chnum
     \hfil$\relax\arraymode\@sharp$\hfil
     \or $\relax\arraymode\@sharp$\hfil
     \or \hfil$\relax\arraymode\@sharp$\fi}}
\def\@array[#1]#2{\setbox\@arstrutbox=\hbox{\vrule
     height\arraystretch \ht\strutbox
     depth\arraystretch \dp\strutbox
     width\z@}\@mkpream{#2}\edef\@preamble{\halign \noexpand\@halignto
\bgroup \tabskip\z@ \@arstrut \@preamble \tabskip\z@ \cr}%
\let\@startpbox\@@startpbox \let\@endpbox\@@endpbox
  \if #1t\vtop \else \if#1b\vbox \else \vcenter \fi\fi
  \bgroup \let\par\relax
  \let\@sharp##\let\protect\relax
  \@arrayskip\@preamble}
\makeatother

\newcommand{\beq}{\begin{eqnarray}}
\newcommand{\eeq}{\end{eqnarray}}

\newcommand{\G}{\Gamma}
\newcommand{\ta}{\tau}

\newcommand{\del}{\partial}

\newcommand{\KO}{{\rm KO}}

\newcommand{\KK}{{\rm KK}}
\newcommand{\KKO}{{\rm KKO}}
\newcommand{\Spin}{{\rm Spin}}
\newcommand{\SO}{{\rm SO}}
\newcommand{\UU}{{\rm U}}
\newcommand{\SU}{{\rm SU}}
\newcommand{\SL}{{\rm SL}}

\newcommand{\ass}{{\rm ass}}
\newcommand{\normx}{{\lVert x \rVert}}
\newcommand{\normy}{{\lVert y \rVert}}
\newcommand{\normxy}{{\lVert x\,y \rVert}}
\newcommand{\norm}{{\lVert 1 \rVert}}
\newcommand{\HR}{\mathcal{H}_{\mathbb{R}}}
\newcommand{\ocat}[1]{\textsf{Or}(#1)}
\newcommand{\subc}[1]{\textsf{Sub}(#1)}
\newcommand{\cat}[1]{\textsf{#1}}

\newcommand{\lmod}[1]{{#1}\!-\!\text{Mod}}
\newcommand{\rmod}[1]{\text{Mod}\!-\!{#1}}
\newcommand{\man}[1]{\text{#1}}
\newcommand{\group}[1]{\text{#1}}
\newcommand{\ram}[1]{\text{#1}}
\newcommand{\rgenus}[1]{\hat{\mathcal{A}}(#1)}
\newcommand{\kred}[2]{\widetilde{\rm K}^{#1}(#2)}
\newcommand{\kgroup}[2]{{\rm K}^{#1}(#2)}
\newcommand{\kogroup}[2]{{\rm KO}^{#1}(#2)}
\newcommand{\redk}[2]{{\rm \widetilde{K}}^{#1}(#2)}
\newcommand{\redko}[2]{{\rm \widetilde{KO}}^{#1}(#2)}
\newcommand{\ckgroup}[2]{{\rm K}_{cpt}^{#1}(#2)}
\newcommand{\ckogroup}[2]{{\rm KO}_{cpt}^{#1}(#2)}
\newcommand{\topsp}[1]{{\rm #1}}
\newcommand{\cliffmod}[1]{\widehat{\mathfrak{M}}^{#1}}
\newcommand{\appCat}{Appendix A}
\newcommand{\appclifford}{Appendix B}
\newcommand{\appcharac}{Appendix C}

\newcommand{\appequiv}{Appendix D} 

\newtheorem{theorem}{Theorem}[chapter]
\newtheorem{lemma}[theorem]{Lemma}
\newtheorem{cor}[theorem]{Corollary}
\newtheorem{proposition}[theorem]{Proposition}
\theoremstyle{definition}
\newtheorem{definition}[theorem]{Definition}
\newtheorem{example}[theorem]{Example}
\newtheorem{remark2}[theorem]{Remark}
\author{Alessandro Valentino}
\title{K-Theory, D-Branes and Ramond-Ramond Fields}
\date{}
\numberwithin{equation}{section}
\begin{document}

\maketitle

\begin{abstract}
This thesis is dedicated to the study of K-theoretical properties of D-branes and Ramond-Ramond fields.\\
We construct abelian groups which define a homology theory on the category of CW-complexes, and prove that this homology theory is equivalent to the bordism representation of KO-homology, the dual theory to KO-theory. We construct an isomorphism between our geometric representation and the analytic representation of KO-homology, which induces a natural equivalence of homology functors. We apply this framework to describe mathematical properties of D-branes in type I String theory.\\
We investigate the gauge theory of Ramond-Ramond fields arising from type II String theory defined on global orbifolds. We use the machinery of Bredon cohomology and the equivariant Chern character to construct abelian groups which generalize the properties of differential K-theory defined by Hopkins and Singer to the equivariant setting, and can be considered as a differential extension of equivariant K-theory for finite groups. We show that the Dirac quantization condition for Ramond-Ramond fieldstrengths on a good orbifold is dictated by equivariant K-theory and the equivariant Chern character, and study the group of flat Ramond-Ramond fields in the particular case of linear orbifolds in terms of our orbifold differential K-theory.   
\end{abstract}
\renewcommand{\abstractname}{Acknowlegdements}
\begin{abstract}
I want to thank my supervisor Prof. Richard Szabo for the support and patience during these years, for the many discussions we had, and for allowing me to always express and develop my own point of view in any aspect of my research.\\
I want to thank my  colleague and friend Rui Reis for the uncountable discussions on maths and science in general, which made me better appreciate many aspects of these subjects.\\
A particular thanks goes to my examiners Jacek Brodzki, Jos\'e Figueroa-O'Farrill, and Des Johnston, whose questions and comments helped improve the quality of the present work.\\
I want to thank U.Bunke, J.Figueroa-O'Farrill, D.Freed, J.Greenlees, J.Howie, A. Konechny, W.L\"{u}ck, T.Schick, P.Turner, and S.Willerton for helpful suggestions and correspondence.\\
I want to thank Dr. Mark Lawson for the discussions on maths, physics, and British culture.\\ 
I want to thank Michele Cirafici and Mauro Riccardi for the many conversations we had in front of a pint of ale.\\
I want to thank Fedele Lizzi for the support during difficult moments, and Patrizia Vitale for the long time interest in my scientific developments.\\ 
I want to thank Giorgos, Henry, and Yorgos for being such great mates, and for the many hours of fun we had.\\
A thank you goes to my officemates Emma, Kenny, Sally, and Singyee, for being the best officemates I could have ever hoped for.\\
I want to thank Christine, Claire and Pat for helping with bureocratic matters of any sort.\\
A special thanks goes to my parents, for letting me always pick any decision by my own, which highly contributed to the person I am now.\\
Last, but in no way least, I want to thank Antonella. Unfortunately, words are useless to describe my deep gratitude, and her relevance to the very existence of this work. They say that ``behind every great man there is always a great woman'': my case trivially suggests that a feminine ``version'' of the previous statement could not possibly hold.
\end{abstract}
\newpage
\thispagestyle{empty}
\quad
\vspace{7.5cm}
\begin{center}
The present work of thesis is based on the following articles
\begin{itemize}
\item[-] R.J.Szabo and A.Valentino, ``Ramond-Ramond Fields, Fractional Branes and Orbifold Differential K-theory'', arXiv:0710.2773
\item[-] R.M.G.Reis, R.J.Szabo, and A.Valentino, ``KO-Homology and Type I String Theory'', arXiv:hep-th/0610177 
\end{itemize}
\end{center}
\pagestyle{plain}
\newpage
\pagenumbering{roman}
\tableofcontents
\newpage
\quad
\thispagestyle{empty}
\pagenumbering{arabic}
\chapter*{Introduction}
\addcontentsline{toc}{chapter}{Introduction}
\setcounter{page}{1}
\pagestyle{plain}
\begin{center}
\emph{``The most powerful method of advance that can be suggested at present is to employ all the resource of pure mathematics in attempts to perfect and generalise the mathematical formalism that forms the existing basis of theoretical physics, and \emph{after} each success in this direction, to try to interpret the new mathematical features in terms of physical entities''}\\
P.A.M. Dirac, 1931
\end{center}

Since the dawn of science, the interaction between physics and mathematics has always been an interesting and debated one, given the apparent difference between these two disciplines. Indeed, on one hand, physics deals with natural phenomena as they happen in an ``objective reality'', external to the observer, and the aim of physics is then to understand and formulate the laws they obey. On the other hand, mathematics appears to deal with a reality which is internal to the human being, populated by objects which need only obey the laws of logic and consistency. Furthermore, their methodology seems rather different: physics proceeds by experiments and particular cases, while mathematics proceeds by deduction and chains of logical statements. As it always happens when different disciplines come in contact, the interaction of these two human activities is bound to enlarge our knowledge of the reality we live in, and of the nature of the human being himself. In this respect, the 20th century has been of a crucial importance. In particular, most part of it has been dominated by the influence of the advances in modern mathematics, such as differential geometry, functional analysis, and algebra in the formulation and understanding of fundamental physical theories, such as General Relativity, Quantum Mechanics, and Quantum Field Theory. This influence was so prominent and surprising that Eugene Wigner was lead to celebrate it in the now classic paper ``The Unreasonable Effectiveness of Mathematics in the Natural Sciences'' \cite{Wigner}. This fruitful interaction has become even stronger in the late part of the last century. Interestingly, though, we have somehow witnessed the shift of influence from mathematics to physics, in such a way that we are lead to wonder about the ``unreasonable effectiveness of physics on mathematics''. The description of Jones' polynomials via quantum field theoretical techniques, which earned Edward Witten a Fields Medal in 1990, is only the most evident result of the above interaction. This is not something new: after all, differential calculus was born mainly out of the necessity to solve problems in mechanics. However, mathematics in the 20th century, in particular in its second half, has evolved independently of any application, following more abstract and formal principles. This makes the possibility of its relation with physics even more exciting and surprising.\\

Despite the difficulties and controversies surrounding it, it is undeniable that String theory has played a major role in these developments, generating an incredible number of interesting problems in different fields of mathematics, in particular algebra, geometry and topology. The rich structure of String theory has lead mathematical physicists and pure mathematicians to work on the same problems, albeit with different attitudes and motivations, allowing both communities to better appreciate the extent of their own subject of study.\\
The present work of thesis belongs to the above trend, finding its place at the interface of geometry, topology, and String theory. Loosely speaking, the ``middle point'' is represented by the circle of ideas surrounding K-theory, a generalized cohomology theory developed by Atiyah, Hirzebruch, and Grothendieck, among others, which can be defined in terms of complex vector bundles on topological spaces. The relevance of K-theory in String theory rests on the fact that D-branes, extended objects present in the theory, have charges which are not classified by the homology cycle of their worldvolumes, as expected, but rather by the K-theory of their normal bundles in the spacetime manifold. In the first part of this thesis, we will exploit the fact that a more natural description of D-branes can be given in terms of K-homology, which is the generalized homology theory dual to K-theory. The interesting mathematical aspect concerning K-homology is the fact that it can very naturally be constructed by using both a geometric and an analytic representation, respectively in terms of ${\rm Spin}^{c}$-manifolds and vector bundles, and of Fredholm modules. In particular, Baum and Douglas \cite{Baumdouglas} showed that an isomorphism can be constructed at the level of the representative cycles, and that such an isomorphism induces a natural equivalence between the geometric and analytic K-homology functors. The authors then showed that the existence of such an isomorphism is equivalent to the Index Theorem of Atiyah and Singer for the canonical Dirac operator, and conjectured that this is the case for any flavour of K-theory.\\
In this thesis we will support this conjecture, by constructing a natural equivalence between geometrical and analytical KO-homology, the dual theory to the K-theory of \emph{real} vector bundles, elucidating its relation with a suitable Index Theorem.\\
The difficult aspect of the above analysis lies in the fact that the index homomorphism takes values in abelian groups with torsion, which does not allow a local description in terms of characteristic classes. Since KO-homology describes D-branes in type I superstring theory, we will be able to construct D-branes with torsion charges which are not associated to any spacetime form field.\\
Indeed, D-branes can also be realised as currents for Ramond-Ramond fields. These ``dual'' objects are gauge fields locally described by differential forms of different degrees, depending on the type of String theory considered. Their fieldstrengths satisfy generalized Maxwell equations, and they interact with D-branes via the usual integral coupling. Since D-brane charges are classified by K-theory/K-homology, it is certainly expected that K-theory plays some role in the Ramond-Ramond gauge theory behaviour. This is indeed the case: the suitable mathematical formalism to describe these fields in nontrivial backgrounds is known as differential K-theory, which is an example of generalized differential cohomology theories recently developed in \cite{Hopkins2005}. In particular, the main ingredient is the Chern character homomorphism from K-theory to the even part of the deRham cohomology ring of the spacetime manifolds, which roughly realizes the Ramond-Ramond current generated by D-branes.\\         
In the second part of this thesis, we will generalize the arguments leading to the above conclusions to the case of type II String theory on orbifolds. An orbifold can be loosely described as a singular object which is locally isomorphic to the quotient of an Euclidean space by a finite subgroup of the group of linear transformations. Despite the presence of the singularities, String theory behaves well on such objects: in particular, D-branes can be introduced, and their charges are classified by equivariant (or orbifold) K-theory, as proposed by \cite{Witten1998}. We will analyse the Ramond-Ramond gauge theory arising from the closed type II String theory on global orbifolds, and demonstrate the role played by equivariant K-theory. In particular, we will construct some abelian groups satisfying all the expected properties to be considered a generalization of differential K-theory for global orbifolds. This will require the introduction of an equivariant cohomology theory not very popular in the physical literature, developed by Bredon in \cite{Bredon}, which nevertheless, as we will argue, captures all the relevant physical properties of Ramond-Ramond fields on orbifolds.\newpage
\begin{center}
\textbf{\large Plan of the Work}
\end{center}
We will know give a brief description of the contents of this thesis, trying to highlight the main new results in mathematics and String theory.

In Chapter 1 we introduce some generalities about String theory, focusing on the structure of its quantum spectrum and on the low energy limit. We have tried throughout to express all the relevant notions in a precise mathematical framework whenever possible, and we have avoided details of the constructions involved, referring the reader to the extensive literature on the subject. In this way, the basic concepts relevant to this work of thesis should be accessible to a mathematically minded audience acquainted with the basics of field theory and quantum mechanics.

In Chapter 2 we introduce D-branes as boundary conditions for open String theory, emphasising the geometric and topological properties of such objects. Indeed, we will focus mainly on properties of D-branes that are expected in any topological nontrivial background, such as the behaviour of the Chan-Paton vector bundle. The main properties of supersymmetric D-branes are briefly reviewed via the analysis of the spinor bundle on the D-brane worldvolume, which is a recurrent ingredient throughout this thesis. We will introduce the gauge theory of Ramond-Ramond fields, and discuss Ramond-Ramond charges and the anomalous couplings with D-branes. Also in this case, we will avoid the rather lengthy computations regarding the inflow mechanism, since these techniques will not play any relevant role in this work of thesis. Finally, we state Sen's conjectures on D-brane decay and mention Witten proposal on D-brane charge classification.

Chapter 3 consists of a quick introduction to topological K-theory. For a given CW-complex $\rm X$, we define the group ${\rm K}^{0}({\rm X})$ as the Grothedieck group associated to the monoid of vector bundles over $\rm X$, and define the higher K-groups via suspension. We describe the multiplicative structure possessed by K-theory, and discuss the Atiyah-Bott-Shapiro isomorphism in some details. We then restrict ourselves to the category of ${\rm Spin^{c}}$-manifolds, discussing the concept of $\rm K$-orientation, Thom isomorphisms, and the Chern character. Finally, we illustrate how the K-theoretical machinery is used in the classification of D-brane charges in type II and type I String theory. This chapter is of great importance, since the concepts therein will be tacitly assumed throughout the rest of the thesis.

Chapter 4 consists mostly of original material regarding KO-homology. After a brief introduction on dual theories and spectral KO-homology, we will recall the basic notions about the theory of \emph{real} ${\rm C}^{*}$-algebras, emphasising the differences with the analogous results in the theory of ordinary complex ${\rm C}^{*}$-algebras. We will then introduce Kasparov's formalism for KKO-theory, which is based on the notions of Hilbert modules and generalized Fredholm modules, and define analytic KO-homology for a topological space $\rm X$ via the ${\rm C}^{*}$-algebra of \emph{real} functions on $\topsp{X}$. We proceed to construct geometric KO-homology in terms of Spin manifolds and real vector bundles. By comparison with the bordism description of spectral K-homology developed in \cite{jakob}, we prove that our construction is a generalized homology theory dual to KO-theory, and discuss the various relevant homological properties. We then introduce the main mathematical result of this chapter, given by the construction of an isomorphism $\mu$ between the geometric and analytic representation of KO-homology which induces a natural equivalence between the geometric and analytic KO-homology functors. To this aim we introduce some index homomorphism on geometric and analytic KO-homology, and prove that the Index theorem for a suitable Dirac operator implies that $\mu$ is indeed an isomorphism. We point out that the proof of the above theorem appears also in \cite{baum-2007-3}, albeit it is completely different from the one we present here, which is more suitable for the applications we present later. We also construct a homological real Chern character, and use it to give an alternative derivation of cohomological index formulas for the canonical $\Cl_{n}$-linear Atiyah-Singer operator. From the physical point of view, we introduce the concept of \emph{wrapping charge} of a \emph{wrapped D-brane}, showing that in type I String theory it is a genuinely different notion from that of an ordinary D-brane. Finally, we construct nontrivial generators for the KO-homology of a point, and interpret these in terms of wrapped D-branes.\\
In Chapter 5 we give a detailed account on (generalized) differential cohomology theories. We motivate this mathematical formalism by discussing the properties of ordinary electromagnetism with Dirac quantization of charges in topologically nontrivial backgrounds. We give a rather extensive treatment of both Cheeger-Simons groups and Deligne, in order to build a solid intuition for these mathematical objects. We then discuss the Moore and Witten argument regarding the charge quantization of Ramond-Ramond fields, as proposed in \cite{Moore2000}. We conclude the chapter by explaining the construction of differential K-theory of Hopkins and Singer, which will be generalized later in the thesis.

In Chapter 6 we present new physical and mathematical results regarding String theory on global orbifolds. We first give a brief introduction to equivariant cohomology theories on the category of $\rm G$-CW complexes, and consider equivariant K-theory as an example. We then proceed to define Bredon cohomology in terms of natural transformations between functors on the orbit category of a finite group. This is the main ingredient used to constructed the equivariant Chern character defined in \cite{Luck1998}, which has the unique property of inducing an isomorphism on rational equivariant K-theory. One of the main new physical result of this chapter is to demonstrate that the Dirac quantization of Ramond-Ramond fields on orbifolds is dictated by the above Chern character, which, after a process of \emph{delocalization}, can be used to describe the couplings of Ramond-Ramond field on global orbifold with \emph{fractional} D-branes, which we will describe in terms of equivariant K-homology. We check this statement on linear orbifolds, which are the usual cases studied in the physics literature. Our approach has the main advantage of being applicable to the case of nonabelian orbifolds and for quotients of general Spin manifolds by finite groups. From the mathematical point of view, we construct abelian groups which have all the desired properties for a generalization of differential K-theory to global orbifolds. This is mathematically needed, since the general results of Hopkins and Singer in \cite{Hopkins2005} hold only for generalized cohomology theories on the category of manifolds. Far from reaching the generality of \cite{Hopkins2005}, we find this an important step towards the general construction of differential extensions of equivariant generalized cohomology  theories. We will use our equivariant (or orbifold) differential cohomology theory to describe Ramond-Ramond fields on orbifolds, and study in particular flat Ramond-Ramond potentials.\\

We conclude the work with Appendices which aim to settle the notations for some of the standard notions used throughout the thesis.

\newpage
\quad
\vspace{50mm}
\begin{flushright}
\parbox[c]{7cm}{\emph{``Freedom is the freedom to say that two plus two make four. If that is granted, all else follows.''}\vspace{-2mm}
\begin{flushright}
G.Orwell, \textbf{1984}
\end{flushright}}
\end{flushright}
\thispagestyle{empty}
\chapter{Generalities on String Theory}
\pagestyle{fancy}
In this chapter we will collect some basic results in String theory that will be used as the starting point for the rest of the thesis.We will review some well known aspects of the perturbative formulation of bosonic and supersymmetric string theories. We direct the reader to \cite{Green},\cite{Polchinski},\cite{quantumath} for more information about these constructions.  
\section{The Bosonic String}
The action functional for a string propagating in spacetime is a direct generalization of the functional describing the motion of relativistic point particle. In the case of a point particle, the action functional for a given curve $\gamma:[0,1]\to\man{M}$, where $\man{M}$ is a d-dimensional Lorentzian manifold, is given by
\begin{equation}
\text{S}[\gamma]:=\int_{\gamma}\mu_{g|\gamma}
\end{equation}
where $\mu_{g|\gamma}$ is the invariant volume form for the spacetime metric restricted to the curve $\gamma$. The action $\text{S}$ computes the length of the curve $\gamma$, the worldline of the point particle, and its stationary points are the geodesics for the metric $g$ on $\man{M}$.\\
\indent It is natural to generalize this action for an extend p-dimensional object propagating in $\man{M}$ as the ``volume'' of the surface swept by the object while propagating. In other words, given $f:\Sigma\to\man{M}$, where $\Sigma$ is a p-dimensional manifold, called the \emph{worldsheet}, and $f$ is a smooth immersion, we have that
\begin{equation}
\text{S}[f]:=\int_{\Sigma}\mu_{f^{*}g}
\end{equation} 
The case of a free string propagating in $\man{M}$ is given by p=2, and by $\Sigma\simeq\text{S}^{1}\times{\mathbb{R}}$ for closed string, and $\Sigma\simeq[0,1]\times{\mathbb{R}}$ for open strings.\\  
The functional $\text{S}$ is called the \emph{Nambu-Goto action} and in a local system of coordinates $\{\sigma^{i}\}$ over $\Sigma$, the functional $\text{S}$ can be represented as
\begin{equation}
\text{S}[f]=\int_{\Sigma}\dd^{p}\sigma\sqrt{-\text{det}(g_{\mu\nu}\frac{\partial{x}^{\mu}}{\partial{\sigma^{i}}}\frac{\partial{x}^{\nu}}{\partial{\sigma^{j}}})}
\end{equation}
where $x^{\mu}(\sigma)$ are local representatives for the function $f$.\\
Unfortunately, this action is non-polynomial in the $x^{\mu}$ and its derivatives, making its quantization difficult to define unambiguously, even in flat spacetime.\\
For this reason, it is generally preferred to use the classically equivalent \emph{Polyakov action} given by
\begin{equation}\label{Polyakov}
\text{S}[f,\gamma]:=k\int_{\Sigma}\mu_{\gamma}<\gamma,f^{*}g>
\end{equation}
In the above expression, $\gamma$ is a Lorentzian metric on $\Sigma$, $<,>$ denotes $||{\rm d}f||^{2}$, defined by considering ${\rm d}f$ as an element in $\Omega^{1}(\Sigma;f^{*}{\rm TM})$, and using both $\gamma$ on ${\rm T^{*}\Sigma}$ and $g$ on ${\rm TM}$, and $k$ is called the \emph{tension}. Notice that now the intrinsic metric $\gamma$ is a dynamical variable, while the metric tensor $g$ is considered as a ``background'', i.e. it is a nondynamical quantity. When also $\gamma$ is held fixed, the action S describes a \emph{nonlinear sigma model}, and its stationary points when p=2 are called \emph{harmonic maps}.\\
In the local system of coordinates as above the Polyakov action can be represented as
\begin{equation}
\text{S}[f,\gamma]=k\int_{\Sigma}\dd^{p}\sigma\sqrt{-\text{det}(\gamma)}\gamma^{ij}(g_{\mu\nu}\frac{\partial{x}^{\mu}}{\partial{\sigma^{i}}}\frac{\partial{x}^{\nu}}{\partial{\sigma^{j}}})
\end{equation}
By construction, the Polyakov action is invariant under $Diff^{+}(\Sigma)$, the group of orientation preserving diffeomorphisms, and under $ISO(g)$, the group of isometries of $g$. Moreover, only in the peculiar case p=2 the functional S is invariant under $C^{\infty}_{+}(\Sigma)$, the group of smooth positive functions on $\Sigma$, acting as \emph{Weyl rescaling} of the metric $\gamma$, i.e. transformations of the type $\gamma\to\rho\cdot\gamma$.\\
 Hence the full group of symmetries of the string action is given by $Diff^{+}(\Sigma)\ltimes{C^{\infty}_{+}(\Sigma)}\times{ISO(g)}$, where $\ltimes$ means semidirect product.\\
 The invariance under diffeomorphims and Weyl scalings is crucial in giving the Polyakov action a more tractable form, allowing to use canonical quantization techniques, and as this is possible only for p=2, this is seen as a reason to rule out higher dimensional extended objects other than strings.\\
 Indeed, a 2-dimensional manifold is always locally \emph{conformally flat}, i.e. there always exist local coordinates $(u,v)$ in which the metric $\gamma$ can be expressed as $\rho\cdot\eta_{ab}$, where $\eta_{ab}$ is the usual Minkowski metric in 2 dimensions. Diffeomorphisms and Weyl invariance then allow to ``pick a gauge'', called the \emph{conformal gauge}, in which the functional S is the action functional for a vibrating string in a curved spacetime.\\
 When $\man{M}$ is the d-dimensional Minkowski spacetime, the equations of motion of the string in the conformal gauge can be completely solved, and its canonical quantization carried on thanks to an additional symmetry enjoyed by the string functional.\\
Indeed, conformal transformations of the Minkowski metric in two dimensions do not spoil the (local) conformal gauge, implying that String theory (at any fixed $\gamma$) is a two dimensional conformal field theory, hence completely integrable by symmetry considerations alone, as the conformal algebra in two dimensions is infinite dimensional.\\
Conformal invariance is so important that it is required to hold also at the quantum level, where an anomaly could possibly spoil it: the cancellation of such a conformal anomaly, indeed, fixes the spacetime dimensionality to d=26, called the \emph{critical dimension}.
\begin{remark}
During this thesis, we will mainly consider both the worldsheet and spacetime manifolds as being equipped with a Riemannian as opposed to a Lorentzian metric tensor. Even if this is not per se physically realistic, in most of the cases one can perform a \emph{Wick rotation} on the spacetime manifold and obtain the relativistic description we have introduced so far.\\
Moreover, we will also consider different topologies for the worldsheet, and in particular we will regard $\Sigma$ as a Riemann surface of genus $g$: this is essentially due to how \emph{interactions} are introduced in String theory at the quantum level. This perturbative formulation of String theory also forces the use of Riemannian worldsheets, as opposed to Lorentzian, since a compact manifold admits  a Lorentzian metric if and only if the Euler number vanishes. For example, the transition amplitude for the propagation of a quantum string will be given by
\begin{displaymath}
A\sim\sum_{\text{topologies of}\:\Sigma}\:\int_{\text{Met}(\Sigma)}\mathcal{D}\gamma\int_{\text{Map}(\Sigma,\man{M})}\mathcal{D}f\:e^{-\text{S}[f,\gamma]}
\end{displaymath}
where $\mathcal{D}\gamma$ and $\mathcal{D}f$ are ``path integral measures'' over a space of metrics and maps, respectively. More precisely, the above expression should be ``gauge fixed'': indeed, the path integral over the space of metrics reduces to an integral over a finite dimensional moduli space. 
\end{remark}
\section{The Supersymmetric String}
Despite the rich structure of its symmetries, the bosonic string has some fundamental flaws. The most relevant ones, apart from the high dimensionality of the spacetimes allowed, are the presence of a tachyonic state, which is a strong signal towards instability of the quantum theory, and the absence of fermions, which is not a feature of an inconsistent theory, but nevertheless fermions are required for physical reasons.\\
A way out of these two problems is the so called \emph{Neveu-Schwarz-Ramond (NSR) formulation} of String theory, which consists in introducing additional fermionic degree  on the woldsheet.\\
More precisely, the Polyakov action can be at first generalized as
\begin{equation}\label{spolyakov}
\text{S}[f,\gamma,\psi]:=k\int_{\Sigma}\mu_{\gamma}\left\{<\gamma,f^{*}g>+\bar{\psi}\Dslash\psi\right\}
\end{equation}
where $\psi$ is a section of $\text{S}\otimes{f^*{\man{TM}}}$, with $\text{S}$ the spin bundle over $\Sigma$ for a given spin structure, and $\Dslash$ is the Dirac operator associated to $\gamma$, coupled to $f^{*}{\man{TM}}$.\\
Notice first that $\Sigma$ is a spin manifold, being 2-dimensional, and it admits, at any genus $g$, $2^{2g}$ inequivalent spin structures which can be distinguished by a $\pm$ sign along the homology cycles of $\Sigma$. Moreover in even dimensions the spin bundle $\text{S}$ decomposes according to the chirality operator as $\text{S}^{+}\oplus\text{S}^{-}$.\\
One also requires $\psi$ to be a Majorana spinor field to ensure that the full action is real. Majorana spinors exist on worldsheets with Lorentzian signature, but not on those with Euclidean signature. Anyway, in the Euclidean case one can use an ordinary chiral spinor field $\psi_{+}$, and choose $\psi_{-}$ to be its complex conjugate, in order to preserve the degrees of freedom.\\
More concretely, $\psi_{+}$ will be a section of $\text{K}^{1/2}$, the square root of the canonical bundle on $\Sigma$, and $\psi_{-}$ a section of $\bar{\text{K}}^{1/2}$ for the same spin structure.\\
In the free (closed) superstring case, the topology of $\Sigma$ admits a single homology cycle, and hence there are two spin structures, which are conventionally referred to as \emph{Ramond}(R) and \emph{Neveu-Schwarz}(NS), and can be characterized by
\begin{displaymath}
\begin{array}{l}
\psi_{\pm}(\tau,\sigma+2\pi)=+\psi_{\pm}(\tau,\sigma)\quad\text{R: periodic conditions}\\
\psi_{\pm}(\tau,\sigma+2\pi)=-\psi_{\pm}(\tau,\sigma)\quad\text{NS: anti-periodic conditions}
\end{array}
\end{displaymath}
On a flat spacetime, the free superstring can be quantized using canonical quantization: the main difference with the bosonic string is given by the appearance of new sectors, NS and R, given by the boundary conditions for the spinor fields $\psi$.\\
In particular, the states in the Fock space $\mathcal{F}^{\text{NS}}$ for Neveu-Schwarz degrees of freedom can be shown to be spacetime bosons, while the states in the Fock space $\mathcal{F}^{\text{R}}$ for Ramond degrees of freedom can be shown to be spacetime fermions.\\[2mm]
Unfortunately, both the above spaces contain \emph{negative norm states} that are not eliminated by any kind of symmetry, in contrast to what happens in the bosonic string case thanks to the Virasoro algebra constraints. To ensure this, one modifies the action (\ref{spolyakov}) by introducting a \emph{gravitino}, a spin 3/2 field, and an interaction term in order for the action to be invariant under \emph{local worldsheet supersymmetry}. In analogy with the bosonic case, the modified superstring action will define a \emph{superconformal} field theory: the preservation of the super conformal symmetry at the quantum level requires the spacetime dimensionality d to be 10.\\
Even with these modifications, the theory is still inconsistent, as there is still a tachyonic state in the NS sector, with no supersymmetric partner state, making it impossible to have spacetime supersymmetry.\\
To overcome this final problem, Gliozzi-Scherk and Olive proposed a procedure for a truncation of the RNS String theory that produces a spectrum with spacetime supersymmetry. This truncation is called the \emph{GSO projection}, which can be thought as a projection on the space of invariant states for the operator $(-1)^{F}$, which assigns to each state the number of fermions present modulo 2.\\
This operator can be carefully defined in both the Ramond and Neveu-Schwarz sector, obtaining that the free open GSO projected string is N=1 supersymmetric, while the free closed GSO projected string is N=2 supersymmetric.
\begin{remark}
In the functional integration formulation, the GSO projection corresponds to a weighted sum over the spin structures of the worldsheet, with the weight choosen in such a way that the resulting amplitudes are invariant under the action of the modular group of $\Sigma$.
\end{remark}
\section{Quantum aspects}\label{quantum}
In this section we will briefly recall the (massless) content of the supersymmetric string space of states, in flat 10-dimensional Minkowski spacetime \cite{Green,quantumath}.\\
Denote by $\mathcal{F}_{k}$ and $\tilde{\mathcal{F}}_{k}$ the Fock spaces for the bosonic degrees of freedom at momentum $k$ for \emph{left} and \emph{right} movers, obtained upon a holomorphic decomposition of the fields in local complex worldsheet coordinates.\\
Then the full RNS Fock space for open and closed strings is given by
\begin{eqnarray}\label{fockspaces}
\mathcal{F}_{\text{open}}&:=&\bigoplus_{k}\mathcal{F}_{k}^{RNS}\nonumber\\
\mathcal{F}_{\text{closed}}&:=&\bigoplus_{k}\mathcal{F}_{k}^{RNS}\otimes{\tilde{\mathcal{F}}_{k}}^{RNS}\qquad k\in\mathbb{R}^{10}
\end{eqnarray} 
where
\begin{displaymath}
\mathcal{F}_{k}^{RNS}:=\mathcal{F}_{k}\otimes(\mathcal{F}^{R}\oplus\mathcal{F}^{NS})
\end{displaymath}
and the same for $\tilde{\mathcal{F}}_{k}^{RNS}$.\\
After imposing the physical constaints required by the superconformal symmetry, we will obtain two positive-definite Hilbert spaces for open and closed strings, denoted by $\mathcal{F}^{\text{phys}}_{open}$ and $\mathcal{F}^{\text{phys}}_{closed}$, respectively.\\
By (\ref{fockspaces}), the open string Fock space contains two sectors, Ramond and Neveu-Schwarz. The ground state in the Ramond sector satisfies a massless Dirac equation, suggesting that it is single particle state for a spacetime fermion field, while the ground state in the Neveu-Schwarz sector is a bosonic tachyon, obtained by a spacetime scalar field.\\
Moreover, the first excited state is a \emph{massless vector state}, with its degree of freedom suggesting it is a one particle state for a Yang-Mills field \cite{Polchinski,Green}.\\
These, and an infinite tower of massive states, are contained in the left-moving sector, but by (\ref{fockspaces}) this suffices to construct the Hilbert space of the open string.\\[2mm]
For the closed string, instead, we need to consider a tensor product for left and right movers Hilbert spaces: we will have four sectors, characterized by the spin structure choice for left and right-moving degrees of freedom.\\
The ground state in the NS-NS sector is a tachyon, as for the open string, while the massless states contains, in the same sense as before, a \emph{graviton}, associated to a symmetric spacetime tensor of type (2,0), the \emph{B-field}, coming from a spacetime two form, and a \emph{dilaton}, a spacetime scalar field. In particular, the graviton state satisfies linearized Einstein equations, hence its name.\\
The R-NS (and NS-R) ground state is a massless state, which is reducible into a spinor state coming from a spinor field $\lambda$, called the \emph{dilatino}, and a \emph{gravitino} state, coming from a spinor-vector field. These are the superpartners of the dilaton and graviton, respectively.\\
Finally, the ground state in the R-R sector is massless, and can be reduced in states that are one particle states for spacetime differential forms of degree 0,\ldots,10, called \emph{Ramond-Ramond fields}. In a certain sense, this is the ``most important'' sector for the content of this thesis: indeed, most of the next chapters  will be devoted to explore the rich mathematical properties of these objects, and the interaction with their ``sources'', called \emph{D-branes}.\\
Of course, all the sectors described above contain also a (infinite) number of massive excited states.\\[2mm]
As was pointed out in the previous section, the theory is still inconsistent, due to the presence of the tachyon state in the NS sector, and does not present spacetime supersymmetry, required by the consistency of the coupling of the \emph{massless} gravitino appearing in the spectrum.\\ 
This enforces the use of the GSO projection, both for the open and closed string. We are left then with three different spacetime supersymmetric string theories, called \emph{type IIA}, \emph{type IIB}, and \emph{type I}, which uses unoriented worldsheets. Moreover, the ``I'' and ``II'' refers to the fact that the theory is spacetime supersymmetric with one and two supercharge generators, respectively.\\[2mm]
Type IIA and type IIB are superstring theories constructed from GSO projecting the Fock space of closed strings: in this case, the GSO projection requires a choice of chirality for the R ground state in the left and right moving sector, but by spacetime parity symmetry the theories obtained by the same choice of chirality coincide. Indeed, under exchange of left and right movers, type IIA is a non-chiral theory, while type IIB is chiral. Moreover, it's a theory of oriented closed strings.\\
The GSO projection modifies, among other things, the content of the massless R-R sector: in type IIA there will be one particle states coming from differential form of \emph{odd} degree, while in type IIB there are states associated to differential forms of \emph{even} degree. Moreover, in type IIB the 5-form fieldstrenght is required to be selfdual.\\[2mm]
Type I is a superstring theory constructed from open \emph{and} closed strings, and will be discussed further on, when we will introduce additional degrees of freedom for open strings, called ``Chan-Paton'' factors, which are essential in obtaining massless Yang-Mills states, i.e the \emph{standard model gauge interactions}.\\[2mm]
To end, we should mention the Heterotic String theory, which is a hybrid theory obtained by combining right movers of type II String theory with bosonic left movers. We will not discuss this theory in this thesis.
\begin{remark}
As pointed out at the beginning of this section, the discussion above refers to the spectrum of a string propagating in 10-dimensional Minkowski spacetime. This is one of the very few cases in which the quantization of the theory can be done ``accurately'', even if in particular gauges and with a particular choices of spacetime coordinates, and the spectrum can be found explicitly\footnote{In case of interests, the choice of a gauge does no affect the quantum theory, as gauge invariance is restored at the quantum level.}. In particular, the associated \emph{classical} fields for the various particle states can be inferred thanks to the Poincar\'{e} symmetry of Minkowski spacetime. It is generally ``assumed'' that the field content of the effective theory does not change for a more generic choice of spacetime manifold: this is a point of view we will adopt in the development of the following chapters.   
\end{remark}
 \section{Background fields and low energy limit}    
 As we mentioned in the previous section, the closed string Fock space contains states in its NS-NS sector that can be associated to a symmetric, an antisymmetric and scalar massless field.  These fields can arise as a modification of the Polyakov action, describing a string propagating in \emph{background fields}.  By background field we mean a (spacetime) field which is not affected by the presence of the propagating string, and that does not represent a dynamical variable; moreover, background fields are not integrated over in the path integral.\\
 The bosonic part of the modified action is given by
 \begin{equation}\label{backgrounds}
 \dfrac{1}{8\pi{l_{s}}^{2}}\left\{\int_{\Sigma}\mu_{\gamma}<\gamma,f^{*}g>+\int_{\Sigma}f^{*}\text{B}+\int_{\Sigma}\mu_{\gamma}\:\text{R}_{\gamma}f^{*}\Phi\right\}
 \end{equation}
 where $g$ is the metric tensor on the spacetime $\man{M}$, $\text{B}$ is locally an element in $\Omega^{2}(\man{M};\mathbb{R})$, $\Phi\in{C^{\infty}(\man{M};\mathbb{R})}$, and $\text{R}_{\gamma}$ is the Ricci scalar for the worldsheet metric. Moreover we have expressed the tension $k$ in term of the worldsheet length scale $l_{s}$.\\
Notice that the first term in the modified string action ``coincides'' with the definition (\ref{Polyakov}), in the sense that in the formulation of the string dynamics as a nonlinear sigma model we have tacitly assumed the presence of a gravitational background.\\
The second term contains the B-field, and for closed strings, i.e in the case $\partial{\Sigma}=\O$, it  is invariant under the gauge transformation $\text{B}\to\text{B}+\dd\lambda$, with $\lambda\in\Omega^{1}(\man{M};\mathbb{R})$.\\
The last term plays an important role, as for any Riemann surface $\Sigma$ of genus $g$ one has
\begin{displaymath}
\chi(\Sigma)=\int_{\Sigma}\mu_{\gamma}\:\text{R}_{\gamma}
\end{displaymath}
where $\chi(\Sigma)=2(1-g)$ is the \emph{Euler number} of $\Sigma$, a topological invariant.\\
This makes the expansion in powers of $l$, called the \emph{low energy expansion}, somehow complicated, as the term containing the dilaton is directly dependent on the loop order.\\
 Moreover, the action (\ref{backgrounds}) is invariant under ${Diff}^{+}(\Sigma)$ and ${Diff}(\man{M})$, but not under Weyl rescaling. Indeed, enforcing these (super) symmetries constrains the choice for the above background fields, providing equations they have to obey at the lowest order in $l$, which are essentially given by the vanishing of $\beta$-functions governing the Weyl scaling of the model described by (\ref{backgrounds}).\\
These equations can be obtained as Euler-Lagrange equations for the following action \cite{Polchinski,Green}
\begin{equation}
2\int_{\man{M}}\mu_{g}\:e^{-2\Phi}(\text{R}_{g}+4\langle\text{D}\Phi,\text{D}\Phi\rangle-\dfrac{1}{2}\text{H}\wedge\star{\text{H}})
\end{equation}   
where $\text{D}$ the usual covariant derivative, $\text{H}$ is the 3-form field strength for the B-field, and $\star$ denotes the Hodge operator on M. Moreover, the critical dimension d of $\man{M}$ is again 10.\\
The important aspect about the above action is that it corresponds to (part of) the bosonic part of N=2 supergravity. More precisely, the above action has to be complemented with an additional set of fields coming from the low energy approximation of String theory, but that cannot be incorporated as background fields. These are essentially the field strengths for the Ramond-Ramond fields, giving a type IIA or type IIB supergravity theory. We will discuss these additional contributions in the next chapters, as they play a prominent role in the study of Ramond-Ramond fields as a (generalized) gauge theory.\\[2mm]
We conclude this section by briefly mentioning the main differences between the perturbative analysis around general field configurations in String theory and quantum field theory.
In quantum field theory, one usually starts with a set of canonical fields $\phi$ and a classical action $\text{S}[\phi]$ which is independent of any quantum parameter, such as $\hslash$. In the functional integral formalism, perturbation theory in the loop expansion is carried out by choosing a classical field configuration $\phi_{0}$, and expanding the fields around this configuration as $\phi=\phi_{0}+\sqrt{\hslash}\varphi$. The Feynman rules are then obtained by expanding the action $S[\phi]$ in powers of $\sqrt{\hslash}$: the propagator for the field $\varphi$ and the interaction vertices will then depend upon the choice of the background field $\phi_{0}$. As $\phi_{0}$ is a classical solution, all the Feynman graphs usually referred to as \emph{tadpoles}, corresponding to the linear term in $\sqrt{\hslash}$, vanish. This works also in the other way, in the sense that the vanishing of all the tadpoles graphs allows one to infer that $\phi_{0}$ is a classical solution. \\
As stated, this procedure has no analogue in String theory, simply because there is no equivalent for the canonical fields $\phi$: indeed, the S-matrix in String theory is not obtained by an action governing string interactions, but it is defined perturbatively through a generalization of Feynman diagrams. Hence it is not possible to check directly that a given configuration for the fields $g$, $\text{B}$ and $\Phi$ (and possibly other fields) is a background configuration, as the string action does not include their equations of motion. One can then think of using the vanishing of tadpole graphs as a criterion to decide when a set of fields defines a background solution.\\
It is remarkable, then, that asking for the action (\ref{backgrounds}) to define a conformal field theory of a given central charge, which is the case in which one is assured that there exists a Hilbert space of states (no nonnegative and zero length vectors), forces the tadpole graphs (of certain fields) to vanish, giving us the equations of motion that the background fields have to satisy.      
\section{Some comments on the B-field}\label{Bfield}
As we have seen in the previous section, an additional term containing a two form $\text{B}$ can be consistently added to the (bosonic) string action. We have also mentioned that the action so extended enjoys a ``gauge symmetry'' with respect to the field B. This gauge simmetry gives the B-field a very rich mathematical interpretation: indeed, even if the B-field ``plays no active role'' for the extent of this thesis, we find it neverthless important to give some informations about its mathematical nature.\\[2mm]
We start by noticing that the description of the $\text{B}$ field as a two form on the spacetime $\man{M}$ need not hold globally, as the $\text{B}$ field is defined up to the transformation $\text{B}\to\text{B}+d\lambda$, since the action (\ref{backgrounds}) is invariant under this process. A local description, then, can be obtained as follows.\\
Let $\left\{\mathcal{U}_{\alpha}\right\}$ be a open cover of $\man{M}$ enjoying the property that $\mathcal{U}_{\alpha}$ and all its intersections are contractible sets. Such a cover is called a \emph{good cover}, and any manifold can be equipped with one.\footnote{Just consider a cover made of geodesically complete open sets for some Riemannian metric.} Moreover, denote with $\mathcal{U}_{\alpha\beta\dots\tau}$ the intersection $\mathcal{U}_{\alpha}\cap\mathcal{U}_{\beta}\cap\dots\cap\mathcal{U}_{\tau}$.\\
As in any open set $\mathcal{U}_{\alpha}$ the form $\text{B}_{\alpha}$ is only defined up to the transfomation above, we have that on any $\mathcal{U}_{\alpha\beta}$ the following equation holds
\begin{equation}
\text{B}_{\beta}-\text{B}_{\alpha}=d\lambda_{\alpha\beta}\qquad\lambda_{\alpha\beta}\in\Omega^{1}(\mathcal{U}_{\alpha\beta};\mathbb{R})
\end{equation}
These equations imply that on a triple overlap $\mathcal{U}_{\alpha\beta\gamma}$ one has
\begin{equation}\label{eq1}
d(\lambda_{\alpha\beta}+\lambda_{\beta\gamma}+\lambda_{\gamma\alpha})=0
\end{equation}
Then by Poincar\'{e}'s lemma one has that
\begin{equation}
\lambda_{\alpha\beta}+\lambda_{\beta\gamma}+\lambda_{\gamma\alpha}=df_{\alpha\beta\gamma}\qquad{f_{\alpha\beta\gamma}\in{C^{\infty}(M;\mathbb{R})}}
\end{equation} 
for some smooth real function $f_{\alpha\beta\gamma}$.\\
Finally, on a quadruple intersections $\mathcal{U}_{\alpha\beta\gamma\sigma}$
\begin{equation}\label{eq4}
f_{\beta\gamma\sigma}-f_{\alpha\gamma\sigma}+f_{\alpha\beta\sigma}-f_{\beta\gamma\alpha}=2\pi\:\omega_{\alpha\beta\gamma\sigma}
\end{equation}
for some real numbers $\omega_{\alpha\beta\gamma\sigma}$.\\
Notice, at this point, that the local field strengths defined as $\text{H}_{\alpha}:=\dd\text{B}_{\alpha}$ can be ``glued'' together to form a globally defined three form on $\man{M}$. If we require $\text{H}$ to have integral periods, then one can choose $\omega_{\alpha\beta\gamma\sigma}$ to be integer numbers. According to this description, then, a B-field configuration is defined as a triple $(\text{B}_{\alpha},\lambda_{\alpha\beta},f_{\alpha\beta\gamma})$ satisying the above equations.\\
 In a mathematical language, the triple above defines a \emph{(abelian) gerbe with connection}, which is a generalization of the idea of principal bundle with connection. As in the case of a connection on a principal bundle, the field strength $\text{H}$ is related to the cohomology of the manifold $\man{M}$: indeed, the class $[\text{H}]$ represented by $\text{H}$ in $\text{H}^{3}(\man{M};\mathbb{R})$ is the image of the class $[\omega_{\alpha\beta\gamma\sigma}]\in\check{\text{H}}^{3}(\man{M};\mathbb{R})$ under the isomorphism between DeRham cohomology and C\v{e}ch cohomology with coefficients in the sheaf of locally constant real functions.\\
For the rest of this thesis, we will always assume that the B-field is ``turned off'', meaning that we will always work under the hypothesis that $\text{B}=0$. We will mention, though, the major modifications that the presence of a nontrivial B-field induces.\\[2mm]   
Notice, at this point, that the above definition of a gerbe has been given in a particular choice of a good cover: this raises the question of the independence of this definition from any particular choice.\\
A more elegant formulation, which manifestly does not rely on any particular local choice,  will be the subject of the next chapters.  

\chapter{D-branes and Ramond-Ramond fields}
In this chapter we will introduce \emph{D-branes} and analyze the main features of open strings in the presence of such objects. We will review the generalized electromagnetic theory of Ramond-Ramond fields, and discuss their anomalous coupling with D-branes. Finally, we will discuss Sen's conjectures, and motivate the relevance of K-theory in the description of D-brane charges.  
\section{D-branes as boundary conditions}
As we have mentioned in the previous chapter, the bosonic string propagating in a flat 26-dimensional Minkowski spacetime describes a two dimensional conformal field theory, in which the spacetime cartesian coordinates are fields defined over a surface of a given genus. We have also seen that String theory can be defined, apart from instabilities and other unpleasant effects, both for closed and open strings: in the latter case, the worldsheet $\Sigma$ has a nonempty boundary $\partial\Sigma$.
Since $\Sigma$ has a boundary, we are faced with the task of specifying boundary conditions for the conformal fields defined on $\Sigma$. Moreover, one has to require that these boundary conditions preserve the conformal invariance of the worldsheet theory.\\
For a general conformal field theory, the classification of local boundary conditions preserving conformal invariance is very complicated to obtain; moreover, very often the boundary conditions have no geometrical interpretation, making it more difficult to understand their relation with the fundamental strings described in the previous chapter.\\
For this reason, we will for the moment focus on those boundary conditions that have a clear geometrical interpretation, postponing the discussion of more general cases to the next chapters.\\[2mm]
Consider the Polyakov action (\ref{Polyakov}) in the case in which $\Sigma$ is a Riemann surface of genus $g$ with boundary $\partial\Sigma$. Recall that the dynamical variables are given by maps from the worldsheet to the spacetime $\man{M}$, and the worldsheet metric $\gamma$.\\
 At this point we can consider the action restricted to a subspace of $\text{Maps}(\Sigma,\man{M})$ specified by the condition
\begin{equation}\label{brane}
f(\partial\Sigma)\subset{\man{Q}}
\end{equation}
where $\man{Q}\subset{\man{M}}$ is a submanifold of the target spacetime.\\
We have then for the moment the following geometrical working definition\\[2mm]
\noindent\textbf{Definition} \emph{A D-brane is a (physical) object whose dynamical evolution is described by a submanifold $Q$ of the spacetime
containing the woldlines of the end points of open strings.}\\[2mm]
In the above definition we are also supposing that the manifold $\man{Q}$ is chosen in such a way that the worldsheet field theory is still conformal, and are usually referred to as \emph{D-submanifold}.\\
The ``D'' in D-brane refers to the fact that some of the coordinates representatives for the map $f$ are locally subject to Dirichlet boundary conditions. More precisely, consider a set of coordinates $\{x^{\mu}\}$ in $\man{M}$ in such a way that $\man{Q}$ can locally be represented by the condition $x^{\alpha}=0, \alpha=p+1,\ldots,26$, with $p=\dim\man{Q}$, and choose coordinates $\{\sigma,\tau\}$ on $\Sigma$ such that $\tau$ is the coordinate along the boundary, for $\sigma=0,\pi$.\\
Then condition (\ref{brane}) implies that
\begin{equation}\label{Dirichlet}
x^{\mu}(\tau,\sigma)|_{\sigma=0}=x^{\mu}(\tau,\sigma)|_{\sigma=0}=0,\quad\mu=p+1,\ldots,26
\end{equation}
which are referred to as Dirichlet boundary conditions.\\
Usually, (\ref{Dirichlet}) are supplemented with local Neumann conditions for the fields representing the coordinates on $Q$. Briefly, these can be enforced by requiring that
\begin{equation}
\dd{f}|_{(\partial\Sigma)_{p}^{\perp}}=0,\qquad{p\in\partial\Sigma}
\end{equation}
where $(\partial\Sigma)_{p}^{\perp}$ denotes the space of normal vectors to the boundary in $p$. In other words, the end points are free to move on the submanifold $\man{Q}$.\\
It is common use to refer to a D-brane represented by a p+1-dimensional manifold $\man{Q}$ as \emph{Dp-brane}. More precisely, in the following we will refer to $\man{Q}$ as the \emph{worldvolume} of the Dp-brane, suggesting a conceptual difference between a D-brane and the submanifold it ``wraps'', i.e. it is represented with. This difference will play an important role in the next chapters of this thesis.\\[2mm]
We conclude this section by stressing out that the above boundary conditions do not exhaust all the possible conditions for a boundary conformal field theory, and they represent a particular well-behaved subset of boundary conditions. Neverthless, their presence introduces remarkable additional features to String theory, as we will see in the next sections.
Also, notice that D-branes can be formally added also to the theory of closed strings, since the condition ($\ref{brane}$) does not apply. In particular, D-branes do not interact if no open strings are present.
\section{The spectrum of open strings on Dp-branes}\label{openDbranes} 
As usual in String theory, open strings in the presence of a D$p$-brane can be quantized only in very specific settings. Indeed, one usually considers open strings propagating in flat $d$-dimensional Minkowski spacetime in the presence of \emph{hyperplanar} D-branes, i.e. D-branes whose worldvolume is specified by linear conditions in (spatial) cartesian coordinates. Then, the spectrum content of the effective theory and its properties are assumed to be the same for more general D-brane configurations.\\
Notice, at this point, that introducing a D-brane generally reduces the spacetime symmetries of the conformal field theory. For instance, in the flat case above, introducing a fixed $p$+1-dimensional hyperplane in Minkowski spacetime destroys the translational symmetry, and reduces the $\group{SO}(1,d)$ action to that of $\group{SO}(1,p)\times{\group{SO}}(d-p-1)$. This is analogous to the fact that introducing an electron in spacetime breaks its translational symmetry, and is not surprising. The analogy is even deeper, in view of the fact that D-branes can arise as solitonic solutions of a low energy effective theory.\\
 Consider then a D$p$-brane in $d$-dimensional Minkowski spacetime represented by the hyperplane $x^{\mu}=\bar{x}^{\mu},\:\mu=p+1,\ldots,d$, and where $x^{0}$ is the time coordinate. In this situation, the Fock space of the quantized open string contains, as in the case without any D-brane, states that can be recognised as one particle states for a quantized classical field. Moreover, these classical fields are constrained to ``live'' on the worldvolume representing the D$p$-brane, in the sense that they can be naturally interpreted as objects defined only on the brane itself.\footnote{ This is due essentially to the fact that they are states generated from the ground state by the action of creation operators associated to the momenta corresponding to the coordinate describing points on the brane.} Notice, however, that statements about the support of these fields are usually dependent on the quantization procedure, and different procedures can yield in principle different answers.\\
In the following, we will just list the first few states arising from the ground state, referring the reader to \cite{Zwiebach} for more details.\\
As in the case of a bosonic open string without the presence of a D-brane, the ground state contains tachyonic states, which arise by a scalar field on the brane, and constitute a section of the normal bundle.\\
The next states are massless states, arising by a vector field on the D-brane, with the degrees of freedom of a gauge field, i.e. they can be described by a 1-form A, which can be shifted by an exact 1-form. This is usually interpreted as a Maxwell field defined on the worldvolume of the D$p$-brane. Actually, this field has a deeper geometrical and topological description, as we will see in the next chapters.\\
Finally, there are massless states for each normal direction to the D$p$-brane, coming from scalar fields defined on the brane itself. These are usually interpreted as the fields generating the excitations of the D$p$-brane, describing fluctuations in the D$p$-brane position in spacetime.\\[2mm]
In the case of general spacetimes and brane configurations, we will assume that the picture described above is still valid locally for any given neighborhood of a point in the D$p$-brane worldvolume.
\section{Chan-Paton factors and Adjoint bundles}\label{chanpaton}
In this section we will consider how the description in the previous section is modified when we allow for the presence of $n$ D$p$-branes represented by the \emph{same} worldvolume $\man{Q}$. This is indeed possible, due to the fact D-branes represent boundary conditions for the worldsheet conformal field theory, and hence they are distinguished by the open string boundary. In this case, the $n$ D-branes are said to be \emph{coincident}.\\
Then, any quantum state of an open string constrained on the worldvolume of a set of $n$ coincident D$p$-branes can be labelled by $|k;ij>$, where $k$ is a given collection of quantum numbers, and $i,j=1,\ldots,n$ label the branes which the end points of the strings are constrained to. Allowing a general superposition of this states, one has that a general state assumes the form $|k;a>=\sum_{i,j}\lambda^{a}_{ij}|k;ij>$: the coefficients $\lambda^{a}_{ij}$ are called the \emph{Chan-Paton factors} of the open strings.\\
For physical reasons (\cite{Polchinski},\cite{quantumath}), the Chan-Paton factors $\lambda^{a}=(\lambda_{ij}^{a})$ are constrained to be antihermitian, i.e. $\lambda^{a\dagger}=-\lambda^{a}$; moreover, they are conserved in interactions. Hence, the Fock space for an open string in the presence of $n$ coincident D$p$-branes is given by
\begin{displaymath}
\mathcal{F}^{\text{D}}_{\text{open}}\otimes\mathfrak{u}(n)
\end{displaymath}
where $\mathcal{F}^{\text{D}}$ is the space of an open string in the presence of a single D$p$-brane, and $\mathfrak{u}(n)$ is the Lie algebra of $\text{U}(n)$. Consequently, all the massless fields described in the previous section will take values, at least locally on the brane, in $\mathfrak{u}(n)$.\\
As the Chan-Paton factors are conserved in interactions, their contribution in the interaction amplitude at any level is of the form of a trace of matrix products. As the trace of a matrix is preserved under the adjoint action of the group of invertible matrices, the transformation $\lambda^{a}\to{g\lambda^{a}g^{\dagger}}$, with $g\in\group{U}(n)$ is a symmetry of the system.\footnote{$\group{U}(n)$ is singled out because it preserves the antihermitian property of $\lambda^{a}$.} Hence the Lie algebra $\mathfrak{u}(n)$ can be seen as the adjoint representation for the matrix group $\group{U}(n)$.\\[2mm]
Consider now a set of $n$ D$p$-branes wrapping a submanifold $\man{Q}$ of the spacetime, and let $\left\{\mathcal{U}_{\alpha}\right\}$ be a good cover for $\man{Q}$. Recall that at low energy the spectrum produces a $\mathfrak{u}(n)$-valued 1-form $\text{A}_{\alpha}$, for any $\alpha$. Consistently with the symmetry of Chan-Paton factors and the properties of $\text{A}_{\alpha}$, on a double intersection $\mathcal{U}_{\alpha\beta}$ we \emph{can} have
\begin{equation}\label{connection}
\text{A}_{\beta}=g_{\alpha\beta}(p)\text{A}_{\alpha}g^{-1}_{\alpha\beta}(p)+g_{\alpha\beta}^{-1}(p)\dd{g_{\alpha\beta}}(p)
\end{equation}
where $g_{\alpha\beta}:\mathcal{U}_{\alpha\beta}\to\group{U}(n)$ is a smooth function. At the same time, other fields will obey similar properties. For instance, the $\mathfrak{u}(n)$-valued scalar field $\text{T}$ describing the tachyon will satisfy on $\mathcal{U}_{\alpha\beta}$
\begin{equation}
\text{T}_{\beta}=g_{\alpha\beta}(p)\text{T}_{\alpha}g^{-1}_{\alpha\beta}(p)
\end{equation}
The function $g_{\alpha\beta}$ must be the same as it arises from a symmetry of the transition amplitude.\\
Consistency on triple overlaps requires that
\begin{equation}
\text{Ad}(g_{\alpha\beta})\text{Ad}(g_{\beta\gamma})\text{Ad}(g_{\gamma\alpha})=1
\end{equation}
The above condition implies that
\begin{equation}\label{cocycle}
g_{\alpha\beta}g_{\beta\gamma}g_{\gamma\alpha}=\omega_{\alpha\beta\gamma}1
\end{equation}
where $\omega_{\alpha\beta\gamma}\in\group{U}(1)$, i.e. in the kernel of the adjoint representation.
At this point one can ask if it is possible to redefine the functions $g_{\alpha\beta}$ in such a way that $\omega_{\alpha\beta\gamma}=1$, i.e. if the vector bundle $\man{E}^{adj}$, of which the tachyon field is a section, is the adjoint bundle associated to a principle bundle $\man{P}$ over $\man{Q}$ defined by the transition functions $g_{\alpha\beta}$. In this case, the $\mathfrak{u}(n)$-valued forms $\text{A}_{\alpha}$ can be seen as the local representatives of a connection defined on the principal bundle P, hence as a Yang-Mills field over $\man{Q}$.\\
This ``lifting'' process is not always possible for general vector bundles: anyway, we will always assume that this is possible. With this we mean that choosing functions $g_{\alpha\beta}$ behaving as transition functions of a principal bundle is \emph{consistent} with the path integral for open strings in the presence of D-branes being well defined and gauge invariant. In this case, the path integral is supplemented with a factor
\begin{equation}
\text{hol}_{\partial\Sigma}(f^{*}\text{A})
\end{equation}
which only makes sense when $\text{A}$ is a connection on some principal bundle.\\
In the case in which a topologically nontrivial B-field is present, though, the equation (\ref{connection}) has to be modified in order to take into account the gauge transformations for the B-field discussed in section \ref{Bfield}: this, in turn, forces a redefinition of the function $g_{\alpha\beta}$, which will now obey condition (\ref{cocycle}), with the cocycle $\omega_{\alpha\beta\gamma}$ directly determined by the topological properties of the B-field. See \cite{Kapustin} for more details.\\[2mm]
Hence, in absence of a B-field, the bundle $\man{E}^{adj}$ over $\man{Q}$ is chosen to be decomposable as $$\man{E}^{adj}~\simeq~\man{E}~\otimes~\bar{\man{E}}$$ where $\man{E}$ is a $\group{U}(n)$-vector bundle, and $\bar{\man{E}}$ is its complex conjugate. It is customary to refer to the vector bundle $\man{E}$ as the \emph{Chan-Paton bundle}, as from this last we can obtain the bundle of Chan-Paton factors $\man{E}^{adj}$.\\
We can then summarize the crucial aspects of this section by the following\\[2mm]
\noindent\textbf{Fact} \emph{In the absence of a nontrivial B-field and in the low energy approximation, a set of n coincident D-branes with worldvolume} $\man{Q}$ \emph{gives rise to a} $\group{U}(n)$-\emph{vector bundle} $\text{E}\to\man{Q}$ \emph{equipped with a linear connection.}\\[2mm] 
Notice that the topology of the vector bundle $\man{E}$ is not determined by the string dynamics, and has to be specified \emph{a priori} when introducing a set of D-branes. Also for a single D-brane, we have to assign a line bundle over its worldvolume: in the case of hyperplanar D-branes discussed in the previuos section, this is hidden by the fact that any line bundle over such a hypersurface is topologically trivial.
\newpage
\section{D-branes and Supersymmetry}\label{super}
In this section we will consider the effect of the presence of a D-brane in superstring theories; we will focus, in particular, on the case of type IIA and IIB superstring theories.\\
Recall that supestring theories, in the RNS formalism, are obtained by introducing degrees of freedom on the worldsheet $\Sigma$ which are sections of $\text{S}(\man{T}\Sigma)\otimes{f^{*}\man{TM}}$, where $\text{S}(\man{T}\Sigma)$ is the spinor bundle associated to $\man{T}\Sigma$, and that the GSO projection ensures spacetime supersymmetry. For this reason, we will require $\man{M}$ to be a spin manifold, where $\man{M}$ is a 10-dimensional euclidean spacetime.\\
We have seen in the previous sections that the defining property of a D$p$-brane with worldvolume $\man{Q}$ is that of constraining the maps $f$ to satisfy $f(\partial\Sigma)\subset\man{Q}$, and we have argued that this requirement modifies the spectrum of the open bosonic string. Indeed, the presence of a D$p$-brane modifies also the fermionic sector for the open supersymmetric string \cite{Szabobook,Polchinski}.\\
Notice, first, that for open strings whose end points are constrained to the submanifold $\man{Q}$, one has $f^{*}\man{TM}\simeq{f^{*}\man{TM}_{|\man{Q}}}$, where $\man{TM}_{|\man{Q}}$ is the restriction of $\man{TM}$ to $\man{Q}$. Moreover, as $\man{Q}$ is a submanifold of $\man{M}$, the following exact sequence holds
\begin{equation}
0\to\man{TQ}\to\man{TM}_{|\man{Q}}\to{\nu_{\rm Q}}\to{0}
\end{equation}
and with a choice of a metric on $\man{TM}$ it splits orthogonally as $\man{TM}_{|\man{Q}}\simeq\man{TQ}\oplus{\nu_{\rm Q}}$, where $\nu_{\rm Q}$ is the normal bundle to $\man{Q}$ in $\man{M}$. Hence, the ground state in the Ramond sector will give rise to a section of
\begin{equation}\label{worldspinor}
\text{S}(\man{TM}_{|\man{Q}})\simeq{\text{S}(\man{TQ})\otimes\text{S}(\nu_{\rm Q})}
\end{equation}
i.e. a worldvolume spinor charged under an internal $\group{SO}(9-p)$ symmetry, the structure group of $\nu_{\rm Q}$\footnote{Actually, the decomposition (\ref{worldspinor}) depends on the parity of the codimension of Q in M, since in general ${\rm C}\ell({\rm V\oplus W})\simeq{\rm C}\ell({\rm V})\hat{\otimes}{\rm C}\ell({\rm W})$, where $\hat{\otimes}$ denotes the $\mathbb{Z}_{2}$-graded product.}.\\
To be more precise, notice that in general neither $\text{S}(\man{TQ})$ nor $\text{S}({\nu_{\rm Q}})$ are \emph{well defined} as vector bundles, even if their tensor product is. Anyway, we will for the moment require that the normal bundle $\nu_{\rm Q}$ admits a spin structure: together with the requirement that $\man{M}$ is a spin manifold, one has that $\man{TQ}$ also admits a spin structure, i.e. $\man{Q}$ is a spin manifold.\\ 
We refer the reader to Appendix B for more details on spin manifolds and Clifford algebras.\\[2mm] 
Because of the properties of a system of $n$ coincident D-branes discussed in section \ref{chanpaton}, the worldvolume spinors obtained in the low energy approximation are sections of $\text{S}(\man{TQ})\otimes\text{S}(\nu_{\rm Q})\otimes\text{E}^{adj}$, and hence they are charged under the gauge field localized on the worldvolume. Then, it is natural to think that, in a certain sense, D$p$-branes in the low energy approximation are described by the $p$+1-dimensional gauge theories localized on their worldvolume.\\[2mm]
As we have mentioned in section \ref{openDbranes}, the introduction of a D-brane can reduce the spacetime symmetries of the vacuum configuration, and hence affect the properties of the open and closed string spectrum. As should be expected, a D-brane affects also the amount of supersymmetry present in the theory.  Indeed, the introduction of a D-brane in an interacting String theory requires the presence of open string states in the spectrum: this is due to the fact that a closed string, though not affected by the boundary conditions, can \emph{break} on the worldvolume of the D-brane into an open string, and open strings have at most N=1 supersymmetry \cite{Polchinski,becker,Szabobook}.\\
A D-brane is said to be \emph{supersymmetric} if the spectrum of open and closed superstrings in its presence is supersymmetric. Moreover, a D-brane is said to be \emph{stable} if the spectrum does not contain a tachyonic state. In this sense, all D-branes in bosonic String theory are unstable. Supersymmetry and stability of a D-brane are very difficult features to determine for general brane configurations and in nontrivial backgrounds, but are well understood in the case of hyperplanar D-branes in Minkowski spacetime.\\
Indeed, one can show that flat D$p$-branes in type IIA on a Minkowski spacetime are supersymmetric for $p$ even, while $p$ must be odd in type IIB. This can be inferred by analyzing the conserved supersymmetry in each case, and by using \emph{T-duality}\cite{Polchinski,Szabobook,becker}. As usual, we will assume these properties of D-branes to hold for a more general type II configurations.\\[2mm]
As suggested by the terminology used above, a D-brane appears not only to be a geometrical region enforcing the boundary conditions for the worldsheet conformal field theory, but also to enjoy particle-like properties, like mass, charge, etc., detected through the behavior of the strings propagating around it. Even if a full theory of quantum D-branes, which would be the right framework to describe quantum properties like decay, etc., has not yet been developed, the superstrings in the presence of a D-brane give useful information about the dynamics of these extended objects.
\begin{remark}
The definition of supersymmetric and stable D-brane given above are sensible in the low energy approximation, and in the case in which we are considering a D-brane that is \emph{fixed}, i.e. we are neglecting its dynamics: indeed, to take into account that a D-brane is a dynamical object, and to describe some aspects of its quantum behavior, a formalism know as the \emph{boundary state formalism} is more suitable. In this formalism, the boundary conditions characterizing the D-brane are imposed \emph{after} the theory of closed superstrings has been quantized: this suggests to identify states of a \emph{quantum} D-brane as coherent states in the Fock space of the closed superstring, making it  more precise the investigation of stability, supersymmetry, etc. of these states. Moreover, it makes it possible to introduce D-branes even when the boundary condition it represents has no clear geometrical description.\\
As in this thesis we will be concerned only with a semiclassical and geometrical description of D-branes, we will not invoke this formalism, directing the reader to \cite{clifford,Polchinski}, and references therein, for a complete review of these topics.
\end{remark}
\section{D-branes and Ramond-Ramond charges}\label{generalized} 
As we have mentioned in section \ref{quantum}, the low energy approximation of type IIA and type IIB superstring theory contains particle states described by a gauge theory of differential forms $\ram{C}^{(p)}$ defined on a $d$-dimensional spacetime manifold $\man{M}$, with $p$ odd in type IIA and $p$ even in type IIB.\\
Moreover, if we denote by $\ram{F}^{p}=\dd\ram{C}^{(p)}$ the Ramond-Ramond field strength, the GSO projection imposes the constraint $\star\ram{F}^{p}=\ram{F}^{d-p}$, representing the fact that not all the forms $\ram{C}^{(p)}$ are \emph{independent} degrees of freedom \cite{Polchinski}.\\
It follows that the Ramond-Ramond field strengths satisfy the linearized equations
\begin{displaymath}
\begin{array}{c}
\dd\star{\ram{F}^{(p)}}=0\\
\dd\ram{F}^{(p)}=0
\end{array}
\end{displaymath}
These equations are a generalization of the Maxwell equations in 4 dimensions.\\
As the field theory they describe plays a prominent role in this thesis, we will digress slightly to review some of its aspects, following \cite{currents,Freed2000a,Freed2002}.
\subsection{Generalized electromagnetism and sources}\label{genmax}
In analogy with Maxwell electromagnetism in 4 dimensions, consider a theory of differential forms $\ram{A}\in\Omega^{n}(\man{M};\mathbb{R})$, where $\man{M}=\mathbb{R}\times\man{Y}$ is a $d$-dimensional oriented Lorentzian manifold with $\man{Y}$ the spatial slice, and whose equations of motion are given by
\begin{displaymath}
\begin{array}{c}
\dd\star\ram{G}=0\\
\dd\ram{G}=0
\end{array}
\end{displaymath}
with $\ram{G}=\dd\ram{A}$.\\
 Such a theory is usually referred to as \emph{generalized electromagnetism}, as the Maxwell equations (in vacuum) are obtained for $n$=1 and $d$=4. The equations of motion are invariant under the transformation $\ram{A}\to{\ram{A}+\omega}$, with $\omega$ a closed $n$-form, describing an abelian gauge theory of $n$-forms. It is natural, then, to introduce an \emph{electric} source for the field $\ram{A}$ in complete analogy with the electromagnetic case, i.e. by modifying the equations of motion as
\begin{equation}\label{genelec}
\begin{array}{c}
\dd\star\ram{G}=j_{e}\\
\dd\ram{G}=0
\end{array}
\end{equation}
where $j_{e}\in\Omega_{s}^{d-n}(\man{M};\mathbb{R})$ is a ($d-n$)-form with compact support on the spatial slice, called the \emph{electric current distribution}.\\
The equations of motion imply that $\dd{j_{e}}=0$, hence $j_{e}$ represents a class 
\begin{displaymath}
[j_{e}]\in\text{H}_{s}^{d-n}(\man{M};\mathbb{R})
\end{displaymath}
in the cohomology of $\topsp{M}$ with real coefficients, and compact support on $\man{Y}$.\\
Notice that the equations of motion \emph{do not} imply that the class $[j_{e}]$ is vanishing, as $\star{\ram{G}}$
 is not required to have compact support on the spatial slice.\\ 
More precisely, if we denote with $i_{t}:\man{Y}\to\man{M}$ the map $\man{Y}\to\{t\}\times{\man{Y}}\subset\man{M}$, the class
\begin{equation}
Q_{e}:=[i_{t}^{*}j_{e}]\in\text{H}_{\text{cpt}}^{d-n}(\man{Y};\mathbb{R})
\end{equation}
is in general non vanishing, and it is called the \emph{total charge} of the electric current distribution. Indeed, the fact that $j_{e}$ is a closed form implies\footnote{This is also due to the particular choice made for $\man{M}$.} that the class $Q_{e}$ does not depend upon the choice of $i_{t}$, hence it is a conserved quantity for the equations of motion.\\
 This cohomological interpretation of the electric charge may sound ``exotic'' at first: but notice that in the ordinary electromagnetism case, with $n$=1, $d$=4, and $\man{Y}=\mathbb{R}^{3}$, one has $\text{H}_{\text{cpt}}^{3}(\mathbb{R}^{3};\mathbb{R})\simeq\mathbb{R}$, and the total electric charge is given by a real number, as usual.\\[2mm]
Equations (\ref{genelec}) can be obtained from the functional
\begin{equation}\label{action}
\text{S}[\ram{A}]:=-\dfrac{1}{2}\int_{\man{M}}\ram{G}\wedge\star\ram{G}-\dfrac{1}{2}\int_{\man{M}}{j_{e}}\wedge\ram{A}
\end{equation}
via a variational principle. This functional is \emph{not} gauge invariant, but one can show that the equations of motion obtained do not depend on the particular choice of a gauge. Moreover, also the coupling term in (\ref{action}) can be given a cohomological interpretation. As $j_{e}$ is a closed form with compact support on $\man{M}$, it will induce a homomorphism
\begin{displaymath}
\psi_{j_{e}}:\text{H}^{n}(\man{M};\mathbb{R})\to\mathbb{R}
\end{displaymath}
defined as
\begin{displaymath}
\psi_{j_{e}}([a]):=<[j_{e}]\cup[a],[\man{M}]>=\int_{\man{M}}j_{e}\wedge{a}
\end{displaymath}
for any \emph{closed} $n$-form $a$.\\
As in cohomology with real coefficients we have that $\text{Hom}(\text{H}^{n}(\man{M};\mathbb{R}))\simeq{\text{H}_{n}(\man{M};\mathbb{R})}$, one has that there exists a class in $[\topsp{Q}]\in\text{H}_{n}(\man{M};\mathbb{R})$ represented by a compact oriented submanifold $\topsp{Q}$ such that
\begin{displaymath}
\psi_{j_{e}}([a])=<[a],[\topsp{Q}]>=\int_{\topsp{Q}}a
\end{displaymath}
The class $[\topsp{Q}]$ is called the \emph{Poincar\'{e} dual} of $[j_{e}]$.\\
Indeed, using this class we could rewrite the coupling term in (\ref{action}) as
\begin{equation}\label{coupling}
\int_{M}j_{e}\wedge{\ram{A}}=\int_{\topsp{Q}}\ram{A}|_{\topsp{Q}}
\end{equation}
which gives rise to the straightforward interpretation of $\topsp{Q}$ as the worldvolume of an extended object acting as a source for the field $\ram{A}$, with coupling given by (\ref{coupling}). Moreover, in the usual physical case in which $\topsp{Q}\simeq\mathbb{R}\times{\widetilde{\topsp{Q}}}$, the total charge of the distribution is represented by $[\widetilde{\topsp{Q}}]$.\\
Unfortunately, as it is stated, equation (\ref{coupling}) is not quite true, as $\ram{A}$ is in general \emph{not} closed.\\ 
The mathematical solution to this problem is to regard $j_{e}$ as a \emph{current} \cite{currents,griffith}, i.e. as a distribution valued differential form, with support on $\topsp{Q}$. This indeed allows equation (\ref{coupling}) to hold for a non-closed differential form $\man{A}$; notice that this is perfectly analogous to the electromagnetic case, where an electron current is described with a Dirac delta distribution having support on the worldline of the electron itself.\\[2mm]
To summarize, in this section we have argued that a generalized abelian gauge theory of $n$-forms admits extended objects as electric sources, whose worldvolume $\topsp{Q}$ are submanifolds of the spacetime, and whose charges can be identified with the homology classes represented by $\topsp{Q}$ in de Rham homology.
\subsection{Ramond-Ramond charges and anomalies}
As we have seen in the previous section, a generalized electromagnetic theory naturally admits extended objects as its electric sources. Since the dynamics of Ramond-Ramond fields are governed by such a gauge theory, it is natural to ask for their sources. Recall that Ramond-Ramond fields are given by  $p$-forms, with $p$ odd  in type IIA String theory, and $p$ even in type IIB: sources for Ramond-Ramond fields in type IIA are necessarily odd-dimensional submanifolds, while in type IIB they need to be even-dimensional. Hence, the natural candidates for the role of sources are the supersymmetric stable D-branes of type II String theory: indeed, we have seen in section \ref{super} that they occur with the right dimension, i.e. as odd-dimensional in type IIA, and as even-dimensional in type IIB. This would indeed take into account the stability of D-branes of certain dimensions.\\
As it is stated here, this could be merely a coincidence, in the sense that the only property of having the right dimension does not necessarily identify D-branes with the sources of Ramond-Ramond fields.\\
On the contrary, in the seminal paper \cite{PolDbrane} strong support was given to the fact that D-branes \emph{do} interact with Ramond-Ramond fields: this was achieved by scattering closed strings with D-branes, and showing that the interaction amplitude contains a term that can be manipulated to a coupling of the form (\ref{coupling}). Moreover, it was shown that D-branes coupling to Ramond-Ramond fields enjoy the \emph{BPS} property: their mass, that in the case of D-branes is called the \emph{tension}, coincides with their charge. D-branes in type II with the ``wrong'' dimension do not couple to Ramond-Ramond fields, and they are unstable: they are called \emph{non-BPS} D-branes.\\
Hence, let $\man{Q}\subset{\man{M}}$ be a Dp-brane: the term
\begin{equation}\label{ramcoup}
\int_{\man{Q}}\ram{C}^{(p+1)}|_{\man{Q}}
\end{equation}
is called the \emph{Ramond coupling}, and the class $[\man{Q}]$ is called the \emph{Ramond charge} of the Dp-brane $\man{Q}$.\\
It turns out \cite{Minasian1997,GHM}, though, that (\ref{ramcoup}) is not yet the right coupling term, and has to be modified in order to take into account some peculiar properties of D-branes.\\[2mm]
\indent Before explaining the reason for this modification, we add a comment on the coupling (\ref{ramcoup}). Indeed, in the above discussion we have willingly neglected the fact that the different Ramond-Ramond fields are dependent upon the relation $\star\ram{F}^{(p)}=\ram{F}^{(d-p)}$. This means that any electric source for the field $\ram{C}^{(p)}$ is a magnetic source for $\ram{C}^{(d-p)}$, which requires a shift in the meaning of $\ram{C}^{(p)}$ themselves, for the appropriate values of $p$. More precisely, in the presence of a Dp-brane, $\ram{C}^{(d-p-1)}$ cannot be a globally defined differential $d-p-1$-form, as the Bianchi identities for its field strength do not hold any longer. Moreover, this poses a serious threat to the quantization of the Ramond-Ramond abelian gauge theory.
For the moment, we will continue to neglect the problems posed by the duality constraints.\\[2mm]
Recall from section \ref{super} that in the low energy approximation the worldvolume of a set of $n$ D$p$-branes carries a dimensionally reduced Yang-Mills theory coupled with a spinor field charged under some additional internal symmetries; moreover, if the D$p$-brane is supersymmetric, the worldvolume gauge theory is an N=1 super Yang-Mills theory, and the tachyon field is absent.\\
A consistency requirement on the gauge theory supported by the D-brane is that the theory does not present any ``quantum'' anomaly, which would spoil its gauge invariance after quantization. Recall, indeed, that this gauge theory is obtained by considering only the massless states in the open string spectrum in the presence of a D-brane, which should represent one particle states for the full quantized gauge theory: an anomaly would render this an inconsistent procedure.\\
 It turns out that the gauge theory on the D-brane suffers from two anomalies. The \emph{abelian} anomaly is related to the fact that the spinor fields on the worldvolume are charged under an internal symmetry, represented by the tensor product $\text{S}(\nu_{\rm Q})\otimes\text{E}^{adj}$, where Q is the D-brane worldvolume. More precisely, in the type IIB case, where $\man{Q}$ is an even-dimensional (spin) manifold, the dynamics for the spinor fields is governed by the Dirac operator
\begin{equation}\label{Diracred}
\Dslash:\text{S}^{+}(\man{TM}|_{\man{Q}})\otimes\text{E}^{adj}\to\text{S}^{-}(\man{TM}|_{\man{Q}})\otimes\text{E}^{adj}
\end{equation}
where the $\pm$ splitting is given by the chiral decomposition, since $\man{M}$ even-dimensional.\\
From equation (\ref{worldspinor}) one can see that
\begin{displaymath}
\text{S}^{\pm}(\man{TM}|_{\man{Q}})\simeq\left[\left(\text{S}^{+}(\man{TQ})\otimes\text{S}^{\pm}(\nu_{\rm Q})\right)\oplus\left(\text{S}^{-}(\man{TQ})\otimes\text{S}^{\mp}(\nu_{\rm Q}\right)\right]
\end{displaymath}
Hence, from the worldvolume perspective the spinor fields do not have a definite chirality: this, in general, leads to the so called ``chiral'' anomaly, and is measured by the index of the Dirac operator in (\ref{Diracred}).\\ 
The other possible source of anomaly is given by the intersection of two D-branes, usually referred to as \emph{I-brane}. The spectrum of the open string in the presence of an I-brane is different from the one obtained in the presence of a D-brane wrapping the same worldvolume \cite{Polchinski}. In particular, the fermion sector on the I-brane develops itself an anomaly, depending on the dimension of the two intersecting branes.\\
Both the anomalies need to be cancelled: the mechanism used is known as the \emph{inflow mechanism}, and essentially consists in modifying the coupling (\ref{ramcoup}) in such a way that its variation under gauge transformations cancel the quantum anomalies above. We refer the reader to \cite{GHM,currents,Szabobook} for a detailed exposition, as the computations involved are rather lengthy, and the techniques used therein will not play any relevant role for the rest of this thesis.\\
The coupling (\ref{ramcoup}) is modified as
\begin{equation}\label{anomalous}
\int_{\man{Q}}i^{*}\text{C}\wedge\text{ch}(\text{E}){i^{*}\sqrt{\rgenus{\man{TM}}}}\dfrac{1}{\rgenus{\man{TQ}}}
\end{equation}
where $i:\man{Q}\to\man{M}$ is the embedding map, $\text{ch}$ denotes the (total) Chern character, and $\hat{\mathcal{A}}$ denotes the A-roof genus\footnote{More precisely, these are the Chern-Weil forms representing  such characteristic classes.}. See {\appcharac} for details on characteristic classes of vector bundles.\\
Moreover, in the coupling (\ref{anomalous}), $\ram{C}$ is the \emph{total} Ramond-Ramond field, defined as the formal sum $\ram{C}:=\ram{C}^{(i)}+\ram{C}^{(i+2)}+\cdots$, with $i=0,1$ in type IIB, type IIA, respectively, and the integration is understood to be 0 when the degrees do not match.\\
The coupling term (\ref{ramcoup}) is referred to as the Wess-Zumino, or Chern-Simons term for D-branes: we will refer to it as the \emph{anomalous coupling}.\\
The anomalous coupling forces us to rethink the charge classification of D-branes: indeed, the current generated by such extended objects is no longer given by the Poincar\'{e} class dual to the homology cycle of their worldvolume, but needs to be corrected according to the equations of motion for the (total) Ramond-Ramond field induced by (\ref{anomalous}). We will refer to the charge of this current as the \emph{D-brane charge}. The modern interpretation for the charge of a D-brane was first exploited in \cite{Minasian1997}, where it was proposed that the correct mathematical tool to represent D-brane charges is not de Rham or singular cohomology, but K-theory, which represents, together with some of its ``flavours'', the main mathematical subject of this thesis.\\[2mm]     
\indent We conclude this section by noticing that the anomalous coupling (\ref{anomalous}) suffers from the same, and actually worse, problems of the coupling (\ref{ramcoup}). Indeed, the anomalous coupling requires that \emph{all} the Ramond-Ramond differential forms should be defined on the worldvolume $\man{Q}$ of the D-brane. This is not the case when we take into account the fact that D-branes are electric \emph{and} magnetic sources: indeed, the usual description for a gauge field $\ram{A}$ in the presence of a magnetic distribution prescribes that $\ram{A}$ should be defined on the \emph{complement} of the distribution's support, which renders the expression (\ref{anomalous}) ill defined \cite{Freed2000a}. Again, its correct description requires a more powerful formalism: we will address a possible solution for this problem in later chapters.
    
\section{D-brane decay and Sen's conjectures}\label{senconjecture}
The fact that Dp-branes are charged under Ramond-Ramond fields allows the introduction of \emph{anti-D-branes}. Of course, such objects should arise in a proper quantum description of D-branes, which we currently lack. In any case, one can give the following semiclassical definition:\\[2mm]
\noindent \textbf{Definition}\emph{ An \emph{anti-D-brane} for a charged Dp-brane} $\man{Q}$ \emph{supporting a Chan-Paton line bundle} $\man{E}$ \emph{is given by a submanifold} $\man{Q}'$ \emph{carrying the opposite Dp-brane charge, and supporting a line bundle} $\man{E}'$ \emph{which is topologically equivalent to} $\man{E}$.\\[2mm]
As D-brane charges are conserved during the dynamics, charged D$p$-branes are stable, and hence cannot decay. Again, the description of the decay process would require a proper quantum theory of D-branes: in this context, we will say that a D$p$-brane has \emph{decayed} if the theory of open and closed strings in the presence of the D$p$-brane is \emph{equivalent}, in some sense, to a theory of only closed strings in the absence of the D$p$-brane itself. This is believed to be described at low energy by the dynamics of the tachyon field living on the worldvolume of the D$p$-brane, which mimics the ``slow rolling'' dynamics of the Higgs field in the standard model of elementary particles. In particular, the decay process of the D$p$-brane ends when the tachyon field reaches a stable minimum of the potential modelling its dynamics \cite{Zwiebach}. Even if at the present stage such a process cannot be described in a satisfactory way, it has greatly contributed to a more basic understanding of D-branes and their charges. More precisely, it is at the base of the so called \emph{Sen's conjectures}.\\
Recall that in the bosonic String theory all D-branes are unstable, as the open string spectrum contains a tachyon. In particular, one has a spacetime filling D25-brane which is unstable.  In \cite{Sen1998rg,Sen1998ii} Sen stated the following\\[2mm] 
\textbf{Conjectures(Sen)} \emph{ The open bosonic String theory in the presence of a spacetime filling D25-brane is such that
\begin{itemize}
\item[a)] The tachyon field potential has a stable local minimum. Moreover, the energy density of this minimum as measured with respect to that of the initial unstable point is equal with opposite sign to the tension of the D25-brane;\\[-8mm]
\item[b)] Lower dimensional Dp-branes can be obtained as solitonic solutions of the field theory living on the D25-brane worldvolume;\\[-8mm]
\item[c)] The stable minimum of the tachyon potential corresponds to the closed string vacuum with no open string excitations  
\end{itemize} }
The above conjecture can be stated also in the case of a lower dimensional unstable D$p$-brane: the advantage of using a spacetime filling D-brane consists in the fact that the spacetime symmetries are not broken.\\
Moreover, the conjecture can be adapted to superstring theory on a 10-dimensional spacetime, but some modifications are needed. For instance, in type IIB String theory the spacetime filling D9-brane is stable, and hence does not obey Sen's conjecture. In any case, an instability appears in a system of a coincident D$p$-brane and its anti-D-brane, denoted as D$\bar{p}$-brane, wrapped on a submanifold $\man{Q}$: at an intuitive level, the system has no conserved D-brane charge, and hence it should be able to annihilate by Sen's conjectures. This is supported by the fact that the spectrum of the open strings ``stretching'' between the D$p$-brane and the D$\bar{p}$-brane contains a tachyonic state, which is not projected out by the GSO projection. This extends to a system of $n$ D$p$-branes and $n$ D$\bar{p}$-branes wrapping the same submanifold of the spacetime.\\
Notice that a system of $n$ D$p$-branes with Chan-Paton bundle E and $m$ anti-D-branes with Chan-Paton bundle F wrapping the submanifold $\man{Q}$ has Chan-Paton bundle $\text{E}\oplus\text{F}$, as the D-branes and the anti-D-branes are distinguished by the open string endpoints. In particular, at low energy the surviving open string tachyon field is described by a section of $\text{E}\otimes\text{F}^{*}$, where $\text{F}^{*}$ denotes the dual bundle.\\[2mm]
\indent Despite the resemblances, there is a major difference between an unstable bosonic D$p$-brane and a brane-antibrane system: indeed, while by Sen's conjectures \emph{all} the bosonic D$p$-branes will eventually decay to the vacuum state, a generic brane-antibrane system can decay to a stable D$p$-brane, if the starting configuration has a net D-brane charge different from ``zero''. In particular, in \cite{sen-1998-9809} Sen was able to construct a D$p$-brane as a decay product of a system of a D$p$+2-brane and a coincindent D$\overline{p+2}$-brane, identifying the worldvolume wrapped by the D$p$-brane as a ``vortex'' for the tachyon field living on the brane-antibrane worldvolume. Indeed, recall that in this case the tachyon field $\ram{T}$ is a section of the complex line bundle $\text{E}\otimes{\text{F}^{*}}$ defined on the brane-antibrane worldvolume: in Sen's construction, the D$p$-brane is identifyed with the submanifold representing the Poincar\'{e} dual class to the (first) Chern class of $\text{E}\otimes{\text{F}^{*}}$, the zero loci for the section $\ram{T}$.\footnote{We are supposing that the Chern class of $\text{E}\otimes{\text{F}^{*}}$ is \emph{nontorsion} and that $\ram{T}$ is a transversal section.}\\[2mm] 
A somewhat simple observation was used by Witten in the seminal paper \cite{Witten1998} to give an elegant mathematical description of the various D-brane configurations. More precisely, consider in type IIB String theory a system of $n$ spacetime filling D9-branes with Chan-Paton bundle E, and a system of $n$ D$\bar{\text{9}}$-branes carrying a Chan-Paton bundle F;\footnote{The number of branes and antibranes needs to be the same in type IIB to cancel the Ramond-Ramond tadpole.} label such a configuration by the pair (E,F). As the process of brane-antibrane creation and annihilation does not change the D-brane charge, we can add any collection of $m$ D9-branes and $m$ D$\bar{\text{9}}$-branes with Chan-Paton bundle H.\footnote{We use the same notation H for both the gauge bundle on the D-brane and that on the anti-D-brane, as they are by definition topologically equivalent.}\\
Hence, the configuration (E,F) should be considered to have the same D-brane charge as the configuration ($\text{E}\oplus\text{H}$,$\text{F}\oplus\text{H}$): pairs of vector bundles with such an equivalence relation constitute, in a nutshell, the basic ingredients for the K-theory group K(M) of the spacetime.\\
In particular, Witten showed that the construction used by Sen can be generalized to arbitrary spacetime and D-brane configurations, and that such a generalization perfectly corresponds to mathematical properties of K-theory. More importantly, Witten's classification of D-branes via K-theory applies not only to type IIB/A superstring theories, but also to type I and to String theory on \emph{orbifolds}: in these latter cases, indeed, the K-theory description makes new and unexpected predictions, as we will see in the following chapters.      
\chapter{K-Theory, an introduction}
\section{The group $\kgroup{0}{\topsp{X}}$}
In this section we introduce the basic definitions used to construct the group $\kgroup{0}{\topsp{X}}$ of a topological space $\topsp{X}$. In particular, we will restrict ourselves to compact Hausdorff topological spaces which can carry a structure of a finite CW-complex. Even if most of the constructions can be extended to more general topological spaces, the choice of CW-complexes is not restrictive for the aim of this thesis, as we will mainly be interested in the K-theory groups of finite dimensional manifolds, which naturally carry a canonical CW-complex structure. In the following exposition we tacitily refer to \cite{Atiyah1967,KAROUBI}, unless otherwise stated. \\[2mm]
\indent Let ${\rm Vect^{F}(\topsp{X})}$ denote the set of isomorphism classes of topological F-vector bundles over $\topsp{X}$, where $\rm F=\mathbb{C},\mathbb{R}$. The direct (Whitney) sum of vector bundles gives ${\rm Vect^{F}(\topsp{X})}$ the structure of an abelian monoid: in a nutshell, the group $\kgroup{0}{\topsp{X}}$ is constructed via a procedure that consists intutively in adding ``inverses'' to ${\rm Vect^{F}(\topsp{X})}$. In particular, the procedure can be generalized to any abelian monoid, which is the case we will illustrate in the following.\\[2mm]
\indent Let $\mathcal{A}$ be an abelian monoid. We can then associate to $\mathcal{A}$ an abelian group $\kgroup{0}{\mathcal{A}}$ and a monoid homomorphism $\alpha:\mathcal{A}\to\kgroup{0}{\mathcal{A}}$ with the the following \emph{universal} property. For any group G, and any morphism of the underlying monoids $f:\mathcal{A}\to{\rm G}$, there is a unique group homomorpshim $\tilde{f}$ such $\tilde{f}\alpha=f$. Because of the uniqueness property, it is  an immediate consequence that if such a $\kgroup{0}{\mathcal{A}}$ exists, then it is unique up to isomorphism. The group $\kgroup{0}{\mathcal{A}}$ is known usually as the \emph{Grothendieck group} of $\mathcal{A}$.\\ 
A possible construction for the group $\kgroup{0}{\mathcal{A}}$ is the following. Denote with $\text{F}(\mathcal{A})$ the free abelian group generated by the elements of $\mathcal{A}$, and let $\text{E}(\mathcal{A})$ be the subgroup of $\text{F}({\mathcal{A}})$ generated by those elements of the form $a+b-(a\oplus{b})$, with $a,b\in\mathcal{A}$, and $\oplus$ denoting the addition in $\mathcal{A}$. Then for $\kgroup{0}{\mathcal{A}}:=\text{F}(\mathcal{A})/\text{E}(\mathcal{A})$, the universality property above holds, with $\alpha:\mathcal{A}\to\kgroup{0}{\mathcal{A}}$ the obvious map.\\[2mm]
 A different and sometimes convenient construction of $\kgroup{0}{\mathcal{A}}$ is the following. Consider the diagonal homomorphism of monoids $\Delta:\mathcal{A}\to\mathcal{A}\times\mathcal{A}$, i.e the map $\Delta(a)=(a,a)$, and denote with $\kgroup{0}{\mathcal{A}}$ the set of cosets of $\Delta(\mathcal{A})$ in $\mathcal{A}\times\mathcal{A}$. It is clearly a quotient monoid, but the interchange of factors induces inverses in $\kgroup{0}{\mathcal{A}}$, giving it a group structure. If we define $\alpha:\mathcal{A}\to\kgroup{0}{\mathcal{A}}$ to be the composition of $a\to(a,0)$ with the natural projection $\mathcal{A}\times\mathcal{A}\to\kgroup{0}{\mathcal{A}}$, then the universality property above holds.\\[2mm]  
\indent The association to a monoid $\mathcal{A}$ of its Grothendieck group as defined above induces in an obvious way a unique covariant functor ${\rm K^{0}}$ from the category $\mathscr{A}$ of abelian monoids with monoid morphisms to the category $\mathscr{G}$ of abelian groups with group morphisms. Indeed, to any morphisim $f:\mathcal{A}\to\mathcal{B}$, the functor ${\rm K^{0}}$ associates the morphism $\kgroup{0}{f}:\kgroup{0}{\mathcal{A}}\to\kgroup{0}{\mathcal{B}}$ satysfying $\alpha_{\mathcal{B}}\circ{f}=\kgroup{0}{f}\circ\alpha_{\mathcal{A}}$, and such a morphism is unique by the universality property. 
\begin{example}
Consider $\mathcal{A}:=(\mathbb{N}_{0},+)$. Then $\kgroup{0}{\mathcal{A}}\simeq{\mathbb{Z}}$
\end{example}
\begin{example}
Consider $\mathcal{A}:=(\mathbb{Z}-\{0\},\cdot)$. Then $\kgroup{0}{\mathcal{A}}\simeq{\mathbb{Q}-\{0\}}$.
\end{example}
A fundamental example of the above construction that allows a higher degree of generality is the following. Let $\mathscr{C}$ an additive category, and denote with $\dot{\rm E}$ the isomorphism class of the object ${\rm E}$. Then the set $\Phi(\mathscr{C})$ of such classes can be provided with a structure of an abelian monoid if one defines $\dot{\rm E}+\dot{\rm F}$ to be $\dot{\rm E\oplus F}$: the well-definedness of the + operation and the algebraic identities derive from the additivity of the category $\mathscr{C}$. In this case one denotes with $\kgroup{0}{\mathscr{C}}$ the group $\kgroup{0}{\Phi(\mathscr{C})}$, called the \emph{Grothendieck group of the category $\mathscr{C}$}. Moreover, if $\varphi:\mathscr{C}\to\mathscr{C}^{'}$ is an additive functor, then $\varphi$ naturally induces a monoid homomorphism $\Phi(\mathscr{C})\to\Phi(\mathscr{C}^{'})$, hence a group homomorphism $\kgroup{0}{\Phi(\mathscr{C})}\to\kgroup{0}{\Phi(\mathscr{C}^{'})}$, denoted with $\varphi_{*}$. The usual composition rule follows.\\[2mm]  
\indent We will now specialize to the case $\mathcal{A}={\rm Vect^{F}(\topsp{X})}$, for $\topsp{X}$ a topological space. We denote with $\kgroup{0}{\topsp{X}}$ the group $\kgroup{0}{{\rm Vect^{F}(\topsp{X})}}$, or equivalently the Grothendieck group of the additive category of topological F-vector bundles on $\topsp{X}$, that we will denote with $\mathscr{V}^{\rm F}$. To follow the notations usually found in literature, we will use $\kgroup{0}{\topsp{X}}$ for $\rm F=\mathbb{C}$, and $\kogroup{0}{\topsp{X}}$ for $\rm F=\mathbb{R}$. When the results \emph{do not} depend on the choice of $\rm F$, we will use the notation for $\rm F=\mathbb{C}$ .\\
The basic constructions for the $\rm K^{0}$ group described above allow to prove the following basic, but very useful results. 
\begin{proposition}\label{decomp} Let $\topsp{X}$ be a compact space. Then every element $x\in\kgroup{0}{\topsp{X}}$ can be written in the form ${\rm [E]-[F]}$, with $\rm E,F$ vector bundles on $\topsp{X}$. Moreover,\\ ${\rm[E]-[F]=[E^{'}]-[F^{'}]}$ in $\kgroup{0}{\topsp{X}}$ if and only if there exists a vector bundle {\rm G} on $\topsp{X}$ such that ${\rm E\oplus{F^{'}}\oplus{G}}\simeq{\rm E^{'}\oplus{F}\oplus{G}}$
\end{proposition}
\begin{proof}
By definition, $x=[({\rm \dot{E},\dot{F}})]$, for some vector bundles $\rm E,F$. Then we have that ${\rm [ (\dot{E},\dot{F})]=[(\dot{E},0)+(0,\dot{F})]=[(\dot{E},0)]+[(0,\dot{F})]=[E]-[F]}$, where with $\rm [E],[F]$ we denote the composition of $\rm {E}\to\dot{E}$ with $\rm\dot{E}\to[\dot{E}]$.\\
Let ${\rm [E]-[F]=[E^{'}]-[F^{'}]}$. Then one has that ${\rm [E\oplus{F^{'}}]=[E^{'}\oplus{F}]}$, which implies $\rm \dot{E}+\dot{F}^{'}+\dot{G}=\dot{E}^{'}+\dot{F}+\dot{G}$, for some $\rm G$. Hence $\rm {E}\oplus{F}^{'}\oplus{G}\simeq{E}^{'}\oplus{F}\oplus{G}$.
\end{proof}
The above proposition allows almost immediately to prove the following
\begin{proposition}
Let $\rm E$ and $\rm F$ be vector bundles over $\topsp{X}$. Then $\rm [E]=[F]$ if and only if ${\rm E}\oplus\theta_{n}\simeq{{\rm F}\oplus\theta_{n}}$, for some $\theta_{n}$ a trivial bundle of rank n. 
\end{proposition}
\begin{proof}
Recall that for any vector bundle $\rm G$ on a compact space $\topsp{X}$ there exists a vector bundle $\rm  G^{'}$ such that $\rm G\oplus{G}^{'}\simeq\theta_{n}$, for some $n$, i.e. any vector bundle is a \emph{projective} object in $\mathscr{V}^{F}$.\\
By Proposition \ref{decomp}, if $\rm [E]=[F]$, then $\rm E\oplus{G}\simeq{F\oplus{G}}$, for some vector bundle $\rm G$. By adding $G^{'}$ chosen as above, one gets that ${\rm E}\oplus{\theta_{n}}\simeq{{\rm F}\oplus{\theta_{n}}}$, which yelds the desired result. 
\end{proof}
An easy consequence of Proposition \ref{decomp} is that any element $x\in\kgroup{0}{\topsp{X}}$ can be written as $\rm [H]-[\theta_{n}]$, for some vector bundle $\rm H$ and some $\theta_{n}$.\\[2mm]
As we have seen, the group $\kgroup{0}{\mathcal{A}}$ ``depends'' covariantly on the monoid $\mathcal{A}$. However, $\kgroup{0}{\topsp{X}}$ depends contravariantly on the topological space $\topsp{X}$, i.e. $\rm K^{0}$ is a contravariant functor from the category $\mathscr{T}$ of topological spaces with continuous maps to $\mathscr{G}$.\\  
More precisely, if we let $f:\topsp{X}\to\topsp{Y}$ be a continuous map, then $f$ induces a monoid morphism $f^{*}:{\rm Vect^{F}(\topsp{Y})}\to{\rm Vect^{F}(\topsp{X})}$ via the pullback of vector bundles, hence a map $\kgroup{0}{\topsp{Y}}\to\kgroup{0}{\topsp{X}}$, which we still denote with $f^{*}$.\\
By the homotopy theory of vector bundles one has immediately the following
\begin{theorem}
Let $\topsp{X}$ and $\topsp{Y}$ be compact topological spaces, and let $f_{0},f_{1}:\topsp{X}\to\topsp{Y}$ be continuous and homotopic maps. Then $f_{0}$ and $f_{1}$ induces the same homomorphism $\kgroup{0}{\topsp{Y}}\to\kgroup{0}{\topsp{X}}$.
\end{theorem}
Hence, the group $\kgroup{0}{\topsp{X}}$ is a topological invariant, in the sense that two isomorphic spaces have isomorphic $\rm K^{0}$ groups. More importantly, if the two spaces can be ``deformed'' into each other, they have isomorphic $\rm K^{0}$: this property is typical of a cohomology theory, and it is intuitively a first indication that the group $\kgroup{0}{\topsp{X}}$ is a building block for a cohomology theory. In this sense, it is important to define the \emph{reduced group} $\redk{0}{\topsp{X}}$.\\
First notice that $\kgroup{0}{\{pt\}}\simeq{\mathbb{Z}}$ , as a vector bundle on a point is uniquely characterized by the dimension of its typical fiber. Then, the inclusion $i:{x_{0}}\to\topsp{X}$ induces the homomorphism 
\begin{displaymath}
i^{*}:\kgroup{0}{\topsp{X}}\to\kgroup{0}{\{x_{0}\}}\simeq\mathbb{Z}
\end{displaymath}
The reduced group $\redk{0}{\topsp{X}}$ is defined as the kernel of $i^{*}$. Moreover, the following exact sequence
\begin{displaymath}
0\to\redk{0}{\topsp{X}}\to\kgroup{0}{\topsp{X}}\xrightarrow{i^{*}}\kgroup{0}{\{x_{0}\}}\to{0}
\end{displaymath}
canonically splits, i.e. $\kgroup{0}{\topsp{X}}\simeq\mathbb{Z}\oplus\redk{0}{\topsp{X}}$.\\
Given a vector bundle $\rm E\to{X}$ with ${\rm E}_{x}$ the fiber of $\rm E$ over $x$, we can define the \emph{rank function of $\rm E$} $\rm rnk(\rm E):\topsp{X}\to\mathbb{N}_{0}$. As E is locally trivial as a vector bundle over $\topsp{X}$, its rank function is a locally constant function over $\topsp{X}$ with values in $\mathbb{N}_{0}$, i.e. an element of the abelian monoid $\rm H^{0}(X;\mathbb{N}_{0})$. Hence, the rank map extends naturally as the homomorphism
\begin{equation}\label{virtual}
\begin{array}{cccc}
\rm rnk :& \kgroup{0}{\topsp{X}}&\to& \rm H^{0}(X;\mathbb{Z})\\
& \rm [E]-[F] & \to & \rm rnk(E)-rnk(F)
\end{array}
\end{equation}
In the case in which $\topsp{X}$ is a connected space, the integer (\ref{virtual}) is called the \emph{virtual dimension} of the class $\rm [(E,F)]$ in $\kgroup{0}{\topsp{X}}$. Moreover, in the same case, the group $\redk{0}{\topsp{X}}$ is isomorphic to the kernel of the rank homomorphism, hence it consists of the subgroup of elements whose virtual dimension is zero, classes $\rm [(E,F)]$ with $\rm E$ and $\rm F$ of equal rank.\\[2mm]
\indent We conclude this section with a proposition which states how the K-theory group behaves under disjoint unions.
\begin{proposition}
Let $\topsp{X}=\coprod_{i=1}^{n}\topsp{X}_{i}$. Then the inclusions of the $\topsp{X}_{i}$ in $\topsp{X}$ induces the decomposition $\kgroup{0}{\topsp{X}}\simeq{\kgroup{0}{\topsp{X}_{1}}\oplus\kgroup{0}{\topsp{X}_{1}}\oplus\dots\oplus\kgroup{0}{\topsp{X}_{n}}}$.
\end{proposition}
\begin{proof}
Use the fact that any vector bundle on $\topsp{X}$ is characterized by its restrictions on the $\topsp{X}_{i}$.
\end{proof}
\quad
\begin{remark}
Notice that this last proposition is \emph{not} true for the functor $\redk{0}{\topsp{X}}$. Indeed, if $\topsp{X}$ is the disjoint union of two points $\{x_{0}\}$ and $\{x_{1}\}$, then $\redk{0}{\topsp{X}}\simeq{\mathbb{Z}}$, but $\redk{0}{\{x_{i}\}}=0$, for $i=0,1$.
\end{remark}
\section{Relative K-theory and higher K-groups}
Starting with the group $\kgroup{0}{\topsp{X}}$ one can naturally define the \emph{relative} K-theory and the \emph{higher} K-theory groups.\\
Let $\mathscr{T}_{\rm P}$ denote the category of \emph{compact pairs} $(\topsp{X},\topsp{Y})$, where $\topsp{X},\topsp{Y}$ are topological spaces, and $\topsp{X}$ is such that it can be equipped with a structure of a CW-complex such that $\topsp{Y}$ is a CW-subcomplex. The morphisms $(\topsp{X},\topsp{Y})\to(\topsp{X}^{'},\topsp{Y}^{'})$ are given by \emph{relative} maps, i.e. continuous functions $f:\topsp{X}\to\topsp{X}^{'}$ such that $f(\topsp{Y})\subset{\topsp{Y}^{'}}$.\\
We define the \emph{relative K-theory group} as
\begin{equation}
\kgroup{0}{\topsp{X},\topsp{Y}}:=\redk{0}{\topsp{X}/\topsp{Y}}
\end{equation}
where $\topsp{X}/\topsp{Y}$ is obtained from $\topsp{X}$ by shrinking $\topsp{Y}$ to a point, with respect to which the reduced K-group is defined. In the case when $\topsp{Y}=\{\O\}$, we define $\topsp{X}/\topsp{Y}$ as $\topsp{X}^{+}:=\topsp{X}\coprod\{\pt\}$. Hence, for a base-pointed space $\topsp{X}$ one has
\begin{equation}
\kgroup{0}{\topsp{X},\O}:=\redk{0}{\topsp{X}^{+}}\simeq\kgroup{0}{\topsp{X}}
\end{equation}
Moreover, for $\{x_{0}\}\subset{\topsp{X}}$ we have $\kgroup{0}{\topsp{X},\{x_{0}\}}=\redk{0}{\topsp{X}}$.\\
Let us denote with $\pi:\topsp{X}\to\topsp{X}/\topsp{Y}$. The map $\pi$ induces the obvious relative map 
\begin{displaymath}
\pi:(\topsp{X},\topsp{Y})\to(\topsp{X}/\topsp{Y},\{pt\})
\end{displaymath}
hence the map
\begin{displaymath}
\pi^{*}:\kgroup{0}{\topsp{X}/\topsp{Y},\{pt\}}\to\kgroup{0}{\topsp{X},\topsp{Y}}
\end{displaymath}
We have the following \emph{excision theorem} 
\begin{theorem} The map $\pi^{*}$ is an isomorphism.
\end{theorem}
It is instructive, at this point, to explain a useful way to describe vector bundles on the space $\topsp{X}/\topsp{Y}$, given a vector bundle $\rm E$ on $\topsp{X}$ with typical fiber $\topsp{F}$. Suppose $E$ is trivializable on $\topsp{Y}$, i.e. there exists a homeomorphism $\alpha:{\rm E|_{Y}\xrightarrow{\simeq}{\rm Y \times{F}}}$ which is linear on the fibres. Consider the equivalence relation on $\rm E|_{\rm Y}$ given by
\begin{displaymath}
v\sim{v^{'}}\:\Leftrightarrow\:p\circ\alpha(v)=p\circ\alpha(v^{'})
\end{displaymath}
where $p:\topsp{Y}\times{\rm F}\to\topsp{Y}$ is the canonical projection. The relation is then extended to the whole $\rm E$. The corresponding set of equivalence classes can be shown to be the total space of a vector bundle over $\topsp{X}/\topsp{Y}$; moreover, every vector bundle on $\topsp{X}$ which is trivial when restricted to $\topsp{Y}$ is isomorphic to the pullback of a vector bundle over $\topsp{X}/\topsp{Y}$.\\
This construction, in particular, allows to prove that the following sequence
\begin{equation}\label{exactseq}
\kgroup{0}{\topsp{X},\topsp{Y}}\xrightarrow{\rho^{*}}\kgroup{0}{\topsp{X}}\xrightarrow{i^{*}}\kgroup{0}{\topsp{Y}}
\end{equation} 
is exact, where the homomorphism $\rho^{*}$ is induced by the map $\rho:\topsp{X}\to\topsp{X}/\topsp{Y}$, and $i:\topsp{Y}\hookrightarrow{\topsp{X}}$ realizes $\topsp{Y}$ as a subspace of $\topsp{X}$ \cite{KAROUBI,Husemoller}.\\[2mm]
\indent To define the higher K-groups, we need to introduce some well known operations on topological spaces. Let $\topsp{X}$ and $\topsp{Y}$ be compact spaces with base points $\{x_{0}\}$ and $\{y_{0}\}$, respectively. The \emph{wedge product} of $\topsp{X}$ and $\topsp{Y}$ is defined as
\begin{equation}
\topsp{X}\vee\topsp{Y}:=(\topsp{X}\amalg\topsp{Y})/\{x_{0}\sim{y_{0}}\}
\end{equation}
while the \emph{smash product} of $\topsp{X}$ and $\topsp{Y}$ is defined as
\begin{equation}
\topsp{X}\wedge\topsp{Y}:=\topsp{X}\times\topsp{Y}/\left\{\topsp{X}\times\{y_{0}\}\cup\{x_{0}\}\times\topsp{Y}\right\}
\end{equation}
The two operations above are related. Indeed, one has natural maps
\begin{displaymath}
\topsp{X}\vee\topsp{Y}\to\topsp{X}\times\topsp{Y}\to\topsp{X}\wedge\topsp{Y}
\end{displaymath}
which allow to write
\begin{displaymath}
\topsp{X}\wedge\topsp{Y}=\topsp{X}\times\topsp{Y}/\topsp{X}\vee\topsp{Y}
\end{displaymath}
Moreover, the operations $\vee$ and $\wedge$ are associative and commutative, and $\wedge$ is distributive over $\vee$. This means, for example, that there is a canonical homeomorphism between $\topsp{X}\wedge\topsp{Y}$ and $\topsp{Y}\wedge\topsp{X}$.\\
An important property of the smash product is that 
\begin{equation}\label{smash}
\topsp{S}^{n}\simeq{\topsp{S}^{1}\wedge\topsp{S}^{1}\wedge\dots\wedge\topsp{S}^{1}}\quad n\text{-times}
\end{equation}
where $\topsp{S}^{n}$ is the standard $n$-sphere with base point.\\
For a given space $\topsp{X}$ with base point, we define the \emph{n-th reduced suspension} $\Sigma^{n}(\topsp{X})$ as
\begin{equation}
\Sigma^{n}(\topsp{X}):=\topsp{S}^{n}\wedge\topsp{X}
\end{equation}
Because of property (\ref{smash}), we have that the $n$-th reduced suspension  $\Sigma^{n}(\topsp{X})$ is the $n$-th iterated reduced suspension of $\topsp{X}$.\\
We have then the following
\begin{definition} For $(\topsp{X},\topsp{Y})$ a compact pair, and for $n\geq{0}$ we define
\begin{displaymath}
\kgroup{-n}{\topsp{X},\topsp{Y}}:=\redk{0}{\Sigma^{n}(\topsp{X}/\topsp{Y})}
\end{displaymath}
For $\topsp{X}$ a compact space we put
\begin{displaymath}
\kgroup{-n}{\topsp{X}}:=\kgroup{-n}{\topsp{X},\O}:=\redk{0}{\Sigma^{n}(\topsp{X}^{+})}
\end{displaymath}
and for $\topsp{X}$ a compact space with base point $\{x_{0}\}$ we put
\begin{displaymath}
\redk{-n}{\topsp{X}}:=\kgroup{-n}{\topsp{X},\{x_{0}\}}:=\redk{0}{\Sigma^{n}(\topsp{X})}
\end{displaymath}
\end{definition}
The ${\rm K}^{-n}$ are contravariant functors, as the reduced suspension induces a covariant functor on $\mathscr{T}$. The relation between higher and reduced K-theory groups is given canonically by
\begin{displaymath}
\kgroup{-n}{\topsp{X}}\simeq\redk{-n}{\topsp{X}}\oplus\kred{-n}{\{x_{0}\}}
\end{displaymath}
with $\{x_{0}\}$ being the base point. By the definitions above $\kgroup{-n}{\{x_{0}\}}\simeq\redk{0}{\topsp{S}^{n}}$: we will compute these groups in section \ref{classifying}.\\
The various higher K-groups are linked together by the following semi-infinite long exact sequence
\begin{equation}\label{exact}
\dots\kgroup{-(n+1)}{\topsp{X}}\to\kgroup{-(n+1)}{\topsp{Y}}\xrightarrow{\delta}\kgroup{-n}{\topsp{X},\topsp{Y}}\to\kgroup{-n}{\topsp{X}}\to\kgroup{-n}{\topsp{Y}}\xrightarrow{\delta}\dots
\end{equation}
where $\delta$ is the boundary homomorphism. For a definition of $\delta$ see \cite{Atiyah1967,KAROUBI}. The sequence above is another typical feature of a cohomology theory.\\
\section{Multiplicative structures on K-theory}
As vector bundles on a space $\topsp{X}$ can be ``multiplied'' together via tensor product, it is natural to look for multiplicative structures on the K-groups.\\
The tensor product of vector bundles on $\topsp{X}$ induces a multiplication
\begin{displaymath}
\kgroup{0}{\topsp{X}}\otimes_{\mathbb{Z}}\kgroup{0}{\topsp{X}}\to\kgroup{0}{\topsp{X}}
\end{displaymath}
defined as
\begin{equation}
\rm [(E,F)]\otimes [(E^{'},F^{'})]:=[(E\otimes{E^{'}}\oplus F\otimes{F^{'}},E\otimes{F^{'}}\oplus F\otimes E^{'})]
\end{equation}
The expression above comes from writing $\rm [(E,F)]$ as $\rm [E]-[F]$, and formally imposing distributivity of the tensor product on virtual bundles. Hence, $\kgroup{0}{\topsp{X}}$ is actually a ring.\\
There is another product, usually called the external tensor product or \emph{cup product}, which is a homomorphism 
\begin{equation}\label{cupprod}
\cup:\kgroup{0}{\topsp{X}}\otimes_{\mathbb{Z}}\kgroup{0}{\topsp{Y}}\to\kgroup{0}{\topsp{X}\times\topsp{Y}} 
\end{equation}
defined as follows. Denote with $\pi_{\topsp{X}}:\topsp{X}\times\topsp{Y}\to\topsp{X}$ and $\pi_{\topsp{Y}}:\topsp{X}\times\topsp{Y}\to\topsp{X}$ the canonical projections.\\
Hence, $\cup(\rm [E],[F])\in\kgroup{0}{\topsp{X}}\otimes_{\mathbb{Z}}\kgroup{0}{\topsp{Y}}$ is the class in $\kgroup{0}{\topsp{X}\times\topsp{Y}}$ defined by
\begin{displaymath}
\cup(\rm [E],[F]):=\pi_{\topsp{X}}^{*}([E])\otimes\pi_{\topsp{Y}}^{*}([F])
\end{displaymath}
and extended by (bi)linearity. Notice that when $\topsp{Y}=\topsp{X}$, we have that 
\begin{displaymath}
\Delta^{*}\cup(\rm [(E,F)],[(E^{'},F^{'})])=[(E,F)]\otimes[(E^{'},F^{'})]
\end{displaymath}
where $\Delta:\topsp{X}\to\topsp{X}\times\topsp{X}$ is the diagonal map.\\
The above exterior product is crucial when dealing with higher K-groups. In particular, when restricted to reduced K-theory, the cup product induces a homomorphism \cite{Atiyah1967}
\begin{displaymath}
\redk{0}{\topsp{X}}\otimes_{\mathbb{Z}}\redk{0}{\topsp{Y}}\to\redk{0}{\topsp{X}\wedge\topsp{Y}}
\end{displaymath}
which immediately induces the homomorphism
\begin{equation}
\redk{0}{\Sigma^{n}(\topsp{X})}\otimes_{\mathbb{Z}}\redk{0}{\Sigma^{m}(\topsp{Y})}\to\redk{0}{\Sigma^{n+m}(\topsp{X}\wedge\topsp{Y})}
\end{equation}
By substituting in the above expression $\topsp{X}^{+}$ and $\topsp{Y}^{+}$ for $\topsp{X}$ and $\topsp{Y}$, respectively, one obtains the homomorphism
\begin{equation}\label{graded}
\kgroup{-n}{\topsp{X}}\otimes_{\mathbb{Z}}\kgroup{-m}{\topsp{Y}}\to\kgroup{-(n+m)}{\topsp{X}\times\topsp{Y}}
\end{equation}
If we denote with $\kgroup{-*}{\topsp{X}}:=\bigoplus_{n\geq{0}}\kgroup{-n}{\topsp{X}}$, the cup product induces, via (\ref{graded}), the structure of a graded ring on $\kgroup{-*}{\topsp{X}}$. Moreover, for any space $\topsp{X}$ with base point ${\{\rm pt\}}$, the cup product makes $\kgroup{-*}{\topsp{X}}$ into a graded module over $\kgroup{-*}{\topsp{pt}}$.\\[2mm]
\indent Notice, at this point, that a priori the ring $\kgroup{-*}{\topsp{X}}$ could be very complicated, even for the case $\topsp{X}={\{\rm pt\}}$. It's an extremely remarkable property of K-theory that this is not the case, as the following results show.\\
For the first time in this introduction, we have to make a difference between complex and real K-theory to introduce a fundamental and deep result.
\begin{theorem}{\rm \textbf{(Bott Periodicity)}}\label{complexpoint} The ring $\kgroup{-*}{\{\rm pt\}}$ is a polynomial algebra generated by an element $u\in\kgroup{-2}{\{\rm pt\}}\simeq\redk{0}{\topsp{S}^{2}}$, i.e. there is a ring isomorphism
\begin{equation}
\kgroup{-*}{\{\text{pt}\}}\simeq{\mathbb{Z}[u]}
\end{equation}
\end{theorem}
The element $u$ can be represented as $u=[\rm H]-[\theta_{1}]$, where H denotes the \emph{tautological} complex line bundle over $\rm S^{2}\simeq\mathbb{CP}^{1}$. For a proof see \cite{Atiyah1967}.\\[2mm]
\indent The above theorem, in particular, says that the map $\mu_{u}:\kgroup{-n}{\{\rm pt\}}\to\kgroup{-n-2}{\{\rm pt\}}$ induced by multiplication by $u$, is an isomorphism for all $n$.\\
The previous theorem generalizes as
\begin{theorem}\label{bott}
Let $\topsp{X}$ be a compact space. Then the map
\begin{displaymath}
\mu_{u}:\kgroup{-n}{\topsp{X}}\xrightarrow{\simeq}\kgroup{-n-2}{\topsp{X}}
\end{displaymath}
given by module multiplication by $u$, is an isomorphism for all $n\geq{0}$
\end{theorem}
The above theorem is usually referred to as the \emph{general Bott periodicity theorem}, and essentially states that there are only two ``independent'' K functors, namely $\rm K^{0}$ and $\rm K^{-1}$. Moreover, theorem (\ref{bott}) can be extended to relative and reduced K-theory. Finally, the original simplest formulation of the periodicity theorem states that for any compact space $\topsp{X}$, there is an isomorphism between $\kgroup{0}{\topsp{X}}\otimes\kgroup{0}{\topsp{S}^{2}}$ and $\kgroup{0}{\topsp{X}\times{\topsp{S}^{2}}}$.\\[2mm]
\indent In the real case things are a bit different, and slightly more complicated. Indeed, one has the following real version of the Bott periodicity theorems.
\begin{theorem}\label{realpoint} The ring $\kogroup{-*}{\{\rm pt\}}$ is generated by elements
\begin{displaymath}
\eta\in\kogroup{-1}{\{{\rm pt} \}},\quad y\in\kogroup{-4}{\{{\rm pt} \}},\quad x\in\kogroup{-8}{\{\rm pt \}}
\end{displaymath}
subject to the relations
\begin{displaymath}
2\eta=0,\quad\eta^{3}=0,\quad\eta{y}=0,\quad y^{2}=4x
\end{displaymath}
i.e. there is a ring isomorphism
\begin{equation}
\kogroup{-*}{\{{\rm pt}\}}\simeq\mathbb{Z}[\eta,y,x]/<2\eta,\eta^{3},\eta{y},y^{2}-4x>
\end{equation}
\end{theorem}
\begin{theorem}\label{realbott}
Let $\topsp{X}$ be a compact space. Then the map
\begin{displaymath}
\mu_{x}:\kogroup{-n}{\topsp{X}}\xrightarrow{\simeq}\kogroup{-n-8}{\topsp{X}}
\end{displaymath}
given by module multiplication by $x$, is an isomorphism for all $n\geq{0}$.
\end{theorem}
The first consequence of theorem (\ref{realpoint}) is that the ring $\kogroup{-*}{\{\rm pt\}}$ is not freely generated, i.e. it contains torsion subgroups.\\[2mm]
\indent Apart from limiting the number of K-groups to compute, the Bott periodicity theorem allows to define K-theory for positive degrees, which is the final important step needed to construct a cohomology theory associated to the K-groups. Indeed, for $n>0$ one defines $\kgroup{n}{\topsp{X},\topsp{Y}}$ as 
\begin{displaymath}
\begin{array}{c}
\kgroup{n}{\topsp{X},\topsp{Y}}:=\kgroup{0}{\topsp{X},\topsp{Y}},\quad\text{for}\:n\:\text{even}\\
\kgroup{n}{\topsp{X},\topsp{Y}}:=\kgroup{-1}{\topsp{X},\topsp{Y}},\quad\text{for}\:n\:\text{odd}
\end{array}
\end{displaymath}
and analogous definitions are given for $\kogroup{n}{\topsp{X},\topsp{Y}}$. In particular, the Bott isomorphism is compatible with the long exact sequence (\ref{exact}), and hence it can be extended to positive degrees.\\
Because of Bott periodicity, it is convenient to define for any pair $(\topsp{X},\topsp{Y})$
\begin{displaymath}
\begin{array}{c}
\kgroup{*}{\topsp{X},\topsp{Y}}:=\kgroup{0}{\topsp{X},\topsp{Y}}\oplus\kgroup{-1}{\topsp{X},\topsp{Y}}\\
\kogroup{*}{\topsp{X},\topsp{Y}}:=\bigoplus_{i=0}^{7}\kogroup{-i}{\topsp{X},\topsp{Y}}
\end{array}
\end{displaymath}
The cohomology theory constructed in this way is referred to as K-theory, and satisfies all the usual axioms of a cohomology theory $\rm E$ on $\mathscr{T}_{\text{P}}$, but the \emph{dimension axiom}, which requires the cohomology $\text{E}$ to satisfy $\text{E}^{i}(\rm pt)=0$, for $i\neq{0}$, and $\text{E}^{0}(\rm pt)=\mathbb{Z}$.\\
Finally, because of Bott periodicity, the long exact sequence (\ref{exact}) for \emph{complex} K-theory can be truncated to the six-term exact sequence 
\begin{displaymath}
\xymatrix{\kgroup{-1}{\topsp{X}}\ar[r] & \kgroup{-1}{\topsp{Y}}\ar[r]^{\delta}&\kgroup{0}{\topsp{X},\topsp{Y}}\ar[d]\\
\kgroup{-1}{\topsp{X},\topsp{Y}} \ar[u]&\kgroup{0}{\topsp{Y}}\ar[l]^{\delta}&\kgroup{0}{\topsp{X}}\ar[l]
}
\end{displaymath}  
\section{K-theory and classifying spaces}\label{classifying} 
As for ordinary cohomology, there is a description of K-theory in terms of classifying spaces. As this approach will play a very important role in later chapters, we will review it in some detail.\\
\indent Let $\topsp{X}$ a compact CW-complex of finite type. It is a very important result in geometry that F-vector bundles of rank $k$ over $\topsp{X}$ can be classified up to isomorphism. Namely, there exists a connected topological space ${\rm BF}_{k}$, equipped with a rank $k$ F-vector bundle $\mathbb{E}_{k}$, called the \emph{universal F-vector bundle}, such that
\begin{displaymath}
{{\rm Vect}_{k}^{F}(\topsp{X})\simeq\left[\topsp{X},{\rm BF}_{k}\right]}
\end{displaymath}
where the isomorphism is induced by assigning to the homotopy class $[f]$ the isomorphism class $[f^{*}\mathbb{E}_{k}]$ \cite{Husemoller}.\\
A model for ${\rm BF}_{k}$ can be defined as follows. Denote with $\text{Gr}(k,m;\text{F})$ the \emph{Grassmannian manifold} of F-linear subspaces in $\text{F}^{m}$ of dimension $k$. For $\rm F=\mathbb{C}$ one has
\begin{displaymath}
\text{Gr}(k,m;\mathbb{C})\simeq\dfrac{\text{U}(m)}{\text{U}(k)\times\text{U}(m-k)}
\end{displaymath}
while for $\rm F=\mathbb{R}$ one has
\begin{displaymath}
\text{Gr}(k,m;\mathbb{R})\simeq\dfrac{\text{O}(m)}{\text{O}(k)\times\text{O}(m-k)}
\end{displaymath}
The inclusion $\text{F}^{m}\subset\text{F}^{m+1}$ induces the natural inclusion $\text{Gr}(k,m;\text{F})\subset\text{Gr}(k,m+1;\text{F})$. Taking the inductive limit over such inclusions allows to define the \emph{infinite} Grassmannian
\begin{displaymath}
\text{Gr}(k,\infty;\text{F}):=\lim_{m\to\infty}\text{Gr}(k,m;\text{F})
\end{displaymath}
which can be used as a model for our classifying space \cite{Husemoller}. We will use the notation $\text{BU}(k)$ for $\text{Gr}(k,\infty;\mathbb{C})$, and $\text{BO}(k)$ for $\text{Gr}(k,\infty;\mathbb{R})$.\\
The universal F-vector bundle $\mathbb{E}_{k}\to\text{Gr}(k,\infty;\text{F})$ is given by assigning to a point $x\in\text{Gr}(k,\infty;\text{F})$ the vector space represented by $x$ itself.\\
To represent K-theory we need a limit of the above construction over the vector bundle rank $k$. More precisely, denote with  $\kgroup{'0}{\topsp{X}}$ the kernel of the rank homomorphism (\ref{virtual}). Recall that when $\topsp{X}$ is connected, $\kgroup{'0}{\topsp{X}}$ is isomorphic to $\redk{0}{\topsp{X}}$. In general  one has the decomposition
\begin{equation}
\kgroup{0}{\topsp{X}}\simeq\kgroup{'0}{\topsp{X}}\oplus\text{H}^{0}(\topsp{X};\mathbb{Z})
\end{equation}
with $\text{H}^{0}(\topsp{X};\mathbb{Z})=[\topsp{X},\mathbb{Z}]$.\\
Notice, at this point, that for $[{\rm E}_{k}]\in{{\rm Vect}_{k}^{\rm F}(\topsp{X})}$, we have that $\rm[E_{k}]-[\theta_{k}]\in\kgroup{'0}{\topsp{X}}$. If we define
\begin{displaymath}
{{\rm Vect}_{\infty}^{\rm F}(\topsp{X}):=\lim_{k\to\infty}{\rm Vect}_{k}^{\rm F}(\topsp{X})}
\end{displaymath}
as the inductive limit over the inclusions ${\rm Vect}_{k}^{\rm F}(\topsp{X})\subset {\rm Vect}_{k+1}^{\rm F}(\topsp{X})$ induced by the maps ${\rm E}_{k}\to{{\rm E}_{k}}\oplus\theta_{1}$, one has that the map $[{\rm E}_{k}]\to[{\rm E}_{k}]-[\theta_{k}]$ is a monoid morphism ${{\rm Vect}_{\infty}^{F}(\topsp{X})\to\kgroup{'0}{\topsp{X}}}$, hence ${\rm Vect}_{\infty}^{F}(\topsp{X})$ can be given the structure of an abelian group \cite{KAROUBI}. Finally, as
\begin{displaymath}
\rm Vect_{\infty}^{F}(\topsp{X})=[\topsp{X},\text{BF}_{\infty}]
\end{displaymath}
one has
\begin{displaymath}
\kgroup{0}{\topsp{X}}\simeq{\rm Vect_{\infty}^{F}(\topsp{X})}\oplus[\topsp{X},\mathbb{Z}]=[\topsp{X},\mathbb{Z}\times\text{BF}_{\infty}]
\end{displaymath}
where $\text{BF}_{\infty}:=\bigcup_{k}\text{BF}_{k}$.\\
Moreover one can obtain reduced K-theory for a base pointed space $\topsp{X}$ as
\begin{displaymath}
\redk{0}{\topsp{X}}=[\topsp{X},\mathbb{Z}\times\text{BF}_{\infty}]_{*}
\end{displaymath}
where the base point in $\mathbb{Z}\times\text{BF}_{\infty}$ has been chosen to lie in $0\times\text{BF}_{\infty}$. Consequently, higher relative K-theory can be obtained as
\begin{displaymath}
\kgroup{-n}{\topsp{X},\topsp{Y}}=[\Sigma^{n}(\topsp{X}/\topsp{Y}),\mathbb{Z}\times\text{BF}_{\infty}]_{*}\simeq[\topsp{X}/\topsp{Y},\Omega^{n}(\mathbb{Z}\times\text{BF}_{\infty})]_{*}
\end{displaymath}
where $\Omega^{n}$ is the iterated loop space functor, and we have used the fact that for CW-complexes $\topsp{X},\topsp{Y}$ one has $[\Sigma(\topsp{X}),\topsp{Y}]_{*}=[\topsp{X},\Omega(\topsp{Y})]_{*}$.\\
In particular, if $\topsp{X}$ is a connected space, then
\begin{displaymath}
\redk{0}{\topsp{X}}=[\topsp{X},\text{BF}_{\infty}]_{*}
\end{displaymath}
In later sections we will see that another classifying space for K-theory is given by the space of Fredholm operators on a infinite dimensional Hilbert space.\\
\subsection{Examples: K-theory of spheres and tori}
The reduced K-theory of spheres can be deduced immediately by the Bott periodicity theorems (\ref{complexpoint}) and (\ref{realpoint}): in this section we want to give a geometric view on the Bott periodicity theorems, at least in the complex case. Indeed, this is actually how the Bott theorem originated.\\
Consider the standard $n$-dimensional sphere $\topsp{S}^{n}$. We will illustrate a procedure, called the \emph{clutching construction}, to construct vector bundles $\text{E}\to\topsp{S}^{n}$. Write the $n$-sphere $\topsp{S}^{n}$ as the union of the upper and lower hemispheres ${\rm D}_{+}^{n}$ and ${\rm D}_{-}^{n}$, such that ${\rm D}^{n}_{+}\cap{{\rm D}_{-}^{n}}\simeq{\text{S}^{n-1}}$. Given a continuous map $f:{\rm S}^{n-1}\to{{\rm GL}_{k}(\mathbb{C})}$, we can define the total space of a complex vector bundle $\text{E}_{f}$ of rank $k$ as
\begin{displaymath}
\text{E}_{f}:={{\rm D}^{n}_{+}\times\mathbb{C}^{k}\coprod {\rm D}^{n}_{-}\times\mathbb{C}^{k}}/\sim
\end{displaymath}
where the identification $\sim$ is between $(x,v)\in\partial{{\rm D}^{n}_{+}\times\mathbb{C}^{k}}$ and $(x,f(x)v)\in\partial{{\rm D}^{n}_{-}\times\mathbb{C}^{k}}$, and the projection $\text{E}_{f}\to\text{S}^{n}$ is the obvious one. Essentially, the clutching construction is a special case of the gluing construction for locally trivial vector bundles, which one can see by considering a ``small'' strip $\text{S}^{n}\times\{-\epsilon,\epsilon\}$, and using the function $f$ on each slice $\text{S}^{n}\times\{t\}$. Moreover, as the structure group for a complex vector bundle of rank $k$ can always be reduced to the unitary group $\text{U}(k)$, upon the introduction of a Hermitian metric, the function $f$ can always be chosen to be a continuous map $f:\text{S}^{n-1}\to\text{U}(k)$.\\
A basic property of the clutching construction is that $\text{E}_{f}\simeq{\text{E}_{g}}$ if $f$ and $g$ are homotopic maps. This allows to define a map
\begin{displaymath}
\Phi:[\text{S}^{n-1},\text{U}(k)]\to\text{Vect}_{k}(\text{S}^{n})
\end{displaymath}
A fundamental result is the following \cite{KAROUBI}
\begin{theorem}
The map $\Phi$ is a bijection.
\end{theorem}
In other words, \emph{all} complex vector bundles on spheres are obtained up to isomorphism by the clutching construction. By taking inductive limits, we finally obtain
\begin{displaymath}
\text{Vect}_{\infty}(\text{S}^{n})=[\text{S}^{n-1},\text{U}(\infty)]
\end{displaymath}
where $\text{U}(\infty)$ denotes the \emph{infinite} unitary group.\\ 
Hence we have
\begin{equation}
\redk{0}{\topsp{S}^{n}}=\pi_{n-1}\left(\text{U}(\infty)\right)
\end{equation}
Because $\kgroup{-n}{\rm pt}\simeq\redk{0}{\text{S}^{n}}$, we see that the Bott periodicity theorem is a statement about the homotopy of the infinite unitary group: indeed, this is how it was originally stated\footnote{To be precise, the statement was on the homotopy groups of the unitary groups $\text{U}(k)$ in the \emph{stable range}.} in \cite{Bott}.\\
The real case requires some modifications in the clutching construction, and, as it is expected from the structure of $\kogroup{*}{\rm pt}$, the relation with the orthogonal groups is more complicated.\\
Notice, also, that as a topological invariant for spheres, K-theory is a very coarse one, as it cannot detect the sphere dimensionality, in contrast to ordinary cohomology.\\
The case for tori, instead, is quite different. To compute the complex K-theory groups for tori, we will use the following result \cite{Atiyah1967}: given two base pointed spaces $\topsp{X}$ and $\topsp{Y}$ we have the following isomorphism
\begin{equation}\label{iso}
\redk{-n}{\topsp{X}\times\topsp{Y}}\simeq\redk{-n}{\topsp{X}\wedge\topsp{Y}}\oplus\redk{-n}{\topsp{X}}\oplus\redk{-n}{\topsp{Y}}
\end{equation}
We now specialize (\ref{iso}) to the case in which $\topsp{Y}=\topsp{S}^{1}$, with the usual base point. In this case, we have that
\begin{displaymath}
\begin{array}{rl}
\redk{-n}{\topsp{X}\times\topsp{S}^{1}}&\simeq\redk{-n}{\topsp{X}\wedge\topsp{S}^{1}}\oplus\redk{-n}{\topsp{X}}\oplus\redk{-n}{\topsp{S}^{1}}\\
&\simeq\redk{-(n+1)}{\topsp{X}}\oplus\redk{-n}{\topsp{X}}\oplus\redk{-(n+1)}{\{\topsp{pt}\}}
\end{array}
\end{displaymath}
where we have used Bott periodicity.\\
Recalling that for odd $n$ the reduced and unreduced K-theory groups coincide, we have that
\begin{equation}\label{relations}
\begin{array}{rl}
\kgroup{-2n}{\topsp{X}\times\topsp{S}^{1}}&\simeq{\kgroup{-1}{\topsp{X}}}\oplus{\kgroup{0}{\topsp{X}}}\\
\kgroup{-(2n+1)}{\topsp{X}\times\topsp{S}^{1}}&\simeq{\kgroup{0}{\topsp{X}}}\oplus{\kgroup{-1}{\topsp{X}}}
\end{array}
\end{equation}
Let $\topsp{T}^{n}\simeq\topsp{S}^{1}\times\dots\times\topsp{S}^{1}$ be the standard n-dimensional torus. The K-theory for $\topsp{T}^{n}$ can be computed using isomorphism (\ref{relations}) by induction with base $n=2$.\\
Indeed, we have that
\begin{displaymath}
\kgroup{0}{\topsp{T}^{2}}\simeq\kgroup{-1}{\topsp{T}^{2}}\simeq\kgroup{0}{\topsp{S}^{1}\times\topsp{S}^{1}}\simeq{\mathbb{Z}\oplus\mathbb{Z}}
\end{displaymath}
As $\topsp{T}^{n}\simeq\topsp{T}^{n-1}\times\topsp{S}^{1}$, we have
\begin{displaymath}
\kgroup{0}{\topsp{T}^{n}}\simeq\kgroup{-1}{\topsp{T}^{n}}\simeq\kgroup{-1}{\topsp{T}^{n-1}}\oplus\kgroup{0}{\topsp{T}^{n-1}}
\end{displaymath}
Then, by induction we find that there is the following non-canonical isomorphism for $n\geq{2}$
\begin{displaymath}
\kgroup{0}{\topsp{T}^{n}}\simeq\kgroup{-1}{\topsp{T}^{n}}\simeq{}\mathbb{Z}^{2^{(n-1)}}
\end{displaymath}
Hence, for tori the K-theory groups can indeed detect their dimension.
\section{The Atiyah-Bott-Shapiro isomorphism}
In this section we will introduce the \emph{Atiyah-Bott-Shapiro} (ABS) isomorphism, which will give explicit representatives for the generators of the rings $\kgroup{*}{\{\rm pt\}}$ and $\kogroup{*}{\{\rm pt\}}$ via the so called ``difference bundle construction''. More importantly, the ABS isomorphism relates complex and real Clifford algebras to K-theory: such a relation is somehow expected, given that the periodicity of K-theory is similar to the periodicity of Clifford algebras. The main reference is the seminal paper \cite{ABS}; moreover, we will refer to {\appclifford} for the basic notions of Clifford algebras.\\[2mm]
\indent To construct the ABS isomorphism, we need a reformulation of the relative K-theory group $\kgroup{0}{\topsp{X},\topsp{Y}}$, which will also prove useful in relation to elliptic operators, and in the description of D-branes, in particular in type IIA String theory. Again, this reformulation is due to \cite{ABS}, and there is no difference between the complex and real case.\\
\begin{definition}
Let $\topsp{X}$ and $\topsp{Y}$ be CW-complexes of finite type. For $n\geq{1}$, denote with $\mathcal{L}_{n}(\topsp{X},\topsp{Y})$ the set of elements $\rm \textbf{E}=(E_{0},E_{1},\cdots,E_{{\it n}};\sigma_{1},\sigma_{2},\cdots,\sigma_{\it n})$, where $\rm E_{\it i}$ is a vector bundle over $\topsp{X}$, $\sigma_{\it i}:{\rm E_{{\it i}-1}|_{\topsp{Y}}\to E_{\it i}|_{\topsp{Y}}}$ is a bundle morphism defined on $\topsp{Y}$, such that
\begin{displaymath}
{\rm 0\to{E_{0}|_{\topsp{Y}}}\xrightarrow{\sigma_{1}}{E_{1}|_{\topsp{Y}}}\xrightarrow{\sigma_{2}}\dots\xrightarrow{\sigma_{\it n}}{E_{\it n}}|_{\topsp{Y}}\to{0}}
\end{displaymath}
is an exact sequence of vector bundles.\\
We will say that two such elements $\rm \textbf{E}$ and $\rm \textbf{E}^{'}$ are \emph{isomorphic} if there are bundle isomorphisms $\varphi_{i}:\rm E_{\it i}\to{E_{\it i}^{'}}$ over $\topsp{X}$ such that the diagram
\begin{displaymath}
\xymatrix{\rm E_{{\it i}-1}|_{\topsp{Y}}\ar[r]^-{\sigma_{\it i}}\ar[d]_{\varphi_{{\it i}-1}} & \rm E_{\it i}|_{\topsp{Y}}\ar[d]^{\varphi_{\it i}}\\
{\rm E^{'}_{{\it i}-1}|_{\topsp{Y}}}\ar[r]^{\sigma^{'}_{\it i}} & \rm E^{'}_{\it i}|_{\topsp{Y}}
}
\end{displaymath}
commutes for every $i$.\\
Finally, an element $\rm \textbf{E}=(E_{0},E_{1},\cdots,E_{\it n};\sigma_{1},\sigma_{2},\cdots,\sigma_{\it n})$ is said to be \emph{elementary} if there is an $i$ such that
\begin{itemize}
\item[a)]$\text{E}_{i}=\text{E}_{i-1}$ and $\sigma_{i}=\text{id}$\\[-10mm]
\item[b)]$\text{E}_{j}=\{0\}$, for $j\neq{i}$ or $i-1$
\end{itemize}
\end{definition}
The Whitney sum $\oplus$ of vector bundles induces naturally an operation on $\mathcal{L}_{n}(\topsp{X},\topsp{Y})$. We define the equivalence relation $\sim$ on $\mathcal{L}_{n}(\topsp{X},\topsp{Y})$ generated by isomorphisms and and addition of elementary elements. Namely, we will say that two elements $\rm \textbf{E},\textbf{E}^{'}$ are equivalent if there are elementary elements $\textbf{P}_{1},\textbf{P}_{2},\ldots,\textbf{P}_{k},\textbf{Q}_{1},\textbf{Q}_{2},\ldots,\textbf{Q}_{l}$ in $\mathcal{L}_{n}(\topsp{X},\topsp{Y})$ such that
\begin{displaymath}
\textbf{E}\oplus\textbf{P}_{1}\oplus\dots\oplus\textbf{P}_{k}\simeq\textbf{E}^{'}\oplus\textbf{Q}_{1}\oplus\dots\oplus\textbf{Q}_{l}
\end{displaymath}
The set of all equivalences classes in $\mathcal{L}_{n}(\topsp{X},\topsp{Y})$ under $\sim$ is denoted by $\text{L}_{n}(\topsp{X},\topsp{Y})$, and is an abelian group under the operation $\oplus$. Moreover, if $\topsp{Y}=\O$, we will use the notation $\text{L}_{n}(\topsp{X})$.\\
Consider the natural map $\text{L}_{n}(\topsp{X},\topsp{Y})\to\text{L}_{n+1}(\topsp{X},\topsp{Y})$ which associates to the element $\rm (E_{0},E_{1},\cdots,E_{\it n};\sigma_{1},\sigma_{2},\cdots,\sigma_{\it n})$ the element $\rm(E_{0},E_{1},\cdots,E_{n},0;\sigma_{1},\sigma_{2},\cdots,\sigma_{n},0)$.\\
We refer to \cite{ABS} for the proof of the following fundamental result
\begin{proposition}
For each $n\geq{1}$, the induced map $\text{L}_{n}(\topsp{X},\topsp{Y})\to\text{L}_{n+1}(\topsp{X},\topsp{Y})$ is an isomorphism.
\end{proposition}
Hence we can focus on the group $\text{L}_{1}(\topsp{X},\topsp{Y})$, whose elements are given by triples $\rm \textbf{E}~=~[E_{0},E_{1};\sigma]$.\\
The \emph{difference bundle construction} allows to associate to any triple \textbf{E} an element $\chi(\textbf{E})$ in $\kgroup{0}{\topsp{X},\topsp{Y}}$ in the following way.\\
First, set $\topsp{X}_{k}=\topsp{X}\times\{k\}$, for $k=0,1$, and consider the space $\topsp{A}=\topsp{X}_{0}\cup_{\topsp{Y}}\topsp{X}_{1}$, obtained from the disjoint union $\topsp{X}_{0}\coprod\topsp{X}_{1}$ by identifying $y\times{0}$ with $y\times{1}$ for any $y\in\topsp{Y}$.\\
Notice that the map
\begin{displaymath}
\rho:\topsp{A}\to\topsp{X}_{1}
\end{displaymath}
is a retraction, i.e $i\circ\rho=id$, $i:\topsp{X}_{1}\hookrightarrow{\topsp{A}}$. In this case, the exact sequence (\ref{exactseq}) becomes the split short exact sequence \cite{Atiyah1967}
\begin{equation}\label{seq}
0\to\kgroup{0}{\topsp{A},\topsp{X}_{1}}\xrightarrow{}\kgroup{0}{\topsp{A}}\xrightarrow{i^{*}}\kgroup{0}{\topsp{X}_{1}}\to{0}
\end{equation}
Moreover, the relative map $(\topsp{X},\topsp{Y})\to(\topsp{A},\topsp{X}_{1})$ which identifies $\topsp{X}$ with $\topsp{X}_{0}$ induces an isomorphism $\varphi:\kgroup{0}{\topsp{A},\topsp{X}_{1}}\xrightarrow{\simeq}\kgroup{0}{\topsp{X},\topsp{Y}}$.\\
Now, from the element $\rm \textbf{E}=[E_{0},E_{1};\sigma]$ we construct (up to isomorphism) a vector bundle over $\topsp{A}$ by setting $\text{F}|_{\topsp{X}_{k}}:=\text{E}_{k}$, and using $\sigma$ to identify over $\topsp{Y}$. Recall, at this point, that the map $i^{*}$ is given essentially by restricting vector bundles from $\topsp{A}$ to $\topsp{X}_{1}$. Hence, setting $\rm F_{1}:=\rho^{*}E_{1}$, the class $\rm [F]-[F_{1}]$ is in $\rm ker({\it i}^{*})\subset{\kgroup{0}{\topsp{A}}}$. By (\ref{seq}), there exists a unique element $\chi(\textbf{E})\in\kgroup{0}{\topsp{X},\topsp{Y}}$ such that
\begin{displaymath}
\pi^{*}\varphi^{-1}\chi(\textbf{E})=\rm [F]-[F_{1}]
\end{displaymath}
In this way we have defined a homomorphism $\chi:\text{L}_{1}(\topsp{X},\topsp{Y})\to\kgroup{0}{\topsp{X},\topsp{Y}}$.\\
The following result allows the desired reformulation of relative K-theory \cite{ABS}
\begin{proposition}
The map $\chi$ is an isomorphism.
\end{proposition}
The above proposition is easy to prove in the case in which $\topsp{Y}=\O$. In this case, indeed, the map $\chi$ satisfies
\begin{displaymath}
\rm \chi([E_{0},E_{1}])=[E_{0}]-[E_{1}]
\end{displaymath}
The surjectivity of $\chi$ is obvious.\\
 Suppose that $\rm \chi([E_{0},E_{1}])=0$: then there exists a vector bundle $\rm G$ such that\\ $\rm E_{0}\oplus{G}\simeq{E_{1}\oplus{G}}$. Hence the element $\textbf{E}\oplus\textbf{G}$, where $\textbf{G}$ is the elementary sequence defined by G, is isomorphic to the elementary sequence defined by $\rm E_{1}\oplus{G}$, and hence it represents 0 in $\text{L}_{1}(\topsp{X})$.\\[2mm]
\indent With this reformulation of relative K-theory, we can describe the ABS isomorphism. Again, we refer the reader to {\appclifford} for the relevant notions of modules of Clifford algebras.\\
Let $\topsp{D}^{n}$ denote the unit disk in $\mathbb{R}^{n}$, and $\topsp{S}^{n-1}$ the boundary $\partial\topsp{D}^{n}$.\\
Let $\rm W=W^{0}\oplus{W^{1}}$ be a $\mathbb{Z}_{2}$-graded  module over the Clifford algebra $\Cl_{n}:=\Cl(\mathbb{R}^{n})$. To the module $\rm W$ we can associate the element
\begin{equation}\label{module}
\varphi(\rm W)=[E_{0},E_{1};\mu]\in\kgroup{0}{\topsp{D}^{\it n},\topsp{S}^{{\it n}-1}}
\end{equation}
where $\rm E_{\it k}:=D^{\it n}\times{W^{\it k}}$, and $\mu$ is the bundle isomorphism $\rm E_{0}\to{E_{1}}$ defined over $\topsp{S}^{n-1}$ by Clifford multiplication
\begin{displaymath}
\mu(x,v):=(x,x\cdot{v})
\end{displaymath}
As the element $\varphi(\rm W)$ depends only on the isomorphism class of W, and the map $\rm W\to\varphi(W)$ is an additive homomorphism, it follows that (\ref{module}) induces a homomorphism
\begin{equation}\label{quasiABS}
\varphi:\cliffmod{\mathbb{C}}_{n}\to\kgroup{0}{\topsp{D}^{n},\topsp{S}^{n-1}}
\end{equation}
where $\cliffmod{\mathbb{C}}_{n}$ is the abelian group freely generated by the irreducible complex graded $\Cl_{n}$-modules.\\
Consider now the homomorphism $i^{*}:\cliffmod{\mathbb{C}}_{n+1}\to\cliffmod{\mathbb{C}}_{n}$ induced by restricting the action from $\Cl_{n+1}$ to $\Cl_{n}$. Now, let W be a graded $\Cl_{n}$-module obtained from a $\Cl_{n+1}$-module by restriction. Then, the isomorphism $\mu$  defined on $\topsp{S}^{n-1}$ can be extended to all of $\topsp{D}^{n}$ by setting
\begin{displaymath}
\tilde{\mu}(x,v):=(x,(x+\sqrt{1-||x||^{2}}\:e_{n+1})\cdot{v})
\end{displaymath}
where $e_{n+1}\in\mathbb{R}^{n+1}$ is a unit vector orthogonal to $\mathbb{R}^{\it n}$. As $\rm E_{0}$ and $\rm E_{1}$ are isomorphic bundles over $\rm D^{\it n}$, the element $\rm [E_{0},E_{1};\mu]$ is 0 in $\kgroup{0}{\topsp{D}^{n},\topsp{S}^{n-1}}$.\\
Hence, the map (\ref{quasiABS}) descends to the homomorphism
\begin{equation}\label{ABSc}
\varphi_{n}:\cliffmod{\mathbb{C}}_{n}/i^{*}\cliffmod{\mathbb{C}}_{n+1}\to\kgroup{0}{\topsp{D}^{n},\topsp{S}^{n-1}}
\end{equation}
In complete analogy, in the real case we have the homomorphism
\begin{equation}\label{ABSr}
\varphi^{\mathbb{R}}_{n}:\cliffmod{}_{n}/i^{*}\cliffmod{}_{n+1}\to\kogroup{0}{\topsp{D}^{n},\topsp{S}^{n-1}}
\end{equation}
Recalling that
\begin{displaymath}
\kgroup{0}{\topsp{D}^{n},\topsp{S}^{n-1}}=\redk{0}{\topsp{D}^{n}/\topsp{S}^{n-1}}\simeq\redk{0}{\topsp{S}^{n}}
\end{displaymath}
Denote with
\begin{displaymath}
\cliffmod{\mathbb{C}}_{*}/i^{*}\cliffmod{\mathbb{C}}_{*+1}:=\bigoplus_{n\geq{0}}\cliffmod{\mathbb{C}}_{n}/i^{*}\cliffmod{\mathbb{C}}_{n+1}
\end{displaymath}
the graded ring induced by the graded tensor product of Clifford modules, and the same for the real case.\\
Finally, we can state the following fundamental result in K-theory
\begin{theorem}\textit{\textbf{Atiyah-Bott-Shapiro Isomorphisms.}} The maps $\varphi_{n}$ and $\varphi^{\mathbb{R}}_{n}$ induce graded ring isomorphisms
\begin{displaymath}
\begin{array}{l}
\varphi_{*}:\cliffmod{\mathbb{C}}_{*}/i^{*}\cliffmod{\mathbb{C}}_{*+1}\to\kgroup{-*}{\{\rm pt\}}\\
\varphi^{\mathbb{R}}_{*}:\cliffmod{}_{*}/i^{*}\cliffmod{}_{*+1}\to\kogroup{-*}{\{\rm pt\}}
\end{array}
\end{displaymath}
\end{theorem}
As the periodicity of the quotients $\cliffmod{\mathbb{C}}_{*}/i^{*}\cliffmod{\mathbb{C}}_{*+1}$ and $\cliffmod{}_{*}/i^{*}\cliffmod{}_{*+1}$ can be deduced by the algebraic properties of Clifford algebras, it may seem that the above theorem gives an algebraic proof of the Bott periodicity theorems. This is not the case, as the proof of the above theorem in \cite{ABS} actually invokes the Bott result.\\[2mm]
\indent The ABS isomorphisms can now be used to obtain explicit representatives for $\kgroup{-*}{\{\rm pt\}}$ and $\kogroup{-*}{\{\rm pt\}}$. Namely, consider $\rm S_{\mathbb{C}}=S_{\mathbb{C}}^{+}\oplus{S_{\mathbb{C}}^{-}}$, the fundamental $\mathbb{Z}_{2}$-graded complex representation of $\Cl_{2n}$. Then, by fundamental results in the theory of Clifford algebras, the group $\cliffmod{\mathbb{C}}_{2n}\simeq{\mathbb{Z}\oplus\mathbb{Z}}$ is generated by $\rm S_{\mathbb{C}}$ and $\rm \tilde{S}_{\mathbb{C}}$, the graded module obtained by interchanging the factors in $\rm S_{\mathbb{C}}$.\\
Moreover, the homomorphism $i^{*}$ maps the generator of the group $\cliffmod{\mathbb{C}}_{2n+1}\simeq\mathbb{Z}$ to $(\rm S_{\mathbb{C}},\tilde{S}_{\mathbb{C}})\in\mathbb{Z}\oplus\mathbb{Z}$. Hence, the generator for $\kgroup{-2n}{\{{\rm pt}\}}\simeq\kgroup{0}{\topsp{D}^{2n},\topsp{S}^{{2n}-1}}$ is given by the element
\begin{displaymath}
\sigma^{\mathbb{C}}_{2n}:=[\rm S^{+}_{\mathbb{C}},{S}^{-}_{\mathbb{C}};\mu]
\end{displaymath}
In the case $n=1$, $\sigma^{\mathbb{C}}_{2}$ is mapped by the isomorphism $\kgroup{0}{\topsp{D}^{2},\topsp{S}^{1}}\simeq\redk{0}{\rm S^{2}}$ to the class $[\rm H]-[\theta_{1}]$, where $\rm H$ is the tautological complex line bundle over $\rm S^{2}\simeq{\mathbb{CP}^{1}}$.\\[2mm]
\indent Due to the presence of torsion, the generators for the ring $\kogroup{-*}{\{\rm pt\}}$ are more complicated to describe. As an example, we will consider the case $n=1$, where we have $\kogroup{-1}{\{\rm pt\}}\simeq\mathbb{Z}_{2}$.\\
First, recall that to any graded module $\rm W=W^{0}\oplus{W^{1}}$ for $\Cl_{n}$ we can assign the \emph{ungraded} module $\rm W^{0}$ for the Clifford algebra $\Cl^{0}_{n}\simeq{\Cl_{n-1}}$, where $\Cl^{0}_{n}$ denotes the even part of $\Cl_{n}$. The converse is also true: given an ungraded module $\rm W^{0}$ for the Clifford algebra $\Cl_{n-1}$, the module
\begin{displaymath}
\rm W:=\Cl_{\it n}\otimes_{\Cl^{0}_{\it n}}W^{0}
\end{displaymath}
is naturally a graded module for the Clifford algebra $\Cl_{n}$.\\
This induces the isomorphism $\cliffmod{}_{n}\simeq{\mathfrak{M}_{n-1}}$, where $\mathfrak{M}_{n-1}$ denotes the ungraded version of $\cliffmod{}_{n-1}$, and consequently the isomorphism 
\begin{displaymath}
\cliffmod{}_{n}/i^{*}\cliffmod{}_{n+1}\simeq{\mathfrak{M}_{n-1}}/i^{*}\mathfrak{M}_{n}
\end{displaymath}
In the case $n=1$, $\Cl_{0}\simeq\mathbb{R}$ and $\Cl_{1}\simeq{\mathbb{C}}$: taking the real and complex dimension of the ungraded modules gives the isomorphisms $\mathfrak{M}_{0}\simeq{\mathbb{Z}}$ and $\mathfrak{M}_{1}\simeq{\mathbb{Z}}$, respectively. The map $i^{*}:\mathfrak{M}_{1}\to\mathfrak{M}_{0}$ is given by considering a complex vector space to be a real vector space under restriction of scalars. It induces a homomorphism $\mathbb{Z}\to\mathbb{Z}$ given by multiplication by 2, hence ${\mathfrak{M}_{0}}/i^{*}\mathfrak{M}_{1}\simeq{\mathbb{Z}_{2}}$. Then, the generator for $\kogroup{-1}{\{\rm pt\}}$ is given by the element in $\kogroup{0}{\topsp{D}^{1},\topsp{S}^{0}}$
\begin{displaymath}
\sigma_{1}:=[\topsp{D}^{1}\times\mathbb{R},\topsp{D}^{1}\times\mathbb{R};\mu]
\end{displaymath}
where
\begin{displaymath}
\begin{array}{c}
\mu(1,v)=v\\
\mu(-1,v)=-v
\end{array}
\end{displaymath}
Moreover, $\sigma_{1}=[\rm H_{1}]-[\theta_{1}]\in\redko{0}{\topsp{S}^{1}}$, where $\rm H_{1}$ is the \emph{infinite M\"{o}bius bundle} over the circle, i.e. the tautological real line bundle over $\mathbb{RP}^{1}\simeq{\topsp{S}^{1}}$.
\section{K-theory and $\text{Spin}^{\text{c}}$ manifolds}\label{Kmanifold}
In this section, we will specialize to the case in which the CW-complex $\topsp{X}$ is a smooth finite dimensional manifold. Being a topological invariant, K-theory does not depend on the presence of any smooth structure: neverthless, restricting to manifolds provides a strong connection between K-theory and the theory of elliptic operators, and, moreover, it is the right framework for the description of D-branes in String theory. Again, all the following results for which a smooth structure is not explicitly required, are to be considered valid in a more general context.\\[2mm]
\indent Before introducing the notions of K-orientation and the Thom isomorphism, we need some basic preliminary results.\\
In the following we will consider the K-theory of total spaces of vector bundles, which are locally compact spaces. For a locally compact space $\topsp{X}$ we define
\begin{displaymath}
\ckgroup{0}{\topsp{X}}:=\redk{0}{\topsp{X}^{+}}
\end{displaymath}
and
\begin{displaymath}
\ckgroup{-n}{\topsp{X}}:=\ckgroup{0}{\topsp{X}\times\mathbb{R}^{n}}
\end{displaymath}
The functors $\text{K}_{cpt}^{-n}$ are defined on the category of locally compact spaces and proper maps, and they constitute the \emph{K-theory with compact support}.\\
Moreover, we have that
\begin{displaymath}
\ckgroup{-n}{\topsp{X},\topsp{Y}}:=\ckgroup{0}{(\topsp{X}-\topsp{Y})\times\mathbb{R}^{n}}
\end{displaymath}
Analogous definitions can be given for $\text{KO}$-theory, and the Bott periodicity theorems assume the form
\begin{displaymath}
\begin{array}{c}
\ckgroup{0}{\topsp{X}}\simeq\ckgroup{0}{\topsp{X}\times\mathbb{C}}\\
\ckogroup{0}{\topsp{X}}\simeq\ckogroup{0}{\topsp{X}\times\mathbb{R}^{8}}
\end{array}
\end{displaymath}
Now, let $\topsp{X}$ be a compact space, and let $\rm E\to{\topsp{X}}$ be a complex vector bundle. The one point compactification $\topsp{E}^{+}$ is called the \emph{Thom space}, or \emph{Thom complex} of $\rm E$.\\
By definition, $\redk{0}{\topsp{E}^{+}}=\ckgroup{0}{\topsp{E}}$.\\
Consider the (unit) \emph{ball} and \emph{sphere} bundle of E, denoted as $\rm B(E)$ and $\rm S(E)$, respectively, and defined as
\begin{displaymath}
\begin{array}{c}
\rm B(E):=\left\{{\it e}\in{E}: \varphi({\it e},{\it e})\leq{1}\right\}\\
\rm S(E):=\left\{{\it e}\in{E}: \varphi({\it e},{\it e})={1}\right\}
\end{array}
\end{displaymath}
for $\varphi$ a Hermitian metric on $\rm E$, and the projection is the obvious one.\\
By noticing that $\rm B(E)-S(E)\simeq{E}$, we have
\begin{displaymath}
\rm B(E)/S(E)\simeq{E^{+}}
\end{displaymath}
and consequently
\begin{displaymath}
\ckgroup{0}{\rm E}:=\redk{0}{\rm E^{+}}\simeq{\redk{0}{\rm B(E)/S(E)}}:=\kgroup{0}{B(E),S(E)}
\end{displaymath}
The isomorphism above does not depend on the choice of the Hermitian metric $\varphi$, and a similar construction follows for $\topsp{KO}$-theory.\\
In the following, we will often interchange $\ckgroup{0}{\rm E}$ and $\rm \kgroup{0}{B(E),S(E)}$ freely using the isomorphism above.
\subsection{K-orientation and Thom isomorphism}\label{Thomsection} 
In the same way that an orientation of a vector bundle $\rm E\to{X}$ induces a class $\tau\in\text{H}_{cpt}^{*}(\text{E})$ and a ring isomorphism $\text{H}_{cpt}^{*}(E)\simeq{\text{H}^{*}(\topsp{X})}$ via cup product with $\tau$, the existence of certain structures on vector bundles will induce analogous results for $\rm{K}$- and $\rm KO$-theory.\\
Let $\pi:\text{E}\to{\topsp{X}}$ be a complex or real vector bundle, and denote with $\rm k^{*}$ the functor $\text{K}_{cpt}^{*}$ or $\text{KO}_{cpt}^{*}$, or, more generally, any multiplicative generalized cohomology theory.\\
Define the map
\begin{displaymath}
j^{*}:\text{k}^{*}(\topsp{X}\times\text{E})\to\text{k}^{*}(\text{E})
\end{displaymath}
induced by the map $j(e)=(\pi(e),e)$.\\
 Then, the composition
\begin{displaymath}
\text{k}^{*}(\topsp{X})\times\text{k}^{*}(\topsp{E})\xrightarrow{\cup}\text{k}^{*}(\topsp{X}\times\topsp{E})\xrightarrow{j^{*}}{\text{k}^{*}(\topsp{E})}
\end{displaymath}
canonically equips $\text{k}^{*}(\topsp{E})$ with a structure of left module over $\text{k}^{*}(\topsp{X})$.\\
In the following we will give a somewhat general definition of orientation for $\text{k}^{*}$, as it will become important later in the context of K-homology.
\begin{definition}
An $n$-dimensional vector bundle $\rm\pi:E\to{X}$ is said to be \emph{$\text{k}^{*}$-orientable} if there exists a class $\tau\in\text{k}^{n}(\topsp{E})$ such that $i_{x}^{*}\tau$ is a generator of $\text{k}^{n}(\topsp{E}_{x})$ for each fiber inclusion $i_{x}:\topsp{E}_{x}\hookrightarrow\topsp{E}$. A choice of such a class $\tau$ is called a \emph{$\text{k}^{*}$-orientation} for $\pi$, and $\pi$ is said to be \emph{$\text{k}^{*}$-oriented} by $\tau$.   
\end{definition}
For a $\text{k}^{*}$-oriented vector bundle one has the following fundamental general result \cite{dyer,switzer}
\begin{theorem}\label{Thomiso}
Let $\pi:\topsp{E}\to\topsp{X}$ be an n-dimensional vector bundle which is ${\rm k}^{*}$-oriented by the class $\tau\in{\rm k}^{n}(\topsp{E})$. Then the homomorphism
\begin{displaymath}
\mathfrak{T}_{\topsp{X},\topsp{E}}:{\rm k}^{i}(\topsp{X})\to{\rm k}^{i+n}(\topsp{E})
\end{displaymath}
defined by
\begin{displaymath}
\mathfrak{T}_{\topsp{X},\topsp{E}}(\xi):=\pi^{*}(\xi)\cup{\tau}
\end{displaymath}
is an isomorphism.
\end{theorem}
The isomorphism $\mathfrak{T}_{\topsp{X},\topsp{E}}$ is called the \emph{Thom isomorphism}. As a corollary of theorem (\ref{Thomiso}), we have that $\text{k}^{*}(\topsp{E})$ is a free $\text{k}^{*}$-module with generator $\tau$, and hence is completely known once $\text{k}^{*}(\topsp{X})$ is.\\[2mm]
\indent Given a vector bundle, one is interested in finding \emph{sufficient}, if not necessary conditions for orientation classes to exist, and, more generally, to investigate which are the possible obstructions to such an existence.\\ 
For this, consider the trivial bundle 
\begin{displaymath}
\topsp{E}=\topsp{X}\times{\mathbb{C}^{m}\xrightarrow{\pi}}\topsp{X}
\end{displaymath}
and consider the class
\begin{displaymath}
\tau(\topsp{E}):=[\Lambda_{\mathbb{C}}^{even}\underline{\mathbb{C}}^{m},\Lambda_{\mathbb{C}}^{odd}\underline{\mathbb{C}}^{m};\sigma]\in\kgroup{0}{\topsp{B(E)},\topsp{S(E)}}
\end{displaymath}
where $\underline{\mathbb{C}}^{m}=\pi^{*}E$, and where
\begin{displaymath}
\sigma_{(x,v)}(\varphi):=v\wedge\varphi-i_{v^{*}}{\varphi}
\end{displaymath}
with $i_{v^{*}}$ given by contraction. See {\appclifford} for details.\\
Now, the restriction of $\tau(\topsp{E})$ to the fiber is given by the class
\begin{displaymath}
[\Lambda_{\mathbb{C}}^{even}\mathbb{C}^{m},\Lambda_{\mathbb{C}}^{odd}\mathbb{C}^{m};\sigma]\in\ckgroup{0}{\mathbb{C}^{m}}=\kgroup{0}{\topsp{D}^{2m},\topsp{S}^{2m-1}}
\end{displaymath}
where we have identified $\mathbb{C}^{m}\simeq{\mathbb{R}^{2m}}$.\\
Upon the identification $\Lambda^{*}\mathbb{C}^{m}\simeq\Cl_{2m}$, and noticing that $\sigma$ is given by Clifford multiplication by the element $v$, the class above can be rewritten as
\begin{displaymath}
[ \rm S^{+}_{\mathbb{C}}, S^{-}_{\mathbb{C}};\mu]
\end{displaymath}
which by the ABS isomorphism is a generator of $\kgroup{0}{\topsp{D}^{2m},\topsp{S}^{2m-1}}$.\\
As any complex vector bundle $\topsp{E}$ is locally trivializable, and there is no obstruction to construct $\Lambda_{\mathbb{C}}^{*}\topsp{E}$ globally, we have the following
\begin{theorem}
Let $\rm \pi:\topsp{E}\to\topsp{X}$ be a complex hermitian vector bundle over a compact space $\topsp{X}$. Then the class
\begin{displaymath}
\tau_{\text{K}}(\topsp{E}):=[\Lambda_{\mathbb{C}}^{even}\pi^{*}\topsp{E},\Lambda_{\mathbb{C}}^{odd}\pi^{*}\topsp{E};\sigma]\in\ckgroup{0}{\topsp{E}}
\end{displaymath}
where $\sigma_{e}(\varphi):=e\wedge\varphi-i_{e^{*}}{\varphi}$, is a $\rm{K}^{*}$-orientation class for $\pi$.
\end{theorem}
Hence, any complex vector bundle is $\text{K}^{*}$-orientable, and there is no obstruction present.\\
The case for real vector bundles is instead different. Indeed, consider the 8$n$-dimensional real trivial vector bundle
\begin{displaymath}
\rm E=X\times{\mathbb{R}^{8{\it n}}}\xrightarrow{\pi}\topsp{X}
\end{displaymath}
and consider the class
\begin{displaymath}
\tau_{KO}(E):=\rm [\pi^{*}S^{+}(E),\pi^{*}S^{-}(E);\mu]\in{\ckogroup{0}{\topsp{E}}}
\end{displaymath}
where $\rm S^{+}(E)\oplus{S^{-}(E)}=S(E)=\topsp{X}\times{W}$, with $\rm W$ the irreducible real graded $\Cl_{8n}$-module, and $\mu$ is Clifford multiplication. Again, by Bott periodicity, the class $\tau_{KO}(E)$ gives the generator of $\ckogroup{0}{\mathbb{R}^{8n}}=\kogroup{0}{\topsp{D}^{8n},\topsp{S}^{8n-1}}$, when restricted to the fiber.\\
Given an arbitrary real vector bundle $\topsp{E}$, the bundle $\topsp{S(E)}$ will exist if and only if $\rm E$ admits a spin structure.\\
Hence, we have the following
\begin{theorem}
Let $\rm \pi:\topsp{E}\to\topsp{X}$ be a real 8n-dimensional vector bundle with a spin structure over a compact space $\topsp{X}$. Then the class
\begin{displaymath}
\tau_{KO}(\rm E):=\rm [\pi^{*}S^{+}(E),\pi^{*}S^{-}(E);\mu]\in{\ckogroup{0}{\topsp{E}}}
\end{displaymath}
where $\mu_{e}(\varphi)=e\cdot{\varphi}$ is Clifford multiplication, is a $\rm KO^{*}$-orientation for $\pi$.
\end{theorem}
In the same way, the following theorem gives sufficient conditions for a real vector bundle to be $\topsp{K}^{*}$-orientable.
\begin{theorem}\label{orientation}
Let $\rm \pi:\topsp{E}\to\topsp{X}$ a real 2n-dimensional vector bundle with a $spin^{c}$ structure over a compact space $\topsp{X}$. Then the class
\begin{displaymath}
\tau_{K}(\rm E):=\rm [\pi^{*}S_{\mathbb{C}}^{+}(E),\pi^{*}S_{\mathbb{C}}^{-}(E);\mu]\in{\ckgroup{0}{\topsp{E}}}
\end{displaymath}
where $\rm S_{\mathbb{C}}=S^{+}_{\mathbb{C}}\oplus S^{-}_{\mathbb{C}}$ is the irreducible spinor bundle associated to the $spin^{c}$ structure on $\topsp{E}$, and $\mu_{e}(\varphi)=e\cdot{\varphi}$ is Clifford multiplication. Then $\tau_{K}(\rm E)$ is a $\topsp{K}^{*}$-orientation for $\pi$.
\end{theorem}
Thanks to the results above, we can construct an orientation class for a $\rm spin^{c}$ vector bundle of arbitrary rank in the following way, with the spin case requiring minor modifications.\\
Let $\pi:\topsp{E}\to\topsp{X}$ be a $\rm spin^{c}$ vector bundle of rank $n$, and consider the Whitney sum $\topsp{E}\oplus{\theta_{p}}$, with $p$ such that $n+p=2m$.\\
If we equip $\theta_{p}$ with its canonical  $spin^{c}$ structure, then the sum $\topsp{E}\oplus{\theta_{p}}$ is a $\rm spin^{c}$ vector bundle of rank $2m$.\\
Hence, by Theorem \ref{orientation} there exists a Thom class
\begin{displaymath}
\tau_{K}(\rm E\oplus\theta_{\it p})\in\kgroup{0}{B(\topsp{E}\oplus{\theta_{\it p}}),S(\topsp{E}\oplus{\theta_{\it p}})}
\end{displaymath}
By using the following isomorphisms \cite{dyer,switzer}
\begin{displaymath}
\begin{array}{rl}
\rm \kgroup{0}{B(\topsp{E}\oplus{\theta_{\it p}}),S(\topsp{E}\oplus{\theta_{\it p}})}:=&\rm \redk{0}{B(\topsp{E}\oplus{\theta_{\it p}})/S(\topsp{E}\oplus{\theta_{\it p}})}\\
\simeq&\rm \redk{0}{\Sigma^{\it p}(B(\topsp{E})/S(\topsp{E})}\\
\simeq&\rm \kgroup{-{\it p}}{B(\topsp{E}),S(\topsp{E})}\\
\simeq&\rm \kgroup{\it n}{B(\topsp{E}),S(\topsp{E})} 
\end{array}
\end{displaymath}
it follows that the class $\tau_{K}(\rm E\oplus\theta_{p})$ is an orientation class for the vector bundle $\pi$.\\
It is important to notice that, in contrast to the results in Theorem \ref{orientation}, the construction of the Thom class above is \emph{not} natural, as we have to make several choices in the process, starting from that of a $\rm spin^{c}$ structure on $\rm E\oplus\theta_{\it p}$. Neverthless, the induced Thom isomorphism will be an essential ingredient in the next chapter, where we will discuss a more natural description of D-branes in terms of K-homology.  
\subsection{Chern Character and Gysin homomorphism}\label{gysinhomo} 
In this section we will briefly introduce the homomorphisms induced by the Chern character and the Thom isomorphism. We refer to {\appcharac} for the basic notions on characteristic classes for vector bundles.\\[2mm]
\indent Let $\topsp{E}\to\topsp{X}$ be a complex vector bundle. The Chern character $\ch(\topsp{E})$ is an element in the rational cohomology $\topsp{H}^{ev}(\topsp{X};\mathbb{Q})$ constructed from the total Chern class of $\topsp{E}$. The importance of the Chern character lies in the fact that it respects the (semi)additive and multiplicative structure on $\text{Vect}(\topsp{X})$. Namely, for $\topsp{E}$ and $\topsp{E}^{'}$ vector bundles on $\topsp{X}$ we have
\begin{displaymath}
\begin{array}{c}
\ch(\topsp{E}\oplus\topsp{E}^{'})=\ch(\topsp{E})+\ch(\topsp{E}^{'})\\
\ch(\topsp{E}\otimes\topsp{E}^{'})=\ch(\topsp{E})\cup\ch(\topsp{E}^{'})
\end{array}
\end{displaymath}    
Hence, we can define the following homomorphism
\begin{equation}\label{Chernhomo}
\begin{array}{rl}
\ch:&\kgroup{0}{\topsp{X}}\to\topsp{H}^{ev}(\topsp{X};\mathbb{Q})\\
&[\topsp{E}]-[\topsp{F}]\to\ch(\topsp{E})-\ch(\topsp{F})
\end{array}
\end{equation}
It is easy to show that the Chern character homomorphism (\ref{Chernhomo}) for a class $x$ does not depend on its representation in terms of vector bundles, hence it is a well defined ring homomorphism.\\
The Chern character homomorphism can be extended to a group homomorphism
\begin{displaymath}
\ch:\kgroup{-n}{\topsp{X},\topsp{Y}}\to\topsp{H}^{*}(\topsp{X},\topsp{Y};\mathbb{Q})
\end{displaymath}
by defining
\begin{equation}\label{cherncar}
\ch(x):=\alpha((\sigma^{n})^{-1}{\rm ch}(x'))\quad{x\in}\kgroup{-n}{\topsp{X},\topsp{Y}}
\end{equation}
where $x'$ is the class corresponding to $x$ in $\redk{0}{\Sigma^{n}(\topsp{X}/\topsp{Y})}$, $\sigma^{n}$ is the suspension isomorphism in cohomology, and $\alpha$ is the canonical isomorphism    
\begin{displaymath}
\alpha:\widetilde{\topsp{H}}^{*}(\topsp{X}/\topsp{Y};\mathbb{Q})\xrightarrow{\simeq}\topsp{H}^{*}(\topsp{X},\topsp{Y};\mathbb{Q})
\end{displaymath}
Moreover, one can prove that the homomorphism (\ref{cherncar}) is compatible with Bott periodicity: this allows to define $\ch(x)$ for $x\in\kgroup{n}{\topsp{X},\topsp{Y}}$, where $n$ is any integer.\\
Finally, by using a spectral sequence argument, the following fundamental result was proven in \cite{Atiyah-Hirzebruch}
\begin{theorem}\label{chern}
Let $\topsp{X}$ be a finite CW-complex. Then the homomorphism
\begin{displaymath}
\ch\otimes{id}:\kgroup{*}{\topsp{X}}\otimes\mathbb{Q}\to\topsp{H}^{*}(\topsp{X};\mathbb{Q})
\end{displaymath}
is an isomorphism and maps $\kgroup{0}{\topsp{X}}\otimes\mathbb{Q}$ onto $\topsp{H}^{ev}(\topsp{X};\mathbb{Q})$ and $\kgroup{1}{\topsp{X}}\otimes\mathbb{Q}$ onto $\topsp{H}^{odd}(\topsp{X};\mathbb{Q})$.
\end{theorem}  
As a corollary of Theorem \ref{chern}, when $\topsp{X}$ is a finite CW-complex, $\kgroup{*}{\topsp{X}}$ is a finitely generated abelian group. This technical simplification in a sense justifies the restriction to CW-complexes.\\[2mm]
\indent In the case in which $\topsp{X}$ is a smooth manifold, the Chern character homomorphism is closely related to the Gysin homomorphism, which we will illustrate in the following.\\
Consider two locally compact manifolds $\topsp{X}$ and $\topsp{Y}$, such that ${\rm dim}\topsp{Y}-{\rm dim}\topsp{X}=0\:\text{mod}\:2$, and let $f:\topsp{X}\to\topsp{Y}$ be a proper embedding\footnote{The embedding condition can be relaxed. See \cite{KAROUBI}.}. Denote with $\nu_{\topsp{X}}$ the normal bundle of $\topsp{X}$ in $\topsp{Y}$, and suppose it is equipped with a $\rm spin^{c}$ structure, which is a restriction on the map $f$. Moreover, identify this normal bundle with a tubular neighbourhood $\topsp{N}$ of $\topsp{X}$ in $\topsp{Y}$. Then, we can define the \emph{Gysin homomorphism} 
\begin{displaymath}
f_{*}:\ckgroup{0}{\topsp{X}}\to\ckgroup{0}{\topsp{Y}}
\end{displaymath}
as the composition
\begin{displaymath}
\ckgroup{0}{\topsp{X}}\xrightarrow{\mathfrak{T}_{\topsp{X},\nu_{\topsp{X}}}}\ckgroup{0}{\nu_{\topsp{X}}}\simeq\ckgroup{0}{\topsp{N}}\xrightarrow{j^{*}}{\ckgroup{0}{\topsp{Y}}}
\end{displaymath}
where the map $j^{*}$ is induced by the map $j:\topsp{Y}^{+}\to\topsp{N}^{+}$ defined as
\begin{displaymath}
\begin{array}{lr}
j(z)&=z,\:\forall\:z\in\topsp{N}\\
j(z)&=\infty,\forall\:z\notin\topsp{N}
\end{array}
\end{displaymath}
with $\infty$ denoting the point that compactifies $\topsp{N}$. Intuitively, $j^{*}$ extends a class in $\ckgroup{0}{\topsp{N}}$ via trivial bundles, and is usually referred to as ``extension by zero''.\\
The Gysin homomorphism is also called the ``wrong way'' morphism, as it ``goes'' in the direction opposite to the contravariance of the K-functor.\\
The Gysin homomorphism does not depend on the tubular neighbourhood $\topsp{N}$, and only depends on the homotopy class of $f$, as the vector bundle $\nu_{\topsp{X}}$ is defined as $f^{*}(\topsp{TY})/\topsp{TX}$, and the K-functor is homotopy invariant.\\
Moreover, the Gysin homomorphism enjoys the following functoriality property, which will be very useful in later chapters.\\
Let $f:\topsp{X}\to\topsp{Y}$ and $g:\topsp{Y}\to\topsp{Z}$ be proper embeddings satisying the hypothesis above. Then
\begin{equation}
(g\circ{f})_{*}=g_{*}f_{*}
\end{equation}
As the Thom isomorphism can be extended to higher K-groups, the Gysin homomorphism can be extended to a homomorphism 
\begin{displaymath}
f_{*}:\ckgroup{-1}{\topsp{X}}\to\ckgroup{-1}{\topsp{Y}}
\end{displaymath}
Moreover, by recalling the discussion on orientation classes in section \ref{Thomsection}, we can relax the condition that ${\rm dim}\topsp{Y}-{\rm dim}\topsp{X}=0\:\text{mod}\:2$. Indeed, in general, we will have a Gysin homomorphism 
\begin{displaymath}
f_{*}:\ckgroup{i}{\topsp{X}}\to\ckgroup{i+n}{\topsp{Y}}
\end{displaymath}
where $n$ is the rank of the normal bundle. Again, the price we pay is that the homomorphism $f_{*}$ so obtained does not induce a natural transformation.\\
Finally, along the same lines, one can define a Gysin homomorphism in ordinary cohomology
\begin{displaymath}
f^{\topsp{H}}_{*}:\topsp{H}^{*}(\topsp{X};\mathbb{Q})\to\topsp{H}^{*}(\topsp{Y};\mathbb{Q})
\end{displaymath}
\indent As mentioned before, the Chern character homomorphism and the Gysin homomorphism are closely related via the Atiyah-Hirzebruch version of the Riemann-Roch theorem \cite{RiemRoch}.\\
Consider two locally compact manifolds $\topsp{X}$ and $\topsp{Y}$ and a proper embedding $f:\topsp{X}\to\topsp{Y}$ as above, and denote with $d(\nu_{X})$ the cohomology class defining the $\rm spin^{c}$ structure on the normal bundle $\nu_{\topsp{X}}$. Then we have the following
\begin{theorem}
(\textbf{Riemann-Roch}) For each class $x\in\kgroup{0}{\topsp{X}}$ we have the relation
\begin{displaymath}
\ch(f_{*}(x))=f^{\topsp{H}}_{*}(\e^{d(\nu_{X})/2}\rgenus{\nu_{\topsp{X}}}\:\ch(x))
\end{displaymath}
\end{theorem}
This form of the Riemann-Roch theorem will be used in the next section, where we will see how the K-theoretic machinery developed so far is related to D-branes in String theory. 
\section{K-theory and type IIA/B D-branes}\label{typeIIbranes} 
Having developed the necessary notions in the previous sections, we will describe how K-theory is related to D-branes in superstring theory. In particular, in this section we will consider type II String theory on a 10-dimensional spin manifold $\topsp{M}$, while in the next section we will investigate type I String theory.\\
We will present two supporting arguments for the K-theoretic description of D-branes, based on the fact that D-branes are currents for Ramond-Ramond fields, and on Sen's conjectures, respectively.\\
We first need some comments on the anomalous coupling (\ref{anomalous}) between the total Ramond-Ramond potential and a D-brane. Recall that in section \ref{super} we made the assumption that the worldvolume $\topsp{Q}$ wrapped by the D-brane is a spin manifold, pointing out that this condition is \emph{not} necessary from a string theoretic point of view. In presence of a D-brane the partition function of the superstring develops an anomaly, called the \emph{Freed-Witten anomaly} \cite{Freed1999}: in a topologically trivial B-field setting, this anomaly cancels exactly when the normal bundle of $\topsp{Q}$ in $\topsp{M}$ admits a $\rm spin^{c}$ structure. Since $\topsp{M}$ is a spin manifold, the manifold $\topsp{Q}$ necessarily admits a $\rm spin^{c}$ structure.\\
Moreover, the anomalous coupling is modified as       
\begin{equation}\label{anom2}
\int_{\man{Q}}i^{*}\text{C}\wedge\e^{d(\nu_{Q})/2}\text{ch}(\text{E}){i^{*}\sqrt{\rgenus{\man{TM}}}}\dfrac{1}{\rgenus{\man{TQ}}}
\end{equation}
which coincides with the coupling (\ref{anomalous}) when $\topsp{Q}$  is spin manifold, as in this case the class $d(\nu_{Q})$ vanishes.\\[2mm]
Consider the cohomology class 
\begin{displaymath}
j_{\topsp{Q}}=\e^{d(\nu_{Q})/2}\text{ch}(\text{E}){i^{*}\sqrt{\rgenus{\man{TM}}}}\dfrac{1}{\rgenus{\man{TQ}}}
\end{displaymath}
As we have discussed in section \ref{generalized}, the class $j_{\topsp{Q}}$ represents the charge associated to the Ramond-Ramond current generated by the D-brane wrapping $\topsp{Q}$. More precisely, for such an interpretation to be valid, the class $j_{\topsp{Q}}$ needs to be ``pushed'' into the spacetime manifold $\topsp{M}$ as
\begin{equation}\label{push}
i^{\topsp{H}}_{*}(j_{\topsp{Q}})
\end{equation}
where we recall that $i:\topsp{Q}\to\topsp{M}$ is the embedding map.
By noticing that 
\begin{displaymath}
\nu_{\topsp{Q}}\simeq{i^{*}\topsp{TM}}/{\topsp{TQ}}
\end{displaymath}
and using the fact that the roof genus is a characteristic class, we have the identity
\begin{displaymath}
\rgenus{\nu_{\topsp{Q}}}=i^{*}\rgenus{\topsp{TM}}/\rgenus{\topsp{TQ}}
\end{displaymath}
Hence, we can rewrite (\ref{push}) as
\begin{displaymath}
i_{*}^{\topsp{H}}(\e^{d(\nu_{Q})/2}\rgenus{\nu_{\topsp{Q}}}\text{ch}(\text{E})\dfrac{1}{i^{*}\sqrt{\rgenus{\man{TM}}}})
\end{displaymath}
By using the following property of the cohomological Gysin homomorphism
\begin{displaymath}
i_{*}^{\topsp{H}}(\alpha\cup{i^{*}\beta})=i_{*}^{\topsp{H}}(\alpha)\cup\beta,\quad\forall{\alpha\in\topsp{H}^{*}(\topsp{Q};\mathbb{Q})},\forall{\beta\in\topsp{H}^{*}(\topsp{M};\mathbb{Q})}
\end{displaymath}
we have that
\begin{displaymath}
i_{*}^{\topsp{H}}(\e^{d(\nu_{Q})/2}\rgenus{\nu_{\topsp{Q}}}\text{ch}(\text{E})\dfrac{1}{i^{*}\sqrt{\rgenus{\man{TM}}}})=i_{*}^{\topsp{H}}(\e^{d(\nu_{Q})/2}\rgenus{\nu_{\topsp{Q}}}\text{ch}(\text{E}))\cup\dfrac{1}{\sqrt{\rgenus{\man{TM}}}}
\end{displaymath}
Finally, by using the Riemann-Roch theorem we have
\begin{equation}\label{charge}
i^{\topsp{H}}_{*}(j_{\topsp{Q}})=\ch{(i_{*}[\rm E])}\cup\dfrac{1}{\sqrt{\rgenus{\man{TM}}}}
\end{equation}
The identity (\ref{charge}) strongly suggests that the charge for a D-brane wrapping a worldvolume $\topsp{Q}$ equipped with a Chan-Paton bundle $\topsp{E}$ is naturally given by the element $i_{*}[\rm E]\in\kgroup{0}{\topsp{M}}$.\\
This argument was developed in \cite{Minasian1997}, and constitutes the first evidence for the importance of K-theory in String theory. However, the argument is strictly speaking only valid over $\mathbb{Q}$, where the map 
\begin{displaymath}
\ch\cup{\dfrac{1}{\sqrt{\rgenus{\man{TM}}}}}:\kgroup{*}{\topsp{M}}\otimes{\mathbb{Q}}\to\topsp{H}^{*}(\topsp{M};\mathbb{Q})  
\end{displaymath}
is a group isomorphism. In this case, then, the cohomological description of D-branes ``coincides'' with the K-theoretical one, with the latter having the advantage of being more natural.\\[2mm]
\indent An argument which relates the full K-theory groups to D-brane charges is based on Sen's conjectures, as discussed in section \ref{senconjecture}: of course, the price we pay consists in the fact that we have to invoke (open) string tachyon condensation, which is not yet fully understood, both from the physical and the mathematical point of view.\\
Recall from section \ref{senconjecture} that Witten made the observation that in type IIB String theory, a configuration of $n$ D9-branes and $n$ $\overline{\rm D9}$-branes $(\rm E,F)$ has to be considered equivalent to the configuration $(\rm E\oplus{H},F\oplus{H})$, because the brane-antibrane system $(\rm H,H)$ is able to decay in the string vacuum, as conjectured by Sen.\\
As we have seen in this chapter, the equivalence classes of configurations $(\rm E,F)$ under brane-antibrane annihilation are elements in $\kgroup{0}{\topsp{M}}$, or more precisely $\redk{0}{\topsp{M}}$, as the bundles $\rm E$ and $\rm F$ have the same rank.\\
Moreover, in \cite{Witten1998} Witten was able to K-theoretically interpret Sen's construction, which allows to obtain D$p$-branes with $p<9$ as the decay product of a system of brane-antibranes of higher dimension. According to this interpretation, the group of possible charges for a D$p$-brane wrapping a $\rm spin^{c}$ submanifold $\topsp{Q}\subset{\topsp{M}}$ is given by 
\begin{displaymath}
\ckgroup{0}{\topsp{N}}\simeq\kgroup{0}{\topsp{B(N)},\topsp{S(N)}}
\end{displaymath}
with $\rm N$ denoting the tubular neighbourhood identifying the normal bundle to $\topsp{Q}$ in $\topsp{M}$.\\
In particular, given a single D$p$-brane wrapping $\topsp{Q}$ with trivial Chan-Paton bundle, the system of D9-$\overline{\rm D9}$-brane decaying to $\topsp{Q}$ is given by
\begin{displaymath}\label{charge-decay}
i_{*}(1)\in\kgroup{0}{\topsp{M}}
\end{displaymath}
where $1$ denotes the class of the trivial line bundle over $\topsp{Q}$.\\
Applying the above machinery to a flat D$p$-brane wrapping $\mathbb{R}^{1,p}\subset\topsp{M}^{10}$, with $\topsp{M}^{10}$ the ten dimensional Minkowski spacetime, we obtain that the its charge group is given by
\begin{displaymath}
\ckgroup{0}{\topsp{B(N)},\topsp{S(N)}}\simeq\ckgroup{0}{\mathbb{R}^{1,p}\times\topsp{D}^{9-p},\mathbb{R}^{1,p}\times\topsp{S}^{9-p}}\simeq\redk{0}{\topsp{S}^{9-p}}
\end{displaymath}
where we have used the compact support relative K-theory, as $\mathbb{R}^{1,p}$ is not compact.\\
As $\redk{0}{\topsp{S}^{9-p}}=0$ for $p$ even, we find that the K-theoretical description agrees with, and in a certain sense justifies the statement  that odd-dimensional (flat) D-branes in type IIB String theory are necessarily unstable, as they cannot carry any D-brane charge.\\[2mm]
The discussion for type IIA is less straightforward. Indeed, the fact that Sen's construction works in even codimension let Witten propose to relate branes not to bundles on $\topsp{M}$, but to bundles on $\topsp{S}^{1}\times\topsp{M}$. Namely, given a D$p$-brane wrapping an odd-dimensional submanifold $\topsp{Q}\subset{\topsp{M}}$, one identifies $\topsp{Q}$ with ${\rm pt}\times{\topsp{Q}}$ in $\topsp{S}^{1}\times\topsp{M}$, where $\rm pt$ is any point in $\topsp{S}^{1}$. Applying Sen's construction, the D$p$-brane wrapping $\topsp{Q}$ determines an element in $\kgroup{0}{\topsp{S}^{1}\times\topsp{M}}$, which turns out to be trivial when restricted to $\rm pt\times\topsp{M}$.\\
By using the isomorphism  
\begin{displaymath}
\kgroup{-1}{\topsp{M}}\simeq\text{ker}\left[\kgroup{0}{\topsp{S}^{1}\times\topsp{M}}\to\kgroup{0}{\rm pt\times\topsp{M}}\right]
\end{displaymath} 
it follows that type IIA brane-antibrane configurations have to be considered equivalent if they determine the same element in $\kgroup{-1}{\topsp{M}}$. Of course, the construction above is not natural, as we have to make a choice of how to embed $\topsp{Q}$ in $\topsp{S}^{1}\times{\topsp{M}}$.\\
Similarly, the group of charges for a D$p$-brane wrapping $\topsp{Q}$ is given by
\begin{displaymath}
\ckgroup{-1}{\topsp{B(N)},\topsp{S(N)}}
\end{displaymath}
which agrees with the known results for flat D$p$-branes in Minkowski spacetime.\\
From our point of view, the appearence of the functor $\topsp{K}^{-1}$ can be motivated in the following way. As we have seen in section \ref{gysinhomo}, a Gysin homomorphism for a map $i:\topsp{X}\to\topsp{Y}$ such that the normal bundle admits a $\rm spin^{c}$ structure \emph{does} always exist. Hence, given a D-brane wrapping an odd codimension manifold $\topsp{Q}$ equipped with a trivial Chan-Paton bundle, its D-brane charge is given by the element
\begin{displaymath}
i_{*}(1)\in\ckgroup{-1}{\topsp{M}}
\end{displaymath}
The classes in $\ckgroup{-1}{\topsp{M}}$ do not represent a brane-antibrane system: instead, they can can be conveniently represented as a system of D9-branes equipped with the tachyon field causing the instability of D9-branes in Type IIA String theory. We refer the reader to \cite{horava-1999-2} for details on this construction.
\section{KO-theory and Type I D-branes: torsion effects}\label{typeIbranes}
As we mentioned in chapter 1, Type I String theory is a theory of open and closed strings, including oriented as well as unoriented worldsheets. The Fock space for a String theory of unoriented worldsheets is obtained by considering states in the Type II Fock space which are invariant under the \emph{worldsheet parity operator} $\Omega$, induced by reversing the orientation on the worldsheet \cite{Polchinski,quantumath}.\\
The operator $\Omega$ induces an action also on the Chan-Paton bundle to a D-brane, forcing it to be an O($n$)-bundle. Moreover, the Ramond-Ramond field content in Type I differs from that of Type II: indeed, the Ramond-Ramond gauge theory in Type I String theory consists of $p$-forms, for $p$=2,6.\\
Finally, Freed-Witten anomaly cancellation imposes that a D-brane can only wrap a spin manifold.\\
The right framework to describe D-brane charges is then $\rm KO$-theory. Namely, a configuration of Type I D9-branes $\rm (E,F)$ modulo brane-antibrane annihilation is represented by a class $x$ in $\kogroup{0}{\topsp{M}}$.\\
Moreover, as in Type II String theory, given a D$p$-brane wrapping a spin manifold $\topsp{Q}\subset\topsp{M}$, the group of its admissable charges is given by
\begin{displaymath}
\kogroup{0}{\topsp{B(N)},\topsp{S(N)}}
\end{displaymath}
Applying this classification to a D$p$-brane wrapping $\mathbb{R}^{1,p}\subset\topsp{M}^{10}$, we obtain
\begin{displaymath}
\ckogroup{0}{\topsp{B(N)},\topsp{S(N)}}\simeq\ckogroup{0}{\mathbb{R}^{1,p}\times\topsp{D}^{9-p},\mathbb{R}^{1,p}\times\topsp{S}^{9-p}}\simeq\redko{0}{\topsp{S}^{9-p}}
\end{displaymath}
where
\begin{displaymath}
\redko{0}{\topsp{S}^{n}}\simeq\left\{
\begin{array}{l}
\mathbb{Z},\:n=0,4\:\text{mod}\:8\\
\mathbb{Z}_{2},\:n=1,2\:\text{mod}\:8\\
0,\:\text{otherwise}
\end{array}\right .
\end{displaymath}
We have $\redko{0}{\topsp{S}^{9-p}}\simeq{\mathbb{Z}}$, for $p=1,5,9$: for $p=1,5$, this agrees with the fact that the corresponding supersymmetric D-branes couple to Ramond-Ramond fields, and hence are stable. Moreover, we would expect these to be the \emph{only} stable D-branes in Type I String theory. A new prediction of K-theory is that this is not the case. Indeed, for $p=7,8$, we have $\redko{0}{\topsp{S}^{9-p}}\simeq\mathbb{Z}_{2}$: hence, D7-branes and D8-branes can carry a charge that would protect them from decaying. Morever, since the group $\mathbb{Z}_{2}$ is a pure torsion group, this charge is not associated to any spacetime field coupling to the D-brane, and has to be considered as a purely topological effect.\\
The K-theoretic description also naturally describes stable objects for $p=-1$ and $p=0$, which are called the \emph{D-instanton} and the\emph{ D-particle}.\\
In contrast to Type II String theory, constructing a system of spacetime filling brane-antibranes which has a given D$p$-brane as its decay product is not very systematic, in the sense that there is no unified procedure, or homomorphism, realizing this construction. Indeed, in \cite{Witten1998} the different stable D$p$-branes are treated with different methods. This is due to the fact that the Thom isomorphism in KO-theory implies that
\begin{equation}\label{thomKO}
\ckogroup{0}{\topsp{B(N)},\topsp{S(N)}}\simeq{\ckogroup{-p}{\topsp{Q}}}
\end{equation}
where $\topsp{Q}$ is the wrapped manifold, and $p$ is the rank of the normal bundle of $\topsp{Q}$ in $\topsp{M}$. As we have seen, elements of $\ckogroup{-p}{\topsp{Q}}$ are not represented by ``differences'' of vector bundles on $\topsp{Q}$, hence there is no natural way to identify them as a system of D-branes. Indeed, the higher KO-groups only play a role through the isomorphism (\ref{thomKO}), which is one of the limitations of the K-theorical description of D-branes. We will address a possible solution to this problem in the next chapter, where we will develop another description of D-branes, based on KO-homology, the dual theory to KO-theory. We refer to \cite{olsenszabodescent} for an alternative interpretation of the higher KO-groups, and to \cite{Olsen1999} for extensive review on the application of K-theory to String theory.

\chapter{KO-homology and Type I D-branes}
As it was emphasized in the last sections of chapter 3, the K-theorical description of D-branes in String theory is based on the Sen-Witten mechanism of brane-antribrane annhilation. However, we have seen that in an ordinary Maxwell theory of p-forms the sources are extended objects, and their charges are related to the homology cycles they represent. We have argued that for the Ramond-Ramond gauge theory this is not the case: anyway, intuitively we still expect the right mathematical framework to be given by some sort of homology theory that would take into account the relation between K-theory and D-branes as exposed in the previous chapter.\\
In Type II String theory, this simple observation suggests that a much more natural description of D-branes can be given in terms of K-homology, the homological theory associated to K-theory, as thouroghly emphasized in \cite{periwal-2000-0007,harvey-2001-42,matsuo-2001-499,szabo-2002-17}.\\
More precisely, K-homology has two equivalent representations: an analytic representation, in terms of $\rm C^{*}$-algebras and Fredholm modules, and a geometric one, constructed by Baum and Douglas in \cite{Baumdouglas,Baumdouglas2}. In particular, the Baum-Douglas construction was extensively used in \cite{Reis2006} to provide a rigorous geometric description of D-branes in Type II String theory in various topologically non trivial backgrounds.\\
In this chapter, we will present new results concerning KO-homology, the homology theory associated to KO-theory.\\
 From the mathematical perspective, we construct a geometric realization of KO-homology, and we argue that it is indeed isomorphic to the homology theory defined via the loop spectrum of KO-theory, which we refer to as \emph{spectral KO-homology}. We will also develop the analytic description of KO-homology, using Kasparov's formalism for real $C^{*}$-algebras, which represents a unified description of both K-theory and K-homology.\\
We then construct an homomorphism between geometric and analytic KO-homology, and give a detailed and explicit proof that such homomorphism induces a natural equivalence between the two representations. We want to comment, at this point, that the equivalence between geometric and analytic K-homology has been a sort of a ``folklore theorem'' until recently \cite{baum-2007-3}, as the original work of Baum and Douglas did not contain any accurate proof of such equivalence. In particular, the work \cite{baum-2007-3} also contains a proof of the equivalence between geometric and analytic KO-homology: neverthless, the proof presented in this chapter is fundamentally different, and the whole construction is more apt for physical applications, as we will see later on. Indeed, the approach followed here has strongly been inspired by the point of view in \cite{Baumdouglas} that index theory is based on the equivalence between geometric and analytic K-homology: this point of view is reinforced by the introduction of certain geometric invariants that will help us to derive some cohomological index formulas in the real case.\\[2mm]
\indent From the physical perspective, we introduce the notion of \emph{wrapped D-brane} and of \emph{wrapping charge} of a D-brane. In particular, we will illustate how the higher K-homology groups can naturally be interpreted in terms of wrapped D-branes, and we will argue that the wrapping charge of a D-brane is a genuinely different concept in Type I String theory..\\
We have mentioned in chapter 2 that the interpretation of a D-brane as a submanifold of the spacetime is not very accurate, and that a distinction should be made somehow between the D-brane itself and the worldvolume it wraps. This point of view emerges here, as we will construct stable torsion D-branes wrapping a single point in the spacetime.                          
\section{Dual theories and spectral KO-homology}
In this section we will explain in which sense KO-homology is ``associated'' to KO-theory, and we will define spectral KO-homology.\\
First we recall some properties of cohomology theories in the category of CW-complexes, referring to \cite{switzer} for a detailed exposition.\\
An \emph{$\Omega$-spectrum}, or \emph{loop spectrum} for a generalized cohomology theory $\rm k^{*}$ is given by a sequence of CW-complexes $\{\topsp{K}_{n}\}_{n\in\mathbb{Z}}$ together with homotopy equivalences 
\begin{equation}\label{loop}
\topsp{K}_{n}\to\Omega\topsp{K}_{n+1}
\end{equation}
where $\Omega$ denotes the loop space functor, such that the functor ${\rm k}^{n}$ can be represented
\begin{displaymath}
{{\rm k}^{n}}(\topsp{X})=[\topsp{X},\topsp{K}_{n}]\quad{\forall\:n\in\mathbb{Z}}
\end{displaymath}
for any CW-complex $\topsp{X}$.\\
By considering the maps $\sigma_{n}:\Sigma(\topsp{K}_{n})\to{\topsp{K}_{n+1}}$ adjoint to the maps in (\ref{loop}), we can define the unreduced generalized homology theory associated to $\rm k_{*}$ by setting
\begin{displaymath}
{\rm k_{i}(\topsp{X})}:={\rm \tilde{k}_{i}(\topsp{X}^{+})}:=\text{lim}_{n}\pi_{n+i}(\topsp{X}^{+}\wedge{\topsp{K}_{n}})
\end{displaymath}
where the inductive limit is taken using the maps $\sigma_{i}$
\begin{displaymath}
\begin{array}{rl}
\pi_{n+i}(\topsp{X}^{+}\wedge{\topsp{K}_{n}})=&[\topsp{S}^{n+i},\topsp{X}^{+}\wedge{\topsp{K}_{n}}]_{*}\xrightarrow{\text{susp}}[\Sigma(\topsp{S}^{n+i}),\Sigma(\topsp{X}^{+}\wedge{\topsp{K}_{n}})]_{*}\\
&\simeq[\topsp{S}^{n+i+1},\topsp{X}^{+}\wedge{\Sigma(\topsp{K}_{n}})]_{*}
\xrightarrow{\sigma_{i}}[\topsp{S}^{n+i+1},\topsp{X}^{+}\wedge{\topsp{K}_{n+1}}]_{*}\\
&=\pi_{n+i+1}(\topsp{X}^{+}\wedge{\topsp{K}_{n+1}})
\end{array}
\end{displaymath}
The relative homology theory can be defined as usual, and we will refer to the homology theory $\rm k_{*}$ as the \emph{dual} theory to $\text{k}^{*}$.\\
In \cite{skewadjoint} it was shown that a suitable spectrum for KO-theory can be defined in the following way. For $n\geq1$, let $\mathcal{H}_\real$ be a real
$\zed_2$-graded separable Hilbert space which is a
$*$-module for the real Clifford algebra $\Cl_{n-1}=\Cl(\real^{n-1})$. Let $\Fred_{n}$ be the space of all
Fredholm operators on $\mathcal{H}_\real$ which are odd, $\Cl_{n-1}$-linear and
self-adjoint. Then $\Fred_{n}$ is the classifying space for
$\KO^{-n}$, and there are homotopy equivalences $\Fred_{n}\to\Omega\Fred_{n-1}$.  For $n\leq0$, we choose $k\in\nat$ such that
$8k+n\geq1$ and define $\Fred_n:=\Fred_{8k+n}$.\\
Then, \emph{spectral KO-homology} can be defined by setting
\begin{equation}
\KO^{s}_{i}(\topsp{X},\topsp{Y}):=\text{lim}_{n}\pi_{n+i}((\topsp{X}/\topsp{Y})\wedge{\Fred_{n}})
\end{equation}
From a general result for spectrally defined generalized homology theories, we have that
\begin{displaymath}
\KO^{s}_{i}(\rm pt)=\redko{0}{\topsp{S}^{i}}
\end{displaymath}
The importance of the functors $\KO^{s}_{n}$ in our context consists in the fact that any set of groups which are isomorphic to spectral KO-homology, in a suitable sense, for any space X, defines necessarily an homology theory. Indeed, this will be the case for geometric KO-homology defined later on, whose homological properties will be deduced by ``comparison'' with spectral KO-homology. 
\section{KKO-theory and analytic KO-homology}
We will give now a detailed overview of the
definition of KO-homology in terms of Kasparov's KK-theory for real
$C^*$-algebras \cite{kasparov}, and describe various properties that we
will need later on. Incidentally, we will eventually give an interpretation of the KK-groups for complex $C^{*}$-algebras in terms of Type II D-branes.
\subsection{Real $\rm C^{*}$-algebras}
A \textit{real algebra} is a ring $A$ which is also an
$\mathbb{R}$-vector space such that $\lambda\,(x\,y)=(\lambda\,x)\,y=x
\,(\lambda\,y)$ for all $\lambda \in \real$ and all $x,y \in A$.
A \textit{real} $*$-\textit{algebra} is a real algebra $A$ equipped
with a linear involution $*: A \rightarrow A$ such that
$(x\,y)^*=y^*\,x^*$ for all $x,y \in A$.
A \textit{real Banach algebra} is a real algebra $A$ equipped with
a norm $\lVert - \rVert:A\to\real$ such that $\normxy \leq
\normx\,\normy$ and such that $A$ is complete in the norm topology.
If $A$ is a unital algebra then we assume $\norm=1$.
A \textit{real Banach} $*$-\textit{algebra} is a real Banach algebra
which is also a real $*$-algebra. A \textit{real} $C^*$-\textit{algebra}
is a real Banach $*$-algebra such that
\begin{romanlist}
\item $\|x^*x\|=\|x\|^{2}$ for all
  $x\in A$; and
\item $1+x^*\,x$ is invertible in $\tilde{A}$ for all $x \in A$.
\end{romanlist}
where $\tilde{A}$ denotes the unitalization of the algebra $A$.
\begin{remark2}
Although in the complex case invertibility of $1+x^*\,x$ for all $x
\in A$ would follow immediately from the $C^*$-algebra structure, in
the real case this is no longer true. For example, consider the real
Banach $*$-algebra $\mathbb{C}$ with involution given by the identity
map. Then $1+\ii^*\,\ii$ is not invertible, where
$\ii:=\sqrt{-1}$. This invertibility condition is fundamental to
obtaining the usual representation theorem below for $C^*$-algebras in
terms of bounded self-adjoint operators on a real
Hilbert space. However, $\mathbb{C}$ with involution given by complex
conjugation is a real $C^*$-algebra. Since the only $\mathbb{R}$-linear
involutions of $\mathbb{C}$ are the identity and complex
conjugation, when we consider $\mathbb{C}$ as a real $C^*$-algebra the
involution will always be implicitly assumed to be complex
conjugation. More generally any complex $C^*$-algebra, regarded as a
real vector space and with the same operations, is a real
$C^*$-algebra.
\end{remark2}
Let us now give a number of examples of real $C^*$-algebras, some of
which we will use later on in representation theorems.

\begin{example}\label{r3}
Let $\HR$ be a real Hilbert space. Then the set of bounded linear
operators $\mathcal{B}(\HR)$ with the usual operations is a real
$C^*$-algebra. Any closed self-adjoint subalgebra of $\mathcal{B}(\HR)$ is also
a real $C^*$-algebra. More generally, any closed self-adjoint subalgebra of
a real $C^*$-algebra is always a real $C^*$-algebra.
\end{example}

\begin{example}\label{r6}
Let $\topsp{X}$ be a locally compact Hausdorff space and $\C_0(\topsp{X},\mathbb{R})$
the space of real-valued continuous functions vanishing at
infinity. Then $\C_0(\topsp{X},\mathbb{R})$ with pointwise operations, the
supremum norm and involution given by the identity map is a real
$C^*$-algebra. As in the complex case, $\C_0(\topsp{X},\mathbb{R})$ is unital
if and only if $\topsp{X}$ is compact.
\label{ex:r6}\end{example}

\begin{example}\label{r7}
With $\topsp{X}$ as in Example~\ref{ex:r6} above, let $\topsp{Y}$ be a closed subspace
of $\topsp{X}$ and $\C_0(\topsp{X},\topsp{Y};\mathbb{R})$ the subspace of $\C_0(\topsp{X},\mathbb{C})$
consisting of maps $f : \topsp{X} \rightarrow \mathbb{C}$ such that $f(\topsp{Y})
\subset\real$. Then with the operations inherited from
$\C_0(\topsp{X},\mathbb{C})$, the subspace $\C_0(\topsp{X},\topsp{Y};\mathbb{R})$ is a real
$C^*$-algebra.
\end{example}

\begin{example}\label{r8}
Let $\topsp{X}$ be a locally compact Hausdorff space with involution $\tau
:\topsp{X} \rightarrow \topsp{X}$, i.e. a homeomorphism such that $\tau \circ \tau =
\textrm{id}_\topsp{X}$, and consider the subset $\C_0(\topsp{X},\tau)$ of
$\C_0(\topsp{X},\mathbb{C})$ consisting of maps $f$ such that $f \circ \tau
=f^*=\overline{f}$. Then $\C_0(\topsp{X},\tau)$, with the operations inherited
from $\C_0(\topsp{X},\mathbb{C})$, is a real $C^*$-algebra. If $\tau
=\textrm{id}_\topsp{X}$ then $\C_0(\topsp{X},\tau)=\C_0(\topsp{X},\mathbb{R})$.
If $\topsp{X}$ is compact and $\topsp{Y}$ is a closed subspace of $\topsp{X}$, then there is
a compact Hausdorff space $Z$ with an involution $\tau$ such that
$\C(\topsp{X},\topsp{Y};\mathbb{R}) \simeq \C(Z,\tau)$. However, the converse does not
hold in general.
\end{example}

\begin{example}\label{r10}
Let $\mathcal{V}$ be a real vector space equipped with a quadratic form
\textit{q}, and consider the associated real Clifford algebra
$\Cl(\mathcal{V},\textit{q})$. Assume, without loss of generality,
that $\textit{q}(v)=\langle v,\phi(v)\rangle$ for all $v \in \mathcal{V}$ with
respect to an inner product on $\mathcal{V}$, where the linear
operator $\phi \in \mathcal{L}(\mathcal{V})$ is symmetric and orthogonal.
We can then define an involution on $\Cl({\mathcal{V}},\textit{q})$ by
$(v_1 \cdots v_k)^*=\phi(v_k) \cdots \phi(v_1)$, i.e. if $v \in {\mathcal{V}}$
then $v^*=\phi(v)$. The isomorphism $\Phi : \Cl({\mathcal{V}} \oplus
{\mathcal{V}},\textit{q} \oplus -\textit{q}) \rightarrow \mathcal{L}(\Lambda
^{*}{\mathcal{V}})$ induces a norm on $ \Cl({\mathcal{V}} \oplus
{\mathcal{V}},\textit{q} \oplus -\textit{q}) $ by
pullback of the operator norm on $\mathcal{L}(\Lambda
^{*}{\mathcal{V}})$, and the inclusion
$\Cl({\mathcal{V}},\textit{q})\hookrightarrow\Cl({\mathcal{V}},
\textit{q})\hat{\otimes}\Cl({\mathcal{V}},-\textit{q})
\simeq\Cl({\mathcal{V}} \oplus {\mathcal{V}},\textit{q} \oplus
-\textit{q})$ given by $x \mapsto x\hat{\otimes}1$ thereby induces a norm
on $\Cl({\mathcal{V}},\textit{q})$. Then
$\Cl({\mathcal{V}},\textit{q})$ with its algebra structure, this
involution and norm is a real $C^*$-algebra.
\end{example}
If $A$, $B$ are real $*$-algebras then a \textit{real} $*$-\textit{algebra
homomorphism} is a real algebra map $\phi : A \rightarrow B$, i.e. an
$\real$-linear ring homomorphism, such that $\phi (x^*)=\phi (x)^*$ for all $x
\in A$. The homomorphism is assumed to be unital if both algebras are
unital.\\
If $A$ is an algebra, we denote by $\mat_n(A)$ the
algebra of $n\times n$ matrices with entries in $A$. Then, we have the following general representation theorem \cite{goodearl}
\begin{theorem}
Let $A$ be a finite-dimensional real $C^*$-algebra. Then there exist\\
$n_1, \ldots, n_k \in \mathbb{N}$ such that $A \simeq
\mat_{n_1}(A_1) \times \cdots \times \mat_{n_k}(A_k)$ as real
$C^*$-algebras with $A_1, \ldots, A_k \in
\{\mathbb{R},\mathbb{C},\mathbb{H}\}$.
\end{theorem}
Analogously to the complex case, real $C^{*}$-algebras are always algebras of operators on some Hilbert space.
\begin{theorem}(\textbf{Ingelstam})
Let $A$ be any real $C^*$-algebra. Then there exists a real Hilbert
space $\mathcal{H}_\mathbb{R}$ such that $A$ is isomorphic as a
real $C^*$-algebra to a closed self-adjoint subalgebra of
$\mathcal{B}(\HR)$.
\end{theorem}\label{Ingelstam}

Consider now a real $C^*$-algebra $A$. We denote by
$A_{\mathbb{C}}:=A\otimes\complex$ the complexification of $A$, which
is a complex algebra containing $A$ as a real algebra. Theorem \ref{Ingelstam} assures that $A_{\mathbb{C}}$ can be given a unique $C^{*}$-algebra norm such that the natural embedding $\theta:A\to A_{\mathbb{C}}$ of $A$ onto its complexification is an isometry.\\
 We can define a map $J_A : A_{\mathbb{C}} \rightarrow A_{\mathbb{C}}$ by $J_A(x+\ii
y)=x-\ii y$ for all $x,y \in A$. The map $J_A$ is a conjugate linear
$*$-isomorphism of the complex $C^*$-algebra $A_{\mathbb{C}}$. If
$\phi : A \rightarrow A$ is a continuous $*$-homomorphism, then the
map $J_A(\phi) : A_{\mathbb{C}} \rightarrow A_{\mathbb{C}}$ defined
by $J_A(\phi)(x+\ii y)=\phi(x)+\ii\phi(y)$ is a continuous
$*$-homomorphism such that $J_A \circ J_A(\phi) = J_A(\phi)\circ
J_A$. Conversely, if $J$ is a conjugate linear $*$-isomorphism of a
complex $C^*$-algebra $B$, then $A= \{ x \in B \, | \, J(x)=x \}$ is
a real $C^*$-algebra. This implies the following result.

\begin{proposition}
Let $\mathscr{C}^*_{\mathbb{R}}$ be the category of real $C^*$-algebras
and continuous $*$-algebra homomorphisms. Let
$\mathscr{C}^*_{\mathbb{C},{\rm cl}}$ be the category of pairs $(A,J)$,
where $A$ is a complex $C^*$-algebra and $J$ is a conjugate linear
$*$-isomorphism of $A$, and continuous $*$-homomorphisms commuting with
$J$. Then the assignments 
\begin{displaymath}
\begin{array}{c}
A \mapsto (A_{\mathbb{C}},J_A)\\
\phi \mapsto J_A(\phi) 
\end{array}
\end{displaymath}
define a functor
$$\mathcal{J} \,:\, \mathscr{C}^*_{\mathbb{R}} ~\longrightarrow~
\mathscr
{C}^*_{\mathbb{C},{\rm cl}}$$
which is an equivalence of categories.
\end{proposition}
The complexification of a real $C^{*}$-algebra $A$ is crucial in generalizing the notion of spectrum of an element. Indeed, we define the \emph{complexified spectrum} ${\rm Sp}^{\mathbb{C}}(x)$ of an element $x\in{A}$ as the spectrum of the element $\theta(x)$ in $A_{\mathbb{C}}$, i.e. the set of $\lambda\in\mathbb{C}$ such that $\lambda-\theta(x)$ is not invertible in $A_{\mathbb{C}}$. This definition of the spectrum assures that the functional calculus in $A$ is well behaved. In the following, we will use the notion of positive element. An element $x$ in a real $C^{*}$-algebra $A$ is said to be \emph{positive} if $x=x^{*}$, and ${\rm Sp}^{\mathbb{C}}(x)\subseteq\mathbb{R}^{+}$.  

As we are interested in the applications of real $C^{*}$-algebras to algebraic topology of CW-complexes,
we will now specialize to the case of commutative algebras.\\
As with complex Banach algebras, a maximal two-sided ideal in a real
Banach algebra $A$ is closed in $A$. If $M$ is a maximal two-sided
ideal of a real Banach algebra $A$, then $A/M$ is isomorphic to one of
$\mathbb{R}$ or $\mathbb{C}$ as real algebras.
A \textit{character} on a real algebra $A$ is a non-zero real
algebra map $\chi : A \rightarrow \mathbb{C}$, assumed unital if $A$
is unital. Let $\Omega_A$ be the space of characters of
$A$. This can be given, as in the complex case, a locally compact
Hausdorff space topology such that $\Omega_A$ is homeomorphic to
$\Omega_{A_\mathbb{C}}$. Furthermore, $A$ is unital if and only if
$\Omega_A$ is compact.

Given $x \in A$, evaluation at $x$ gives a continuous map $
\Gamma(x) : \Omega_A \rightarrow \mathbb{C}$ called the
\textit{Gel'fand transform of $x$}. From this we obtain the
\textit{Gel'fand transform of $A$}, $\Gamma : A \rightarrow
\C_0(\Omega_A,\mathbb{C})$, which is a continuous real algebra
homomorphism of unit norm. If $A$ is a real $*$-algebra, then $\Gamma$
is a $*$-algebra homomorphism.\\ 
The following important results on the
representation of commutative real $C^*$-algebras allow to unify the treatment of real commutative $C^{*}$-algebras and topological spaces with involution.
\begin{theorem}
Let $A$ be a commutative real $C^*$-algebra. Then:
\begin{romanlist}
\item The map $\tau : \Omega_A
\rightarrow \Omega_A$ defined by $\tau(\chi)=\overline{\chi}$ is an
involution; and
\item The Gel'fand transform $\Gamma : A \rightarrow
  \C_0(\Omega_A,\tau)$ is a real $C^*$-algebra isomorphism.
\end{romanlist}
\end{theorem}

\begin{proof}
(i) The map $\tau$ is a bijection. The collection of sets
$$U_{x,V}= \bigl\{\chi\in \Omega_A ~ | ~ \chi(x) \in V  \bigr\}$$
for every $x \in A$ and $V$ open in $\mathbb{C}$ is a sub-basis for
the topology of $\Omega_A$. The complex conjugate $\overline{V}$ of
$V$ is an open set and $\tau^{-1}(U_{x,V})=U_{x,\overline{V}}$. Thus
$\tau$ is continuous.

\noindent
(ii) The map $\Gamma$ is a real $*$-algebra map with $\lVert \Gamma(x)
\rVert =\normx$. One also has \bea\G(x) \circ
\ta(\chi)&=&\G(x)(\,\overline{\chi}\,)\nonumber\\ &=&
\overline{\chi(x)}~=~\G(x)^*(\chi) \ , \nonumber\eea and so
$\G(x) \circ \ta=\G(x)^*$ and $\G(A) \subset \C_0(\Omega_A,\ta)$.
Let $\theta : A \rightarrow A_\mathbb{C}$ be the $C^*$-algebra
embedding of $A$ into its complexification. The map $\vartheta :
\Omega_{A_\mathbb{C}} \rightarrow\Omega_A$ given by $\vartheta(f)=f \circ
\theta$ is a homeomorphism and there is a commutative diagram
$$\xymatrix{ A \ar[r]^{\!\!\!\!\!\!\!\!\!\!\!\!\G}
\ar[d]_{\theta} & \C_0(\Omega_A,\mathbb{C}) \ar[d]^{\vartheta^*}\\
             A_\mathbb{C} \ar[r]_{\!\!\!\!\!\!\!\!\!\!\!\!\G}
&  \C_0(\Omega_{A_\mathbb{C}},\mathbb{C}) \ .}$$
Using this one then shows that $\G(A)=\C_0(\Omega_A,\ta)$; see \cite{goodearl}.
\end{proof}

\begin{cor}
Let $A$ be a commutative real $C^*$-algebra with trivial
involution. Then $A$ is $*$-isomorphic to $\C_0(\Omega_A,\mathbb{R})$.
\end{cor}
\subsection{Hilbert Modules}
In this section we will present a generalization of the notion of Hilbert space, in which the scalar product takes value in a general real $C^{*}$-algebra.\\
Let $A$ be a (not necessarily commutative) real $C^*$-algebra. A
$\textit{pre-Hilbert module}$ \textit{over} $A$  is a (right)
$A$-module $\mathcal{E}$
equipped with an $A$-$\textit{valued inner product}$, i.e. a
bilinear map $(-, -) : \mathcal{E} \times \mathcal{E}
\rightarrow A$ such that
\begin{romanlist}
\item $(x,x) \geq 0$ for all $x \in \mathcal{E}$ and $(x,x)=0$ if and
  only if $x=0$;
\item $(x,y)=(y,x)^*$ for all $x,y \in \mathcal{E}$; and
\item $(x,y\,a)=(x,y)\,a$ for all $x,y \in \mathcal{E}, \, a \in A$.
\end{romanlist}
For $x \in \mathcal{E}$ we define $\normx_{\mathcal{E}}:= \lVert (x,x)
\rVert ^{1/2}$. This
defines a norm on ${\mathcal{E}}$ satisfying the Cauchy-Schwartz
inequality. If ${\mathcal{E}}$ is complete under this norm, then it is
called a \textit{Hilbert module over}~$A$.\\
As a straight generalization from the ordinary Hilbert space case, we have the following examples of Hilbert modules.
\begin{example}
A real $C^{*}$-algebra $A$ can be given the structure of a Hilbert module over itself by defining $(a,b)=a^{*}b$, for any $a,b\in\:{A}$. More generally, any closed right ideal of $A$ is a Hilbert module over $A$.
\end{example}
\begin{example}
Let $\mathcal{E}$ consists of all sequences $(a_{n})_{n\in\mathbb{N}}$, $a_{n}\in{A}$, such that 
\begin{displaymath}
\sum_{n} \lVert a_{n} \rVert^{2}<\infty 
\end{displaymath}
with inner product $((a_{n}),(b_{n}))=\sum_{n}a_{n}^{*}b_{n}$.\\
$\mathcal{E}$ is called the \emph{Hilbert space over A}, and is often denoted with $A^{\infty}$.
\end{example}

Let ${\mathcal{E}},{\mathcal{F}}$ be Hilbert $A$-modules and $\topsp{T} :
{\mathcal{E}} \rightarrow {\mathcal{F}}$ an
$A$-linear map. We call a map $\topsp{T}^* : {\mathcal{F}} \rightarrow
{\mathcal{E}}$ such that
$(\topsp{T}x,y)_{\mathcal{F}}=(x,\topsp{T}^*y)_{\mathcal{E}}$ for all $x \in
{\mathcal{E}}, \, y \in {\mathcal{F}}$ the
\textit{adjoint} of $\topsp{T}$. If it exists the adjoint is unique.\\
In contrast to ordinary bounded operators on separable Hilbert spaces, not every $A$-linear map between Hilbert
$A$-modules has an adjoint. We denote the set
of all $A$-linear maps $\topsp{T} : {\mathcal{E}} \rightarrow {\mathcal{F}}$
admitting an adjoint by
$\mathcal{L}({\mathcal{E}},{\mathcal{F}})$.
Elements of $\mathcal{L}({\mathcal{E}},{\mathcal{F}})$ are bounded
$A$-linear maps and
$\mathcal{L}({\mathcal{E}}):=\mathcal{L}({\mathcal{E}},{\mathcal{E}})$
is a $C^*$-algebra with the operator
norm and involution given by the adjoint.\\[2mm]
Notice that a submodule of a Hilbert $A$-module $\mathcal{E}$ in general need not be complemented, i.e. there is generally no projection in $\mathcal{L}({\mathcal{E}})$ onto the given submodule.\\ 
However, one can define some special ``rank 1'' operators as follows. Given $x \in {\mathcal{F}}, \, y \in
{\mathcal{E}}$ we define an operator $\theta_{x,y} \in
\mathcal{L}({\mathcal{E}},{\mathcal{F}})$ by
$\theta_{x,y}(z)=x\,(y,z)_{\mathcal{E}}$. These operators generate an
$\mathcal{L}({\mathcal{E}})-\mathcal{L}({\mathcal{F}})$-bimodule whose
norm closure in $\mathcal{L}({\mathcal{E}},{\mathcal{F}})$ is denoted
$\mathcal{K}({\mathcal{E}},{\mathcal{F}})$. Elements of
$\mathcal{K}({\mathcal{E}},{\mathcal{F}})$ are called
$\textit{generalized compact operators}$.  We will use the notation $\mathcal{K}({\mathcal{E}})$ for $\mathcal{K}({\mathcal{E}},{\mathcal{E}})$.
\begin{example}$\mathcal{K}(A)=A$ for any real $C^{*}$-algebra $A$. If $A$ is unital, we have also that $\mathcal{L}(A)=A$.
\end{example}
\begin{example}
If $\mathcal{E}$ is a Hilbert $A$-module, then $\mathcal{L}(\mathcal{E}^{n})\simeq\mathbb{M}_{n}(\mathbb{R})\otimes\mathcal{L}(\mathcal{E})$, and $\mathcal{K}(\mathcal{E}^{n})\simeq\mathbb{M}_{n}(\mathbb{R})\otimes\mathcal{K}(\mathcal{E})$
\end{example}
\begin{example}
 $\mathcal{K}({A}^{\infty})\simeq{A}\otimes\mathcal{K}_{\mathbb{R}}$, where
$\mathcal{K}_{\mathbb{R}}:=\mathcal{K}(\mathcal{H}_{\mathbb{R}})$.
\end{example}
For a real $C^*$-algebra $A$, the \textit{multiplier algebra} of
$A$, $\textsf{M}(A)$, is the maximal $C^*$-algebra
containing $A$ as an essential ideal. Equivalently, by representing $A
\subset \mathcal{L}(\mathcal{H}_\real)$ one has
$$\textsf{M}(A)= \{ \, \topsp{T} \in \mathcal{L}(\mathcal{H}_\real) ~|~ \topsp{T}\,\topsp{S},
\topsp{S}\,\topsp{T} \in A \quad \textrm{for all} ~ \topsp{S} \in A \} \ . $$
The multiplier algebra $\textsf{M}(A)$ is a $C^*$-algebra which is
$*$-isomorphic to the $C^*$-algebra of double centralizers, i.e. pairs
$(\topsp{T}_1,\topsp{T}_2) \in \mathcal{L}(A) \times \mathcal{L}(A)$ such
that $a\,\topsp{T}_1(b)=\topsp{T}_2(a)\,b$, $\topsp{T}_1(a\,b)=\topsp{T}_1(a)\,b$ and
$\topsp{T}_2(a\,b)=a\,\topsp{T}_2(b)$ for all $a,b \in A$.\\
If $A$ is unital, then
$\textsf{M}(A)=A$. Furthermore, $\textsf{M}(\mathcal{K}_{\mathbb{R}})=
\mathcal{L}(\mathcal{H}_{\mathbb{R}})$, and
$\textsf{M}(\C_0(\topsp{X},\mathbb{R}))=\C_{\rm b}(\topsp{X},\mathbb{R})$ is the
$C^*$-algebra of real-valued bounded continuous functions on a locally
compact Hausdorff space $\topsp{X}$.\\
We have then the following proposition, whose proof follows from the analogous result for complex Hilbert modules \cite{blackadar}
\begin{proposition}
Let ${\mathcal{E}}$ be a Hilbert $A$-module. Then there is an
isomorphism $$\mathcal{L}\bigl({\mathcal{E}}\bigr)\simeq
{\sf M}\bigl(\mathcal{K}({\mathcal{E}})\bigr) \ . $$
\end{proposition}
As with ordinary Hilbert spaces, one can define tensor products of Hilbert modules in
the following way \cite{blackadar}.\\
Let $\mathcal{E}_{i}$ be a Hilbert $B_{i}$-module, for i=1,2, and let $\phi:B_{1}\to{\mathcal{L}(\mathcal{E}_{2})}$ be a *-homomorphism. If we regard $\mathcal{E}_{2}$ as a left $B_{1}$-module via $\phi$, we can form the algebraic tensor product $\mathcal{E}_{1}\otimes_{B_{1}}\mathcal{E}_{2}$, which is a right $B_{2}$-module. Finally, we define the $B_{2}$-valued pre-inner product on the algebraic tensor product by $$(x_{1}\otimes{x_{2}},y_{1}\otimes{y_{2}}):=(x_{2},\phi((x_{1},y_{1})_{\mathcal{E}_{1}})y_{2})_{\mathcal{E}_{2}}$$
The completion of the algebraic tensor product with respect to this inner product, with vectors of length 0 divided out, is called the \emph{tensor product} of $\mathcal{E}_{1}$ and $\mathcal{E}_{2}$, and is denoted with $\mathcal{E}_{1}\otimes_{\phi}\mathcal{E}_{2}$.\\
As for ordinary Hilbert spaces, there is a natural homomorphism from $\mathcal{L}(\mathcal{E}_{1})$ to $\mathcal{L}(\mathcal{E}_{1}\otimes_{\phi}\mathcal{E}_{2})$, and we will denote the image of $\topsp{T}\in\mathcal{L}{(\mathcal{E}_{1})}$ as $\topsp{T}\otimes{1}$, or $\phi_{*}(\topsp{T})$ when we want to emphatize the homomorphism $\phi$. However, there is no homomorphism from $\mathcal{L}(\mathcal{E}_{2})$ to $\mathcal{L}(\mathcal{E}_{1}\otimes_{\phi}\mathcal{E}_{2})$ in general.\\
Finally, if ${\mathcal{E}}$ is a pre-Hilbert module over the real $C^*$-algebra $A$, we
assume that the complexification ${\mathcal{E}} \otimes \mathbb{C}$ is a
pre-Hilbert module over $A_{\mathbb{C}}$. This means that the
$A$-valued inner product extends to a sesquilinear map. We assume
that sesquilinear maps are linear in the second variable.\\
\subsection{Kasparov's formalism for KKO-theory}
We are now ready to define KKO-theory by using Kasparov's approach, developed in \cite{kasparov}. In the following we will assume that a real $C^{*}$-algebra $A$ is separable and a real $C^{*}$-algebra $B$ is $\sigma$-unital, i.e. $B$ contains an element $h$ such that $\phi(h)>0$, for every character $\phi$ of $B$. This technical requirements will be useful in the following.
\begin{definition}
A \textit{(Kasparov)} $(A,B)$\textit{-module} is a triple
$({\mathcal{E}},\rho,\topsp{T})$, where ${\mathcal{E}}$ is a countably
generated $\mathbb{Z}_{2}$-graded Hilbert $B$-module, $\rho : A \rightarrow
\mathcal{L}({\mathcal{E}})$ is an even $*$-homomorphism and $\topsp{T} \in
\mathcal{L}({\mathcal{E}})$ such that
\beq
\left(\topsp{T}-\topsp{T}^*\right)\rho(a)~, \, \left(\topsp{T}^2-1
\right)\rho(a)~, \, \bigl[\topsp{T}\,,\,\rho(a)\bigr] \in
\mathcal{K}({\mathcal{E}})
\label{Kmoduledef}\eeq
for all $a \in A$, and T is odd with respect to the grading on $\mathcal{E}$. A Kasparov module $({\mathcal{E}},\rho,\topsp{T})$ is called {\it
  degenerate} if all operators in (\ref{Kmoduledef}) are zero.
Two Kasparov modules $({\mathcal{E}}_i,\rho_i,\topsp{T}_i)$, $i=1,2$ are said to be
\textit{orthogonally equivalent} if there is
an isometric isomorphism $\topsp{U} \in
\mathcal{L}({\mathcal{E}}_1,{\mathcal{E}}_2)$ such that
$\topsp{T}_1=\topsp{U}^*\,\topsp{T}_2\,\topsp{U}$ and $\rho_1(a)=\topsp{U}^*\,\rho_2(a)\,\topsp{U}$ for all $a \in
A$.
\end{definition}
Orthogonal equivalence is an equivalence relation on the set of
Kasparov modules. We denote the set of Kasparov modules by
${\sf E}(A,B)$. The subset containing degenerate modules is denoted
${\sf D}(A,B)$. Direct sum makes ${\sf E}(A,B)$ and
${\sf D}(A,B)$ into monoids.
\begin{definition}
Let $({\mathcal{E}}_i,\rho_i,\topsp{T}_i) \in {\sf E}(A,B)$ for
$i=0,1$, $({\mathcal{E}},\rho,\topsp{T}) \!\in\!
{\sf E}(A,\!B\otimes \C([0,1],\real))$, and let $f_t :
B\otimes \C([0,1],\real) \rightarrow B$ be the evaluation map
$f_t(g)=g(t)$. Then $({\mathcal{E}}_0,\rho_0,\topsp{T}_0)$ and
$({\mathcal{E}}_1,\rho_1,\topsp{T}_1)$ are
said to be \textit{homotopic} and $({\mathcal{E}},\rho,\topsp{T})$ is called a
\textit{homotopy} if $({\mathcal{E}} \otimes_{f_i}B,f_i \circ \rho,f_{i*}(\topsp{T}))$
is orthogonally equivalent to $({\mathcal{E}}_i,\rho_i,\topsp{T}_i)$ for $i=0,1$.
\label{homotopydef}\end{definition}
Homotopy is an equivalence relation on ${\sf E}(A,B)$ and we denote
the equivalence classes by $[{\mathcal{E}},\rho,\topsp{T}]$.
It is useful to consider special kinds of homotopy. If
${\mathcal{E}}=\C([0,1],{\mathcal{E}}_0)$,
${\mathcal{E}}_0={\mathcal{E}}_1$ and the induced maps $t\mapsto \topsp{T}_t,
\, t\mapsto \rho_t(a)$ for all $a \in A$ are strongly $*$-continuous,
then we call $({\mathcal{E}},\rho,\topsp{T})$ a \textit{standard homotopy}. If
in addition $\rho_t=\rho$ is constant and $\topsp{T}_t$ is norm continuous, then
$({\mathcal{E}},\rho,\topsp{T})$ is called an \textit{operator homotopy}. The following result holds, whose proof can be obtained by the analogous result in complex case \cite{blackadar}.
\begin{proposition}
Let $(\mathcal{E},\phi,\topsp{T})$ be an element in ${\sf D}(A,B)$. Then $(\mathcal{E},\phi,\topsp{T})$ is homotopic to the zero module.
\end{proposition}
We can now give the definition of the Kasparov's KKO-groups.
\begin{definition}
The set of equivalence classes in ${\sf E}(A,B)$ with respect to
homotopy of $(A,B)$-modules is denoted $\KKO(A,B)$ or
$\KKO_0(A,B)$. For $p,q \geq 0$ we define $$\KKO_{p,q}(A,B):=\KKO(A,B
\otimes \Cl_{p,q}) \ , $$ where $\Cl_{p,q}:=\Cl(\real^{p,q})$ is the
real Clifford algebra of the vector space $\real^{p+q}$ with quadratic
form of signature $(p,q)$.
\label{KKOgendef}\end{definition}
The equivalence relation allows us to simplify the $(A,B)$-modules
required to define $\KKO(A,B)$ \cite{blackadar}. Indeed, we need only consider modules
of the form $(B^\infty,\rho,\topsp{T})$ with $\topsp{T}=\topsp{T}^*$. If $A$ is unital,
we can further assume that $\lVert \topsp{T} \rVert \leq 1$ and
$\topsp{T}^2-1 \in \mathcal{K}(B^\infty)$.\\

There is another equivalence relation that we can define on
${\sf E}(A,B)$. We say that two $(A,B)$-modules $({\mathcal{E}}_i,\rho_i,\topsp{T}_i)$,
$i=0,1$ are \textit{stably operator homotopic},
$({\mathcal{E}}_0,\rho_0,\topsp{T}_0)\simeq_{\rm
  oh}({\mathcal{E}}_1,\rho_1,\topsp{T}_1)$, if there exist
$({\mathcal{E}}_i^{\prime},\rho_i^{\prime},\topsp{T}_i^{\prime})\in{\sf D}(A,B)$ such that
$({\mathcal{E}}^{\phantom{\prime}}_0 \oplus
{\mathcal{E}}_0^{\prime},\rho^{\phantom{\prime}}_0 \oplus
\rho_0^{\prime},\topsp{T}^{\phantom{\prime}}_0 \oplus \topsp{T}_0^{\prime})$ and
$({\mathcal{E}}^{\phantom{\prime}}_1 \oplus
{\mathcal{E}}_1^{\prime},\rho^{\phantom{\prime}}_1
\oplus \rho_1^{\prime},\topsp{T}^{\phantom{\prime}}_1 \oplus \topsp{T}_1^{\prime})$
are operator homotopic up to orthogonal equivalence.One can proove  that the set of
equivalence classes with respect to $\simeq _{\rm oh}$ coincides with
the set $\KKO(A,B)$ defined above. For this result to hold the hypothesis that $B$ is $\sigma$-unital is of particular importance.\\
The set $\KKO(A,B)$ is an abelian group, for any separable real $C^{*}$-algebra $A$ and any real $\sigma$-unital $C^{*}$-algebra $B$. Moreover, $\KKO(-,-)$ is a covariant bifunctor from the category of separable $C^{*}$-algebras into the category of abelian groups which is additive, i.e.
\bea
\KKO(A_1 \oplus A_2,B)&=&\KKO(A_1,B) \oplus \KKO(A_2,B) \ ,
\nonumber\\ \KKO(A,B_1 \oplus B_2)&=&\KKO(A,B_1) \oplus \KKO(A,B_2) \
. \nonumber
\eea
Namely,  any two $*$-homomorphisms $f:A_2 \rightarrow A_1$ and $g:B_1 \rightarrow
B_2$ induce group homomorphisms
\bea
f^*\,:\,\KKO(A_1,B) &\longrightarrow& \KKO(A_2,B) \ , \nonumber\\
g_*\,:\, \KKO(A,B_1) &\longrightarrow& \KKO(A,B_2) \nonumber
\eea
defined by
\bea
f^*[{\mathcal{E}},\rho,T]&=&[{\mathcal{E}},\rho \circ f,T] \ , \nonumber\\
g_*[{\mathcal{E}},\rho,T]&=&[{\mathcal{E}} \otimes _g B_2,\rho \otimes
1,T \otimes 1]  . \nonumber
\eea
Finally, any two homotopies $f_t: A_2 \rightarrow A_1$ and $g_t: B_1 \rightarrow
B_2$ induce the same homomorphism for all $t\in[0,1]$,
i.e. $f_t^*=f_0^*$ and $g_{t*}=g_{0*}$.
\subsection{Analytic KO-homology}
As mentioned in the introduction to this chapter, Kasparov's formalism represents a unified description of both KO-theory and KO-homology. Indeed, if we denote with $\KO_{p}(B)$ the algebraic KO-theory groups of a unital real $C^{*}$-algebra $B$, we have the following result \cite{blackadar,kasparov}.
\begin{theorem}
Let $B$ be a unital real $C^{*}$-algebra. Then, $\KKO(\mathbb{R},B)\simeq{\KO_{0}(B)}$ and $\KKO_{p,q}(\mathbb{R},{B})\simeq\KO_{p-q}(B)$
\end{theorem}
Recall that for a real unital $C^{*}$-algebra $A$, the algebraic K-theory group ${\rm KO}_{0}(A)$ is defined as the Grothendieck group of the monoid of unitarily equivalent projectors in $M_{\infty}(A)$, which is defined as the direct limit of $A$-valued matrix algebras $M_{n}(A)$ under the embedding $a\to{{\rm diag}(a,0)}$. The higher algebraic K-theory group can be defined by ${\rm KO}_{p}(A):={\rm KO}_{0}(C_{0}(\mathbb{R}^{p})\otimes{A})$.\\
If $\topsp{X}$ is a compact Hausdorff space, $\KKO_{p}(\mathbb{R},\C(\topsp{X},\mathbb{R}))\simeq{}\KO_{p}(\C(\topsp{X},\mathbb{R}))\simeq{\KO^{p}(\topsp{X})}$. On the other hand, using the Gel'fand transform the contravariant functor
$(\topsp{X},\ta) \mapsto \C(\topsp{X},\ta)$ induces an equivalence of categories
between the category of compact Hausdorff spaces with involution and
the category of commutative real $C^*$-algebras. Since
$\KKO_\sharp(-,\mathbb{R})$ is also a contravariant functor, it
follows that their composition $(\topsp{X},\ta)
\mapsto\KKO_\sharp(\C(\topsp{X},\ta),\mathbb{R})$ is a covariant functor. This motivates the following
\begin{definition}
Let $(\topsp{X},\ta)$ be a compact Hausdorff space with involution. The
\textit{analytic KO-homology groups of $(\topsp{X},\ta)$} are defined by
$$\KO_n^{\rm
  a}\bigl(\topsp{X}\,,\,\ta\bigr):=\KKO_{n,0}\bigl(\C(\topsp{X},\ta)\,,\,\mathbb{R}
\bigr)=\KKO\bigl(\C(\topsp{X},\tau)\,,\,\Cl_n\bigr). $$
\end{definition}
In the following, we will illustrate an alternative description of Kasparov's KO-homology groups for a real $C^{*}$-algebra $A$, referring to \cite{Higson} for more details.\\
This description is based on triples $({\mathcal{H}_\real},\rho,\topsp{T})$ which are defined
by the data:
\begin{romanlist}
\item ${\mathcal{H}_\real}$ is a separable real Hilbert space;
\item $\rho:A\to\mathcal{L}({\mathcal{H}_\real})$ is a *-representation of $A$; and
\item $\topsp{T}$ is a bounded linear operator on ${\mathcal{H}_\real}$.
\end{romanlist}
These are assumed to satisfy the following conditions:
\begin{romanlist}
\item ${\mathcal{H}_\real}$ is equipped with a
  $\mathbb{Z}_{2}$-grading such that $\rho(a)$ is even for all $a\in
  A$ and $\topsp{T}$ is odd;
\item For all $a\in A$ one has
\beq
\left(\topsp{T}^{2}-1\right)\,\rho(a)~,~
  \left(\topsp{T}-\topsp{T}^{*}\right)\,\rho(a)~, ~
\topsp{T}\,\rho(a)-\rho(a)\,\topsp{T} \in \mathcal{K}_\real \ ;
\label{TrhoKOarels}\eeq
and
\item There are odd $\real$-linear operators
  $\varepsilon_{1},\dots,\varepsilon_{n}$ on ${\mathcal{H}_\real}$
  with the $\Cl_n$ algebra relations
\beq
\varepsilon^{\phantom{*}}_{i}=-\varepsilon_{i}^{*} \ ,
\quad \varepsilon^{2}_{i}=-1 \ , \quad
\varepsilon^{\phantom{*}}_{i}\,\varepsilon^{\phantom{*}}_{j}+
\varepsilon^{\phantom{*}}_{j}\,\varepsilon^{\phantom{*}}_{i}=0
\label{Clnrels}\eeq
for $i\neq j$ such that $\topsp{T}$ and $\rho(a)$ commute with each
$\varepsilon_{i}$.
\end{romanlist} 
We shall refer to the triple $(\mathcal{H}_\real,\rho,\topsp{T})$ as
an \textit{$n$-graded Fredholm module}.\\
Let us denote by $\Gamma{\rm O}^{n}(A)$ the set of all $n$-graded
Fredholm modules over $A$. Consider the equivalence relation $\sim$ on
$\Gamma{\rm O}^{n}(A)$ generated by the relations:
\begin{itemize}
\item[]\emph{Orthogonal equivalence}:
  $({\mathcal{H}_\real^{\phantom{\prime}}},\rho,\topsp{T})\sim({\mathcal{H}_\real'},
\rho',\topsp{T}'\,)$ if and only if there exists an isometric
degree-preserving linear operator
$U:{\mathcal{H}^{\phantom{\prime}}_\real}\to{{\mathcal{H}_\real'}}$
such that $U\,\rho(a)=\rho'(a)\,U$ for all $a\in A$, $U\,\topsp{T}=\topsp{T}'\,U$, and
$U\,\varepsilon^{\phantom{\prime}}_{i}=\varepsilon'_{i}\,U$; and
\item[]\emph{Homotopy equivalence}:
  $({\mathcal{H}_\real},\rho,\topsp{T})\sim({\mathcal{H}_\real},\rho,\topsp{T}'\,)$ if
  and only if there exists a norm continuous function
$t\mapsto{\topsp{T}_{t}}$ such that $({\mathcal{H}_\real},\rho,\topsp{T}_{t})$ is a
Fredholm module for all $t\in[0,1]$ with $\topsp{T}_{0}=\topsp{T}$, $\topsp{T}_{1}=\topsp{T}'$.
\end{itemize}
We define the {\it direct sum} of two Fredholm modules
$({\mathcal{H}^{\phantom{\prime}}_\real},\rho,\topsp{T})$ and
$({\mathcal{H}_\real'},\rho',\topsp{T}'\,)$ to be the Fredholm module
$({\mathcal{H}^{\phantom{\prime}}_\real}\oplus{{\mathcal{H}_\real'}},
\rho\oplus\rho',\topsp{T}\oplus{\topsp{T}'}\,)$.

We may now define $\KO^{n}(A)$ as the free abelian group generated
by elements in $\Gamma{\rm O}^{n}(A)/{\sim}$ and quotiented by the ideal
generated by the set $\{[x_{0}\oplus{x_{1}}]-[x_{0}]-[x_{1}] \, | \,
[x_{0}],[x_{1}]\in\Gamma{\rm O}^{n}(A)/{\sim} \}$. In $\KO^{n}(A)$ the
\emph{inverse} of a class represented by the module
$({\mathcal{H}^{\phantom{\rm o}}_\real},\rho,\topsp{T})$ is given by
$({\mathcal{H}^{\rm o}_\real},\rho,\topsp{T})$, where ${\mathcal{H}^{\rm o}_\real}$
is the Hilbert space ${\mathcal{H}^{\phantom{\rm o}}_\real}$ with the opposite
$\zed_{2}$-grading and where the operators $\varepsilon_{i}$ reverse
their signs. Moreover, the operator $\topsp{\topsp{T}}$ for a triple $({\mathcal{H}_\real},\rho,\topsp{T})$ representing an element in $\KO^{n}(A)$ can be taken to be a Fredholm operator without loss of generality \cite{blackadar}.\\
For a compact Hausdorff space $\topsp{X}$ we define
\begin{displaymath}
\KO_{n}^{\rm a}\bigl(\topsp{X}\bigr):={\KO^{n}\bigl(\C(\topsp{X},\real)\bigr)}=
\KKO\bigl(\C(\topsp{X},\real)\,,\,\Cl_n\bigr) \ .
\end{displaymath}
Moreover, for a compact pair $(\topsp{X},\topsp{Y})$, we can define the higher relative KO-homology groups as
\begin{displaymath}
\KO_{n}^{\rm a}\bigl(\topsp{X},\topsp{Y}\bigr):={\KO^{n}\bigl(\C(\topsp{X/Y},\real)\bigr)}=
\KKO\bigl(\C(\topsp{X/Y},\real)\,,\,\Cl_n\bigr) \ .
\end{displaymath}
The groups $\KO_{n}^{\rm a}\bigl(\topsp{X},\topsp{Y}\bigr)$ enjoy Bott periodicity
\begin{displaymath}
\KO_{n}^{\rm a}\bigl(\topsp{X},\topsp{Y}\bigr)\simeq\KO_{n+8}^{\rm a}\bigl(\topsp{X},\topsp{Y}\bigr), \forall\:n\geq{0}
\end{displaymath}
which can be proven by using the periodicity of real Clifford algebras, and the isomorphisms
\begin{displaymath}
\KKO(A \otimes \Cl_{p,q},B\otimes \Cl_{r,s}) \simeq \KKO(A \otimes
\Cl_{p,q} \otimes \Cl_{r,s},B) \simeq \KKO(A \otimes
\Cl_{p-q+s-r,0},B)
\end{displaymath}
induced by the \emph{intersection product}, which we will not attempt to define here.\\
The term ``KO-homology'' for the groups $\KO_{n}^{\rm a}\bigl(\topsp{X},\topsp{Y}\bigr)$ is justified by the following theorem \cite{baum-2007-3,Higson}.
\begin{theorem}\label{kasparov}\textbf{(Kasparov)} There are connecting homomorphisms
\begin{displaymath}
\del :{\KO_{n}^{\rm a}\bigl(\topsp{X},\topsp{Y}\bigr)}\to\KO_{n-1}^{\rm a}\bigl(\topsp{X},\topsp{Y}\bigr)
\end{displaymath} 
which are compatible with Bott periodicity, and give Kasparov KO-homology the structure of a $\mathbb{Z}_{8}$-graded homology theory on the category of CW-complex pairs $(\topsp{X},\topsp{Y})$. On the subcategory of finite CW-complex pairs Kasparov's KO-homology is isomorphic to spectral KO-homology.
\end{theorem}
In particular, Theorem \ref{kasparov} implies that
\begin{displaymath}
\KO_{n}^{\rm a}({\rm pt})\simeq\KO^{-n}({\rm pt})
\end{displaymath}
This can indeed be proven by directly computing the groups $\KKO(\mathbb{R},\Cl_{n})$; see \cite{blackadar} for such a computation in the complex case.
\section{Geometric KO-homology}
As we have seen in the previous section, analytic KO-homology gives a representation of spectral KO-homology based on $C^{*}$-algebras and linear operators on Hilbert spaces. In this section we will develop a geometric version of KO-homology, which is analogous to the Baum-Douglas construction of K-homology \cite{Baumdouglas,Baumdouglas2,Reis2006}. Indeed, we will prove directly its homological properties by comparing it with other formulations of KO-homology.
\begin{definition}
Let $\topsp{X}$ be a finite CW-complex. A {\it KO-cycle} on $\topsp{X}$ is a triple $(\topsp{M},\topsp{E},\phi)$ where
\begin{romanlist}
\item $\topsp{M}$ is a compact spin manifold without boundary;
\item $\topsp{E}$ is a real vector bundle over $\topsp{M}$; and
\item $\phi : \topsp{M} \rightarrow \topsp{X}$ is a continuous map.
\end{romanlist}
\label{kcycle}\end{definition}
\noindent
There are no connectedness requirements made upon $\topsp{M}$, and hence the
bundle $\topsp{E}$ can have different fibre dimensions on the different
connected components of $\topsp{M}$. It follows that disjoint union
$$(\topsp{M}_1,\topsp{E}_1, \phi_1) \amalg (\topsp{M}_2,\topsp{E}_2,\phi_2):=(\topsp{M}_1 \amalg \topsp{M}_2,\topsp{E}_1 \amalg \topsp{E}_2,
\phi_1 \amalg \phi_2)$$ is a well-defined operation on the set of
KO-cycles on $\topsp{X}$, which we will denote with $\Gamma{\rm O}(\topsp{X})$.\\
In the following we will consider some equivalence relations on the set $\Gamma {\rm O}(\topsp{X})$.
\begin{definition}
Two KO-cycles $(\topsp{M}_1,\topsp{E}_1,\phi_1)$ and $(\topsp{M}_2,\topsp{E}_2,\phi_2)$ on $\topsp{X}$ are
  {\it isomorphic} if there exists a diffeomorphism $h : \topsp{M}_1 \rightarrow \topsp{M}_2$
  such that
\begin{romanlist}
\item $h$ preserves the spin structures;
\item $h^{*}(\topsp{E}_2) \simeq \topsp{E}_1$ as real vector bundles; and
\item The diagram
  $$\xymatrix{ \topsp{M}_1 \ar[r]^h  \ar[rd]_{\phi_1} & \topsp{M}_2 \ar[d]^{\phi_2}\\
     & \topsp{X} }$$
commutes.
\end{romanlist}
\label{isoK-cycles} \end{definition}
\begin{definition}
Two KO-cycles $(\topsp{M}_1,\topsp{E}_1,\phi_1)$ and $(\topsp{M}_2,\topsp{E}_2,\phi_2)$ on $\topsp{X}$ are
{\it spin bordant} if there exists a compact spin manifold $\topsp{W}$ with
boundary, a real vector bundle $\topsp{E} \rightarrow \topsp{W}$, and a continuous
map $\phi : \topsp{W} \rightarrow \topsp{X}$ such that the two KO-cycles
$$\bigl(\partial \topsp{W}\,,\, \topsp{E}|_{\partial \topsp{W}}\,,\,\phi|_{\partial \topsp{W}}\bigr)
\quad , \quad \bigl(\topsp{M}_1\amalg (-\topsp{M}_2)\,,\, \topsp{E}_1 \amalg \topsp{E}_2\,,
\, \phi_1 \amalg \phi_2\bigr)$$ are
isomorphic, where $-\topsp{M}_2$ denotes $\topsp{M}_2$ with the opposite class in $\text{\topsp{H}}^{1}(\topsp{\topsp{M}};\mathbb{Z}_{2})$ representing the spin structure on
its tangent bundle $\topsp{TM}_2$. The triple $(\topsp{W},\topsp{E},\phi)$ is called
a {\it spin bordism} of KO-cycles. \label{bord}
\end{definition}
We finally introduce the last equivalence relation we will need to define geometric KO-homology. Let $\topsp{M}$ be a spin manifold and $\topsp{F}\to \topsp{M}$ a $\C^\infty$ real
spin vector bundle with fibres of dimension $n:=\textrm{dim}_\real \, \topsp{F}_p
\equiv 0~\textrm{mod}~8$ for $p\in \topsp{M}$. Let
${{1\!\!1}}_{\topsp{M}}^\real:=\topsp{M} \times \mathbb{R}$ denote the trivial real line
bundle over $\topsp{M}$. Then $\topsp{F} \oplus {{1\!\!1}}_{\topsp{M}}^\real$ is a real vector
bundle over $\topsp{M}$ with fibres of dimension $n+1$ and projection map
$\lambda$. By choosing a $\C^{\infty}$ metric on it, we may define
 \beq 
 \widehat{\topsp{M}}=\S\bigl(\topsp{F} \oplus \id_{\topsp{M}}^\real\bigr)
 \label{hatMdef}
 \eeq  
where ${\S}\bigl(\topsp{F} \oplus {{1\!\!1}}_{\topsp{M}}^\real\bigr)$ denotes the spere bundle of $\topsp{F} \oplus {{1\!\!1}}_{\topsp{M}}^\real$. The tangent bundle
of $\topsp{F} \oplus {{1\!\!1}}_{\topsp{M}}^\real$ fits into an exact sequence of bundles given
by
$$
0~\longrightarrow~\lambda^*\bigl(\topsp{F} \oplus {{1\!\!1}}_{\topsp{M}}^\real\bigr)~
\longrightarrow~\topsp{T}\bigl(\topsp{F} \oplus {{1\!\!1}}_{\topsp{M}}^\real\bigr)~
\longrightarrow~\lambda^*\bigl(\topsp{T}\topsp{M}\bigr)~\longrightarrow~0 \ ,
$$
as for a vector bundle $\topsp{E}\xrightarrow{\pi}\topsp{M}$ the vertical tangent bundle to the fibration $\pi$ is isomorphic to the vector bundle $\pi^{*}\topsp{E}\to\topsp{E}$. Upon choosing a splitting, it follows that
\begin{displaymath}
\topsp{T}\bigl(\topsp{F}\oplus{{1\!\!1}}_{\topsp{M}}^\real\bigr)\simeq\lambda^{*}\bigl(\topsp{TM}\bigr)\oplus\lambda^{*}\bigl(\topsp{F} \oplus {{1\!\!1}}_{\topsp{M}}^\real\bigr)
\end{displaymath}
and hence the spin structures on $\topsp{TM}$ and $\topsp{F}\oplus{{1\!\!1}}_{\topsp{M}}^\real$ determine a spin structure on $\topsp{T}\bigl(\topsp{F}\oplus{{1\!\!1}}_{\topsp{M}}^\real\bigr)$. Since the sphere bundle $\S\bigl(\topsp{F} \oplus {{1\!\!1}}_{\topsp{M}}^\real\bigr)$ is the boundary of the disk bundle $\B\bigl(\topsp{F} \oplus {{1\!\!1}}_{\topsp{M}}^\real\bigr)$, and using the fact that we can equip $\B\bigl(\topsp{F} \oplus {{1\!\!1}}_{\topsp{M}}^\real\bigr)$ with the spin structure induced by that on the total space of $\topsp{T}\bigl(\topsp{F}\oplus{{1\!\!1}}_{\topsp{M}}^\real\bigr)$, it follows that $\widehat{\topsp{M}}$ is a compact spin manifold. By construction, $\widehat{\topsp{M}}$ is a sphere
bundle over $\topsp{M}$ with $n$-dimensional spheres $\topsp{S}^n$ as fibres. We
denote the bundle projection by \beq \pi \,:\, \widehat{\topsp{M}}
~\longrightarrow~ \topsp{M} \ . \label{pidef}\eeq We may
regard the total space $\widehat{\topsp{M}}$ as consisting of two
  copies $\B^\pm(\topsp{F})$, with opposite spin structures, of the unit ball
  bundle $\B(\topsp{F})$ of $\topsp{F}$ glued together by the identity map
  $\Id_{\S(\topsp{F})}$ on its boundary so that
\beq
\widehat{\topsp{M}} = \B^+(\topsp{F}) \cup_{\S(\topsp{F})} \B^-(\topsp{F}) \ .
\label{hatMBF}\eeq

Since $n \equiv 0~\textrm{mod}~8$, the group Spin$(n)$ has two
irreducible real half-spin representations. The spin structure on $\topsp{F}$
associates to these representations real vector bundles $\topsp{S}_0(\topsp{F})$ and
$\topsp{S}_1(\topsp{F})$ of equal rank $2^{n/2}$ over $\topsp{M}$. Their Whitney sum $\topsp{S}(\topsp{F})=\topsp{S}_0(\topsp{F})\oplus
\topsp{S}_1(\topsp{F})$ is a bundle of real Clifford modules over $\topsp{TM}$ such that
$\Cl(\topsp{F})\simeq{\rm End}\,\topsp{S}(\topsp{F})$, where $\Cl(\topsp{F})$ is the real Clifford
algebra bundle of $\topsp{F}$. Let $\nslash{\topsp{S}}^+(\topsp{F})$ and $\nslash{\topsp{S}}^-(\topsp{F})$ be
the real spinor bundles over $\topsp{F}$ obtained from pullbacks to $\topsp{F}$ by the
  bundle projection $\topsp{F} \rightarrow \topsp{M}$ of $\topsp{S}_0(\topsp{F})$ and $\topsp{S}_1(\topsp{F})$,
  respectively. Clifford multiplication induces a bundle map $\topsp{F}
  \otimes \topsp{S}_0(\topsp{F})\rightarrow \topsp{S}_1(\topsp{F})$ that defines a vector bundle map
  $\sigma :\nslash{\topsp{S}}^+(\topsp{F})\rightarrow\nslash{\topsp{S}}^-(\topsp{F})$ covering $\Id_\topsp{F}$
  which is an isomorphism outside the zero section of $\topsp{F}$. Since the
  ball bundle $\B(\topsp{F})$ is a sub-bundle of $\topsp{F}$, we may form real spinor
  bundles over $\B^\pm(\topsp{F})$ as the restriction bundles $\Delta^\pm(\topsp{F})
=\nslash{\topsp{S}}^\pm(\topsp{F})|_{\B^\pm(\topsp{F})}$. We can then glue $\Delta^{+}(\topsp{F})$ and
$\Delta^{-}(\topsp{F})$ along $\S(\topsp{F})=\partial\B(\topsp{F})$ by the Clifford
multiplication map $\sigma$ giving a real vector bundle over
$\widehat{\topsp{M}}$ defined by
\beq
\topsp{H}(\topsp{F}) = \Delta^{+}(\topsp{F})\cup_{\sigma}\Delta^{-}(\topsp{F}) \ .
\label{HFdef}\eeq
For each $p\in \topsp{M}$, the bundle $\topsp{H}(\topsp{F})|_{ \pi^{-1}(p)}$ is the real
  Bott generator vector bundle over the $n$-dimensional sphere
  $\pi^{-1}(p)$.\\
In particular, the class $[\topsp{H(F)}]\in\KO^{0}(\widehat{\topsp{M}})
$ is the image of the Thom class of $\topsp{F}$ $\tau_{\topsp{F}}\in\KO^{0}(\B^{+}(\topsp{F}),\S(\topsp{F}))$ under the composition of the homomorphisms
\begin{displaymath}
\KO^{0}(\B^{+}(\topsp{F}),\S(\topsp{F}))\to\KO^{0}(\widehat{\topsp{M}},\B^{-}(\topsp{F}))\to\KO^{0}(\widehat{\topsp{M}})
\end{displaymath}
where the first map is given by excision, and the second map is given by restriction.
\begin{definition}
Let $(\topsp{M},\topsp{E},\phi)$ be a KO-cycle on $\topsp{X}$ and $\topsp{F}$ a $\C^\infty$ real
spin vector bundle over $\topsp{M}$ with fibres of dimension $\textrm{dim}_\real
\, \topsp{F}_p \equiv 0~\textrm{mod}~8$ for $p\in \topsp{M}$. Then the process of obtaining
the KO-cycle $(\,\widehat{\topsp{M}},\topsp{H}(\topsp{F}) \otimes \pi^{*}(\topsp{E}),\phi \circ
\pi)$ from $(\topsp{M},\topsp{E},\phi)$ is called \it{real vector bundle modification}.
    \label{VBM} \end{definition}
We are now ready to define the geometric KO-homology groups of
the space $\topsp{X}$.

\begin{definition}
  The \textit{geometric KO-homology group of
  $\topsp{X}$} is the abelian group obtained from quotienting $\Gamma
{\rm O}(\topsp{X})$ by the equivalence relation $\sim$ generated by the
relations of
\begin{romanlist}
\item isomorphism;
\item spin bordism;
\item direct sum: if $\topsp{E}= \topsp{E}_1 \oplus \topsp{E}_2$, then $(\topsp{M},\topsp{E},\phi) \sim (\topsp{M},\topsp{E}_1,\phi)
  \amalg (\topsp{M},\topsp{E}_2,\phi)$; and
\item real vector bundle modification.
\end{romanlist}
  The group operation is induced by disjoint union of KO-cycles. We denote this
  group by $\KO_\sharp^{\rm t}(\topsp{X}):=\Gamma {\rm O}(\topsp{X}) \, / \sim$,
  and the homology class of the KO-cycle $(\topsp{M},\topsp{E},\phi)$ by
  $[\topsp{M},\topsp{E},\phi]\in\KO^{\rm t}_\sharp(\topsp{X})$. \label{Tgroup}
\label{KOtdef}\end{definition}
\indent Since the equivalence relation on $\Gamma {\rm O}(\topsp{X})$
preserves the dimension of $\topsp{M}$ $\textrm{mod}~8$ in KO-cycles
$(\topsp{M},\topsp{E},\phi)$, one can define the subgroups $\KO_n^{\rm t}(\topsp{X})$
 consisting of classes of KO-cycles
$(\topsp{M},\topsp{E},\phi)$ for which all connected components $\topsp{M}_i$ of $\topsp{M}$ are
of dimension $\textrm{dim} \, \topsp{M}_i \equiv n~\textrm{mod}~8$.
Then
\beq
\KO_\sharp^{\rm t}(\topsp{X}) = \bigoplus _{n=0}^7 \,\KO_{n}^{\rm t}(\topsp{X})
\label{KOtgrade}\eeq
has a natural $\zed_8$-grading.\\
The geometric construction of KO-homology is functorial. If $f:\topsp{X}
\rightarrow \topsp{Y}$ is a continuous map, then the induced
  homomorphism $$f_{*} \,:\, \KO_\sharp^{\rm t}(\topsp{X}) ~\longrightarrow~
  \KO_\sharp^{\rm t}(\topsp{Y})$$
of $\zed_8$-graded abelian groups is given on classes of KO-cycles
$[\topsp{M},\topsp{E},\phi]\in\KO^{\rm
  t}_\sharp(\topsp{X})$ by $$f_{*}[\topsp{M},\topsp{E},\phi] := [\topsp{M},\topsp{E},f\circ\phi] \ . $$
One has $(\Id_\topsp{X})_*=\Id_{\KO^{\rm t}_\sharp(X)}$ and $(f\circ
g)_*=f_*\circ g_*$. Since real vector bundles over $\topsp{M}$ extend to real
vector bundles over $\topsp{M}\times[0,1]$, it follows by spin bordism that
induced homomorphisms depend only on their homotopy classes.
If $\pt$ denotes a one-point topological space, then the
collapsing map $\zeta:\topsp{X} \rightarrow \pt$ induces an
epimorphism \beq \zeta_{*} \,:\, \KO_\sharp^{\rm t}(\topsp{X})
~\longrightarrow ~ \KO_\sharp^{\rm t}(\pt)  \ .
\label{collapseepi}\eeq The {\it reduced} geometric KO-homology
group of $\topsp{X}$ is \beq \widetilde{\KO}{}_\sharp^{\,\rm t}(\topsp{X}):=
\ker\zeta_{*} \ . \label{redKhom}\eeq Since the map
(\ref{collapseepi}) is an epimorphism with left inverse induced by
the inclusion of a point $\iota:\pt\hookrightarrow \topsp{X}$, one has
$\KO_\sharp^{\rm t}(\topsp{X})\simeq\KO_\sharp^{\rm
t}(\pt)\oplus\widetilde{\KO}{}_\sharp^{\,\rm t}(\topsp{X})$ for any
space $\topsp{X}$.\\

As for the complex case \cite{Reis2006}, the abelian group $\KO_\sharp^{\rm t}(\topsp{X})$ is generated by classes $[\topsp{M},\topsp{E},\phi]$ where $\topsp{M}$ is connected, and if $\{\topsp{X}_{j}\}_{j \in J}$ is the set of connected components of $\topsp{X}$
then $$\KO_\sharp^{\rm t}(\topsp{X})= \bigoplus _{j \in J}\,\KO_\sharp^{\rm
t}(\topsp{X}_{j}) \ . $$
Moreover, the homology class of a cycle $(\topsp{M},\topsp{E},\phi)$ on $\topsp{X}$ depends only
  on the KO-theory class of $\topsp{E}$ in $\KO^{0}(\topsp{M})$, and on the homotopy class of $\phi$ in $[\topsp{M},\topsp{X}]$.
\subsection{Homological properties of $\KO_\sharp^{\rm t}$}
In the previous section we have constructed a covariant functor $\KO_\sharp^{\rm t}$ from the category of finite CW-complexes to the category of abelian groups, which is homotopy invariant. We will now establish that this construction actually yields a (generalized) homology theory, and in particular is the dual homology to KO-theory. The main strategy consists in ``comparing'' the functors $\KO_{n}^{\rm t}$ with a realization of spectral KO-homology developed in \cite{jakob}, which we will denote with $\KO_{\sharp}^{'}$. Namely, for each pair $(\topsp{X},\topsp{Y})$ we will construct a map $\mu^{\rm s}
:\KO_n^{\rm t}(\topsp{X},\topsp{Y}) \rightarrow\KO^{\prime}_n(\topsp{X},\topsp{Y})$ for each $n\in\mathbb{Z}$ which defines a
natural equivalence between functors on the category of topological
spaces having the homotopy type of finite CW-pairs.\\
The set of cycles for $\KO_{\sharp}^{'}(\topsp{X})$ is given by triples $(\topsp{M},x,\phi)$ as in Definition \ref{kcycle}, but with $x\in\KO^{i}(\topsp{M})$ being a KO-theory class over $\topsp{M}$ such that dim \topsp{M}$\equiv i+n$ mod 8. The equivalence relations are as in the previous section, apart from real vector bundle modification, which is modified from Definition \ref{VBM} as follows. The nowhere zero section
\begin{equation}\label{sigmadef}
\Sigma^\topsp{F}\,:\,\topsp{M}~\longrightarrow~ \topsp{F}\oplus\id_{\topsp{M}}^\real
\end{equation}
defined by
$$
\Sigma^\topsp{F}(p)=0_p\oplus1
$$
for $p\in \topsp{M}$ induces an embedding
\beq
\Sigma^\topsp{F}\,:\,\topsp{M}~\hookrightarrow~\widehat{\topsp{M}} \ .
\label{secembind}\eeq
Then real vector bundle modification is replaced by the relation
$$\bigl(\topsp{M}\,,\,x\,,\,\phi\bigr)\sim\bigl(\,\widehat{\topsp{M}}\,,\,
\Sigma_!^\topsp{F}(x)\,\,,\phi\circ\pi\bigr) \ ,
$$ where the functorial homomorphism
$\Sigma_!^\topsp{F}:\KO^i(\topsp{M})\to\KO^{i+r}(\,\widehat{\topsp{M}}\,)$ is the Gysin map induced
by the embedding (\ref{secembind}), with $r={\rm rank(F)}$. By construction, the normal bundle to $\topsp{M}$ in $\widehat{\topsp{M}}$ can be identified with $\topsp{F}$. Since $\topsp{H}(\topsp{F})$ is the image of the Thom class of $\topsp{F}$, by the definition of Gysin morphism in section \ref{Kmanifold} it follows that on stable isomorphism classes of
real vector bundles $[\topsp{E}]\in\KO^0(\topsp{M})$ one has
\beq
\Sigma_!^\topsp{F}\bigl[\topsp{E}\bigr]=\bigl[\topsp{H}(\topsp{F})\otimes\pi^*(\topsp{E})\bigr] \ .
\label{secstable}\eeq  
To compare $\KO_{\sharp}^{\rm t}$ with $\KO_{\sharp}^{'}$, we define $\KO_{n+8k}^{\rm t}(\topsp{X}) :=  \KO_{n}^{\rm t}(\topsp{X})$ for
all $k \in \mathbb{Z}$, $ 0\leq n \leq 7$. Moreover, we give a spin bordism description of the relative geometric KO-homology groups $\KO_n^{\rm t}(\topsp{X},\topsp{Y})$ as follows. We consider the set $\Gamma{\rm O}(\topsp{X},\topsp{Y})$ of isomorphism
classes of triples $(\topsp{M},\topsp{E},\phi)$ where
\begin{romanlist}
\item $\topsp{M}$ is a compact spin manifold with (possibly empty) boundary;
\item $\topsp{E}$ is a real vector bundle over $\topsp{M}$; and
\item $\phi : \topsp{M} \rightarrow \topsp{X}$ is a continuous map with
  $\phi(\partial \topsp{M})\subset \topsp{Y}$.
\end{romanlist}
The set $\Gamma{\rm O}(\topsp{X},\topsp{Y})$ is then quotiented by relations of
relative spin bordism, which is modified from Definition~\ref{bord} by
the requirement that $\topsp{M}_1\amalg(-\topsp{M}_2)\subset\partial \topsp{W}$ is a regularly
embedded submanifold of codimension~$0$ with $\phi(\partial \topsp{W}\setminus
\topsp{M}_1\amalg(-\topsp{M}_2))\subset \topsp{Y}$, direct sum, and real vector bundle
modification, which is applicable in this case since $\S(\topsp{F}
\oplus \id_{\topsp{M}}^\real)$ is a compact spin manifold with boundary
$\S(\topsp{F} \oplus \id_{\topsp{M}}^\real) | _{\partial \topsp{M}}$. The
collection of equivalence classes is a $\zed_8$-graded abelian
group with operation induced by disjoint union of relative
KO-cycles. One has $\KO_i^{\rm t}(\topsp{X},\emptyset) = \KO_i^{\rm t}(\topsp{X}).$\\
We can finally prove the following
\begin{theorem}\label{equiv}
The map
$$\mu^{\rm s} \,:\, \KO_n^{\rm t}(\topsp{X},\topsp{Y})
~\longrightarrow~\KO^{\prime}_n(\topsp{X},\topsp{Y})$$
defined on classes of KO-cycles by
$$\mu^{\rm s}\bigl[\topsp{M}\,,\,\topsp{E}\,,\,\phi\bigr]_{\rm t}=
\bigl[\topsp{M}\,,\,[\topsp{E}]\,,\,\phi\bigr]_{'}$$
is an isomorphism of abelian groups which is natural with
respect to continuous maps of pairs.
\label{specequivthm}\end{theorem}
\begin{proof}
Taking into account the equivalence relations on $\Gamma
{\rm O}(\topsp{X},\topsp{Y})$ used to define both KO-homology groups, the map
$\mu^{\rm s}$ is well-defined and a group homomorphism.
Let $[\topsp{M},x,\phi]_{\prime} \in \KO^{\prime}_n(\topsp{X},\topsp{Y})$ with $m:= \dim \topsp{M}$. We may
assume that  $\topsp{M}$ is connected and $x$ is non-zero in $\KO^i(\topsp{M})$. Then
$m-i \equiv n~\rm{mod}~8$. Consider the trivial spin vector bundle
$\topsp{F}=\topsp{M} \times \mathbb{R}^{n+7m}$ over $\topsp{M}$. In this case the
sphere bundle (\ref{hatMdef}) is $\widehat{\topsp{M}}=\topsp{M}\times\S^{n+7m}$ and
the associated Gysin homomorphism in KO-theory is a map
$$ \Sigma_{!}^{\topsp{F}} \,:\, \KO^i\big(\topsp{M}\big)~
  \longrightarrow~\KO^{i+7m+n}\big(\,\widehat{\topsp{M}}\,\big) \ . $$
  Since $i+7m+n\equiv(i+7m+m-i)~{\rm mod}~ 8\equiv 0~{\rm mod}~
  8$, one has $\KO^{i+7m+n}(\,\widehat{\topsp{M}}\,) \simeq
  \KO^{0}(\,\widehat{\topsp{M}}\,)$. It follows that there are real vector bundles
  $\topsp{E},\topsp{H} \rightarrow \widehat{\topsp{M}}$ such that
  $\Sigma_{!}^{\topsp{F}}(x)=[\topsp{E}]-[\topsp{H}]$, and so by real vector bundle modification
  one has $$[\topsp{M},x,\phi]_{\prime}=[\,\widehat{\topsp{M}},[\topsp{E}],\phi \circ
  \pi]_{\prime}-[\,\widehat{\topsp{M}},[\topsp{H}],\phi \circ \pi]_{\prime}$$ in
  $\KO^{\prime}_n(\topsp{X},\topsp{Y})$. Therefore $\mu^{\prime}(\,[\,\widehat{\topsp{M}},\topsp{E},\phi \circ
  \pi]_{\rm t}-[\,\widehat{\topsp{M}},\topsp{H},\phi \circ\pi]_{\rm
    t}\,)=[\topsp{M},x,\phi]_{\prime}$, and we conclude
  that $\mu^{\rm s}$ is an epimorphism.

Now suppose that $\mu^{\rm s}[\topsp{M}_1,\topsp{E}_1,\phi_1]_{\rm t}=\mu^{\rm
  s}[\topsp{M}_2,\topsp{E}_2,\phi_2]_{\rm t}$ are identified in $\KO'_n(\topsp{X},\topsp{Y})$ through
real vector bundle modification. Then, for instance, there is a real spin
vector bundle $\topsp{F} \rightarrow \topsp{M}_1$ such that $\topsp{M}_2 = \widehat{\topsp{M}_1}$ and
$[\topsp{E}_2]=\Sigma_{!}^{\topsp{F}}[\topsp{E}_1]$. Since $$\Sigma_!^\topsp{F}\bigl[\topsp{E}_{1}\bigr]=\bigl[\topsp{H}(\topsp{F})\otimes\pi^*(\topsp{E}_{1})\bigr]$$ and since the class $[\topsp{M}_2,\topsp{E}_2,\phi_2]_{\rm t}$ depends only on the KO-theory class $[\topsp{E}_{2}]$, it follows that the homology
  classes $[\topsp{M}_1,\topsp{E}_1,\phi_1]_{\rm t}$ and $[\topsp{M}_2,\topsp{E}_2,\phi_2]_{\rm t}$ are also identified in $\KO_n^{\rm t}(\topsp{X},\topsp{Y})$ through real
  vector bundle modification. As this is the only relation in
  $\KO^{\prime}_n(\topsp{X},\topsp{Y})$ that might identify these classes without
  identifying them as KO-cycles, we conclude that
  $\mu^{\rm s}$ is a monomorphism and therefore an isomorphism.
\end{proof}
Since $\KO_{\sharp}^{'}$ is a homological realization of the homology theory associated with KO-theory, we have thus established that geometric KO-homology is a generalized homology theory which is equivalent to spectral KO-homology. In particular, it enjoys the standard homological properties, such as the existence of a long exact sequence for any pair $(\topsp{X},\topsp{Y})$. Moreover, the connecting homomorphism  
$$
\partial\,:\,\KO^{\rm t}_n(\topsp{X},\topsp{Y})~\longrightarrow~\KO^{\rm t}_{n-1}(\topsp{Y})
$$
is given by the boundary map
\beq
\partial[\topsp{M},\topsp{E},\phi]:=[\partial \topsp{M},\topsp{E}|_{\partial \topsp{M}},\phi|_{\partial \topsp{M}}]
\label{bdrymap}\eeq
on classes of KO-cycles and extended by linearity.\\
Other homological properties are direct translations of those of the
complex case provided by \cite{Reis2006}, to which we refer for details.
\subsection{Products and Poincar\'{e} Duality}
Since we have established that geometric KO-homology is a representation of spectral KO-homology, we can define products and dualities with KO-theory.\\
The \emph{cap product} is defined as 
the $\zed_8$-degree
preserving bilinear pairing $$\,\cap\,\,:\,\KO^{0}(\topsp{X})\otimes \KO_\sharp^{\rm
t}(\topsp{X})~\longrightarrow~\KO_\sharp^{\rm t}(\topsp{X})$$ given for any real
vector bundle $\topsp{F}\to \topsp{X}$ and KO-cycle class $[\topsp{M},\topsp{E},\phi]\in\KO^{\rm
  t}_\sharp(\topsp{X})$ by 
\begin{equation}\label{cap}
[\topsp{F}] \,\cap\, [\topsp{M},\topsp{E},\phi]:=[\topsp{M},\phi^{*}\topsp{F}\otimes \topsp{E} ,\phi]
\end{equation}
and extended linearly. In particular, it makes $\KO_\sharp^{\rm t}(\topsp{X})$ into a
module over the ring $\KO^0(\topsp{X})$.\\
As in the complex case, this product can be
extended to a bilinear form $$ \,\cap\,\,:\,\KO^i(\topsp{X})\otimes\KO^{\rm
  t}_j(\topsp{X},\topsp{A})~\longrightarrow~ \KO^{\rm t}_{j-i}(\topsp{X},\topsp{A}) \ . $$
defined as
\begin{equation}
x\cap{[\topsp{M},\topsp{E},\phi]}:=(\mu^{s})^{-1}([\topsp{M},x\cup\phi^{*}[\topsp{E}],\phi]_{'})
\end{equation}
which coincides with (\ref{cap}) when $x=[\topsp{F}]$.\\

If $\topsp{X}$ and $\topsp{Y}$ are spaces, then the {\it exterior product}
$$ \times \,:\, \KO^{\rm t}_{i}(\topsp{X}) \otimes \KO^{\rm t}_{j}(\topsp{Y}) ~\longrightarrow~
\KO^{\rm t}_{i+j}(\topsp{X} \times \topsp{Y})$$ is given for classes of
KO-cycles $[\topsp{M},\topsp{E},\phi]\in\KO^{\rm t}_i(\topsp{X})$ and
$[\topsp{N},\topsp{F},\psi]\in\KO^{\rm t}_j(\topsp{Y})$ by $$\big[\topsp{M},\topsp{E},\phi\big]\times
\big[\topsp{N},\topsp{F},\psi\big]:=\big[\topsp{M} \times \topsp{N},\topsp{E} \boxtimes \topsp{F},
(\phi,\psi)\big] \ , $$ where $\topsp{M} \times \topsp{N}$ has the product spin
structure uniquely induced by the spin structures on $\topsp{M}$ and $\topsp{N}$,
and $\topsp{E} \boxtimes \topsp{F}$ is the real vector bundle over $\topsp{M} \times \topsp{N}$ with
fibres $(\topsp{E} \boxtimes \topsp{F})_{(p,q)}=\topsp{E}_p \otimes \topsp{F}_q$ for $(p,q)\in
\topsp{M}\times \topsp{N}$. This product is natural with respect to continuous
maps. Unfortunately, in contrast to the complex case, we don't have 
a version of the K\"unneth theorem for KO-homology. Indeed, should such
a formula exist, one could use it to show that
$\KO_{\sharp}(\pt)\otimes{\KO_{\sharp}}(\pt)$ has to be a tensor
product as modules over the ring $\KO^{\sharp}(\pt)$. But this does not
work correctly as pointed out by Atiyah \cite{AtiyahKunneth}.\\
Let $\topsp{M}$ be a spin manifold of dimension $n$. The class 
\begin{displaymath}
[\topsp{M}]:=[\topsp{M},\id_{\topsp{M}}^{\real},\Id_\topsp{M}]\in\KO^{\rm t}_{n}(\topsp{M},\partial \topsp{M}) 
\end{displaymath}
is called the \emph{fundamental class}, and it induces the following \emph{Poincar\'e Duality} isomorphism \cite{jakob}
\begin{displaymath}
\begin{array}{cccc}
 \Phi_\topsp{M}\,:&\,\KO^i(\topsp{M})&{\xrightarrow{\:\:\simeq\:\:}}&
 \KO^{\rm t}_{n-i}(\topsp{M},\partial \topsp{M})\\
 &\xi&{\longrightarrow}&\xi\cap[\topsp{M}]
\end{array}
\end{displaymath}
The fundamental class is uniquely determined.\\
Finally, let $\topsp{V}\xrightarrow{\pi}{\topsp{X}}$ be a KO-oriented vector bundle of rank $r$ over a space $\topsp{X}$ with Thom class $\tau_{\topsp{V}}$. In analogy to K-theory, we have the homological version of Thom isomorphism
\begin{equation}
\mathfrak{T}_{\topsp{X},\topsp{V}}^{*}\,:\,\KO^{\rm t}_{i+r}\big(\B(\topsp{V})\,,\,\S(\topsp{V})\big) ~
\stackrel{\approx}{\longrightarrow}~\KO^{\rm
  t}_{i}\big(\topsp{X}\big) \ .
\label{ThomhomdefV}\end{equation}
defined as
\begin{equation}\label{Thomdef}
\mathfrak{T}_{\topsp{X},\topsp{V}}^{*}([\topsp{M},\topsp{E},\phi]):=\tilde{\pi}_{*}(\tau_{\topsp{V}}\cap[\topsp{M},\topsp{E},\phi])
\end{equation}
where $\tilde{\pi}:\B(\topsp{V})\to{\topsp{X}}$ is the bundle projection induced by $\pi$.\\
We conclude this section by noticing that all the above constructions have an equivalent description in analytic KO-homology.
\section{K-homology and Index Theorems}
We have seen in the previous sections that both analytic KO-homology and geometric KO-homology are representations of spectral KO-homology. It follows straightforwardly from compositions of natural equivalences that analytic KO-homology and geometric KO-homology are naturally equivalent. It is interesting, at this point, to ask if such natural isomorphism can be induced by a map defined \emph{at the level of the cycles}. That this is indeed the case can be viewed as the primordial formulation of the index theorem, and it is the philosophy proposed in \cite{Baumdouglas,Baumdouglas2}. Since this point of view will be crucial in both the construction and the proof of the equivalence carried in the next section, we will briefly make the above statement more precise by illustrating the case of even complex K-homology, directing the reader to \cite{Baumdouglas,Baumdouglas2,Higson} for more information.\\  
Let $\topsp{M}$ be an even dimensional compact $\rm spin^{c}$ manifold, and let $\topsp{E}$ be a complex vector bundle on $\topsp{M}$. The canonical Dirac operator on $\topsp{M}$ induces the elliptic differential operator
\begin{displaymath}
\nslash{\topsp{D}}_{\topsp{M}}\otimes{\topsp{I}_{\topsp{E}}}:\Gamma(\nslash{\topsp{S}}\otimes{\topsp{E}})\to\Gamma(\nslash{\topsp{S}}\otimes{\topsp{E}})
\end{displaymath}
where $\nslash{S}$ denotes the spinor bundle associated to $\topsp{TM}$. After a choice of a smooth metric $g$ on $\topsp{M}$, we can construct the Hilbert space
\begin{displaymath}
\mathcal{H}^{\topsp{M}}_{\topsp{E}}:=\topsp{L}^{2}(\Gamma(\nslash{\topsp{S}}\otimes{\topsp{E}});dg^{\topsp{M}})
\end{displaymath}
and define the bounded Fredholm operator $\topsp{T}^{\topsp{E}}_{\topsp{M}}$ as the partial isometry in the polar decomposition for $\nslash{\topsp{D}}_{\topsp{M}}\otimes{\topsp{I}_{\topsp{E}}}$. Since $\topsp{M}$ is
 even-dimensional, $\mathcal{H}^{\topsp{M}}_{\topsp{E}}$ is $\mathbb{Z}_{2}$-graded, and there is a *-homomorphism 
\begin{displaymath}
\rho^{\topsp{M}}_{\topsp{E}}:\topsp{C}(\topsp{M};\mathbb{C})\to\mathcal{L}(\mathcal{H}^{\topsp{M}}_{\topsp{E}})
\end{displaymath}
as the space of section of a vector bundle on a manifold is equipped with a module structure over the algebra of functions of the manifold itself. Hence we have a correspondence
\begin{displaymath}
(\topsp{M},\topsp{E},\phi)\to(\mathcal{H}^{\topsp{M}}_{\topsp{E}},\rho^{\topsp{E}}_{\topsp{M}}\circ\phi^{*},\topsp{T}^{\topsp{E}}_{\topsp{M}})
\end{displaymath}
between cycles for geometric K-homology and cycles for analytic K-homology.
\begin{theorem}(\cite{Baumdouglas})
Let $\topsp{X}$ be a finite CW-complex. Then the correspondence
\begin{displaymath}
(\topsp{M},\topsp{E},\phi)\to(\mathcal{H}^{\topsp{M}}_{\topsp{E}},\rho^{\topsp{E}}_{\topsp{M}}\circ\phi^{*},\topsp{T}^{\topsp{E}}_{\topsp{M}})
\end{displaymath}
induces a natural isomorphism
\begin{displaymath}
\mu^{a}:\K_{0}^{t}(\topsp{X})\to\K^{a}_{0}(\topsp{X})
\end{displaymath}
commuting with the cap product between K-theory and K-homology.
\end{theorem}
The isomorphism $\mu$ can be now used to give an elegant formulation of the Atiyah-Singer index theorem. Namely, consider a closed even-dimensional smooth manifold $\topsp{M}$, and let $\topsp{T}^{*}\topsp{M}$ denote its cotangent bundle. By the results in \cite{AtiyahSinger}, any elliptic pseudo-differential operator $\topsp{D}$ between the sections of vector bundles on $\topsp{M}$ can be assigned to a class in $\K_{cpt}^{0}(\topsp{T}^{*}\topsp{M})$. More precisely, let 
\begin{displaymath}
\topsp{D}:\Gamma(\topsp{E})\to\Gamma(\topsp{F})
\end{displaymath}
be an elliptic differential operator of order $m$, where the vector bundles $\topsp{E}$ and $\topsp{F}$ have rank $p$ and $q$ respectively. The differential operator $\topsp{D}$ can be expressed in local coordinates $(x_{1},\ldots,x_{n})$ on $\topsp{M}$ as
\begin{displaymath}
\topsp{D}=\sum_{|\alpha|\leq m} A^{\alpha}(x)\dfrac{\partial^{|\alpha|}}{\partial x^{\alpha}}
\end{displaymath}
where $|\alpha|:=\sum_{k}\alpha_{k}$ for a $n$-tuple on nonnegative integers $\alpha=(\alpha_{1},\ldots,\alpha_{n})$, and where for each $\alpha$ $A^{\alpha}(x)$ is a $q\times{p}$ matrix of smooth complex-valued functions on $\topsp{M}$ with $A^{\alpha}(x)\neq{0}$ for some $\alpha$ such that $|\alpha|=m$. The \emph{principal symbol} of $\topsp{D}$ is defined to be the section $\sigma(\topsp{D})$ of the bundle $(\odot^{m}\topsp{T}\topsp{M})\otimes\topsp{Hom}(\topsp{E},\topsp{F})$ represented by the coefficients $\{A^{\alpha}\}_{|\alpha|=m}$, and where $\odot$ denotes the symmetric tensor product. Since  $\odot^{m}V$ is canonically isomorphic to the space of homogeneous polynomial functions of degree $m$ on $V^{*}$, for any vector space $V$, it follows that for each cotangent vector $\xi\in \topsp{T}_{x}^{*}(\topsp{M})$, the principal symbol gives a homomorphism
\begin{displaymath}
\sigma_{\xi}(\topsp{D}):\topsp{E}_{x}\to{\topsp{F}_{x}}
\end{displaymath}
For elliptic operators, by definition, such a homomorphism is invertible for each non-zero cotangent vector $\xi$ and any point $x\in \topsp{M}$. The principal symbol $\sigma(\topsp{D})$ can be considered to live on the cotangent bundle, i.e. it defines a bundle map
\begin{displaymath}
\sigma(\topsp{D}):\pi^{*}\topsp{E}\to\pi^{*}\topsp{F}
\end{displaymath}
where $\pi:\topsp{T}^{*}\topsp{M}\to{\topsp{M}}$. In particular, if $\topsp{D}$ is elliptic, $\sigma(\topsp{D})$ is an isomorphism away from the zero section, hence we can assign to $\topsp{D}$ the class
\begin{displaymath}
i(\topsp{D}):=[\pi^{*}\topsp{E},\pi^{*}\topsp{F};\sigma(\topsp{D})]\in\K^{0}(\B(\topsp{T}^{*}\topsp{M}),\S(\topsp{T}^{*}\topsp{M}))
\end{displaymath}
Conversely, to any class $[\topsp{E},\topsp{F};\mu]\in\K^{0}(\B(\topsp{T}^{*}\topsp{M}),\S(\topsp{T}^{*}\topsp{M}))$ one can assign a pseudo-differential operator on $\topsp{M}$ with total symbol $\mu$. See \cite{LM} for details on this construction.\\
By using the pseudo-differential operator associated to any class in $\K_{cpt}^{0}(\topsp{T}^{*}\topsp{M})$, one can define isomorphisms
\begin{displaymath}
\begin{array}{c}
\text{ind}_{t}:\K_{cpt}^{0}(\topsp{T}^{*}\topsp{M})\to{\K^{t}_{0}(\topsp{M})}\\
\text{ind}_{a}:\K_{cpt}^{0}(\topsp{T}^{*}\topsp{M})\to{\K^{a}_{0}(\topsp{M})}
\end{array}
\end{displaymath}
which in the case that $\topsp{M}$ is a $\rm spin^{c}$ manifold are given by the composition of Thom isomorphism and Poincar\'{e} duality. The above isomorphisms can be used to state the following elegant version of the index theorem.
\begin{theorem}\label{equivalence}
\textbf{(Atiyah-Singer)} Let $\topsp{M}$ be a closed smooth manifold. Then the following diagram commutes
\begin{displaymath}
\xymatrix{ &\K_{cpt}^{0}(\topsp{T}^{*}\topsp{M}) \ar[dr]^{\text{ind}_{t}}\ar[dl]_{\text{ind}_{a}}& \\
\K^{t}_{0}(\topsp{M})\ar[rr]_{\mu^{a}}& &\K^{a}_{0}(\topsp{M})
}
\end{displaymath}
\end{theorem}
To recover the usual form of the Atiyah-Singer index theorem, we consider the following composition of commutative diagrams
\begin{equation}\label{diagram}
\begin{array}{c}
\xymatrix{ &\K_{cpt}^{0}(\topsp{T}^{*}\topsp{M}) \ar[dr]^{\text{ind}_{a}}\ar[dl]_{\text{ind}_{t}}& \\
\K^{t}_{0}(\topsp{M})\ar[dr]_{\zeta_{*}} \ar[rr]_{\mu^{a}}& &\K^{a}_{0}(\topsp{M})\ar[dl]^{\zeta_{*}}\\
&\K_{0}({\rm pt})\simeq{\mathbb{Z}}&\\
}
\end{array}
\end{equation}
where $\zeta$ is the collapsing map, and the commutativity of the bottom diagram is granted by the fact that the isomorphism $\mu$ is natural. The homomorphisms\\ $\zeta_{*}\circ\text{ind}_{t}$, $\zeta_{*}\circ\text{ind}_{a}$ coincide with the topological and analytical index homomorphism, respectively, defined by Atiyah and Singer, and the commutativity of the two triangles implies that 
\begin{displaymath}
\zeta_{*}\circ\text{ind}_{t}=\zeta_{*}\circ\text{ind}_{a}
\end{displaymath} 
which is the original formulation of the index theorem.\\
Theorem \ref{equivalence} motivated the authors in \cite{Baumdouglas} to state that for any flavour $k$ of  K-theory, the equivalence between geometric and analytic $k$ -homology is related to an index theorem for the given theory $k$. In the next section we will reinforce the above statement: we will construct a natural morphism $\mu^{a}$ between $\KO^{\rm t}_{\sharp}$ and $\KO^{\rm a}_{\sharp}$ defined at the level of the cycles, and use a suitable index theorem to prove that $\mu^{a}$ is indeed a natural equivalence. 
\section{The equivalence between $\KO^{\rm t}_{\sharp}$ and $\KO^{\rm a}_{\sharp}$}
As illustrated in the previous section, our
primary goal is to prove the following result.
\begin{theorem}\label{isomorphism}
There is a natural equivalence
\begin{displaymath}
\mu^{\rm a}\,:\,\KO^{\rm t}~\stackrel{\approx}
{\longrightarrow}~\KO^{\rm a}
\end{displaymath}
between the topological and analytic KO-homology functors.
\end{theorem}
As we have seen, geometric and analytic KO-homology are generalized cohomology theories defined on the category of finite CW-pairs $(X,Y)$. For such theories, the following general result holds \cite{daviskirk}.
\begin{theorem}
Let $h_{\sharp}$ and $k_{\sharp}$ be generalized homology theories defined on the category of finite CW-pairs, and let
\begin{displaymath}
\phi:h_{\sharp}\to{k}_{\sharp}
\end{displaymath}
a natural transformation such that
\begin{displaymath}
\phi:h_{n}({\rm pt})\to{k_{n}({\rm pt})}
\end{displaymath}
is an isomorphism for any $n\in{\mathbb{Z}}$. Then $\phi$ is a natural equivalence.
\end{theorem}
We have explained in the previous section how the fact that the natural isomorphism $\mu^{a}$ between geometric and analytic K-homology induces the commutativity of the bottom triangle in (\ref{diagram}), and consequently the Atiyah-Singer index theorem. Our strategy to prove Theorem \ref{isomorphism} will be instead opposite: namely, we will construct surjective
``index'' homomorphisms $\ind_{n}^{\rm t}$ and $\ind_{n}^{\rm a}$ such
that the diagram
\beq
\xymatrix{\KO^{\rm t}_{n}(\pt)\ar[r]^{\mu^{\rm a}}\ar[rd]_{\ind_{n}^{\rm t}} &
\KO^{\rm a}_{n}(\pt)\ar[d]^{\ind_{n}^{\rm a}}\\ &\KO^{-n}(\pt)}
\label{muadiagmain}\eeq
commutes for every $n$, thanks to a suitable index theorem. Recall that $$\KO^{t,a}_{n}({\rm pt})\simeq{\KO^{-n}({\rm pt})}$$ and since the
groups $\KO^{-n}({\pt})$ are equal to
either $0$, $\mathbb{Z}$ or $\mathbb{Z}_2$ depending on the particular
value of $n$, the commutativity of the diagram (\ref{muadiagmain})
along with surjectivity of the index maps are sufficient to prove that
$\mu^{\rm a}$ is an isomorphism\footnote{Recall that surjective group homomorphisms of $\mathbb{Z}$ or any finite group are isomorphisms}. Moreover, the index homomorphism will play a crucial role in the applications exploited in later sections, in particular in the construction of cycles representing the generators for $\KO_{\sharp}^{t}({\rm pt})$.  
\subsection{The natural transformation $\mu^{a}$} 
Let $(\topsp{M},\topsp{E},\phi)$ be a topological KO-cycle on $\topsp{X}$ with $\dim \topsp{M}=n$. Recall that $\topsp{M}$ is a compact spin manifold. We
construct a corresponding class in $\KO^{\rm a}_{n}(\topsp{X})$ as
follows. Consider the associated vector bundle
\begin{displaymath}
\,\nslash{\mathfrak{S}}(\topsp{M}):={P_{\,\Spin}}(\topsp{M})\times_{\lambda_n}\Cl_{n}
\end{displaymath}
where $\Cl_{n}={\Cl(\mathbb{R}^n)}$, $\lambda_n:\Spin(n)\to{\rm
  End}(\Cl_{n})$ is given by left multiplication with
$\Spin(n)\subset{\Cl_{n}^{0}\subset{\Cl_{n}^{\phantom{0}}}}$, and $P_{\,\Spin}(\topsp{M})$ is
the principal $\Spin(n)$-bundle over $\topsp{M}$ associated to the spin structure on
the tangent bundle $\topsp{T}\topsp{M}$. Since
$\Cl^{\phantom{0}}_{n}=\Cl^{0}_{n}\oplus{\Cl^{1}_{n}}$ is a
$\zed_2$-graded algebra, it follows that
\beq
\,\nslash{\mathfrak{S}}(\topsp{M})=\,\nslash{\mathfrak{S}}^{0}(\topsp{M})\oplus\,
\nslash{\mathfrak{S}}^{1}(\topsp{M})
\label{Cliffbungrad}\eeq
is a $\mathbb{Z}_{2}$-graded real vector bundle over $\topsp{M}$ with respect to the
$\Cl(\topsp{T}\topsp{M})$-action. The Clifford algebra $\Cl_{n}$ acts by right
multiplication on the fibres whilst preserving the bundle grading
(\ref{Cliffbungrad}). Since $\:\nslash{\mathfrak{S}}(\topsp{M})$ is a vector bundle associated to $P_{\,\Spin}(\topsp{M})$, it carries the canonical Riemannian connection associated to the spin connection, and hence it is a Dirac bundle, i.e. it admits a canonical Dirac operator which is selfadjoint on the space of $\rm L^{2}$-sections of ${\:\nslash{\mathfrak{S}}(\topsp{M})}$ with respect to the given metric, and which has finite dimensional kernel. Indeed, as a vector bundle, $\:\nslash{\mathfrak{S}}(\topsp{M})$ is isomorphic to the direct sum of irreducible real spinor bundles on $\topsp{M}$. Hence, after choosing a $\rm C^{\infty}$ Riemannian metric $g^{\topsp{M}}$ on $\topsp{T}\topsp{M}$, we consider the canonical Dirac operator
\begin{displaymath}
\,\nslash{\mathfrak{D}}^\topsp{M}:\C^\infty(\topsp{M},\,\nslash{\mathfrak{S}}(\topsp{M}))\to
\C^\infty(\topsp{M},\,\nslash{\mathfrak{S}}(\topsp{M}))
\end{displaymath}
which we will refer to as the \emph{Atiyah-Singer} operator \cite{skewadjoint} defined locally by
\beq
\,\nslash{\mathfrak{D}}^\topsp{M}=\sum_{i=1}^n\,e_{i}\cdot\nabla^\topsp{M}_{e_{i}}
\ ,
\label{ASop}\eeq
where $\{e_{i}\}_{1\leq i\leq n}$ is a local orthonormal basis of sections of the
tangent bundle $\topsp{T}\topsp{M}$, $\nabla^\topsp{M}_{e_i}$ are the corresponding
components of the spin connection $\nabla^\topsp{M}$, and the dot denotes
Clifford multiplication. The operator $\:\,\nslash{\mathfrak{D}}^\topsp{M}$ is a
$\Cl_{n}$-operator \cite{LM}, i.e. one has
\begin{displaymath}
\,\nslash{\mathfrak{D}}^\topsp{M}(\Psi\cdot\varphi)=\,\nslash{\mathfrak{D}}^\topsp{M}(\Psi)
\cdot\varphi
\end{displaymath}
for all $\Psi\in\C^\infty(\topsp{M},\,\nslash{\mathfrak{S}}(\topsp{M}))$ and all
$\varphi\in{}\Cl_{n}$, where $\cdot\,\varphi$ denotes right
multiplication by $\varphi$. In particular, with respect to decomposition (\ref{Cliffbungrad}), the operator $\nslash{\mathfrak{D}}^\topsp{M}$ is of the form
\begin{displaymath}
\nslash{\mathfrak{D}}^\topsp{M}=\left(\begin{array}{cc}
0 & \nslash{\mathfrak{D}}_{1}^\topsp{M}\\
\nslash{\mathfrak{D}}_{0}^\topsp{M} & 0
\end{array}\right)
\end{displaymath}
where 
\begin{displaymath}
\nslash{\mathfrak{D}}_{0}^\topsp{M}:\Gamma(\nslash{\mathfrak{S}}^{0}(\topsp{M}))\to\Gamma(\nslash{\mathfrak{S}}^{1}(\topsp{M}))
\end{displaymath}
is a real, elliptic first-order differential operator which commutes with the action of $\Cl_{n}^{0}\simeq\Cl_{n-1}$ on $\:\nslash{\mathfrak{S}}(\topsp{M})$. Since $\,\nslash{\mathfrak{D}}^\topsp{M}$ commutes
with the $\Cl_{n}$-action, the vector space $\ker~\nslash{\mathfrak{D}}^\topsp{M}$ is a finite dimensional graded $\Cl_{n}$-module.\\
We can now construct a triple $(\mathcal{H}^\topsp{M}_{\topsp{E}},\rho^{\topsp{M}}_{\topsp{E}},
\topsp{T}^{\topsp{M}}_{\topsp{E}})$ comprising the following data:
\begin{itemize}
\item[(i)] The separable real Hilbert space $\mathcal{H}^{\topsp{M}}_{\topsp{E}}:={\rm
    L}^{2}_\real(\topsp{M},\,\nslash{\mathfrak{S}}(\topsp{M})\otimes{\topsp{E}};\dd g^\topsp{M})$;
\item[(ii)] The $*$-homomorphism
  $\rho^{\topsp{M}}_{\topsp{E}}:\C(\topsp{M},\real)\to{\mathcal{L}(\mathcal{H}^{\topsp{M}}_{\topsp{E}})}$
  defined by
\begin{displaymath}
\bigl(\rho^{\topsp{M}}_{\topsp{E}}(f)(\Psi)\bigr)(p)=f(p)\,\Psi(p)
\end{displaymath}
for $f\in\C(\topsp{M},\real)$,
$\Psi\in\C^\infty(\topsp{M},\,\nslash{\mathfrak{S}}(\topsp{M})\otimes{\topsp{E}})$ and
$p\in{\topsp{M}}$; and
\item[(iii)] The bounded Fredholm operator
\beq
  \topsp{T}^{\topsp{M}}_{\topsp{E}}:=\frac{\,\nslash{\mathfrak{D}}^\topsp{M}_{\topsp{E}}}
  {\sqrt{1+\big(\,\nslash{\mathfrak{D}}^\topsp{M}_{\topsp{E}}\big)^{2}}}
\label{TEMdef}\eeq
acting on $\mathcal{H}_\topsp{E}^\topsp{M}$, where $\,\nslash{\mathfrak{D}}^\topsp{M}_{\topsp{E}}$ is
the Atiyah-Singer operator (\ref{ASop}) twisted by the real vector
bundle $\topsp{E}\to \topsp{M}$.
\end{itemize}
This triple satisfies the following properties \cite{Higson}:
\begin{itemize}
\item[(i)] $\mathcal{H}^{\topsp{M}}_{\topsp{E}}$ is $\mathbb{Z}_{2}$-graded according
  to the splitting (\ref{Cliffbungrad}) of the Clifford bundle;
\item[(ii)] $\rho^{\topsp{M}}_{\topsp{E}}(f)$ is an even operator on
  $\mathcal{H}^{\topsp{M}}_{\topsp{E}}$ for all $f\in\C(\topsp{M},\real)$;
\item[(iii)] Since $\topsp{M}$ is compact, $\topsp{T}^{\topsp{M}}_{\topsp{E}}$ is an odd Fredholm
  operator which obeys the compactness conditions (\ref{TrhoKOarels})
  with $\rho_\topsp{E}^\topsp{M}(f)$; and
\item[(iv)] There are odd operators $\varepsilon_{i}$, $i=1,\dots,n$
  commuting with both $\rho^{\topsp{M}}_{\topsp{E}}(f)$ and $\topsp{T}_\topsp{E}^\topsp{M}$ which generate a
  $\Cl_n$-action on $\mathcal{H}_\topsp{E}^\topsp{M}$ as in (\ref{Clnrels}), and which
  are given explicitly as right multiplication by elements
  $e_{i}$ of a basis of the vector space $\mathbb{R}^{n}$.
\end{itemize}
It follows that $(\mathcal{H}^{\topsp{M}}_{\topsp{E}},\rho^{\topsp{M}}_{\topsp{E}},\topsp{T}^{\topsp{M}}_{\topsp{E}})$ is a
well-defined $n$-graded Fredholm module over the real $C^*$-algebra
$\C(\topsp{M},\mathbb{R})$.\\
We now define the map $\mu^{\rm a}$ in (\ref{muadiagmain}) by
\beq
\mu^{\rm a}\bigl(\topsp{M}\,,\,\topsp{E}\,,\,\phi\bigr):=\phi_{*}\bigl(
\mathcal{H}^{\topsp{M}}_{\topsp{E}}\,,\,\rho^{\topsp{M}}_{\topsp{E}}\,,\,
\topsp{T}^{\topsp{M}}_{\topsp{E}}\bigr)=\bigl(\mathcal{H}^{\topsp{M}}_{\topsp{E}}\,,\,\rho^{\topsp{M}}_{\topsp{E}}\circ\phi^{*}
\,,\,\topsp{T}^{\topsp{M}}_{\topsp{E}}\bigr) \ ,
\label{muacycledef}\eeq
where $\phi^{*}:\C(\topsp{X},\real)\to\C(\topsp{M},\real)$ is the real $C^*$-algebra
homomorphism induced by the map $\phi$. At this stage the map
$\mu^{\rm a}$ is only defined on KO-cycles. More precisely, we can consider the map $\mu^{a}:\Gamma{\rm O}_{n}(\topsp{X})\to\KO_{n}^{a}(\topsp{X})$ induced by the equivalence relations on the set of Fredholm modules. We need to prove, at this point, that the map $\mu^{a}$ gives  a well defined homomorphism
\begin{displaymath}
\mu^{a}:\KO_{n}^{t}(\topsp{X})\to\KO_{n}^{a}(\topsp{X})
\end{displaymath}
 We will first recall some useful results concerning the above construction. Let us  denote with $[\:\nslash{\mathfrak{D}}^\topsp{M}_{\topsp{E}}]$ the class corresponding to the element $\mu^{a}(\topsp{M},\topsp{E},{\rm id}_{\topsp{M}})\in\KO_{n}^{a}(\topsp{M})$. We can then state the following result \cite{Higson}
\begin{theorem}\label{connective}
Let $\topsp{M}-\partial{\topsp{M}}$ be the interior of a spin manifold $\topsp{M}$ of dimension $n$ with boundary $\partial \topsp{M}$, and let $\topsp{E}$ be a real vector bundle on $\topsp{M}$. Equip the boundary $\partial{\topsp{M}}$ with the spin structure induced by that on $\topsp{M}$. Then
\begin{displaymath}
\partial[\:\nslash{\mathfrak{D}}^{\topsp{M}-\partial{\topsp{M}}}_{\topsp{E}|_{\topsp{M}-\partial{}\topsp{M}}}]=[\:\nslash{\mathfrak{D}}^{\partial{\topsp{M}}}_{\topsp{E}|_{\partial{}\topsp{M}}}]
\end{displaymath}
where $\partial:\KO_{n}^{a}(\topsp{M}-\partial{\topsp{M}})\to\KO^{a}_{n-1}(\partial{\topsp{M}})$ is the boundary homomorphism.
\end{theorem}
Notice that in the theorem above we have used that $\KO_{n}^{a}(\topsp{M}-\partial{\topsp{M}})\simeq\KO_{n}(\topsp{M},\partial{\topsp{M}})$ via excision. Moreover, one can prove that the class $[\:\nslash{\mathfrak{D}}^{\topsp{M}}]:=[\:\nslash{\mathfrak{D}}^{\topsp{M}}_{\id}]$ represents the fundamental class of $\topsp{M}$ in $\KO_{n}^{a}(\topsp{M})$, and that 
\begin{displaymath}
[\:\nslash{\mathfrak{D}}^{\topsp{M}}_{\topsp{E}}]=[\topsp{E}]\cap[\:\nslash{\mathfrak{D}}^{\topsp{M}}]
\end{displaymath}
See \cite{Bunke1995} for details.\\
We are now ready to prove the following
\begin{theorem}\label{welldefined}
The map $\mu^{a}:\Gamma{\rm O}_{n}(\topsp{X})\to\KO_{n}^{a}(\topsp{X})$ induces a well defined homomorphism
\begin{displaymath}
\mu^{a}:\KO_{n}^{t}(\topsp{X})\to\KO_{n}^{a}(\topsp{X})
\end{displaymath}
for any CW-complex $\topsp{X}$, and any $n\in\mathbb{Z}$.
\end{theorem}
\begin{proof}
That the map $\mu^{a}$ respects the algebraic sum and independence of the direct sum relation follows straightforwardly from the fact that
\begin{displaymath}
\Gamma(\:\nslash{\mathfrak{S}}(\topsp{M}\sqcup{\topsp{N}})\otimes(\topsp{E}\sqcup{\topsp{F}}))=\Gamma(\:\nslash{\mathfrak{S}}(\topsp{M})\otimes{\topsp{E}})\oplus\Gamma(\:\nslash{\mathfrak{S}}({\topsp{N}})\otimes{\topsp{F}})
\end{displaymath}
for any compact spin manifolds $\topsp{M}$, $\topsp{N}$, with vector bundles $\topsp{E}$, $\topsp{F}$ over $\topsp{M}$, $\topsp{N}$ respectively, and by the definition of direct sum of Fredholm modules.\\
To prove that the map homomorphism $\mu^{a}$ does not depend on the bordism relation is equivalent to proving that $\mu^{a}(\topsp{M},\topsp{E},\phi)=0$ for every bord $n$-cycle $(\topsp{M},\topsp{E},\phi)$, i.e. a cycle for which there exists an $n+1$-cycle $(\topsp{W},\topsp{F},\psi)$ in the appropriate relative group such that 
$$(\topsp{M},\topsp{E},\phi)=(\partial {\topsp{W}},\topsp{F}|_{\partial {\topsp{W}}}, \psi|_{\partial {\topsp{W}}})$$
First, we notice that the map $\psi_{\partial \topsp{W}}$ factors as $\psi|_{\partial {\topsp{W}}}=\psi\circ{i}$, where
\begin{displaymath}
i:\partial{\topsp{W}}\hookrightarrow{\topsp{W}}
\end{displaymath}
denotes the inclusion of the boundary. We have that 
\begin{displaymath}
\begin{array}{rl}
\mu^{a}(\topsp{M},\topsp{E},\phi)&=(\psi|_{\partial \topsp{W}})_{*}([\:\nslash{\mathfrak{D}}^{\partial{\topsp{W}}}_{\topsp{F}|_{\partial{\topsp{W}}}}])\\
&=(\psi)_{*}\circ{i_{*}}([\:\nslash{\mathfrak{D}}^{\partial{\topsp{W}}}_{\topsp{F}|_{\partial{\topsp{W}}}}])\\
&=0
\end{array}
\end{displaymath}
which follows by Theorem \ref{connective}, and the long exact sequence for the pair $(\topsp{W},\partial \topsp{W})$
\begin{displaymath}
\cdots\to\KO_{n+1}^{a}(\topsp{W}-\partial {\topsp{W}})\xrightarrow{\partial}\KO_{n}^{a}(\partial \topsp{W})\xrightarrow{i_{*}}\KO_{n}^{a}(\topsp{W})\to\cdots
\end{displaymath}
To conclude the proof, we need to show that for any KO-cycle $(\topsp{M},\topsp{E},\phi)$ and any real spin vector bundle $\topsp{F}$ of rank $8r$ over $\topsp{M}$ we have
\begin{displaymath}
\mu^{a}(\topsp{M},\topsp{E},\phi)=\mu^{a}(\widehat{\topsp{M}},\topsp{H}(\topsp{F})\otimes\pi^{*}\topsp{E},\phi \circ \pi)
\end{displaymath}
By definition, we have
\begin{displaymath}
\mu^{a}(\widehat{\topsp{M}},\topsp{H}(\topsp{F})\otimes\pi^{*}\topsp{E},\phi \circ \pi):=\phi_{*}\pi_{*}\bigl([\:\nslash{\mathfrak{D}}^{\widehat{\topsp{M}}}_{\topsp{H}(\topsp{F})\otimes\pi^{*}\topsp{E}}]\bigr)
\end{displaymath}
By using that
\begin{displaymath}
\begin{array}{rl}
[\:\nslash{\mathfrak{D}}^{\widehat{\topsp{M}}}_{\topsp{H}(\topsp{F})\otimes\pi^{*}\topsp{E}}]&=[\topsp{H}(\topsp{F})\otimes\pi^{*}\topsp{E}]\cap[\:\nslash{\mathfrak{D}}^{\widehat{\topsp{M}}}]\\
&=([\topsp{H}(\topsp{F})]\cup[\pi^{*}\topsp{E}])\cap[\:\nslash{\mathfrak{D}}^{\widehat{\topsp{M}}}]\\
&=\pi^{*}[\topsp{E}]\cap\bigl( [\topsp{H}(\topsp{F})]\cap[\:\nslash{\mathfrak{D}}^{\widehat{\topsp{M}}}]\bigr)
\end{array}
\end{displaymath}
we have
\begin{displaymath}
\begin{array}{rl}
\mu^{a}(\widehat{\topsp{M}},\topsp{H}(\topsp{F})\otimes\pi^{*}\topsp{E},\phi \circ \pi)&=\phi_{*}\pi_{*}\bigl(\pi^{*}[\topsp{E}]\cap\bigl( [\topsp{H}(\topsp{F})]\cap[\:\nslash{\mathfrak{D}}^{\widehat{\topsp{M}}}]\bigr)\bigr)\\
&=\phi_{*}\bigl([\topsp{E}]\cap\pi_{*}\bigl( [\topsp{H}(\topsp{F})]\cap[\:\nslash{\mathfrak{D}}^{\widehat{\topsp{M}}}]\bigr)\bigr)
\end{array}
\end{displaymath}
By recalling that $[\topsp{H}(\topsp{F})]$ is the image of the Thom class of $\topsp{F}$, and by the equation (\ref{Thomdef}), the class $\pi_{*}\bigl( [\topsp{H}(\topsp{F})]\cap[\:\nslash{\mathfrak{D}}^{\widehat{\topsp{M}}}]\bigr)\bigr)$ is the image of $[\:\nslash{\mathfrak{D}}^{\widehat{\topsp{M}}}]$ under the isomorphism
\begin{displaymath}
\KO_{n+8r}(\widehat{\topsp{M}})\to\KO_{n}(\topsp{M})
\end{displaymath}
for the spherical fibration $\widehat{\topsp{M}}\xrightarrow{\pi}\topsp{M}$ induced by the Thom isomorphism. Since the above isomorphism maps the fundamental class of $\widehat{\topsp{M}}$ to that of $\topsp{M}$, we have
\begin{displaymath}
\pi_{*}\bigl( [\topsp{H}(\topsp{F})]\cap[\:\nslash{\mathfrak{D}}^{\widehat{\topsp{M}}}]\bigr)=[\:\nslash{\mathfrak{D}}^{\topsp{M}}]
\end{displaymath}
The following equalities conclude the proof
\begin{displaymath}
\begin{array}{rl}
\mu^{a}(\widehat{\topsp{M}},\topsp{H}(\topsp{F})\otimes\pi^{*}\topsp{E},\phi \circ \pi)&=\phi_{*}\bigl([\topsp{E}]\cap[\:\nslash{\mathfrak{D}}^{\topsp{M}}]\bigr)\\
&=\phi_{*}\bigl([\:\nslash{\mathfrak{D}}_{\topsp{E}}^{\topsp{M}}]\bigr)\\
&=\mu^{a}(\topsp{M},\topsp{E},\phi)
\end{array}
\end{displaymath}
\end{proof}
We refer the reader to \cite{baum-2007-3} for an alternative proof of Theorem \ref{welldefined}.
\subsection{The analytic index map ${\rm ind}_{n}^{\rm a}$\label{indna}}
Let $(\hilR,\rho,\topsp{T})$ be an $n$-graded Fredholm module over the real
$C^*$-algebra $\C(\topsp{X},\real)$ such that $\topsp{T}$ is a Fredholm operator. Since $\topsp{T}$
commutes with $\varepsilon_{i}$ for $i=1,\dots,n$, the kernel $\ker
\topsp{T}\subset\hilR$ is a real $\Cl_{n}$-module with $\zed_2$-grading
induced by the grading of $\hilR$. Thus we can define \beq
\ind_{n}^{\rm a}(\topsp{T}):=\:[\ker \topsp{T}]\in\: 
\widehat{\mathfrak{M}}_{n}/ \imath^{*}\widehat{\mathfrak{M}}_{n+1}
\label{indnaT}\eeq
where $\widehat{\mathfrak{M}}_{n}$ is the Grothendieck group of irreducible Clifford modules defined in chapter 3.\\
By the Atiyah-Bott-Shapiro isomorphism, we have
\begin{displaymath}
\widehat{\mathfrak{M}}_{n}/ \imath^{*}\widehat{\mathfrak{M}}_{n+1}\simeq\KO^{-n}({\rm pt})
\end{displaymath}
We will call (\ref{indnaT}) the \emph{analytic} or \emph{Clifford
index} of the Fredholm operator $\topsp{T}$. An important property of this
definition is the following result \cite{LM}.
\begin{theorem}
The analytic index
\begin{displaymath}
\ind^{\rm a}_{n}\,:\,\Fred_{n}~\longrightarrow~\KO^{-n}(\pt)
\end{displaymath}
is surjective and constant on the connected components of $\Fred_{n}$.
\label{anindthm}\end{theorem}
Given two Fredholm modules $(\hil_\real,\rho,\topsp{T})$ and $(\hilR,\rho,\topsp{T}'\,)$
over a real $C^*$-algebra $A$, we will say that $\topsp{T}$ is a {\it compact
perturbation} of $\topsp{T}'$ if $(\topsp{T}-\topsp{T}'\,)\,\rho(a)\in\mathcal{K}_\real$ for
all $a\in{A}$. We then have the following elementary result.
\begin{lemma}
If $\topsp{T}$ is a compact perturbation of $\topsp{T}'$, then the Fredholm modules
$(\hil_\real,\rho,\topsp{T})$ and $(\hilR,\rho,\topsp{T}'\,)$ are operator homotopic
over $A$.
\label{lemmacompact}\end{lemma}
\begin{proof}
Consider the path $\topsp{T}_{t}=(1-t)\,\topsp{T}+t\,\topsp{T}'$ for $t\in[0,1]$. Then the map
$t\mapsto{\topsp{T}_{t}}$ is norm continuous. We will show that for any $t\in[0,1]$,
the triple $(\hilR,\rho,\topsp{T}_{t})$ is a Fredholm module over $A$, i.e.
that the operator $\topsp{T}_{t}$ satisfies
\beq
\left(\topsp{T}_{t}^{2}-1\right)\,\rho(a)~,~
\left(\topsp{T}^{\phantom{*}}_{t}-\topsp{T}_{t}^{*}\right)\,\rho(a)~,~
\topsp{T}^{\phantom{*}}_{t}\,\rho(a)-\rho(a)\,\topsp{T}^{\phantom{*}}_{t}\in\mathcal{K}_\real
\label{TtFred}\eeq
for all $a\in A$. The last two inclusions in (\ref{TtFred}) are
easily proven because the path $\topsp{T}_{t}$ is ``linear'' in the operators
$\topsp{T}$ and $\topsp{T}'$. To establish the first one, for any $t\in[0,1]$ and
$a\in A$ we compute
\begin{equation}
\begin{array}{rl}
\left(\topsp{T}_{t}^{2}-1\right)\,\rho(a)=&\Bigl[\left(\topsp{T}^{2}-1\right)+t^{2}\,
\left(\topsp{T}-\topsp{T}'\,\right)^{2}-t\,\left(\topsp{T}^{2}-1\right)\\
&-t\,\left(\topsp{T}-\topsp{T}'\,
\right)^{2}+t\,\left(\topsp{T}^{\prime\,2}-1\right)\Bigr]\,\rho(a) \
. \\
\end{array}
\label{compcompute}\end{equation}
By using the fact that $(\hilR,\rho,\topsp{T})$ and $(\hilR,\rho,\topsp{T}'\,)$ are
Fredholm modules, that $\topsp{T}$ is a compact perturbation of $\topsp{T}'$, and that
$\mathcal{K}_\real$ is an ideal in $\mathcal{L}(\hilR)$, one easily
verifies that the right-hand side of (\ref{compcompute}) is a compact
operator. This implies that $(\hilR,\rho,\topsp{T}_{t})$ is a well-defined
family of Fredholm modules over $A$.
\end{proof}
By Lemma \ref{lemmacompact}, we can choose the operator $\topsp{T}$ in the class $[\hilR,\rho,\topsp{T}]$ to be selfadjoint without loss of generality. Indeed, given a Fredholm module we simply replace the operator $\topsp{T}$ with  $\tilde{\topsp{T}}:=\frac{1}{2}(\topsp{T}+\topsp{T}^{*})$. Moreover, such a choice of $\tilde{\topsp{T}}$ is compatible with operator homotopy. Hence in the following we will always assume that the operator $\topsp{T}$ is selfadjoint and Fredholm. 
\begin{proposition}
The induced map
\begin{displaymath}
\ind^{\rm a}_{n}\,:\,\KO^{\rm a}_{n}(\topsp{X})~\longrightarrow~{\KO^{-n}(\pt)}
\end{displaymath}
given on classes of $n$-graded Fredholm modules by
$$
\ind^{\rm a}_{n}[\hilR,\rho,\topsp{T}]=[\ker \topsp{T}]
$$
is a well-defined surjective homomorphism for any $n\in\nat$.
\end{proposition}
\begin{proof}
We first show that to the direct sum of two Fredholm modules
$(\hilR^{\phantom{\prime}},\rho,\topsp{T})$ and $(\mathcal{H}_\real',\rho',\topsp{T}'\,)$ over
$A=\C(\topsp{X},\real)$, the map $\ind^{\rm a}_n$ associates the class $[\ker
\topsp{T}]+[\ker \topsp{T}'\,]\in\widehat{\mathfrak{M}}_{n}/\imath^{*}
\widehat{\mathfrak{M}}_{n+1}\simeq\KO^{-n}(\pt)$. The kernel
\begin{displaymath}
\ker(\topsp{T}\oplus{\topsp{T}'}\,)=\ker(\topsp{T})\oplus{\ker(\topsp{T}'\,)}
\end{displaymath}
is a real graded $\Cl_{n}$-module. By the definition of the group
$\widehat{\mathfrak{M}}_{n}$ and of its quotient by
$\imath^{*}\widehat{\mathfrak{M}}_{n+1}$, one thus has $\ind_n^{\rm
  a}(\topsp{T}\oplus \topsp{T}'\,)=[\ker \topsp{T}]+[\ker \topsp{T}'\,]$ and so the map $\ind^{\rm
  a}_{n}$ respects the algebraic structure on $\Gamma{\rm O}_{n}(A)$.

Consider now two Fredholm modules $(\hilR^{\phantom{\prime}},\rho,\topsp{T})$ and
$(\mathcal{H}_\real',\rho',\topsp{T}'\,)$ which are orthogonally
equivalent. Then there exists an even isometry
$U:\hilR^{\phantom{\prime}}\to\mathcal{H}_\real'$ such that
$$
 \topsp{T}'=U\,\topsp{T}\,U^* \ , \quad
\varepsilon'_{i}=U\,\varepsilon^{\phantom{\prime}}_{i}\,U^* \ .
$$
This implies that $\ker \topsp{T}'=U(\ker \topsp{T})$, and that the graded $\Cl_{n}$
representations given respectively by $\varepsilon'_{i}$ and
$\varepsilon^{\phantom{\prime}}_{i}$ are equivalent. In particular,
they represent the same class in
$\widehat{\mathfrak{M}}_{n}/\imath^{*}\widehat{\mathfrak{M}}_{n+1}$.

Finally, consider two homotopic $n$-graded
Fredholm modules $(\hilR,\rho,\topsp{T})$, $(\hilR,\rho,\topsp{T}'\,)$ over $A$. There exists by definition a continuous path $t\to{\topsp{T}_{t}}$ connecting $\topsp{T}$ and $\topsp{T}'$ in ${\rm Fred}_{n}$. Hence, by Theorem \ref{anindthm} ${\rm ind}^{a}_{n}(\topsp{T})={\rm ind}^{a}_{n}(\topsp{T}')$.
\end{proof}
\subsection{The topological index map ${\rm ind}_{n}^{\rm t}$\label{indnt}}
Given a KO-cycle $(\topsp{M},\topsp{E},\phi)$ on $\topsp{X}$ with $\topsp{M}$ an $n$-dimensional
compact spin manifold, we can assign to it the
\textit{associated Atiyah-Milnor-Singer (AMS) invariant}~\cite{LM} defined by
\beq
\widehat{\mathcal{A}}_{\topsp{E}}(\topsp{M})=\beta\circ{f_{!}}([\topsp{E}])\in{\KO^{-n}(\pt)}
\label{AMSinvdef}\eeq
where $f_{!}$ is the Gysin homomorphism for the smooth embedding
\begin{equation}\label{embeddmap}
f:\topsp{M}\hookrightarrow{\mathbb{R}^{n+8k}}
\end{equation}
for some $k\in\mathbb{N}$ sufficiently large, and $\beta$ is induced by the Bott periodicity isomorphism
\begin{displaymath}
\beta:\KO_{\rm cpt}(\mathbb{R}^{n+8k})\to\KO^{-n}(\rm pt)
\end{displaymath}
This definition does not depend on the embedding (\ref{embeddmap}) nor on the integer $k$.\\
We define
\beq
\ind^{\rm t}_{n}(\topsp{M},\topsp{E},\phi):=\widehat{\mathcal{A}}_{\topsp{E}}(\topsp{M}) \ .
\label{indntdef}\eeq

\begin{proposition}
The map
\begin{displaymath}
\ind^{\rm t}_{n}\,:\,\KO^{\rm t}_{n}(\topsp{X})~\longrightarrow~{\KO^{-n}}(\pt)
\end{displaymath}
induced by (\ref{indntdef}) is a well-defined surjective homomorphism
for any $n\in\nat$.
\end{proposition}
\begin{proof}
We first prove that the map $\ind^{\rm t}_{n}$ respects the algebraic
structure on the abelian group $\KO^{\rm t}_{n}(\topsp{X})$. Given two
$n$-dimensional compact spin manifolds $\topsp{M}_{1}$ and $\topsp{M}_{2}$, let
$\topsp{M}=\topsp{M}_{1}\amalg{\topsp{M}_{2}}$. Embed $\topsp{M}$ in the Euclidean space $\mathbb{R}^{n+8k}$ for some
$k$ sufficiently large as in (\ref{embeddmap}). Recall now that the Gysin homomorphism of $f$ is the composition of Thom isomorphism with respect to the normal bundle of $\topsp{M}$ in $\mathbb{R}^{n+8k}$ with the map ``extending by zero'', and that the Thom isomorphism is the isomorphism induced by the Thom class. Moreover, notice that the normal bundle $\nu$ to the embedding of $\topsp{M}$ is given by $\nu_{1}\amalg\nu_{2}$, where $\nu_{1}$ and $\nu_{2}$ are respectively the normal bundles to the embeddings of $\topsp{M}_{1}$ and $\topsp{M}_{2}$ induced by the embedding of $\topsp{M}$. The Thom class of $\nu$ is given by
\begin{eqnarray*}
\tau_{\nu}&:=&\left[\varpi^{*}\,\nslash{S}^{+}(\nu)\,,\,
\varpi^{*}\,\nslash{S}^{-}(\nu)\,;\,\sigma\right]\\
&=&\left[\varpi_{1}^{*}\,\nslash{S}^{+}(\nu_{1})\amalg\varpi_{2}^{*}\,
\nslash{S}^{+}(\nu_{2})\,,\,\varpi_{1}^{*}\,\nslash{S}^{-}(\nu_{1})
\amalg\varpi_{2}^{*}\,\nslash{S}^-(\nu_{2})\,;\,\sigma_{1}
\amalg\sigma_{2}\right]\\
&=&\tau_{\nu_{1}}+\tau_{\nu_{2}}\in\:\KO^0(\topsp{B(\nu)},\topsp{S(\nu)})\simeq
\KO^0(\topsp{B(\nu_{1})},\topsp{S(\nu_{1})})\oplus{}\KO^0(\topsp{B(\nu_{1})},\topsp{S(\nu_{2})}) \ .
\end{eqnarray*}
where $\varpi=\varpi_{1}\amalg\varpi_{2}:\nu_{1}\amalg{\nu_{2}}\to
\topsp{M}_{1}\amalg{\topsp{M}_{2}}$ is the normal bundle projection.
Let $\topsp{E}_{1}$ and $\topsp{E}_{2}$ be real vector bundles over $\topsp{M}_{1}$ and
$\topsp{M}_{2}$, respectively, and let $\topsp{E}=\topsp{E}_{1}\amalg{\topsp{E}_{2}}$. Then in
$\KO^0(\topsp{B(\nu)},\topsp{S(\nu)})$ one has
\begin{eqnarray*}
\tau_{\nu}(\topsp{E})&=&\tau_{\nu}\,\smile\,[\varpi^{*}\topsp{E}]\\
&=&\tau_{\nu}\,\smile\,[\varpi_{1}^{*}\topsp{E}_1\amalg\varpi_{2}^{*}\topsp{E}_2]\\
&=&\tau_{\nu_{1}}\,\smile\,[\varpi_{1}^{*}\topsp{E}_{1}]+\tau_{\nu_{2}}\,\smile\,
[\varpi_{2}^{*}\topsp{E}_{2}]~=~\tau_{\nu_{1}}(\topsp{E}_{1})+\tau_{\nu_{2}}(\topsp{E}_{2}) \ .
\end{eqnarray*}
Since the map extending by zero is a homomorphism, one then finds
\begin{eqnarray*}
\ind^{\rm t}_{n}\bigl((\topsp{M}_{1},\topsp{E}_{1},\phi_{1})\amalg(\topsp{M}_{2},\topsp{E}_{2},\phi_{2})\bigr)
&:=&\widehat{\mathcal{A}}_{\topsp{E}}(\topsp{M})\\
&=&\widehat{\mathcal{A}}_{\topsp{E}_{1}}(\topsp{M}_{1})+\widehat{\mathcal{A}}_{\topsp{E}_{2}}(\topsp{M}_{2})\\
&=&\ind^{\rm t}_{n}(\topsp{M}_{1},\topsp{E}_{1},\phi_{1})+\ind^{\rm
  t}_{n}(\topsp{M}_{2},\topsp{E}_{2},\phi_{2})\in{\KO^{-n}(\pt)} \ ,
\end{eqnarray*}
showing that ${\rm ind^{t}}_{n}$ respects the algebraic sum of cycles.\\
Next we have to check that the map $\ind^{\rm t}_{n}$ is independent
of the choice of representative of a homology class in $\KO^{\rm
  t}_{n}(\topsp{X})$. The independence of the direct sum relation follows
from the discussion above, while spin bordism independence is
guaranteed by the property that the AMS invariant
$\widehat{\mathcal{A}}_{\topsp{E}}(\topsp{M})$ is a spin cobordism
invariant \cite{LM}. Finally, we have to verify that the map
$\ind^{\rm t}_{n}$ does not depend on real vector bundle
modification. Let $\topsp{M}$ be a smooth $n$-dimensional compact spin manifold and let $\topsp{E}\to
\topsp{M}$ be a smooth real vector bundle. Let $\topsp{F}$ be a real spin
vector bundle over $\topsp{M}$ with fibres of real dimension
$8l$ for some $l\in\nat$. Consider the corresponding sphere bundle
(\ref{hatMdef}) with projection (\ref{pidef}). Real
vector bundle modification of a KO-cycle $(\topsp{M},\topsp{E},\phi)$ on $\topsp{X}$ induced
by $\topsp{F}$ produces the KO-cycle
$(\,\widehat{\topsp{M}},\widehat{\topsp{E}},\phi\circ{\pi})$, where
$\widehat{\topsp{E}}=\topsp{H}(\topsp{F})\otimes\pi^*(\topsp{E})$ is the real vector bundle over
$\widehat{\topsp{M}}$ such that
\begin{displaymath}
\big[\,\widehat{\topsp{E}}\,\big]=\Sigma_{!}^\topsp{F}\big[\topsp{E}\big]
\end{displaymath}
with $[\topsp{E}]\in{}\KO^0(\topsp{M})$, $[\,\widehat{\topsp{E}}\,]\in{\KO^0(\,\widehat{\topsp{M}}\,)}$, and $\Sigma^{\rm F}$ defined as in (\ref{secembind}). We may compute the AMS
invariant for the pair $(\,\widehat{\topsp{M}}\,,\,\widehat{\topsp{E}}\,)$ by choosing
an embedding
\begin{displaymath}
\widehat{f}\,:\,\widehat{\topsp{M}}~\hookrightarrow~\mathbb{R}^{n+8k+8l}
\end{displaymath}
so that
\begin{eqnarray*}
\widehat{\mathcal{A}}_{\widehat{\topsp{E}}}\big(\,\widehat{\topsp{M}}\,\big)&=&
\beta\circ\widehat{f}_{!}\bigl([\,\widehat{\topsp{E}}\,]\bigr)\\
&=&\beta\circ\widehat{f}_{!}\circ\Sigma^\topsp{F}_{!}[\topsp{E}]~=~
\beta\circ\big(\,\widehat{f}\circ\Sigma^\topsp{F}\,\big)_{!}[\topsp{E}] \ ,
\end{eqnarray*}
where in the last equality we have used functoriality of the Gysin
map. Notice that
\begin{displaymath}
\widehat{f}\circ{\Sigma^\topsp{F}}\,:\,\topsp{M}~\hookrightarrow~
\mathbb{R}^{n+8k+8l}=:\mathbb{R}^{n+8m}
\end{displaymath}
is an embedding of $\topsp{M}$ into a ``large enough'' Euclidean space. Since
$\widehat{\mathcal{A}}_{\topsp{E}}(\topsp{M})$ is independent of the embedding and
the integer $m$, we have
\begin{displaymath}
\widehat{\mathcal{A}}_{\widehat{\topsp{E}}}\bigl(\,\widehat{\topsp{M}}\,\bigr)=
\widehat{\mathcal{A}}_{\topsp{E}}\bigl(\topsp{M}\bigr)
\end{displaymath}
as required.
\end{proof}
\subsection{The Isomorphism Theorem}
We can now assemble the constructions of the previous sections to finally establish our main
result. Notice first of all that since
$\ker~\nslash{\mathfrak{D}}^{{\topsp{M}}}_\topsp{E}\simeq\ker T^{\topsp{M}}_{\topsp{E}}$, one
has
\beq
\ind^{\rm a}_{n}\circ\mu^{\rm a}
\big(\topsp{M}\,,\,\topsp{E}\,,\,\phi\big)=\ind^{\rm a}_{n}
\big(\,\nslash{\mathfrak{D}}^\topsp{M}_\topsp{E}\big)
\label{indmuarel}\eeq
for any KO-cycle $(\topsp{M},\topsp{E},\phi)$ on $\topsp{X}$ with $\dim \topsp{M}=n$.
At this point we can use an important result from spin
geometry called the \textit{$\Cl_{n}$-index theorem}~\cite{LM}.
\begin{theorem}
Let $\topsp{M}$ be a compact spin manifold of dimension $n$ and let $\topsp{E}$ be a
real vector bundle over $\topsp{M}$. Let
\begin{displaymath}
\,\nslash{\mathfrak{D}}^\topsp{M}_{\topsp{E}}\,:\,\C^\infty\bigl(\topsp{M}\,,\,\,
\nslash{\mathfrak{S}}(\topsp{M})\otimes{\topsp{E}}\bigr)~\longrightarrow~
\C^\infty\bigl(\topsp{M}\,,\,\,\nslash{\mathfrak{S}}(\topsp{M})\otimes{\topsp{E}}\bigr)
\end{displaymath}
be the $\Cl_{n}$-linear Atiyah-Singer operator with coefficients in
$\topsp{E}$. Then
\begin{displaymath}
\ind^{\rm
  a}_{n}\big(\,\nslash{\mathfrak{D}}^\topsp{M}_{\topsp{E}}\big)=
\widehat{\mathcal{A}}_{\topsp{E}}\big(\topsp{M}\big) \ .
\end{displaymath}
\label{spinthm}\end{theorem}
\noindent
The proof of Theorem~\ref{isomorphism} is now completed once we establish the
following result.
\begin{proposition}
The map
\begin{displaymath}
\mu^{\rm a}\,:\,\KO^{\rm t}_{n}(\pt)~\longrightarrow~{\KO^{\rm a}_{n}(\pt)}
\end{displaymath}
is an isomorphism for any $n\in\nat$.
\end{proposition}
\begin{proof}
As noticed at the beginning of this section, it suffices to establish
the commutativity of the diagram (\ref{muadiagmain}), i.e. that
\begin{displaymath}
\ind^{\rm t}_{n}=\ind^{\rm a}_{n}\circ\mu^{\rm a} \ .
\end{displaymath}
Let $[\topsp{M},\topsp{E},\phi]$ be the class of a KO-cycle over ${\pt}$ with $\dim
{\topsp{M}}=n$. Using Theorem~\ref{spinthm} and (\ref{indmuarel}) we have
\begin{eqnarray*}
\ind^{\rm t}_{n}\big[\topsp{M}\,,\,\topsp{E}\,,\,\phi\big]&:=&
\widehat{\mathcal{A}}_{\topsp{E}}\big(\topsp{M}\big)\\
&=&\ind^{\rm a}_{n}\big(\,\nslash{\mathfrak{D}}^\topsp{M}_{\topsp{E}}\big)
~=~\ind^{\rm a}_{n}\circ\mu^{\rm a}\big[\topsp{M}\,,\,\topsp{E}\,,\,\phi\big]
\end{eqnarray*}
as required.
\end{proof}
\section{The Real Chern Character}
In this section we will describe the natural complexification map from
geometric KO-homology to geometric K-homology and use it to define the
Chern character homomorphism in topological KO-homology. We describe
various properties of this homomorphism, most notably its intimate
connection with the AMS invariant which was the crux of the
isomorphism of the previous section.\\
Let $\topsp{X}$ be a compact topological space. Consider the topological,
generalized homology groups $\K_{\sharp}^{\rm t}(\topsp{X})$ and
$\KO_{\sharp}^{\rm t}(\topsp{X})$, along with the corresponding K-theory and
KO-theory groups. The complexification of a real vector bundle over
$\topsp{X}$ is a complex vector bundle over $\topsp{X}$ which is isomorphic to its own
conjugate vector bundle. The complexification map is compatible with
stable isomorphism of real and complex vector bundles, and thus
defines a homomorphism from stable equivalence classes of real vector
bundles to stable equivalence classes of complex vector bundles. It
thereby induces a natural transformation of cohomology theories
$$(\otimes\,\complex)^* \,:\, \KO^{*}(\topsp{X}) ~\longrightarrow~
\K^{*}(\topsp{X})$$
induced by
$$ [\topsp{E}]-[\topsp{F}] ~\longmapsto~ [\topsp{E}_{\mathbb{C}}]-[\topsp{F}_{\mathbb{C}}]$$
where $\topsp{E}_{\mathbb{C}}:=\topsp{E}\otimes\complex$ is the complexification of
the real vector bundle $\topsp{E}\to \topsp{X}$.

We can also define a complexification morphism relating the
homology theories
\beq
(\otimes\,\complex)_* \,:\, \KO_{\sharp}^{\rm t}(\topsp{X}) ~\longrightarrow~
\K_{\sharp}^{\rm t}(\topsp{X})
\label{complexmaphom}\eeq
by $$ [\topsp{M},\topsp{E},\phi]\otimes\complex:=[\topsp{M},\topsp{E}_{\mathbb{C}},\phi]$$
and extended by linearity, where on the right-hand side we regard $\topsp{M}$
endowed with the spin$^c$ structure induced by its spin structure as
a KO-cycle. One can easily see that
\beq
[\topsp{M},\topsp{E},\phi]\otimes\complex
=\phi_*\bigl([\topsp{E}_\complex]\,\frown\, [\topsp{M}]_\K\bigr)
\label{complexexpl}\eeq
where $[\topsp{M}]_\K\in\K_{\sharp}(\topsp{M})$ denotes the K-theory fundamental class of
$\topsp{M}$. Thus in the case when $\topsp{X}$ is a compact spin manifold, the homomorphism
$(\otimes\,\complex)_*$ is just the Poincar\'e dual of
$(\otimes\,\complex)^*$. This is clearly a natural transformation of
homology theories.

A related natural transformation between cohomology theories is the
realification morphism
$$(\,\underline{\phantom{\topsp{E}}}
\,\real)^* \,:\, \K^{\sharp}(\topsp{X}) ~\longrightarrow~ \KO^{\sharp}(\topsp{X})$$
induced by assigning to a complex vector bundle over $\topsp{X}$ the
underlying real vector bundle over $\topsp{X}$. Because a spin$^c$ manifold is not
necessarily spin, we cannot implement this transformation in the
homological setting in general. Rather, we must assume that $\topsp{X}$ is a
compact spin manifold. In this case the K-homology group $\K^{\rm
  t}_\sharp(\topsp{X})$ has generators $[\topsp{X} \times \topsp{S}^n,\topsp{E}_i,\pr_1]-[\topsp{X}
\times\topsp{S}^n,\topsp{F}_i,\pr_1]$, $0 \leq n \leq 7$, where $\pr_1:\topsp{X}\times\topsp{S}^n\to
\topsp{X}$ is the projection onto the first factor \cite{Reis2006}. We can then define the
morphism
$$(\,\underline{\phantom{\topsp{E}}}\,\real)_* \,:\, \K_{\sharp}^{\rm t}(\topsp{X})
~\longrightarrow~ \KO_{\sharp}^{\rm t}(\topsp{X})$$
by
$$\bigl([\topsp{X} \times \topsp{S}^n,\topsp{E}_i,\pr_1]-[\topsp{X} \times \topsp{S}^n,\topsp{F}_i,\pr_1]\bigr)
\,\underline{\phantom{\topsp{E}}}\,\real:=
[\topsp{X} \times \topsp{S}^n,\topsp{E}_i\,\underline{\phantom{\topsp{E}}}\,\real,\pr_1]-
[\topsp{X} \times \topsp{S}^n,\topsp{F}_i\,\underline{\phantom{\topsp{E}}}\,\real,\pr_1]$$ and
extending by linearity. Since this definition depends on a choice of
generators for $ \K_{\sharp}^{\rm t}(\topsp{X})$, the transformation is not
natural. As for the complexification morphism, the morphism
$(\,\underline{\phantom{\topsp{E}}}\,\real)_*$ thus defined is Poincar\'e dual
to $(\,\underline{\phantom{\topsp{E}}}\,\real)^*$. It follows that the
composition $(\,\underline{\phantom{\topsp{E}}}\,\real)_* \circ
(\otimes\,\complex)_*$ is multiplication by 2.

We can use the natural transformation provided by the complexification
homomorphism (\ref{complexmaphom}) to define a real homological Chern
character
\beq
\ch_{\bullet}^{\mathbb{R}} \,:\, \KO_{\sharp}^{\rm t}(\topsp{X}) ~\longrightarrow~
\H^{\phantom{1}}_{\sharp}(\topsp{X},\mathbb{Q})
\label{realCherndef}\eeq
by
$$\ch_{\bullet}^{\mathbb{R}}(\xi)=\ch^{\phantom{1}}_{\bullet}(\xi\otimes\complex) $$
for $\xi\in\KO_\sharp^{\rm t}(\topsp{X})$, where on the right-hand side we use
the K-homology Chern character $\ch^{\phantom{1}}_\bullet:\K^{\rm
  t}_\sharp(\topsp{X})\to\H^{\phantom{1}}_\sharp(\topsp{X},\rat)$, which can be defined as
\begin{displaymath}
\ch^{\phantom{1}}_\bullet(\topsp{M},\topsp{E},\phi)=\phi_{*}\bigl((\ch(\topsp{E})\cup{e^{d(\topsp{TM})/2}\hat{\topsp{A}}(\topsp{TM})})\cap[\topsp{M}]\bigr)
\end{displaymath}
where the cohomology class $d(\topsp{TM})$ is defined in {\appcharac}.\\
The real Chern character (\ref{realCherndef}) is a
natural transformation of homology theories, and it preserves the cap product, i.e. the diagram
\begin{displaymath}
\xymatrix{{\KO^{*}_{\rm t}(\topsp{X})\otimes\KO_\sharp^{\rm t}(\topsp{X})}\ar[r]^{\:\:\qquad\cap} \ar[d]_{\ch^{\bullet}\otimes\ch^{\phantom{1}}_\bullet }& \KO_\sharp^{\rm t}(\topsp{X})\ar[d]^{\ch^{\phantom{1}}_\bullet}\\
{\rm \topsp{H}}^{*}_{\rm t}(\topsp{X};\mathbb{Q})\otimes{\rm \topsp{H}}_\sharp^{\rm t}(\topsp{X};\mathbb{Q})\ar[r]^{\:\:\qquad\cap}& \H_\sharp^{\rm t}(\topsp{X};\mathbb{Q})
}
\end{displaymath}
commutes for any space $\topsp{X}$. This property is indeed guaranteed by the fact that $\ch^{\mathbb{R}}_\bullet$ is the natural transformation in KO-homology induced by the real Chern character in KO-theory. See \cite{jakob} for details. \topsp{F}inally, tensoring with
$\mathbb{Q}$ gives a map
$$\ch_{\bullet}^{\mathbb{R}} \otimes \Id_{\mathbb{Q}}=
\big(\ch^{\phantom{1}}_{\bullet} \otimes
\Id_{\mathbb{Q}}\big) \circ \big((\otimes\,\complex)_* \otimes
\Id_{\mathbb{Q}}\big)\,:\, \KO_{\sharp}^{\rm t}(\topsp{X})\otimes_{\mathbb{Z}}
\mathbb{Q}~\longrightarrow~
\H^{\phantom{1}}_{\sharp}(\topsp{X},\mathbb{Q})$$  
In constrast with the complex case, this map is \emph{not} an isomorphism.
\section{Cohomological Index formulas}
We will now explore the relation between the homological real
Chern character and the topological index defined in
(\ref{indntdef}). In particular, we will be able to give cohomological $\Cl_{n}$-index formulas.  We first show that up to Poincar\'e duality the
topological index is the homological morphism induced by the
collapsing map. Recall that up to isomorphism, the AMS invariant is
given by
\begin{displaymath}
\widehat{\mathcal{A}}_{\topsp{E}}(\topsp{M})=\tilde\zeta^{\,\KO}_{!}[\topsp{E}]
\end{displaymath}
where $\topsp{M}$ is a compact spin manifold of dimension $n$, $\topsp{E}$ is a real
vector bundle over $\topsp{M}$, $\tilde\zeta:\topsp{M}\to\pt$ is the collapsing map on
$\topsp{M}$, and $\tilde\zeta_!^{\,\KO}$ is the corresponding Gysin
homomorphism\footnote{We are using the Gysin homomorphism definition  for continuous and proper maps, which are not necessarily embeddings. See \cite{KAROUBI} for details.} on KO-theory. In this case we have
\begin{displaymath}
\tilde\zeta_{*}^{\,\KO}=\Phi^{\phantom{1}}_{\pt}\circ\tilde
\zeta^{\,\KO}_{!}\circ\Phi^{-1}_{\topsp{M}}
\end{displaymath}
where $\tilde\zeta^{\,\KO}_{*}$ is the induced morphism on $\KO^{\rm
  t}_{\sharp}(\topsp{M})$, and $\Phi_{\pt}$ and $\Phi_{\topsp{M}}$ are the Poincar\'e
duality isomorphisms on $\pt$ and $\topsp{M}$, respectively. It then follows
that
\begin{eqnarray}
\Phi_{\pt}^{\phantom{1}}\circ{}\ind_n^{\rm t}(\topsp{M},\topsp{E},\phi)&=&
\Phi^{\phantom{1}}_{\pt}\circ\tilde\zeta^{\,\KO}_{!}[\topsp{E}]\nonumber\\
&=&\Phi^{\phantom{1}}_{\pt}\circ\tilde\zeta^{\,\KO}_{!}\circ
\Phi^{-1}_{\topsp{M}}(\topsp{M},\topsp{E},\Id_\topsp{M})\nonumber\\
&=&\tilde\zeta_{*}^{\,\KO}[\topsp{M},\topsp{E},\Id_\topsp{M}]\nonumber\\
&=&[\topsp{M},\topsp{E},\tilde\zeta\,]~=~\zeta_{*}^\KO[\topsp{M},\topsp{E},\phi]
\label{topindhommorph}\end{eqnarray}
where $\zeta:\topsp{X}\to{\pt}$ is the collapsing map on $\topsp{X}$ with
$\tilde\zeta=\zeta\circ\phi$.\\
We will next describe how the real Chern character can be used to give
a characteristic class description of the map $\ind_n^{\rm t}$ in the
torsion-free cases. Consider first the case $n\equiv4~{\rm mod}\,
8$. We begin by showing that there is a commutative diagram
\beq
\xymatrix{\KO^{\rm t}_{4}(\topsp{X})\ar[r]^{\zeta^{\KO}_{*}}
\ar[rd]_{\zeta^{\H}_{*}\circ{}\ch^{\mathbb{R}}_{\bullet}~} &
\KO^{\rm t}_{4}(\pt)\ar[d]^{\ch^{\mathbb{R}}_{\bullet}}\\ &
\H_{0}(\pt,\mathbb{Q})}
\eeq
where $\zeta_*^\H$ is the induced morphism on homology. Recall that
$\ch^{\mathbb R}_{\bullet}=\ch^{\phantom{1}}_{\bullet}\circ(\otimes\,\complex)_*$,
where $(\otimes\,\complex)_{*}$ is the complexification map
(\ref{complexmaphom}), hence it can be expressed as
\begin{displaymath}
\ch^{\mathbb{R}}_\bullet(\topsp{M},\topsp{E},\phi)=\phi_{*}\bigl((\ch(\topsp{E}_{\mathbb{C}})\cup\hat{\topsp{A}}(\topsp{TM}))\cap[\topsp{M}]\bigr)
\end{displaymath}
 Then one has
\begin{eqnarray*}
\zeta^{\H}_{*}\circ{}\ch^{\mathbb{R}}_{\bullet}(\topsp{M},\topsp{E},\phi)&=&
\zeta^{\H}_{*}\circ{}\phi_{*}\bigl((\ch^{\bullet}(\topsp{E}_{\mathbb{C}})\,\cup\,{}
\widehat{{\topsp{A}}}(\topsp{TM}))\,\cap\,[\topsp{M}]\bigr)\\
&=&(\zeta\circ{\phi})_{*}\bigl((\ch^{\bullet}(\topsp{E}_{\mathbb{C}})\,\cup\,
\widehat{{\topsp{A}}}(\topsp{TM}))\,\cap\,[\topsp{M}]\bigr)\\
&=&\ch^{\mathbb{R}}_{\bullet}(\topsp{M},\topsp{E},\tilde\zeta\,)~=~
\ch^{\mathbb{R}}_{\bullet}\circ\zeta^{\KO}_{*}[\topsp{M},\topsp{E},\phi] \ .
\end{eqnarray*}

Now we use the fact that the map
$\ch^{\mathbb{R}}_{\bullet}:\KO^{\rm t}_{4}(\pt)\to{\H_{0}(\pt,\mathbb{Q})}$
sends $\mathbb{Z}\to{2\mathbb{Z}}\subset{\mathbb{Q}}$ \cite{KAROUBI}. On its image,
the homomorphism $\ch^{\mathbb{R}}_{\bullet}$ is thus invertible and
its inverse is given as division by 2. This can be proven as follows. We have
\begin{eqnarray}
\zeta^{\H}_{*}\circ{}\ch^{\mathbb{R}}_{\bullet}(\topsp{M},\topsp{E},\phi)&=&
\zeta^{\H}_{*}\circ{\Phi_{\topsp{M}}}\bigl(\ch^{\bullet}(\topsp{E}_{\mathbb{C}})\,\cup\,
{\widehat{{A}}(\topsp{TM})}\bigr)\nonumber\\
&=&\Phi_{\pt}\circ{}\zeta^{\H}_{!}\bigl(\ch^{\bullet}(\topsp{E}_{\mathbb{C}})\,
\cup\,\widehat{{A}}(\topsp{TM})\bigr)\nonumber\\
&=&\bigl\langle\ch^{\bullet}(\topsp{E}_{\mathbb{C}})\,\cup\,{\widehat{{A}}(\topsp{TM})}\,,\,
[\topsp{M}]\bigr\rangle \ ,
\label{div2expl}\end{eqnarray}
where $\langle-,-\rangle:\H^\sharp(\topsp{M},\rat)\times\H_\sharp(\topsp{M},\rat)\to\rat$ is
the canonical dual pairing between cohomology and homology. In
(\ref{div2expl}) we have used the fact that $\Phi_{\pt}$ is the
identity on $\H_{0}(\pt,\mathbb{Q})\simeq\mathbb{Q}\,$, and the proof of
the last equality uses the Atiyah-Hirzebruch version of the
Grothendieck-Riemann-Roch theorem \cite{KAROUBI}. Recall that for a spin manifold $\topsp{M}$ of dimension $4k+8$,
one has
$$\bigl\langle\ch^{\bullet}(\topsp{E}_{\mathbb{C}})\,\cup\,{\widehat{{A}}(\topsp{TM})}\,,\,
[\topsp{M}]\bigr\rangle\in2\zed$$ After using the isomorphism
$\KO_{4}(\pt)\simeq\mathbb{Z}$, we thus deduce that
$$\zeta_*^\KO[\topsp{M},\topsp{E},\phi]=
\frac12\,\bigl\langle\ch^{\bullet}(\topsp{E}_{\mathbb{C}})\,\cup\,
{\widehat{{A}}(\topsp{TM})}\,,\, [\topsp{M}]\bigr\rangle$$ and from
(\ref{topindhommorph}) we arrive finally at
\begin{displaymath}
\ind_n^{\rm t}(\topsp{M},\topsp{E},\phi)=\mbox{$\frac12$}
\,\bigl\langle\ch^{\bullet}(\topsp{E}_{\mathbb{C}})\,\cup\,
{\widehat{{A}}(\topsp{TM})}\,,\, [\topsp{M}]\bigr\rangle \ 
\end{displaymath}
When $n\equiv0\:{\rm mod}\:8$, one obtains a similar result but now
without the factor $\frac{1}{2}$, since in that case
$\ch^{\mathbb{R}}_{\bullet}:\KO^{\rm t}_0(\pt)\to\H_0(\pt,\rat)$ is the
inclusion $\mathbb{Z}\hookrightarrow\mathbb{Q}$ \cite{KAROUBI}.
In the remaining non-trivial cases $n\equiv1,2~{\rm mod}~8$
the homological Chern character is of no use, as $\KO^{-n}(\pt)$ is
the pure torsion group $\mathbb{Z}_{2}$, and there is no cohomological
formula for the AMS invariant in these instances. However, by using
Theorem~\ref{spinthm} one still has an interesting mod~2 index formula
for the topological index in these cases as well~\cite{LM}. We can
summarize our homological derivations of these index formulas as
follows.

\begin{theorem}
Let $[\topsp{M},\topsp{E},\phi]\in\KO^{\rm t}_n(\topsp{X})$, and let
$~\nslash{\mathfrak{D}}^\topsp{M}_{\topsp{E}}$ be the Atiyah-Singer operator on $\topsp{M}$ with
coefficients in $\topsp{E}$. Let
$~\nslash{\mathcal{H}}^\topsp{M}_\topsp{E}:=\ker~\nslash{{D}}^\topsp{M}_{\topsp{E}}$ denote the
vector space of real harmonic $\topsp{E}$-valued spinors on $\topsp{M}$, where $\:\nslash{{D}}^\topsp{M}_{\topsp{E}}$ is the Atiyah-Singer operator for the irreducible real spinor bundle on $\topsp{M}$. Then one has
the $\Cl_n$-index formulas
$$
\ind_n^{\rm t}(\topsp{M},\topsp{E},\phi)~=~\left\{\begin{array}{cc}
\bigl\langle\ch^{\bullet}(\topsp{E}_{\mathbb{C}})\,\cup\,
{\widehat{{A}}(\topsp{TM})}\,,\, [\topsp{M}]\bigr\rangle \ , & \quad n\equiv0\:{\rm
  mod}\:8 \ , \\ \dim_\complex~\nslash{\mathcal{H}}^\topsp{M}_\topsp{E}~{\rm
  mod}\:2 \ , & \quad n\equiv1\:{\rm mod}\:8 \ , \\
\dim_{\mathbb{H}}~\nslash{\mathcal{H}}^\topsp{M}_\topsp{E}~
{\rm mod}\:2 \ , & \quad n\equiv2\:{\rm mod}\:8 \ , \\
\mbox{$\frac12$}\,\bigl\langle\ch^{\bullet}(\topsp{E}_{\mathbb{C}})\,\cup\,
{\widehat{{A}}(\topsp{TM})}\,,\, [\topsp{M}]\bigr\rangle \ , & \quad n\equiv4\:{\rm
  mod}\:8 \ , \\ 0 \ , & {\rm otherwise} \ . \end{array} \right.
$$
\label{Clnindexthm}\end{theorem}
\section{D-branes and K-homology}\label{wrappedbranes}
In this section we will show how geometric K-homology can be used to describe D-branes in type II and type I String theory in a topological nontrivial spacetime. We will introduce the notion of \emph{wrapped D-brane} on a given submanifold of spacetime, we will define the group of charges of wrapped D-branes, and construct explicit examples of wrapped D-branes which have torsion charge.\\   
Recall by section \ref{typeIIbranes} that the group of topological charges of a D$p$-brane realized as a $\topsp{spin}^{c}$ submanifold $\topsp{W}\subset{\topsp{X}}$ is given by
\begin{displaymath}
\K_{\rm cpt}^{0}({\nu_{\topsp{W}}})\simeq\K^{0}(\topsp{B(N(W))},\topsp{S(N(W))})
\end{displaymath}
as proposed by Witten. According to the Sen-Witten construction, the classes in $\K_{\rm cpt}^{0}({\nu_{W}})$ are interpreted as systems of $\rm D9-D\bar{9}$ branes which are unstable, and will decay onto the worldvolume $\topsp{W}$, which correspond to the zero loci of the appropriate tachyon field. In particular, this process happens in spacetime, and it depends on how the worldvolume is embedded in it. On the other hand, the role played by the Chan-Paton vector bundle on the D$p$-brane is not manifest in this classification. However, there is a natural way of classifying the D$p$-branes on $\topsp{W}$ by means which manifestly takes into account the Chan-Paton bundle contribution. Indeed, from the D$p$-brane data, we can naturally construct the Baum-Douglas cycle $(\topsp{W},\topsp{E},{\rm id})$, where $\topsp{E}$ denotes the Chan-Paton bundle, and declare that its charge is given by the class $[\topsp{W},\topsp{E},{\rm id}]\in\K_{p+1}(\topsp{W})$. As the group $\K_{p+1}(\topsp{W})$ contains no information about the embedding of the worldvolume $\topsp{W}$ in $\topsp{X}$, we can intuitively think the charge $[\topsp{W},\topsp{E},{\rm id}]$ takes into account how the D-brane wraps the submanifold $\topsp{W}$. Notice that this analogous to the charge classification of an extended object in an abelian gauge theory via the homology cycle of its worldvolume. Finally, the equivalence relation that defines the group $\K_{p+1}(\topsp{W})$ are very natural from the physical point of view. See \cite{Reis2006} for details. At this point, we notice that by definition the elements of $\K_{p+1}^{\rm t}(\topsp{W})$ are given by (differences of) classes $[\topsp{M},\topsp{E},\phi]$ where $\topsp{M}$ is a $p+1$-dimensional manifold. However, it is \emph{not} always possible to choose the map $\phi$ in $[\topsp{M},\topsp{W}]$ in such a way that $\phi$ is a diffeomorphism. This motivates the following
\begin{definition} Let $\topsp{X}$ be a type II String theory spacetime, described by a 10-dimensional spin manifold, and let $\topsp{W}\subset{\topsp{X}}$ be a $\topsp{spin}^{c}$ submanifold. A D$p$-brane \emph{wrapping} the worldvolume $\topsp{W}$ is defined as the K-cycle $(\topsp{M},\topsp{E},\phi)$, where dim $\topsp{M}$ =$p+1$, and $\phi(\topsp{M})\subset{\topsp{W}}$. We will call $\topsp{E}$ the Chan-Paton bundle associated to the wrapped D$p$-brane, and we will say that the D$p$-branes \emph{fills} $\topsp{W}$ if $\phi(\topsp{M})=\topsp{W}$. The charge of the wrapped D$p$-brane is given by the class $[\topsp{M},\topsp{E},\phi]$ in the group $\K_{p+1}^{\rm t}(\topsp{W})$.
\end{definition} 
Notice that in the above definition we have relaxed the condition that ${\rm dim} \topsp{W}=p+1$, as we are not requiring that the wrapping preserves the dimension of the D-brane. This is an attempt to take into account, at least at the topological level, the well-known fact that D-branes are not always representable as submanifolds equipped with vector bundles, since they are boundary conditions for a superconformal field theory, and that a distinction should be made between the wrapping D-brane, in this case identified with a K-cycle representing a particular type of boundary conditions, and the worldvolume it wraps. Notice also that the group of charges of wrapped D$p$-branes does \emph{not} depend on how the manifold $\topsp{W}$ is embedded into the spacetime, and hence it seems to represent a genuine worldvolume concept. In particular, as mentioned above, the wrapped D-brane definition is very natural  in the ordinary case of a D-brane realized as a submanifold $\topsp{W}$ of spacetime equipped with a Chan-Paton bundle $\topsp{E}$, as it only depends on how the vector bundle is defined on the submanifold, and not on the procedure used to ``extend'' it to the spacetime. Finally, in the case of ordinary D-branes wrapping  $\topsp{W}$ with ${\rm dim} \topsp{W}=p+1$, the group $\K_{p+1}(\topsp{W})$ coincides with the group of charges of type IIB D$p$-branes that can be obtained via the Sen-Witten construction, i.e. via brane-antibrane decay. This can be shown as follows.  Since the normal bundle $\nu_{\topsp{W}}\to{\topsp{W}}$ is a $\topsp{spin}^{c}$ vector bundle, we can use the Thom isomorphism in K-theory to establish that
\begin{displaymath}
\K_{\rm cpt}^{0}({\nu_{\topsp{W}}})\simeq{\K^{0}(\topsp{W})}
\end{displaymath}
As $\topsp{W}$ is a $\topsp{spin}^{c}$ submanifold of the spacetime, we can use Poincar\'e Duality to get
\begin{displaymath}
\K^{0}(\topsp{W})\simeq{\K^{t}_{p+1}(\topsp{W})}
\end{displaymath}
where $p+1={\rm dim}\:\topsp{W}$. This suggests that for ordinary type II D$p$-branes the wrapping charge is completely determined by the decay of the tachyon field. It is natural at this point to extend the notion of wrapped D-brane and of wrapping charge to type I String theory. In this case, though, the two notions of charge do not coincide, as we will show in the following.\\
Recall that in type I String theory the group of topological charges of a D$p$-brane realized as a $\topsp{spin}$ submanifold $\topsp{W}\subset{\topsp{X}}$ is given by
\begin{equation}\label{braneanti}
\KO_{\rm cpt}^{0}({\nu_{\topsp{W}}})\simeq\KO^{0}(\topsp{B(N(W))},\topsp{S(N(W))})
\end{equation} 
By using the Thom isomorphism in KO-theory, we have that
\begin{displaymath}
\KO_{\rm cpt}^{0}({\nu_{\topsp{W}}})\simeq{\KO^{p-9}(\topsp{W})}
\end{displaymath}
Finally, by Poincar\'e Duality, we get
\begin{displaymath}
\KO^{p-9}(\topsp{W})\simeq{\KO_{10}(\topsp{W})}
\end{displaymath}
The group $\KO_{10}(\topsp{W})$ is in general not isomorphic to the group $\KO_{p+1}(\topsp{W})$, and explicitly depends on the dimension of the spacetime. We can physically interpret the elements of $\KO_{10}(\topsp{W})$ as equally charged systems of wrapping $\rm D9-D\bar{9}$-branes decaying on the submanifold $\topsp{W}$, and via the inclusion $i:\topsp{W}\hookrightarrow{\topsp{X}}$ they can be related to the D9-branes used in the Sen-Witten construction. This is not surprising, as the decay mechanism is somehow at the heart of the spacetime D-brane charge classification, and it reinforces the statement that the group (\ref{braneanti}) encodes spacetime properties of the D$p$-brane. After introducing the K-theoretical description of Ramond-Ramond fields in the next chapter, we will argue that wrapped D-branes can in principle couple to Ramond-Ramond fields.\\

We comment at this point that the definition of wrapped branes and of wrapped charge as presented in this section is very natural from the mathematical point of view, but is still heuristic in nature. Indeed, stronger evidences for wrapped D-branes and their coupling to Ramond-Ramond fields should come from the boundary conformal field theory describing type I String theory in topologically nontrivial settings, which has not yet been investigated in full generality.   
\section{The group $\KO_{\sharp}^{\rm t}(\rm pt)$ and torsion branes}
In this section we will find explicit generators for the group $\KO_{\sharp}^{\rm t}(\rm pt)$, which, as explained in the previous section, we can interpret as stable D-branes wrapping a point in spacetime. The stability of such D-branes is related to the fact that their charge is the ``lightest'' in the group of possible charges, hence the conservation of charge does not allow these D-branes to decay. To this end, we first give sufficient combinatorial criteria on the rational homology of $\topsp{X}$ which ensure that torsion-free D-branes can wrap non-trivial spin cycles of the spacetime $\topsp{X}$, which is an adaptation of an analogous result in \cite{Reis2006}
\begin{theorem}
 Let $\topsp{X}$ be a compact connected finite CW-complex of dimension
$n$ whose rational homology can be presented as
$$\H_\sharp(\topsp{X},\mathbb{Q})= \bigoplus_{p=0}^{n}~ \bigoplus_{i=1}^{m_p}
\,\big[\topsp{M}^{p}_{i}\big] {}~\mathbb{Q} \ , $$ where $\topsp{M}_i^{p}$ is a
$p$-dimensional compact connected spin submanifold of $\topsp{X}$ without
boundary and with orientation cycle $[\topsp{M}_i^{p}]$ given by the spin
structure. Then the KO-homology lattice $\Lambda_{\KO_\sharp^{\rm t}(\topsp{X})}:=\KO_\sharp^{\rm
t}(\topsp{X})\,/\,\tor_{\KO_\sharp^{\rm
    t}(\topsp{X})}$ contains a set of linearly independent elements given by
the classes of KO-cycles
$$\big[\topsp{M}_{i}^p,\id^\real_{\topsp{M}_i^p},\iota_i^p\big] \ , \quad 0 \leq p\leq n
\ , ~~ 1 \leq i\leq m_p \ . $$
\label{HomKOcycles}\end{theorem}

\begin{proof} By~\cite{Reis2006} the cycles
  $\big[\topsp{M}_{i}^p,\id^\complex_{\topsp{M}_i^p},\iota_i^p\big]$ form a rational basis for
  the lattice
$\Lambda_{\K_\sharp^{\rm t}(\topsp{X})}:=\K_\sharp^{\rm
t}(\topsp{X})\,/\,\tor_{\K_\sharp^{\rm
    t}(\topsp{X})}$ in K-homology. The conclusion follows from the fact that
    $$\big[\topsp{M}_{i}^p,\id^\real_{\topsp{M}_i^p},\iota_i^p\big]\otimes\complex
=\big[\topsp{M}_{i}^p,\id^\complex_{\topsp{M}_i^p},\iota_i^p\big] \ , $$
     i.e. that the elements $\ch_{\bullet}^{\mathbb{R}}
     \big(\topsp{M}_{i}^p,\id^\reals_{\topsp{M}_i^p},\iota_i^p\big)$
     form a set of generators of $\H_\sharp(\topsp{X},\mathbb{Q})$.
\end{proof}
Recall now, that we have $\KO_{n}^{\rm t}(\rm pt)\simeq{\mathbb{Z}}$ for $n=0,4$, $\KO_{n}^{\rm t}(\rm pt)\simeq{\mathbb{Z}_{2}}$ for $n=1,2$, and 0 otherwise. For $n=0,4$ Theorem \ref{HomKOcycles} and the Chern character assure that the classes $[\rm{pt},\id_{\rm{pt}}^{\real},\Id^{\phantom{1}}_\pt]$ and
$[\topsp{S}^{4},\id_{\topsp{S}^4}^{\real},\zeta]$ are generators of the groups
$\KO^{\rm t}_{0}(\rm{pt})\simeq\zed$ and $\KO^{\rm
  t}_{4}(\rm{pt})\simeq\zed$, respectively. 
Let us now consider the group $\KO_1^{\rm t}(\pt)$. Consider the
circle $\topsp{S}^{1}$ and assign to it a Riemannian metric. Since there is
only one unit tangent vector at any point of $\topsp{S}^{1}$, one has
$P_{{\,\topsp{S}O}}(\topsp{S}^{1})\simeq{\topsp{S}^{1}}$, where $P_{{\,\topsp{S}O}}(\topsp{S}^{1})$ is the orthonormal frame bundle of $\topsp{S}^{1}$. A spin structure on $\topsp{S}^{1}$ is thus
given by a double covering
\begin{displaymath}
P_{{\,\Spin}}\big(\topsp{S}^{1}\big)~\longrightarrow~{\topsp{S}^{1}}
\end{displaymath}
and by the fibration
\begin{displaymath}
\xymatrix{
{\mathbb{Z}_{2}}~\ar[r]~&~P_{{\,\Spin}}\big(\topsp{S}^{1}\big)\ar[d] \ . \\ &{\topsp{S}^{1}}}
\end{displaymath}
There are only two double coverings of the circle, one disconnected
and the other connected, given respectively by
\begin{displaymath}
\topsp{S}^{1}\times{\mathbb{Z}_{2}}~\longrightarrow~{\topsp{S}^{1}} \ , \qquad
\topsp{S}^{1}_{\rm \topsp{M}}~\longrightarrow~{\topsp{S}^{1}}
\end{displaymath}
where $\topsp{S}^{1}_{\rm \topsp{M}}$ is the total space of the principal
$\mathbb{Z}_{2}$-bundle associated to the M\"obius strip. 
We will call these two spin structures the ``interesting'' and
the ``uninteresting'' spin structures, respectively.

Corresponding to these two spin structures (labelled `i' and `u',
respectively), we construct classes in $\KO^{\rm t}_{1}(\rm{pt})$
given by $[\topsp{S}_{\rm i}^{1},\id_{\topsp{S}^{1}}^{\real},\zeta]$ and $[\topsp{S}_{\rm
  u}^{1},\id_{\topsp{S}^{1}}^{\real},\zeta]$ where $\zeta:\topsp{S}^{1}\to{\rm{pt}}$
is as usual the collapsing map. We will now compute the topological
indices in detail, finding the AMS invariants~\cite{LM}
\begin{displaymath}
\widehat{\mathcal{A}}_{\id_{{\topsp{S}}^{1}}^{\real}}\big(\topsp{S}_{\rm i}^{1}
\big)~=~1 \ , \qquad
\widehat{\mathcal{A}}_{\id_{{\topsp{S}}^{1}}^{\real}}\big(\topsp{S}_{\rm u}^{1}\big)~=~0
\end{displaymath}
in $\KO^{-1}(\rm{pt})\simeq{\mathbb{Z}_{2}}$. Hence the two classes
above represent the elements of $\KO^{\rm
  t}_{1}(\rm{pt})\simeq{\mathbb{Z}_{2}}$. In particular,
$[\topsp{S}_{\rm i}^{1},\id_{{\topsp{S}}^{1}}^{\real},\zeta] $ is a generator, analogous to the
non-BPS Type~I D-particle that arises from tachyon condensation.

Let us first consider the circle with the interesting spin structure.
Since $\Cl_{1}\simeq{\mathbb{C}}$, one has
$~\nslash{\mathfrak{S}}(\topsp{S}^1):=P_{{\,\Spin}}(\topsp{S}^{1})\times_{\zed_2}\Cl_{1}\simeq
{\topsp{S}^{1}\times{\mathbb{C}}}$. By decomposing
$\mathbb{C}=\mathbb{R}\oplus{\ii\mathbb{R}}$, one has the
identifications
$~\nslash{\mathfrak{S}}^{0}(\topsp{S}^1)=\topsp{S}^{1}\times{\mathbb{R}}$ and
$~\nslash{\mathfrak{S}}^{1}(\topsp{S}^1)=\topsp{S}^{1}\times{\ii\mathbb{R}}$. As the
Clifford bundle is trivial, its space of sections is given by
$\C^\infty(\topsp{S}^1,~\nslash{\mathfrak{S}}(\topsp{S}^1))=
\C^{\infty}(\topsp{S}^{1},\mathbb{C})$. 
By coordinatizing the circle $\topsp{S}^{1}$ with arc length $s$, the
Atiyah-Singer operator  can be expressed as
\beq
\,\nslash{\mathfrak{D}}^{\topsp{S}^1}=\ii\,\mbox{$\frac{\dd}{\dd s}$}
\label{ASopS1}\eeq
where $e_{1}=\ii$ is a generator of the Clifford algebra $\Cl_{1}$. To
compute the topological index
$\widehat{\mathcal{A}}_{\id_{{\topsp{S}}^{1}}^{\real}}(\topsp{S}_{\rm i}^{1})$, we
use the $\Cl_1$-index Theorem~\ref{spinthm} and hence determine the vector
space $\ker~\nslash{\mathfrak{D}}^{\topsp{S}^1}$, or equivalently the chiral
subspace $\ker(\,\nslash{\mathfrak{D}}^{\topsp{S}^1})^{0}$. Since
$\C^\infty(\topsp{S}^1,~\nslash{\mathfrak{S}}^{0})=
\C^{\infty}(\topsp{S}^{1},\mathbb{R})$, the kernel of the chiral
Atiyah-Singer operator $(\,\nslash{\mathfrak{D}}^{\topsp{S}^1})^{0}:
\C^\infty(\topsp{S}^1,~\nslash{\mathfrak{S}}^{0}) \rightarrow{}
\C^\infty(\topsp{S}^1,~\nslash{\mathfrak{S}}^{1})$ is given by the space of
real-valued constant functions on $\topsp{S}^{1}$. The dimension of this
vector space, as a module over $\Cl_{1}^{0}\simeq{\mathbb{R}}$, is 1
and hence
\begin{displaymath}
\ind_1^{\rm t}\big(\topsp{S}_{\rm
  i}^{1}\,,\,\id_{\topsp{S}^{1}}^{\real}\,,\,\zeta\big)~=~
\big[\ker(\,\nslash{\mathfrak{D}}^{\topsp{S}^1})^{0}\big]~=~1
\end{displaymath}
in ${\mathfrak{M}}_{0}/\imath^{*}{\mathfrak{M}}_{1}\simeq
{\KO^{-1}(\rm{pt})} \simeq{\mathbb{Z}_{2}}$. (Note that here we are
using \emph{ungraded} Clifford modules.)

We now turn to the uninteresting spin structure on $\topsp{S}^{1}$.
This time the bundle $~\nslash{\mathfrak{S}}(\topsp{S}^{1})$ is the
(infinite complex) M\"obius bundle. It can be described by a
trivialization made of three charts $U_{1}$, $U_{2}$ and $U_{3}$ with
$\mathbb{Z}_{2}$-valued transition functions $g_{12}=1$, $g_{23}=1$
and $g_{31}=-1$. In this case, the vector space
$\ker(\,\nslash{\mathfrak{D}}^{\topsp{S}^1})^{0}$ consists of locally constant
real-valued functions $\psi_{i}$ defined on $U_{i}$ which satisfy
$\psi_{j}=g_{ji}\,\psi_{i}$ on the intersections
$U_{i}\cap{U_{j}}\neq{\emptyset}$. Because of the non-trivial
transition function $g_{31}$, there are no non-zero solutions $\psi$
to the equation $(\,\nslash{\mathfrak{D}}^{\topsp{S}^1})^{0}\psi=0$. The kernel
$\ker(\,\nslash{\mathfrak{D}}^{\topsp{S}^1})^{0}$ is thus trivial, and so
\begin{displaymath}
\ind_1^{\rm t}\big(\topsp{S}_{\rm
  u}^{1}\,,\,\id_{\topsp{S}^{1}}^{\real}\,,\,\zeta\big)=0 \ .
\end{displaymath}
Let us now consider the structure of the group $\KO_2^{\rm
  t}(\pt)$. Analogously to the construction above, one can equip the
torus $\topsp{T}^2=\topsp{S}^{1}\times{\topsp{S}^{1}}$ with an ``interesting'' spin
structure and show that
\begin{displaymath}
\widehat{\mathcal{A}}_{\id_{\topsp{T}^2}^{\real}}\big(\topsp{T}^2\big)=1 \ ,
\end{displaymath}
and also that
\begin{displaymath}
\widehat{\mathcal{A}}_{\id_{\topsp{S}^2}^{\real}}\big(\topsp{S}^{2}\big)={0}
\end{displaymath}
in $\KO^{-2}(\rm{pt})\simeq\zed_2$. It follows that the classes
$[\topsp{T}^2,\id_{\topsp{T}^2}^{\real},\zeta]$ and
$[\topsp{S}^{2},\id_{\topsp{S}^2}^{\real},\zeta]$ represent the elements of the
group $\KO^{\rm t}_{2}(\rm{pt})\simeq\zed_2$. In particular,
$[\topsp{T}^2,\id_{\topsp{T}^2}^{\real},\zeta]$ is a generator, and
it is analogous to the Type~I non-BPS D-instanton which is usually
constructed as the $\Omega$-projection of the Type~IIB D$(-1)$
brane-antibrane system. We
will now give some details of these results.

Equip $\topsp{T}^2$ with the flat metric
$\dd\theta_{1}\otimes{\dd\theta_{1}}+\dd\theta_{2}\otimes{\dd\theta_{2}}$,
where $(\theta_{1},\theta_{2})$ are angular coordinates on
$\topsp{S}^{1}\times\topsp{S}^1$. Since $\topsp{T}^2$ is a Lie group, its tangent bundle is
trivializable, and hence the oriented orthonormal frame bundle is
canonically given by $P_{{\,\SO}}(\topsp{T}^2)=\topsp{T}^2\times{\topsp{S}^{1}}$. Consider the
spin structure on $\topsp{T}^2$ given by
\begin{displaymath}
P_{{\,\Spin}}\big(\topsp{T}^2\big)=\topsp{T}^2\times{\topsp{S}^{1}}~\xrightarrow{\Id_{\topsp{T}^2}
\times{z^{2}}}~\topsp{T}^2\times{\topsp{S}^{1}} \ .
\end{displaymath}
Since $\Cl^{\phantom{0}}_{2}\simeq{\mathbb{H}}$ and
$\Cl^{0}_{2}\simeq{\mathbb{C}}$, the corresponding Clifford bundles are
$~\nslash{\mathfrak{S}}(\topsp{T}^2)=\topsp{T}^2\times\mathbb{H}$ and
$\nslash{\mathfrak{S}}^{0}(\topsp{T}^2)=\topsp{T}^2\times\mathbb{C}$. In the
riemannian coordinates $(\theta_{1},\theta_{2})$, the Atiyah-Singer
operator can be expressed as
\begin{displaymath}
\,\nslash{\mathfrak{D}}^{\topsp{T}^2}=\sigma_{1}\,\mbox{$\frac{\partial}
{\partial{\theta_{1}}}$}+\sigma_{2}\,\mbox{$
\frac{\partial}{\partial{\theta_{2}}}$}
\end{displaymath}
where the Pauli spin matrices
\begin{displaymath}
\sigma_{1}=\left(\begin{array}{cc}
0&{1}\\
1&{0}
\end{array}\right) \ , \qquad \sigma_{2}=\left(\begin{array}{cc}
0&{-\ii}\\
\ii&{0}
\end{array}\right)
\end{displaymath}
represent the generators $e_{1},\:e_{2}$ of $\Cl_{2}$, acting by left
multiplication. The chiral operator $(\,\nslash{\mathfrak{D}}^{\topsp{T}^2})^{0}$ is
locally the Cauchy-Riemann operator, and hence its kernel consists of
holomorphic sections of the chiral Clifford bundle
$~\nslash{\mathfrak{S}}^{0}(\topsp{T}^2)$. These are simply the
complex-valued constant functions on $\topsp{T}^2$, as the torus is a compact
complex manifold. As a module over $\Cl^{0}_{2}$, this vector space is
one-dimensional and so
\begin{displaymath}
\ind^{\rm t}_{2}\big(\topsp{T}^2\,,\,\id_{\topsp{T}^2}^{\real}\,,\,\zeta
\big)~=~\big[\ker(\,\nslash{\mathfrak{D}}^{\topsp{T}^2})^{0}
\big]~=~1
\end{displaymath}
in ${\mathfrak{M}}_{1}/\imath^{*}{\mathfrak{M}}_2
\simeq{\KO^{-2}(\rm{pt})} \simeq{\mathbb{Z}_{2}}$.

Consider now the two-sphere $\topsp{S}^{2}$ as a riemannian manifold.
It is not difficult to see that
\begin{displaymath}
P_{{\,\SO}}\big(\topsp{S}^{2}\big)=\SO(3)~\longrightarrow~\SO(3)/{\SO(2)}
\simeq{\topsp{S}^{2}}
\end{displaymath}
is the oriented orthonormal frame bundle over $\topsp{S}^{2}$.
The (unique) spin structure on $\topsp{S}^{2}$ is thus given by
\begin{displaymath}
\xymatrix{P_{{\,\Spin}}\big(\topsp{S}^{2}\big)\simeq\SU(2)~\ar[r]^{h}~\ar[rd]_{\UU(1)} &
~P_{{\,\SO}}\big(\topsp{S}^{2}\big)\simeq\SO(3)\ar[d]^{\SO(2)}\\ &\topsp{S}^{2}}
\end{displaymath}
with $h:\SU(2)\rightarrow{\SO(3)}$ the usual double covering, and by
\begin{displaymath}
\xymatrix{\UU(1)~\ar[r]~&{~P_{{\,\Spin}}\big(\topsp{S}^{2}\big)}\ar[d]\\
& {\topsp{S}^{2}} }
\end{displaymath}
which is the Hopf fibration of $\topsp{S}^2$. Recall that the group
$\Spin(2)\simeq{\U(1)}\simeq{\SO(2)}$ acts on $\Cl_{2}\simeq{\mathbb{H}}$
as multiplication by
\begin{displaymath}
\left(\begin{array}{cc}
\e^{\ii\theta}&0\\
0&\e^{-\ii\theta}
      \end{array}\right)\quad , \quad  \theta\in[0,2\pi) \ .
\end{displaymath}
If one gives the sphere $\topsp{S}^{2}$ the structure of the complex
projective line $\complex\P^1$, then there are isomorphims
$~\nslash{\mathfrak{S}}^{0}(\topsp{S}^{2})=P_{{\,\Spin}}(\topsp{S}^{2})\times_{\UU(1)}\mathbb{C}
\simeq{T^{1,0}\complex\P^1}$ since the bundle
$~\nslash{\mathfrak{S}}^{0}(\topsp{S}^{2})$ has the same transition functions
as the Hopf fibration. In other words,
$~\nslash{\mathfrak{S}}^{0}(\topsp{S}^{2})$ is isomorphic to the canonical
line bundle $L_\complex$ over $\complex\P^1$. The vector space
$\ker(\,\nslash{\mathfrak{D}}^{\topsp{S}^2})^{0}$ thus consists of the holomorphic
sections of $L_\complex$. The only such section on $\complex\P^1$ is the zero
section \cite{miranda}, and we finally find
 \begin{displaymath}
 \ind^{\rm t}_{2}\big(\topsp{S}^{2}\,,\,\id_{\topsp{S}^2}^{\real}\,,\,\zeta
\big)~=~\big[\ker(\,\nslash{\mathfrak{D}}^{\topsp{S}^2})^{0}\big]~=~0
 \end{displaymath}
in ${\mathfrak{M}}_{1}/\imath^{*}{\mathfrak{M}}_2
\simeq{\mathbb{Z}_{2}}$.
\begin{remark}
As we have seen above, the
problem of finding generators of the geometric KO-homology groups of a
space $X$ becomes
increasingly involved at a very rapid rate. Even in the case of
spherical D-branes, we have not been able to find a nice explicit
solution. Nevertheless, at least in these cases we can find a
formal solution as follows, which also illustrates the generic
problems at hand.

Suppose that we want to construct generating branes for the group $\KO^{\rm
  t}_{k}(\topsp{S}^{n})$ for some $n>0$. Poincar\'e duality gives the map
\beq
\KO^{n-k}\big(\topsp{S}^{n}\big)~\longrightarrow~{\KO^{\rm t}_{k}\big(
\topsp{S}^{n}\big)} \ , \qquad
\xi~\longmapsto~\xi\,\cap\,\big[\topsp{S}^n\,,\,\id_{\topsp{S}^n}^\real\,,\,
\Id_{\topsp{S}^n}^{\phantom{1}}\big] \ .
\label{PDKOSn}\eeq
As Poincar\'e duality is a group isomorphism, picking a generator in
$\KO^{n-k}(\topsp{S}^{n})$ will give a generator in $\KO^{\rm t}_{k}(\topsp{S}^{n})$. But
the problem is that the class $\xi$ is not a (virtual or stable)
vector bundle over $\topsp{S}^{n}$ in the cases of interest $k<n$. To this
end, we rewrite the cap product in (\ref{PDKOSn}) by using the
suspension isomorphism $\Sigma$ and the desuspension $\Sigma^{-1}$ to
get
\begin{displaymath}
\xi\,\cap\,\big[\topsp{S}^n\,,\,\id_{\topsp{S}^n}^\real\,,\,
\Id_{\topsp{S}^n}^{\phantom{1}}\big]=
\Sigma^{-1}\Big(\Sigma\big(\xi\big)\,\cap\,\Sigma\big[\topsp{S}^n\,,\,
\id_{\topsp{S}^n}^\real\,,\,\Id_{\topsp{S}^n}^{\phantom{1}}\big]\Big) \ .
\end{displaymath}
As we are interested only in generators, we can substitute
$\Sigma(\xi)$ with the generators of the KO-theory group
$\KO^0(\Sigma^{n-k}\topsp{S}^{n})=\widetilde{\KO}{}^0(\topsp{S}^{2n-k})$. The
generators of the latter groups are given by~\cite{KAROUBI} the canonical
line bundle $L_{\mathbb{F}}$ over the projective line $\mathbb{F}\P^1$, with
$\mathbb{F}$ the reals $\real$ for $k=2n-1$, the complex numbers
$\complex$ for $k=2n-2$, the quaternions $\mathbb{H}$ for $k=2n-4$ and
the octonions $\mathbb{O}$ for $k=2n-8$.
\end{remark}

\chapter{Abelian Gauge Theories and Differential Cohomology}
We have seen in the previous chapters that D-brane charges require the use of K-theory and K-homology to be properly described. Morever, since D-branes are electric and magnetic sources  for Ramond-Ramond fields, we expect that some form of K-theory will play a relevant role also in the description of these generalized gauge fields. In this chapter we introduce some basic notions of \emph{generalized differential cohomology theories}, which we will see constitute a powerful mathematical machinery to describe abelian gauge theories of differential forms, and in particular the theory of Ramond-Ramond fields. We first motivate the use of this formalism in the case of ordinary electromagnetism, following \cite{Freed2000,Freed2006}.
\section{An example: the electromagnetic case}\label{generalizedmaxwell}
In ordinary electromagnetism formulated on the four dimensional Minkowski spacetime $\topsp{M}^{4}=\mathbb{R}_{t}\times{\mathbb{R}^{3}}$, the Maxwell equations are given by
\begin{equation}\label{maxwell}
\begin{array}{c}
{\rm dF=0}\\
{\rm d\star{F}}=j_{e}
\end{array}
\end{equation} 
where ${\rm F\in\Omega^{2}(M;\mathbb{R})}$, and where $j_{e}\in{\rm \Omega^{3}(M;\mathbb{R})}$ is the electric current distribution. Since $\rm M^{4}$ is a contractible space\footnote{More precisely, we only use that every 2-sphere is the boundary of a 3-ball.}, equations (\ref{maxwell}) and Poincar\'e's Lemma imply that there exists a form $\rm A\in\Omega^{1}(M;\mathbb{R})$, called the \emph{vector potential}, such that
\begin{equation}\label{vecpot}
\rm F=dA
\end{equation}
As we have seen in section \ref{generalized}, the total electric charge of the distribution can be identified with the class $[j_{e}|_{\mathbb{R}^{3}}]$ in $\rm H^{3}_{cpt}(\mathbb{R}^{3};\mathbb{R})\simeq{\mathbb{R}}$.\\
We can modify the equations (\ref{maxwell}) by introducing a magnetic current $j_{m}\in{\rm\Omega^{3}(M;\mathbb{R})}$, and allowing the equations
\begin{equation}\label{maxwell2}
\begin{array}{c}
{\rm dF}=j_{m}\\
{\rm d\star{F}}=j_{e}
\end{array}
\end{equation}
The magnetic charge of the system $Q_{m}$ is given by the class $[j_{m}|_{\mathbb{R}^{3}}]$ in $\rm H^{3}_{cpt}(\mathbb{R}^{3};\mathbb{R})$.\\
The equations (\ref{maxwell2}) have changed the ``global properties'' of $\rm F$: indeed, it is no longer a closed form, and equation (\ref{vecpot}) no longer holds. If we denote with $\topsp{W}_{e}$ and $\topsp{W}_{m}$ the support of the forms $j_{e}$ and $j_{m}$ respectively, and supposing $\topsp{W}_{e}\cap\topsp{W}_{m}=\O$, we can consider the following equations defined on $\topsp{M}^{4}-\topsp{W}_{m}$
\begin{equation}\label{maxwell3}
\begin{array}{c}
{\rm dF}=0\\
{\rm d\star{F}}=j_{e}
\end{array}
\end{equation}
Notice at this point that equations (\ref{maxwell3}) do not imply equation (\ref{vecpot}), since the space $\topsp{M}^{4}-\topsp{W}_{m}$ is in general not contractible. However, for any contractible open set $\mathcal{U}_{\alpha}\subset{\topsp{M}^{4}-\topsp{W}_{m}}$ we have
\begin{displaymath}
\rm F|_{\mathcal{U}_{\alpha}}=dA_{\alpha}
\end{displaymath}
and on overlaps $\mathcal{U}_{\alpha}\cap\mathcal{U}_{\beta}$ we have
\begin{displaymath}
{\rm A_{\alpha}-A_{\beta}=dg_{\alpha\beta}}\quad{g_{\alpha\beta}\in{C^{\infty}(\topsp{M};\mathbb{R})}}
\end{displaymath}
Finally, on triple overlaps $\mathcal{U}_{\alpha}\cap\mathcal{U}_{\beta}\cap\mathcal{U}_{\gamma}$
\begin{displaymath}
g_{\alpha\beta}+g_{\beta\gamma}+g_{\gamma\alpha}=c_{\alpha\beta\gamma}
\end{displaymath} 
where $c_{\alpha\beta\gamma}$ is a real constant over the entire triple overlap. The fact that in the presence of magnetic charges the vector potential is not globally defined requires that the coupling term 
\begin{displaymath}
\int_{\topsp{M}^{4}}{\rm A}\wedge{j_{e}}=\int_{\topsp{W}_{e}}{\rm A}
\end{displaymath}
be carefully defined, as $\topsp{W}_{e}$ could intersect more patches $\mathcal{U}_{\alpha}$. However, one can show that the classical action evaluated over a given worldline $\topsp{W}_{e}$ is ambiguous up to a constant $c_{\alpha\beta\gamma}$ \cite{Alvarez}, and this does not affect the classical equations of motion. Indeed, equations (\ref{maxwell3}) only depend on the fieldstrength $\rm F$, which is a globally defined two form. At the quantum level, instead, this classical ambiguity leads to inconsistencies, unless some restrictions are imposed on the collection of all $\{c_{\alpha\beta\gamma}\}$. For example, in the path integral quantization, an ambiguous phase factor of $\exp\{i\:c_{\alpha\beta\gamma}\}$ is potentially present at each non-empty triple intersection of patches: the only way to avoid this ambiguity is to require that $c_{\alpha\beta\gamma}=2\pi\:n_{\alpha\beta\gamma}$, where $n_{\alpha\beta\gamma}$ are \emph{integer} numbers. Moreover, the constants $c_{\alpha\beta\gamma}$ can be directly related to the total magnetic flux, hence the total magnetic charge, implying that the two form $\rm F$ restricted to $\topsp{M}^{4}-\topsp{W}_{m}$ has \emph{integral periods}. This is the so called \emph{Dirac quantization condition}, and was proposed by Dirac in \cite{dirac}, albeit in a different way. Mathematically, the Dirac quantization condition states that the class $[\rm F]$ in ${\rm H}^{2}(\topsp{M}^{4}-\topsp{W}_{m};\mathbb{R})$ lies in the image of the map
\begin{displaymath}
{\rm H}^{2}(\topsp{M}^{4}-\topsp{W}_{m};\mathbb{Z})\to {\rm H}^{2}(\topsp{M}^{4}-\topsp{W}_{m};\mathbb{R})
\end{displaymath}
induced in cohomology by the inclusion $\mathbb{Z}\hookrightarrow\mathbb{R}$.\\
Notice at this point that the above argument applies even to the case of Maxwell equations on a spacetime $\topsp{M}=\mathbb{R}_{t}\times{\topsp{N}}$ in then absence of any magnetic current, provided that ${\rm H}^{2}(\topsp{N};\mathbb{R})\neq{0}$. Also in this case, the Dirac quantization condition requires that the fieldstrength $\rm F$ has integer periods. Differential geometry provides a beautiful solution to the necessity of having a local vector potential $\rm A$ in the quantum theory, without forgetting the obstructions that prevent its global existence, and including the Dirac quantization condition. Indeed, one regards the fieldstrength $\rm F$ as the curvature of a connection $\rm A$ defined on a principal U(1)-bundle $\pi:\mathcal{L}\to\topsp{M}$ with first Chern class $c_{1}(\mathcal{L})=[c_{\alpha\beta\gamma}]\in\topsp{H}^{2}(\topsp{M};\mathbb{Z})$. The relevant space\footnote{It is actually a \emph{groupoid}, where the morphisms are given by connection preserving gauge transformations.} of fields for the quantum theory is then given by the space of all principal U(1)-bundles with connection over $\topsp{M}$. Notice that this space extends the space of classical solutions of Maxwell equations by the flat connections, which may contribute nontrivially to the quantum theory.\\ 

As we have seen through the argument above, the Dirac quantization of charges is a required condition in order to have a well defined quantum theory. We have also seen that the theory of principal bundles can be used to geometrically encode this condition. However, this framework cannot be applied to the case of generalized electromagnetic theories introduced in section \ref{generalized}: indeed, since the fieldstrength $\rm F$ is given by a $p$-form, it cannot be realized as the curvature of some connection. One can then resort to a local description of these gauge fields, as done for the B-field in section \ref{Bfield}. This approach has several disadvantages, though: it is usually difficult in this framework to determine the space of gauge equivalent field configurations, over which the path integral should be performed, and in particular it is difficult to introduce a coupling of these fields with their sources, as one needs a proper notion of pullback, integration, etc. In the following section, we will introduce the proper mathematical formalism to treat these fields, whose recent development has been greatly motivated by the very problems mentioned above.    
\section{Differential Cohomology}
In the past 30 years, Differential Cohomology has appeared in the mathematical literature as the theory of differential characters, Deligne cocycles, sparks, and more recently differential functions \cite{cheeger,Brylinski,sparks,Hopkins2005}. From the mathematical point of view, it has provided a refinement of the theory of characteristic classes and characteristic forms, which, in the appropriate contexts, gives rise to obstruction to conformal immersion of Riemannian manifolds in euclidean spaces; more recently, in combination with differential K-theory, it has been used to construct an index theorem for flat bundles. As mentioned above, from the physical point of view it constitutes a powerful formalism to describe gauge theories of $p$-forms whose Dirac quantization condition is dictated by integer cohomology. In the following we will focus on two particular descriptions of Differential Cohomology, which use Cheeger-Simons characters and Deligne cocycles, respectively. We will illustrate the relation between differential cohomology and electromagnetism with Dirac quantization of charges, and we will give a definition of \emph{generalized abelian gauge theories} in terms of these objects.
\subsection{Cheeger-Simons characters}\label{cheeger}
Let $\topsp{M}$ denote a smooth manifold, and let $\mathcal{C}_{k}\supset\mathcal{Z}_{k}\supset\mathcal{B}_{k}$ denote the groups of smooth singular chains, cycles, and boundaries with $\partial:\mathcal{C}_{k}\to\mathcal{C}_{k-1}$ and $\delta:\mathcal{C}^{k}\to\mathcal{C}^{k+1}$ the usual boundary and coboundary operators, respectively. Let $\Omega^{k}_{\mathbb{Z}}(\topsp{M})\subset\Omega^{k}(\topsp{M})$ denote the lattice of closed $k$-forms with integral periods. Notice that if we denote with $r$ the map induced on cohomology by the inclusion $\mathbb{Z}\hookrightarrow{\mathbb{R}}$, then a $k$-form $\omega$ has integral periods if and only if $[\omega]=r(u)$, for some $u\in{\rm H}^{k}(\topsp{M};\mathbb{Z})$. Given $\omega\in\Omega^{k}(\topsp{M})$, we have the map
\begin{equation}\label{formmap}
\Omega^{k}(\topsp{M})\to\mathcal{C}^{k}(\topsp{M};\mathbb{R}/\mathbb{Z})
\end{equation}
which assigns to $\omega$ the $\mathbb{R}/\mathbb{Z}$-valued cochain defined as
\begin{displaymath}
\tilde{\omega}(\sigma):=\int_{\sigma}\omega\:\text{mod}\:\mathbb{Z},\quad \forall{\sigma\in\mathcal{C}_{k}(\topsp{M})}
\end{displaymath}
As the integral of a generic differential form over the set of all cycles never takes values only in $\mathbb{Z}$, the map (\ref{formmap}) is an injection, and we denote the image of $\Omega^{k}(\topsp{M})$ in $\mathcal{C}^{k}(\topsp{M};\mathbb{R}/\mathbb{Z})$ with $\widetilde{\Omega}^{k}(\topsp{M})$. We have then the following
\begin{definition}
The \emph{n-th Cheeger-Simons group} of a smooth manifold $\topsp{M}$ is the group
\begin{displaymath}
\check{\rm H}^{n}(\topsp{M}):=\left\{f\in{\rm Hom}\left(\mathcal{Z}_{n-1},\mathbb{R}/\mathbb{Z}\right):\:f\circ\partial\:\in\:\widetilde{\Omega}^{n}(\topsp{M})\right\}
\end{displaymath}
We set $\check{\rm H}^{0}(\topsp{M})=\mathbb{Z}$. The elements of $\check{\rm H}^{*}(\topsp{M}):=\oplus_{n}\check{\rm H}^{n}(\topsp{M})$ are called \emph{differential characters}.
\end{definition}
A smooth map $\phi:\topsp{M}\to\topsp{N}$ naturally induces a homomorphism 
\begin{displaymath}
\phi^{*}:\check{\rm H}^{n}(\topsp{N})\to\check{\rm H}^{n}(\topsp{M})
\end{displaymath}
Hence, $\check{\rm H}^{n}$ is a contravariant functor from the category of smooth manifolds to the category of abelian groups.\\
A main result is the following \cite{cheeger}
\begin{proposition}
There exist surjective maps 
\begin{displaymath}
\begin{array}{c}
\check{\rm H}^{n}(\topsp{M})\xrightarrow{\rm F}\Omega^{n}_{\mathbb{Z}}(\topsp{M})\\
\check{\rm H}^{n}(\topsp{M})\xrightarrow{\rm c}{\rm H}^{n}(\topsp{M};\mathbb{Z})
\end{array}
\end{displaymath}
for any smooth manifold $\topsp{M}$ and any $n\in\mathbb{N}$.\\
The map $\rm F$ is called the \emph{fieldstrength map}, and the map ${\rm c}$ is called the \emph{characteristic class map}.
\end{proposition}
\begin{proof}
Let $f\in\check{\rm H}^{n}(\topsp{M})$. Consider an element $\rm T^{'}\in{\rm Hom}\left(\mathcal{Z}_{n-1},\mathbb{R}\right)$ such that 
\begin{displaymath}
\tilde{\rm T}^{'}=f
\end{displaymath}
where the tilde means mod $\mathbb{Z}$. Consider now the following (split) exact sequence
\begin{displaymath}
0\to\mathcal{Z}_{n-1}\xrightarrow{i}\mathcal{C}_{n-1}\xrightarrow{\partial}\mathcal{B}_{n-2}\to{0}
\end{displaymath}
Since ${\rm Hom}(\:\cdot\:,\mathbb{R})$ is an exact functor, the following exact sequence holds
\begin{equation}\label{sequence}
0\to{\rm Hom}(\mathcal{B}_{n-2},\mathbb{R})\xrightarrow{\delta}{\rm Hom}(\mathcal{C}_{n-1},\mathbb{R})\xrightarrow{i^{*}}{\rm Hom}(\mathcal{Z}_{n-1},\mathbb{R})\to{0}
\end{equation}
where the map $i^{*}$ is given by restriction. Hence, there exist $\rm T\in{\rm Hom}(\mathcal{C}_{n-1},\mathbb{R})$ such that ${\rm T|_{\mathcal{Z}_{n-1}}=T^{'}}$, and consequently 
\begin{displaymath}
\tilde{\rm T}|_{\mathcal{Z}_{n-1}}=f
\end{displaymath}
Since by construction $\delta{\rm T}:={\rm T}\circ\partial$, and $\widetilde{\delta{\rm T}}=\delta{\tilde{\rm T}}$, we have that 
\begin{displaymath}
\delta{\tilde{\rm T}}=f\circ\partial
\end{displaymath}
By assumption $\delta{\tilde{\rm T}}$ is an element of $\widetilde{\Omega}^{n}(\topsp{M})$. Any such element can be written as $\omega-c$, where $\omega$ is a $n$-form regarded as a real cochain by integration, and $c\:\in\mathcal{C}^{n}(\topsp{M},\mathbb{Z})$. Since $\delta{\rm T}=\omega-c$, we have   
\begin{equation}\label{classes}
0=\delta^{2}{\rm T}={\delta}\omega-\delta{c} ={\rm d}\omega-\delta{c}
\end{equation}
where $\rm d$ denotes the deRham differential. Since the map (\ref{formmap}) is an injection, equation (\ref{classes}) implies that ${\rm d}\omega=0$ and $\delta{c}=0$.\: Finally, since $\delta{\rm T}=\omega-c$ we have that $[\omega]=r([c])$, with $[c]\in{\rm H}^{n}(\topsp{M};\mathbb{Z})$, which implies that $\omega\in\Omega^{n}_{\mathbb{Z}}(\topsp{M})$. The elements $\omega$ and $[c]$ do not depend on the \emph{lift} $\rm T$. Indeed, let $\rm T^{'}$ be another real cochain such that $\tilde{\rm T}^{'}|_{\mathcal{Z}_{n-1}}=f$. Then $\widetilde{\rm T-T^{'}}|_{\mathcal{Z}_{n-1}}=0$, and the sequence (\ref{sequence}) implies that 
\begin{equation}\label{lift}
{\rm T^{'}=T+\delta{\alpha}+\beta}
\end{equation}
with $\alpha\in\mathcal{C}^{n-2}(\topsp{M};\mathbb{R})$ and $\beta\in\mathcal{C}^{n-1}(\topsp{M};\mathbb{Z})$. By the argument above there exist $\omega^{'}$ and $c^{'}$ such that $\delta{\rm T^{'}}=\omega^{'}-c^{'}$. Hence we have
\begin{displaymath}
\omega^{'}-c^{'}=\delta{\rm T^{'}}=\delta{\rm T}+\delta{\beta}=\omega-c+\delta{\beta}
\end{displaymath}
which implies
\begin{displaymath}
\omega^{'}-\omega=c^{'}-c+\delta{\beta}\in\Cl({\rm V};q)
\end{displaymath}
By the same reason as before, we have $\omega^{'}=\omega$ and $[c^{'}]=[c]$.\\
We define $\rm F(f)=\omega$ and ${\rm c}(f)=[c]$, where $\omega$ and $c$ are obtained as above. Moreover, we will refer to the form $\omega$ as the $\emph{fieldstrength}$ and to the class $[\rm c]$ as \emph{the characteristic class}; the use of this terminology will be clear later.\\

To prove that the maps $\rm F$ and $\rm c$ are surjective, notice that any $\omega\in\Omega^{n}_{\mathbb{Z}}(\topsp{M})$ determines a $u\in{\rm H}^{n}(\topsp{M};\mathbb{Z})$ such that $[\omega]=r(u)$, and conversely for any $u$ we can find such a $\omega$. Let $[c]=u$. Then $\omega-c$ is exact as a real cochain, and there exists $\rm T$ such that $\delta{\rm T}=\omega-c$. Then $\tilde{\rm T}|_{\mathcal{Z}_{n-1}}$ defines a homomorphism $f:\mathcal{Z}_{n-1}\to\mathbb{R}/\mathbb{Z}$ with the property that $f\circ\partial\:\in \widetilde{\Omega}^{n}(\topsp{M})$.
\end{proof}
We will now investigate the kernels of the fieldstrength and characteristic class maps.\\
Let $f\in\check{\rm H}^{n}(\topsp{M})$ satisfy ${\rm F}(f)=0$. Then the real cochain $\rm T$ satisfies $\delta{\rm T}=-c$. Hence $\delta{\tilde{\rm T}}=0$, and thus $\tilde{\rm T}$ is an $\mathbb{R}/\mathbb{Z}$-valued cocycle. If we consider another lift $\rm T^{'}$, equation (\ref{lift}) implies that $\tilde{\rm T}^{'}=\tilde{\rm T}+\delta\tilde{\alpha}$, and hence $[\rm T^{'}]=[\rm T] \in{{\rm H}^{n-1}(\topsp{M};\mathbb{R}/\mathbb{Z}})$. So an element $f$ with ${\rm F}(f)=0$ determines an $\mathbb{R}/\mathbb{Z}$ cohomology class. Conversely, given an $\mathbb{R}/\mathbb{Z}$ class represented by a cocyle $s$, we have that $s|_{\mathcal{Z}_{n-1}}$ defines a differential character $f$, and such a definition is independent of the choice of the representing cocycle $s$.\\
Finally, let $f\in\check{\rm H}^{n}(\topsp{M})$ satisfy ${\rm c}(f)=0$. Then $\delta{\rm T}=\omega-c$, with $c=\delta{e}$, for some $e\in\mathcal{C}^{n-1}(\topsp{M};\mathbb{Z})$. Hence, $\delta({\rm T}-e)=\omega$. Since $[\omega]=r([c])$, by the deRham theorem we have that $\omega={\rm d}\theta$, for some $\theta\in\Omega^{n-1}(\topsp{M})$. Since $\delta({\rm T}-e-\theta)=0$, we have that ${\rm T}-e-\theta=z$, for some cocycle $z\in\mathcal{Z}^{n-1}(\topsp{M};\mathbb{R})$. Again by the deRham theorem, there exists a closed form $\phi\in\Omega^{n-1}(\topsp{M})$ such that $\phi|_{\mathcal{Z}_{n-1}}=z|_{\mathcal{Z}_{n-1}}$, which implies that 
\begin{displaymath}
f=\tilde{{\rm T}}|_{\mathcal{Z}_{n-1}}=\widetilde{\theta+\phi+e}=\widetilde{\theta+\phi}
\end{displaymath}
Hence $f$ is in the image of the map $\omega\to\tilde{\omega}|_{\mathcal{Z}_{n-1}}$, with $\omega\in\Omega^{n-1}(\topsp{M})$, whose kernel is given by $\Omega^{n-1}_{\mathbb{Z}}(\topsp{M})$.\\
We have then proved the following \cite{cheeger}
\begin{proposition}
There are natural exact sequences
\begin{displaymath}
\begin{array}{c}
0\to{\rm H}^{n-1}(\topsp{M};\mathbb{R}/\mathbb{Z})\to\check{\rm H}^{n}(\topsp{M})\xrightarrow{{\rm F}}\Omega_{\mathbb{Z}}^{n}(\topsp{M})\to{0}\\
 0\to{\Omega^{k-1}(\topsp{M})/\Omega^{k-1}_{\mathbb{Z}}(\topsp{M})}\to\check{\rm H}^{n}(\topsp{M})\xrightarrow{{\rm c}}{\rm H}^{n}(\topsp{M};\mathbb{Z})\to{0}
\end{array}
\end{displaymath}
\end{proposition}
Another useful exact sequence can be obtained as follows. Consider the group
\begin{displaymath}
{\rm A}^{n}(\topsp{M}):=\left\{(\omega,u)\in\Omega_{\mathbb{Z}}^{n}(\topsp{M})\times{\rm H}^{n}(\topsp{M};\mathbb{Z}):[\omega]=r(u)\right\}
\end{displaymath}
We have then a surjective map
\begin{displaymath}
\check{\rm H}^{n}(\topsp{M})\xrightarrow{({\rm F,c})}{\rm A}^{n}(\topsp{M})
\end{displaymath}
The kernel of this map is given by elements $f\in\check{\rm H}^{n}(\topsp{M})$ such that ${\rm F}(f)=0$ and ${\rm c}(f)=0$. By the same argument as before, we have that $\delta({\rm T}-e)=0$, hence the element $f$ determines a class in ${\rm H}^{n-1}(\topsp{M};\mathbb{R})$ represented by the closed form $\phi$ that satisfies ${\rm T}|_{\mathcal{Z}_{n-1}}=\phi+e$. As $\phi$ is a closed form, $\tilde{\rm T}|_{\mathcal{Z}_{n-1}}$ is determined only by the class $[\phi]\in{\rm H}^{n-1}(\topsp{M};\mathbb{R})$. Conversely, any class $[\phi]$ determines a differential character $f=\tilde{\phi}|_{\mathcal{Z}_{n-1}}$, and the kernel of such an assignment is given precisely by $r({\rm H}^{n-1}(\topsp{M};\mathbb{Z}))$. We have then showed that the following exact sequence holds
\begin{equation}
0\to{\rm H}^{n-1}(\topsp{M};\mathbb{R})/r({\rm H}^{n-1}(\topsp{M};\mathbb{Z}))\to\check{\rm H}^{n}(\topsp{M})\xrightarrow{({\rm F,c})}{\rm A}^{n}(\topsp{M})\to{0}
\end{equation}
The above sequences are very important, as they are usually the only computational technique available since the groups $\check{\rm H}^{n}(\topsp{M})$ are usually infinite dimensional.\\

Notice that the group ${\rm A}^{*}(\topsp{M}):=\oplus_{n}{\rm A}^{n}(\topsp{M})$ carries an obvious ring structure given by
\begin{displaymath}
(\omega,u)\cdot(\phi,v):=(\omega\wedge\phi,u\cup{v})
\end{displaymath}
We expect then that the group $\check{\rm H}^{*}(\topsp{M})$ carries an analogous product. This is indeed the case. However, the ring product of differential characters is more subtle to define: this is due to the fact that given two forms $\omega_{1}$ and $\omega_{2}$ of degrees $l$ and $p$, respectively, in general we have 
\begin{displaymath}
\omega_{1}\wedge\omega_{2}\neq\omega_{1}\cup\omega_{2}
\end{displaymath}
where on the right hand side $\omega_{1}$ and $\omega_{2}$ are realized as real cochains. However, one has that
\begin{displaymath}
\delta{E(\omega_{1},\omega_{2})}=\omega_{1}\wedge\omega_{2}-\omega_{1}\cup\omega_{2}
\end{displaymath}
for some $E(\omega_{1},\omega_{2})\in\mathcal{C}^{l+p-1}(\topsp{M};\mathbb{R})$. The technical difficulty consists exactly in constructing the cochain $E(\omega_{1},\omega_{2})$ in a canonical way, and we refer the reader to \cite{cheeger} for details. A definition of the ring product can then be given in the following way. Let $f\in\check{\rm H}^{l}(\topsp{M})$ and $g\in\check{\rm H}^{p}(\topsp{M})$, and choose the lifts ${\rm T}_{f}\in\mathcal{C}^{l-1}(\topsp{M};\mathbb{R})$ and ${\rm T}_{g}\in\mathcal{C}^{p-1}(\topsp{M};\mathbb{R})$ for $f$ and $g$ respectively. Then we can define \cite{cheeger}
\begin{displaymath}
f\star{g}:=\widetilde{{\rm T}_{f}\cup\omega_{f}}-(-1)^{l-1}\widetilde{\omega_{f}\cup{\rm T}_{g}}-\widetilde{{\rm T}_{f}\cup\delta{\rm T}_{g}}+E(\omega_{f},\omega_{g})|_{\mathcal{Z}_{l+p-1}}
\end{displaymath}
and it can be shown that $f\star{g}$ is independent of the choice of ${\rm T}_{f}$ and ${\rm T}_{g}$. Moreover, the product $\star$ is associative, graded commutative, and it is such that the fieldstrength and characteristic class maps are ring homomorphisms. Even if the product $\star$ is fairly complicated in general, there are special cases in which it greatly simplifies. Indeed, if $g\in{\rm ker\:F}\subset\check{\rm H}^{p}(\topsp{M})$, then $f\star{g}$ is the image of $(-1)^{l}{\rm c}(f)\cup{g}$, for $f\in\check{\rm H}^{l}(\topsp{M})$, while if $g\in{\rm ker\:c}\subset\check{\rm H}^{p}(\topsp{M})$, $f\star{g}$ is the image of $(-1)^{l}{\rm F}(f)\wedge{g}$.\\
 
In the following we will give some examples.
\begin{example}
Let $\topsp{M}={\rm S}^{1}$ and consider the group $\check{\rm H}^{1}({\rm S}^{1})$. Since $\mathcal{Z}_{0}(\topsp{S}^{1})$ is freely generated by the points of ${\topsp{S}^{1}}$, we have that $\check{\rm H}^{1}({\rm S}^{1})\simeq{C^{\infty}(\topsp{S}^{1},\topsp{S}^{1})}$, where we have realized $\mathbb{R}/\mathbb{Z}$ as $\topsp{S}^{1}$. Consider then a smooth function
\begin{displaymath}
f:\topsp{S}^{1}\to\topsp{S}^{1}
\end{displaymath}
Any such function $f$ can be expressed as
\begin{displaymath}
f(p)={\rm e}^{2\pi{i}\Theta(f(p))}
\end{displaymath}
 where $\Theta=\theta\circ{f}$, with $\theta$ the usual local angular coordinate on the circle. Then a lift $\rm T$ of $f$ is given by
\begin{displaymath}
{\rm T}=\dfrac{1}{2\pi i}{\rm log}\:f=\Theta(p)+\eta(p)
\end{displaymath}
where $\eta$ is a locally constant $\mathbb{Z}$-valued function. By assumption, we know that on any curve $\gamma$ on $S^{1}$ 
\begin{displaymath}
\delta\tilde{\rm T}(\gamma)={\rm e}^{2\pi{i}(\Theta(b)-\Theta(a))}
\end{displaymath}
where $a$ and $b$ are the endpoints of the curve. The right hand side of the above equation can be written as
\begin{displaymath}
{\rm exp}\left\{2\pi{i}\int_{\gamma}{\rm d}\Theta\right\}
\end{displaymath}
Notice that ${\rm d}\Theta$ is not an exact 1-form, but it is integral, as ${\rm d}\Theta=f^{*}{\rm d}\theta$. Hence we have shown that 
\begin{displaymath}
\dd\left(\dfrac{1}{2\pi{i}}{\rm log}f\right)
\end{displaymath}
 is the fieldstrength of $f$. Since the group ${\rm H}^{1}(\topsp{S}^{1};\mathbb{Z})$ contains no torsion, the fieldstrength determines the characteristic class: hence, we have 
\begin{displaymath}
[c]=f^{*}[\rm d{\theta}]
\end{displaymath}
The same argument can be applied to the group $\check{\rm H}^{1}({\rm M})$ for any smooth manifold $\topsp{M}$.
\end{example}
\begin{example}
For any smooth manifold $\topsp{M}$ with $n={\rm dim}\:\topsp{M}$ we have
\begin{displaymath}
\begin{array}{c}
\check{\rm H}^{n+1}({\rm M})\simeq{\rm H}^{n}(\topsp{M};\mathbb{R}/\mathbb{Z})\\
\check{\rm H}^{k}({\rm M})=0,\:k>n+1
\end{array}
\end{displaymath}
This can be shown by using the exact sequences introduced above.
\end{example} 
\begin{example}
Let $\topsp{M}={\rm pt}$. Then
\begin{displaymath}
\check{\rm H}^{n}({\rm M})=\left\{\begin{array}{l}
\mathbb{Z},\:n=0\\
\mathbb{R}/\mathbb{Z},\:n=1\\
0,\:{\rm otherwise}
\end{array}\right .
\end{displaymath}
\end{example}
\begin{example}
Finally, we present an example coming from differential geometry, which is in a certain sense the ``canonical'' one.\\
Consider a complex line bundle $\mathcal{L}\to\topsp{M}$ with connection $\nabla$, and denote with ${\rm F}_{\nabla}$ its curvature form. Since $\dfrac{1}{2\pi}{\rm F}_{\nabla}$ represents the real first Chern class of $\mathcal{L}$, we have that ${\rm F}_{\nabla}\in\Omega^{2}_{\mathbb{Z}}(\topsp{M})$. Let $\gamma$ be a closed curve, and define 
\begin{displaymath}
f(\gamma):={\rm hol}_{\nabla}(\gamma)\in\mathbb{R}/\mathbb{Z}
\end{displaymath}
where ${\rm hol}_{\nabla}$ denotes the holonomy of the connection $\nabla$. We can extend the homomorphism $f$ to the whole of $\mathcal{Z}_{1}(\topsp{M})$ as follows. Let $z\in\mathcal{Z}_{1}(\topsp{M})$.  Represent $z$ as $z=\sum_{i}n_{i}\gamma_{i}+\partial{y}$, where $\gamma_{i}$ is a closed curve, and $y\in\mathcal{C}_{2}(\topsp{M})$. Then we can define
\begin{displaymath}
f(z):=\prod_{i}f(\gamma_{i})^{n_{i}}+\dfrac{1}{2\pi}\tilde{\rm F}_{\nabla}(y)
\end{displaymath}
The homomorphism $f$ above is well defined, and independent of the presentation of $z$, since ${\rm F}_{\nabla}$ has integral periods. Since $f\circ\partial=\dfrac{1}{2\pi}{\rm F}_{\nabla}$, we have that $f\in\check{\rm H}^{2}({\rm M})$. Moreover, one can check that 
\begin{displaymath}
\begin{array}{c}
{\rm F}(f)=\dfrac{1}{2\pi}\tilde{\rm F}_{\nabla}\\
{\rm c}(f)=c_{1}(\mathcal{L})
\end{array}
\end{displaymath}
Notice that the character $f$ contains more information than ${\rm F}_{\nabla}$ and $c_{1}(\mathcal{L})$ together, since both may vanish when $f$ does not, e.g. when $\topsp{M}=\topsp{S}^{1}$.
\end{example}
It is immediate to realize in the above example that if we perform a connection preserving gauge transformation on $\mathcal{L}$, the homomorphism $f$ does not change, hence it is a gauge invariant quantity. Conversely, given an element $f\in\check{\rm H}^{2}({\rm M})$ we can construct up to isomorphism a line bundle $\mathcal{L}$ classified by ${\rm c}(f)$, equipped with a connection $\nabla$ defined by requiring that
\begin{displaymath}
{\rm hol}_{\nabla}(\gamma):=f(\gamma)
\end{displaymath}
for any closed curve $\gamma$. Hence, for any smooth manifold $\topsp{M}$ the group $\check{\rm H}^{2}({\rm M})$ is equivalent to the space of all gauge equivalent complex line bundles with connection over $\topsp{M}$: the group structure on the latter space is induced by tensor product of line bundles. We have then realized that the group $\check{\rm H}^{2}({\rm M})$ is the set of orbits of the group of gauge transformations acting on the space of connections on \emph{all} line bundles over $\rm M$: noncanonically, this space can be expressed as
\begin{displaymath}
\bigcup_{c_{1}\in{\rm H}^{2}(\topsp{M};\mathbb{Z})}\mathcal{A}(\mathcal{L}_{c_{1}})/\mathcal{G}
\end{displaymath} 
where $\mathcal{A}(\mathcal{L}_{c_{1}})$ is the affine space of connections on the line bundle $\mathcal{L}_{c_{1}}$ classified by the class $c_{1}$, and $\mathcal{G}$ is the group of gauge transformations. As seen in section \ref{generalizedmaxwell}, this is precisely the space of all possible gauge inequivalent solutions to Maxwell equations on $\topsp{M}$, taking into account the Dirac quantization condition. Realizing this space as the group $\check{\rm H}^{2}(\topsp{M})$ naturally suggests the following 
\begin{definition}\label{gen}
A \emph{generalized abelian gauge theory} on a smooth manifold $\topsp{M}$ is a field theory whose space of gauge inequivalent configurations is given by $\check{\rm H}^{n}(\topsp{M})$, for some $n\in\mathbb{N}$.
\end{definition}
The generalized electromagnetism discussed in section \ref{genmax}, once provided with a suitable Dirac quantization condition, is a generalized abelian gauge theory. Because of the relation between the group $\check{\rm H}^{n}(\topsp{M})$ and gauge theories of $n$-forms, one usually refers to ${\rm H}^{n-1}(\topsp{M};\mathbb{R}/\mathbb{Z})$ as the group of \emph{flat fields}, while $\Omega^{n-1}(\topsp{M})/\Omega^{n-1}_{\mathbb{Z}}(\topsp{M})$ is called the group of \emph{topologically trivial fields}. Consequently, we can see that the group ${\rm H}^{n-1}(\topsp{M};\mathbb{R})/r({\rm H}^{n-1}(\topsp{M};\mathbb{Z}))$ classifies the (gauge equivalence classes of) \emph{flat and topologically trivial fields}. Notice that if we classify inequivalent field configurations only by curvature and characteristic class, we are not taking into account the effect of flat and topologically trivial fields, which may be nonvanishing, according to the topology of $\topsp{M}$.\\

Since the electromagnetic field can be described with complex line bundles and connections thereon, this suggets that elements of $\check{\rm H}^{n}(\topsp{M})$ should represent isomorphism classes of ``higher'' line bundles with ``higher'' connections. This is indeed the case: for example, the group $\check{\rm H}^{3}(\topsp{M})$ is given by equivalence classes of bundle gerbes with connections. However, to be able to talk about the gauge fields and higher line bundles whose equivalent classes are the elements in $\check{\rm H}^{n}(\topsp{M})$, we need a \emph{model} representing such a group. Recall indeed, that the group $\check{\rm H}^{2}(\topsp{M})$ can be obtained as the set of isomorphism classes for the groupoid of connections on $\topsp{M}$. To give such a model is the aim of the next section.  
\subsection{Deligne cohomology}
In this section we will describe a cochain model for the Cheeger-Simons groups. This is done by introducing \emph{Deligne cohomology}, which is defined via the cohomology of a certain complex. The Deligne cohomology groups are isomorphic to the Cheeger-Simons groups: however, we will not prove this result, and will focus instead on a detailed construction of the cochain model, showing how it can naturally describe the gauge fields in a generalized abelian gauge theory in topologically nontrivial backgrounds. In the following we will describe the \emph{smooth} Deligne cohomology\footnote{Deligne cohomology originated in the context of algebraic geometry.}.\\
In a nutshell, the differential complex used to define Deligne cohomology is a ``modification'' of the \v{C}ech-de Rham complex: we will recall some basics about the \v{C}ech-de Rham complex in order to clearly show the nature of such a modification.\\
Let $\topsp{M}$ be a paracompact smooth manifold with $n={\text{dim}\:\topsp{M}}$, and consider a good cover $\mathfrak{U}=\left\{U_{\alpha}\right\}_{\alpha\in{}\rm I}$, where the index set $\rm I$ may be infinite. Denote
\begin{displaymath}
U_{\alpha_{0}\alpha_{1}\cdots\alpha_{p}}:=U_{\alpha_{0}}\cap{U_{\alpha_{1}}}\cap\dots\cap{U_{\alpha_{p}}}
\end{displaymath}
Let $\Omega^{p}(\topsp{M})$ be the set of real valued $p$-forms, and denote with $\Omega^{0}(\topsp{M})$ the space of smooth real functions $C^{\infty}\left(\topsp{M};\mathbb{R}\right)$. Since each $U_{\alpha_{0}\alpha_{1}\cdots\alpha_{p}}$ is contractible, Poincar\'e's Lemma implies that for any sequence $\alpha_{0},\alpha_{1},\cdots,\alpha_{p}$ we have an exact sequence 
\begin{displaymath}
0\to{C^{\infty}_{lc}(U_{\alpha_{0}\alpha_{1}\cdots\alpha_{p}};\mathbb{R})}\to\Omega^{0}(U_{\alpha_{0}\alpha_{1}\cdots\alpha_{p}})\xrightarrow{d}\Omega^{1}(U_{\alpha_{0}\alpha_{1}\cdots\alpha_{p}})\xrightarrow{d}\dots\xrightarrow{d}\Omega^{n}(U_{\alpha_{0}\alpha_{1}\cdots\alpha_{p}})\to{0}
\end{displaymath}
where $C^{\infty}_{lc}(\topsp{M};\mathbb{R})$ denotes the space of smooth locally constant real valued functions.\\
Define the group of $p$-cochains of the cover $\mathfrak{U}$ with values in the $q$-forms as 
\begin{displaymath}
\check{C}^{p}(\mathfrak{U};\Omega^{q}):=\prod_{\alpha_{0}<\alpha_{1}<\dots<\alpha_{p}}\Omega^{q}(U_{\alpha_{0}\alpha_{1}\cdots\alpha_{p}})
\end{displaymath} 
Notice that the inclusion 
\begin{displaymath}
U_{\alpha_{0}\cdots\alpha_{p}}\hookrightarrow{U_{\alpha_{0}\cdots\hat{\alpha_{i}}\cdots\alpha_{p}}}
\end{displaymath}
induces a map
\begin{displaymath}
\phi_{\alpha_{i}}:\Omega^{q}(U_{\alpha_{0}\cdots\hat{\alpha_{i}}\cdots\alpha_{p}})\to\Omega^{q}(U_{\alpha_{0}\cdots\alpha_{p}})
\end{displaymath}
which is given by restriction. We can then define the homomorphism 
\begin{displaymath}
\delta:\check{C}^{p}(\mathfrak{U};\Omega^{q})\to\check{C}^{p+1}(\mathfrak{U};\Omega^{q})
\end{displaymath}
which on $\omega\in\check{C}^{p}(\mathfrak{U};\Omega^{q})$ is defined as
\begin{displaymath}
\delta\omega:=\sum_{i=0}^{p+1}(-1)^{i}\phi_{\alpha_{i}}(\omega)
\end{displaymath}
The homomorphism $\delta$ satisfies $\delta^{2}=0$, hence it can be used to the define a cohomology theory. However, we have the following Generalized Mayer-Vietoris exact sequence \cite{BottTu}
\begin{displaymath}
0\to\Omega^{q}(\topsp{M})\to\prod\Omega^{q}(U_{\alpha_{0}\alpha_{1}})\xrightarrow{\delta}\prod\Omega^{q}(U_{\alpha_{0}\alpha_{1}\alpha_{3}})\xrightarrow{\delta}\cdots
\end{displaymath}
for any $q\geq{0}$, which tells us that the $\delta$-cohomology of the complex $\check{C}^{*}(\mathfrak{U};\Omega^{q})$ vanishes identically. Notice that for $q=0$ this is related to the fact that ${\rm H}^{p}(\topsp{M};\underline{\mathbb{R}})$, the $p$-th cohomology group of $\topsp{M}$ valued in the sheaf of smooth real valued functions, is 0 for all $p\geq{0}$, since $\underline{\mathbb{R}}$ is a fine sheaf\footnote{A fine sheaf is loosely speaking a sheaf with a ``partition of unity''.}.\\ 
Since the de Rham differential
\begin{displaymath}
{\rm d}:\check{C}^{p}(\topsp{U};\Omega^{q})\to\check{C}^{p}(\mathfrak{U};\Omega^{q+1})
\end{displaymath}
commutes with the coboundary homomorphism $\delta$, we can form the following double complex, called the \emph{\v{C}ech-de Rham} complex
\begin{displaymath}
  \check{C}^{*}(\topsp{M};\Omega^{*}):=\bigoplus_{p,q\geq{0}}\check{C}^{p}(\mathfrak{U};\Omega^{q})
\end{displaymath}
with coboundary operator
\begin{displaymath}
D=\delta+(-1)^{deg(\cdot)}{\rm d}
\end{displaymath}
where for $\omega\in\check{C}^{p}(\mathfrak{U};\Omega^{q})$, $deg(\omega):=p$.\\
As usual for differential double complexes, we can consider the diagonal subcomplex 
\begin{displaymath}
K^{n}:=\bigoplus_{p+q=n}\check{C}^{p}(\mathfrak{U};\Omega^{q})
\end{displaymath}
By ``tic-tac-toeing'', i.e by chasing diagrams, one can prove that
\begin{displaymath}
{\rm H}_{\rm dR}^{*}(\topsp{M})\simeq{\rm H}^{*}(K^{*};D)\simeq\check{\rm H}_{\rm Ch}^{*}(\topsp{M};\mathbb{R})
\end{displaymath}
where $\check{\rm H}_{\rm Ch}^{*}(\topsp{M};\mathbb{R})$ denote \v{C}ech cohomology. The above isomorphism can be seen as a consequence of the fact that \emph{all} the ``rows'' and ``columns'' of the \v{C}ech-de Rham complex are exact\footnote{Indeed, ${\rm H}_{\rm dR}^{*}(\topsp{M})$ and $\check{\rm H}_{\rm Ch}^{*}(\topsp{M};\mathbb{R})$ are the cohomology groups of the augmented column and row, respectively.}.\\
The basic idea behind Deligne cohomology consists in constructing a complex analog to the \v{C}ech-de Rham one, but based on the following exact sequence
\begin{displaymath}
\Omega_{\rm U}^{0}(U_{\alpha_{0}\alpha_{1}\cdots\alpha_{p}})\xrightarrow{\rm d\:log}\Omega^{1}(U_{\alpha_{0}\alpha_{1}\cdots\alpha_{p}})\xrightarrow{\rm d}\dots\xrightarrow{\rm d}\Omega^{l}(U_{\alpha_{0}\alpha_{1}\cdots\alpha_{p}})
\end{displaymath} 
where $\Omega_{\rm U}^{0}(V):=C^{\infty}(V;{\rm U(1)})$ for any open set $V$, and $0\leq{l}\leq{n}$. In other words, we truncate the complex of forms on the right for some $l$, and we substitute in the 0 degree real valued functions with circle valued ones. We can then define a modified \v{C}ech-de Rham complex as
\begin{displaymath}
 \check{C}^{*}[l](\topsp{M};\Omega^{*}):=\bigoplus_{p\geq{0},0\leq{q}\leq{l}}\check{C}^{p}(\mathfrak{U};\Omega^{q})
\end{displaymath} 
and the associated single complex
\begin{displaymath}
K[l]^{n}:=\bigoplus_{p+q=n}\check{C}^{p}(\mathfrak{U};\Omega^{q})
\end{displaymath}
The \v{C}ech-de Rham coboundary $D$ is unchanged, and we will refer to an element $\omega\in K[l]^{*}$ with $D\omega=0$ as a \emph{Deligne cochain}. For a given $l$, the Deligne cohomology is given by
\begin{displaymath}
\mathbb{H}^{*}(\topsp{M};\mathcal{D}^{l}):={\rm H}^{*}(K[l]^{*};D)
\end{displaymath}
As for \v{C}ech cohomology, the result does not depend on the particular choice of good cover.\\ 
Notice that it is no longer true that all the rows of the double complex $\check{C}^{*}[l](\topsp{M};\Omega^{*})$ are exact. Indeed, for $q=0$ the sequence
\begin{displaymath}
\prod\Omega^{0}_{U}(U_{\alpha_{0}})\xrightarrow{\delta}\prod\Omega^{0}_{U}(U_{\alpha_{0}\alpha_{1}})\xrightarrow{\delta}\prod\Omega^{0}_{U}(U_{\alpha_{0}\alpha_{1}\alpha_{2}})\xrightarrow{\delta}\cdots
\end{displaymath}
is in general not exact, and the $\delta$-cohomology is given by $\check{\rm H}_{\rm Ch}^{*}(\topsp{M};\underline{\rm U(1)})$, the sheaf cohomology with value in the circle functions. Moerover, one can prove that
\begin{displaymath}
\check{\rm H}_{\rm Ch}^{*}(\topsp{M};\underline{\rm U(1)})\simeq\check{\rm H}_{\rm Ch}^{*+1}(\mathfrak{U};\mathbb{Z})
\end{displaymath}
for any good cover $\mathfrak{U}$. This innocuous modification implies interesting properties for the groups $\mathbb{H}^{n}(\topsp{M};\mathcal{D}^{l})$. Indeed, consider the case $n<l$. Then a Deligne class is determined by an $n$-tuple
\begin{displaymath}
(\omega_{0}^{n},\omega^{n-1}_{1},\cdots,\omega_{n}^{0}),\quad\omega^{i}_{j}\in\prod\Omega^{i}(U_{\alpha_{0}\alpha_{1}\dots\alpha_{j}})
\end{displaymath}
satisfying the equations
\begin{displaymath}
\begin{array}{c}
{\rm d}\omega^{n}_{0}=0\\
\delta{\omega^{n}_{0}}-{\rm d}\omega^{n-1}_{1}=0\\
\vdots\\
\delta{\omega^{1}_{n-1}}+(-1)^{n}{\rm d}\log\omega^{0}_{n}=0\\
\delta{\omega^{0}_{n}}=0
\end{array}
\end{displaymath}     
Since the first equation implies that $\omega^{n}_{0}={\rm d}{\alpha^{n-1}_{0}}$, for some $\alpha_{0}^{n-1}$, one can ``descend'' through the equations, showing that $\omega^{0}_{n}=f+\delta{\alpha^{0}}$, where $f\in\prod C_{lc}^{\infty}(U_{\alpha_{0}\alpha_{1}\dots\alpha_{n}};{\rm U}(1))$. Moreover, by the last equation $\delta{f}=0$, hence $f$ represents a class in $\check{\rm H}_{Ch}^{n}(\topsp{M};\mathbb{R}/\mathbb{Z})$. Conversely, any class $[f]$ in $\check{\rm H}_{\rm Ch}^{n}(\topsp{M};\mathbb{R}/\mathbb{Z})$ determines a Deligne class up to $D$-exact terms. We have then
\begin{displaymath}
\mathbb{H}^{n}(\topsp{M};\mathcal{D}^{l})\simeq{\check{\rm H}_{\rm Ch}^{n}(\topsp{M};\mathbb{R}/\mathbb{Z})},\:n<l
\end{displaymath}
Consider now the case $n>l$. In this case a Deligne class is represented by an $l$-tuple
\begin{displaymath}
(\omega_{n-l}^{l},\omega^{l-1}_{n-(l-1)},\cdots,\omega_{n}^{0})
\end{displaymath}
satisfying the equations
\begin{displaymath}
\begin{array}{c}
\delta{\omega^{l}_{n-l}}+(-1)^{n-(l-1)}{\rm d}\omega^{l-1}_{n-(l-1}=0\\
\delta\omega^{l-1}_{n-{l-1}}+(-1)^{n-(l-2)}{\rm d}\omega^{l-2}_{n-(l-2)}=0\\
\vdots\\
\delta\omega^{1}_{n-1}+(-1)^{n}{\rm d}\log\omega_{n}^{0}=0\\
\delta{\omega^{0}_{n}}=0
\end{array}
\end{displaymath}     
Notice that in this case the element $\omega^{l}_{n-l}$ is not in general ${\rm d}$-closed. Then $\omega^{0}_{n}$ represents a class in $\check{\rm H}_{\rm Ch}^{n}(\topsp{M};\underline{\rm U(1)})\simeq{\check{\rm H}_{\rm Ch}^{n+1}(\topsp{M};\mathbb{Z})}$. Conversely, given a class $[f]$ in $\check{\rm H}_{\rm Ch}^{n}(\topsp{M};\underline{\rm U(1)})$ one can reconstruct the Deligne class up to $D$-exact terms. This is possible thanks to the fact that any element $\omega^{i}_{j}$ for $j\geq{1}$ which is $\delta$-closed, is also $\delta$-exact. Hence we have that
\begin{displaymath}
\mathbb{H}^{n}(\topsp{M};\mathcal{D}^{l})\simeq{\check{\rm H}_{\rm Ch}^{n+1}(\topsp{M};\mathbb{Z})},\:n>l
\end{displaymath}
The case $n=l$ is the most interesting one. Indeed, in this case the arguments above cannot be applied. Instead, we have the following \cite{Brylinski}
\begin{theorem}\label{delignecs}
Let $\topsp{M}$ be a smooth manifold. Then
\begin{displaymath}
\mathbb{H}^{n}(\topsp{M};\mathcal{D}^{n})\simeq{\check{\rm H}^{n+1}(\topsp{M})}
\end{displaymath}
for any $n\geq{0}$.
\end{theorem}
The complex above defined to compute Deligne cohomology constitutes then a local description of differential characters, which instead are defined in an intrinsic and global way. One can easily define the fieldstrength and the characteristic class for the groups $\mathbb{H}^{n}(\topsp{M};\mathcal{D}^{n})$ in the following way. A Deligne class $\xi\in\mathbb{H}^{n}(\topsp{M};\mathcal{D}^{n})$ is represented as   
\begin{displaymath}
\xi=[(\omega_{0}^{n},\omega^{n-1}_{1},\cdots,\omega_{n}^{0})]
\end{displaymath}
We can then define the fieldstrength of $\xi$ as the $n+1$-form
\begin{displaymath}
{\rm F}(\xi):={\rm d}\omega^{n}_{0}
\end{displaymath}
This is a globally defined $n\!+\!1$-form, since 
\begin{displaymath}
\delta{d\omega^{n}_{0}}={\rm d}\delta\omega^{n}_{0}={\rm d}^{2}\omega^{n-1}_{1}=0
\end{displaymath}
The characteristic class of $\xi$ is instead defined as 
\begin{displaymath}
{\rm c}(\xi):=[\omega^{0}_{n}]\in{\rm H}^{n}(\topsp{M};\mathbb{Z})
\end{displaymath}
By chasing diagrams one can prove that
\begin{displaymath}
[{\rm F}(\xi)]=r({\rm c}(\xi))
\end{displaymath}
The groups $\mathbb{H}^{n}(\topsp{M};\mathcal{D}^{n})$ satisfy the same exact sequences as the Cheeger-Simons groups. Moreover, the relation between the Cheeger-Simons groups and Deligne cohomology clarifies the very use of the term ``cohomology''. In other words, the groups $\check{\rm H}^{n+1}(\topsp{M})$ can be realized as the cohomology of a differential complex, but this does \emph{not} mean that the collection of functors $\check{\rm H}^{*}(\cdot)$ constitute a cohomology theory on the category of manifolds in the sense of the Eilenberg-Steenrod axioms. First, notice that the various groups $\check{\rm H}^{n+1}(\topsp{M})$  are the cohomologies of \emph{different} complexes, for different $n$. Morever, they do not satisfy important cohomological properties like homotopy invariance. This is clear, since the groups $\check{\rm H}^{n+1}(\topsp{M})$ are extensions of the groups of closed forms, which are not homotopy invariant.\\   

As we have seen in section \ref{cheeger}, the first few Cheeger-Simons groups have a geometric interpretation. Indeed, for $n=0$ we have shown that 
\begin{displaymath}
\check{\rm H}^{1}(\topsp{M})\simeq{C^{\infty}(\topsp{M},\topsp{S}^{1}})
\end{displaymath}
A Deligne class for $\mathbb{H}^{0}(\topsp{M};\mathcal{D}^{0})$ is determined by an element
\begin{displaymath}
f\in\prod{C^{\infty}(U_{\alpha_{0}};{\rm U(1)})}
\end{displaymath}
satisfying $\delta{f}=0$, which implies that $f$ is a globally defined function from $\topsp{M}$ to $\rm U(1)$.\\
For $n=1$, the group $\mathbb{H}^{1}(\topsp{M};\mathcal{D}^{1})$  is given by isomorphism classes of line bundles with connection. An element in $\mathbb{H}^{1}(\topsp{M};\mathcal{D}^{1})$ can be represented as a pair
\begin{displaymath}
(\omega^{1}_{0},\omega^{0}_{1})
\end{displaymath}
satisying the equations
\begin{displaymath}
\begin{array}{c}
\delta{\omega^{1}_{0}}-{\rm d}\log\omega^{0}_{1}=0\\
\delta{\omega^{0}_{1}}=0
\end{array}
\end{displaymath}
We can then see that $\omega^{1}_{0}$ gives local representatives $\left\{{\rm A}_{\alpha}\right\}$ of a connection on the line bundle with transition functions $g_{\alpha\beta}$ given by $\omega^{0}_{1}$. The representatives in the same Deligne class differ by $D$-exact terms
\begin{displaymath}
(\tilde{\omega}^{1}_{0},\tilde{\omega}^{0}_{1})=(\omega^{1}_{0},\omega^{0}_{1})+D\xi^{0}_{0}=(\omega^{1}_{0}+{\rm d\:log}\xi^{0}_{0},\omega^{0}_{1}+\delta{\xi_{0}^{0}})
\end{displaymath}
which correspond to the gauge transformations
\begin{displaymath}
\begin{array}{c}
\tilde{\rm A}_{\alpha}={\rm A}_{\alpha}+{\rm d}\log{f_{\alpha}}\\
\tilde{g}_{\alpha\beta}=f^{-1}_{\beta}g_{\alpha\beta}f_{\alpha}
\end{array}
\end{displaymath}
for some circle valued function $f_{\alpha}$. Notice that the gauge transformations above correctly include a ``change'' in the transition functions of the line bundle.\\
The case $n=2$ is in a sense the first interesting case for the application to abelian gauge theories of $p$-forms. Indeed, a Deligne class in $\mathbb{H}^{2}(\topsp{M};\mathcal{D}^{2})$ can be represented by a triple 
\begin{displaymath}
(\omega^{2}_{0},\omega^{1}_{1},\omega^{0}_{2})
\end{displaymath}
satisfying the equations
\begin{displaymath}
\begin{array}{c}
\delta{\omega^{2}_{0}}-{\rm d}\omega^{1}_{1}=0\\
\delta{\omega^{1}_{1}}+{\rm d}\log\omega^{0}_{2}=0\\
\delta{\omega^{0}_{2}}=0
\end{array}
\end{displaymath}
The triple $(\omega^{2}_{0},\omega^{1}_{1},\omega^{0}_{2})$ satisfies the equations (\ref{eq1})-(\ref{eq4}) we used in section \ref{Bfield} to define the B-field, apart from the fact that the functions $\omega^{0}_{2}$ are circle valued, rather than real valued. This is due to the fact that Deligne classes in $\mathbb{H}^{2}(\topsp{M};\mathcal{D}^{2})$ represent gauge equivalence classes of B-fields with a Dirac quantization condition as dictated by the group ${\rm H}^{3}(\topsp{M};\mathbb{Z})$. Having realized the B-field as a Deligne cochain, we also have determined the correct notion of gauge transformation, which consists simply in adding terms of the form $D\alpha$, with $\alpha=(\alpha^{1}_{0},\alpha^{0}_{1})$.\\
Using the Deligne complex as a model for the groups $\check{\rm H}^{n+1}(\topsp{M})$ clarifies the sense of Definition \ref{gen}, since, as we have shown above, the groups $\mathbb{H}^{p}(\topsp{M};\mathcal{D}^{p})$ have a natural description in terms of gauge equivalence classes of $p$-form gauge fields. From the mathematical point of view, it is very natural to interpret the elements of $\mathbb{H}^{p}(\topsp{M};\mathcal{D}^{p})$ as \emph{higher circle bundles with connections} \cite{freedlocality}. In contrast to ordinary circle bundles, these objects are \emph{not} fiber bundles over a base, since there is no total space which makes sense as a manifold. Neverthless, usual operations like pullback can be performed. More importantly, a theory of integration can be constructed for such objects, i.e one can define the notion of the holonomy of a Deligne class around a suitable cycle. This is of major importance for physical applications, as the holonomy usually describes the coupling term of the gauge field to the relative current. In String theory, in particular, this is needed to make sense of the term      
\begin{displaymath}
\int_{\Sigma}f^{*}{\rm B}
\end{displaymath}
describing the coupling of a fundamental string with the B-field. In this sense, the holonomy of a Deligne class $\xi$ is exactly the differential character assigned to $\xi$ by the isomorphism in Theorem \ref{delignecs}.\\
For more details on Deligne cohomology we refer the reader to \cite{Brylinski,holonomy,freedlocality}. 
\section{Ramond-Ramond fields and charge quantization}\label{moorewitten}
As we have seen in the previous sections, differential cohomology, in the Cheeger-Simons or the Deligne approach, represents a suitable formalism to describe fields in generalized abelian gauge theories, and we have also seen how the formalism automatically incorporates the Dirac quantization condition. It is natural then to try to describe within this formalism the gauge theory of Ramond-Ramond fields. However, because of the fact that D-branes, which are sources for Ramond-Ramond fields, have charges taking values in K-theory, it turns out that Ramond-Ramond fields themselves, in the absence of branes, are classified by K-theory. In the following we will present the main arguments developed in \cite{Moore2000} which support this statement, and the fact that the Dirac quantization condition for Ramond-Ramond fields should be expressed in a K-theoretic language. Some of the arguments we will present here are of a heuristic  nature, and should be understood as an educated guess for the use of differential K-theory, which will be introduced in the next section.\\    
Let us first consider the case of ordinary generalized electromagnetism on a $d$-dimensional spacetime $\topsp{M}=\mathbb{R}\times\topsp{Y}$, where $\rm Y$ is noncompact. As discussed in section \ref{genmax}, in presence of a magnetic source $j_{m}$ we have the equation 
\begin{equation}\label{chargeeq}
{\rm dG}=j_{m}
\end{equation}
where the fieldstrength $\rm G$ is an $n$-form.\\
Recall that the (total) magnetic charge is given by the class $[i^{*}_{t}j_{m}]\in{\rm H}_{\rm cpt}^{n+1}(\topsp{Y};\mathbb{R})$. However, to enforce equation (\ref{chargeeq}), the class $[i^{*}_{t}j_{m}]$ must be in the kernel of the natural map
\begin{displaymath}
i:{\rm H}_{\rm cpt}^{n+1}(\topsp{Y};\mathbb{R})\to{\rm H}^{n+1}(\topsp{Y};\mathbb{R})
\end{displaymath}
which ``forgets'' the compact support condition. Let us define $\topsp{N}$ as the boundary of $\topsp{Y}$: the term boundary is used in a loose way, e.g. when $\topsp{Y}=\mathbb{R}^{d-1}$, the boundary is considered to be the sphere $\topsp{S}^{d-2}$ ``at infinity''. We have then that 
\begin{displaymath}
{\rm H}_{\rm cpt}^{n+1}(\topsp{Y};\mathbb{R})\simeq{\rm H}^{n+1}(\topsp{Y},\topsp{N};\mathbb{R})
\end{displaymath}
By using the long exact cohomology sequence for the pair $(\topsp{Y},\topsp{N})$ we have
\begin{displaymath}
\cdots\to{\rm H}^{n}(\topsp{Y};\mathbb{R})\xrightarrow{j}{\rm H}^{n}(\topsp{N};\mathbb{R})\to{\rm H}^{n+1}(\topsp{Y},\topsp{N};\mathbb{R})\xrightarrow{i}{\rm H}^{n+1}(\topsp{Y};\mathbb{R})\to\cdots
\end{displaymath}
where $j$ is given by restriction. The sequence above can be ``broken'', and we have
\begin{displaymath}
{\ker\:i}\simeq{\rm H}^{n}(\topsp{N};\mathbb{R})/j({\rm H}^{n}(\topsp{Y};\mathbb{R}))
\end{displaymath}
The above isomorphism leads to the following interpretation: the total charge of the source $j_{m}$ can be detected by classes of fields on the boundary $\topsp{N}$, i.e. at infinity, that are not restrictions of fields defined on $\topsp{Y}$. Moreover, the group ${\rm H}^{n}(\topsp{Y};\mathbb{R})$ classifies field configurations $\rm G$ which do not contribute to the total charge. In other words, it classifies gauge fields in the absence of sources.\\
The same reasoning can now be applied to type II String theory, where the sources are classified by the integral K-theory group ${\rm K}(\topsp{Y})$. Indeed, in type IIB String theory, for D-branes with compact support in space, the brane charge can be realized as an element in ${\rm K}_{\rm cpt}^{0}(\topsp{Y})$. To ensure that the equations of motion for the Ramond-Ramond field have a solution, the brane charge should take values in the kernel of the map
\begin{displaymath}
i:{\rm K}_{\rm cpt}^{0}(\topsp{Y})\to{\rm K}^{0}(\topsp{Y})
\end{displaymath}
As before, the long exact sequence
\begin{displaymath}
\cdots\to{\rm K}^{-1}(\topsp{Y})\xrightarrow{j}{\rm K}^{-1}(\topsp{N})\to{\rm K}^{0}(\topsp{Y},\topsp{N})\xrightarrow{i}{\rm K}^{0}(\topsp{Y})\to\cdots
\end{displaymath}
implies that
\begin{displaymath}
{\ker\:i}\simeq{\rm K}^{-1}(\topsp{N})/j({\rm K}^{-1}(\topsp{Y}))
\end{displaymath}
In analogy with the cohomological case, we can interpret the above isomorphism as the fact that in type IIB, gauge equivalence classes of Ramond-Ramond fields in the absence of D-branes are classified topologically by the group $\rm K^{-1}(\topsp{Y})$. By the same argument, we find that in type IIA, Ramond-Ramond fields are topologically classified by the group ${\rm K}^{0}(\topsp{Y})$.\\
Having extablished a relation between Ramond-Ramond fields and K-theory, we still must determine the relation between Ramond-Ramond fields and cohomology. Indeed, to be able to write equations of motion for such fields, we need to specify the de Rham cohomology class associated to an element $x\in{\rm K}(\topsp{Y})$ that determines (the class of) a Ramond-Ramond fieldstrength $\rm G$. Let us then first consider the type IIA case. Recall that in type IIA the total Ramond-Ramond potential is of the form
\begin{displaymath}
\rm C=C_{1}+C_{3}+\cdots
\end{displaymath}
where ${\rm C}_{i}$ is a locally defined $i$-form on $\topsp{M}=\mathbb{R}\times\topsp{Y}$. Let us introduce a collection of 8-branes and $\overline{8}$-branes with worldwolume ${\rm p}\times\topsp{Y}$, with $\rm p$ a point in $\mathbb{R}$, and with arbitrary Chan-Paton bundles $(\rm E,\rm F)$ over $\topsp{Y}$. These D-branes are in a sense instantonic configurations. The (electric) coupling term is given by
\begin{equation}\label{coup}
\int_{{\rm p}\times{\topsp{Y}}}{\rm C}\wedge\dfrac{1}{\sqrt{\hat{\rm A}(\topsp{TY})}}\left({\rm ch}(\rm E)-{\rm ch}(\rm F)\right)
\end{equation}
where we have used $\sqrt{\hat{\rm A}(\topsp{TM})}=\sqrt{\hat{\rm A}(\topsp{TY})}$. The Ramond-Ramond equations of motion can be formally written as
\begin{displaymath}
{\rm d\star{G}}=\delta({\rm p})\dfrac{1}{\sqrt{\hat{\rm A}(\topsp{TY})}}\left({\rm ch}(\rm E)-{\rm ch}(\rm F)\right)
\end{displaymath}
where $\delta(\rm p)$ is a Dirac delta distribution supported on $\rm p\times\topsp{Y}$. By integrating both sides, we have the relation
\begin{equation}\label{jump}
\star{\rm G}_{R}-\star{\rm G}_{L}=\dfrac{1}{\sqrt{\hat{\rm A}(\topsp{TY})}}\left({\rm ch}(\rm E)-{\rm ch}(\rm F)\right)
\end{equation}
where ${\rm G}_{R}$ and ${\rm G}_{L}$ denote the value of the field $\rm G$ on spatial slices on the ``right'' and on the ``left'' of the D-brane, where the notion of left and right is given by the orientation of $\topsp{M}$ and $\topsp{Y}$. Now, by imposing the selfduality constraint on the total Ramond-Ramond fieldstrength $\rm G$, we have the following relation
\begin{equation}\label{jump2}
{\rm G}_{R}-{\rm G}_{L}=\dfrac{1}{\sqrt{\hat{\rm A}(\topsp{TY})}}\left({\rm ch}(\rm E)-{\rm ch}(\rm F)\right)
\end{equation}
Notice that this at best a heuristic conclusion, since requiring the Ramond-Ramond field to be selfdual implies that the coupling (\ref{coup}) is \emph{not} defined. Moreover, notice that the right hand side of (\ref{jump}) and (\ref{jump2}) are not selfdual, while the left hand side is supposed to be. A way of circumnavigating the problem is to use selfduality to eliminate half the fields in $\rm G$. For $\rm G_{0},G_{2},G_{4}$ one would introduce the magnetic coupling as arising from the equation  
\begin{equation}\label{motion}
{\rm d{G}}=\delta({\rm p})\dfrac{1}{\sqrt{\hat{\rm A}(\topsp{TY})}}\left({\rm ch}(\rm E)-{\rm ch}(\rm F)\right)
\end{equation}
and introduce the magnetic coupling for $\rm G_{6},G_{8},G_{10}$ via the electric coupling (\ref{coup}) \cite{Moore2000}.\\
Apart from these difficulties, equation (\ref{jump2}) gives a good educated guess for the cohomology class of a Ramond-Ramond field. Indeed, let $b,a\in{\rm K}^{0}(\topsp{Y})$ classify the Ramond-Ramond field on the right and on the left of the collection of D-branes, respectively. Then we have the following equation in de Rham cohomology
\begin{equation}\label{diraccond}
[{\rm G}_{R}(b)]-[{\rm G}_{L}(a)]=\big[{\sqrt{\hat{\rm A}(\topsp{TY})}}\big]^{-1}\cup[{\rm ch}(x)]
\end{equation}
where $x=[\rm E]-[\rm F]$.
Now, by considering the limit ${\rm p}\to{-\infty}$, and requiring that for $a=0$, the associated de Rham cohomology class vanishes,  equation (\ref{diraccond}) determines the class of $\rm G$ on all of the spacetime $\topsp{M}$ as
\begin{equation}\label{condition}
[{\rm G}(x)]=\big[{\sqrt{\hat{\rm A}(\topsp{TY})}}\big]^{-1}\cup[{\rm ch}(x)]
\end{equation}
where we have used that the fact that the Chern character is an isomorphism over the reals. In other words, in this particular configuration, the element $x\in{\rm K}^{0}(\topsp{Y})$ classifies the class of Ramond-Ramond fields whose fieldstrength is a solution of the equation (\ref{motion}). Morover, the fieldstrength form $\rm G$ associated to $x$ must satisfy the condition (\ref{condition}). An analogous argument can be formulated for type IIB String theory, by using the Chern character definition as in section \ref{gysinhomo}. The condition (\ref{condition}) is the Dirac quantization condition for Ramond-Ramond fieldstrengths in type II String theory. It is a ``quantization'' condition since the Chern character \\  
\begin{displaymath}
{\rm ch}:{\rm K^{0,-1}}(\topsp{Y})\to{\rm H}^{\rm ev,odd}(\topsp{Y};\mathbb{R})
\end{displaymath}
maps the K-theory group to a lattice in real cohomology, since the group ${\rm K^{0,-1}}(\topsp{Y})$ is a $\mathbb{Z}$-module. \\
At this point, one assumes that the K-theoretical classification of Ramond-Ramond fields in the absence of branes, and the condition (\ref{condition}), are valid for general spacetimes $\topsp{M}$ not of the form $\mathbb{R}\times{\topsp{Y}}$, and not only in the situation used to motivate it.\\
We are then left with the following mathematical problem: assign to any manifold $\topsp{M}$ abelian groups that can naturally be interpreted as classes of gauge inequivalent fields carrying a topological charge with values in the group ${\rm K^{0,-1}}(\topsp{M})$, and such that the associated fieldstrength satisfies the condition $(\ref{condition})$. We see that the differential cohomology defined in the previous section is not the suitable formalism to solve this problem, since objects in $\check{\rm H}^{p}(\topsp{M})$ or $\mathbb{H}^{p}(\topsp{M};\mathcal{D}^{p})$ carry a topological charge in integer cohomology. We then intuitively need a ``generalized'' version of differential cohomology, which is the subject of the next section. 
\section{Generalized differential cohomology}\label{diffKHS}
As we have seen in the previous section, the K-theoretical description of Ramond-Ramond fields requires a new framework which generalizes the differential cohomology formalism. We need indeed a theory of some sort which is able to describe objects with local degrees of freedom, and at same time takes into account global properties of such objects. A hint on how to define such a theory is given by a closer inspection of the Dirac quantization condition, as appeared in the examples before, along the lines of the arguments proposed in \cite{Freed2000,Freed2002}. Essentially, given a gauge theory of fields $\rm \check{A}$ with fieldstrength $\omega\in\Omega^{*}(\topsp{M})$, to impose a Dirac quantization condition is tantamount to requiring that   
\begin{displaymath}
[\omega]\in\Lambda\subset{\rm H}_{\rm dR}^{*}(\topsp{M};\mathbb{R})
\end{displaymath}
where $\Lambda$ is a lattice in ${\rm H}_{\rm dR}^{*}(\topsp{M};\mathbb{R})$. For instance, in the case of ordinary electromagnetism, the lattice $\Lambda_{\rm H}$ is given by the image of the map
\begin{displaymath}
i:{\rm H}^{*}(\topsp{M};\mathbb{Z})\hookrightarrow{\rm H}_{\rm dR}^{*}(\topsp{M};\mathbb{R})
\end{displaymath}
while for Ramond-Ramond fields in type II String theory, the lattice $\Lambda_{\rm K}$ is given by the image of the map
\begin{displaymath}
{\rm ch}:{\rm K}^{0,1}(\topsp{M})\to{\rm H}_{\rm dR}^{\rm ev,odd}(\topsp{M};\mathbb{R})
\end{displaymath}
up to a ``scale'' factor.\\
Notice that both the above maps induce an \emph{isomorphism} when tensored over the real field. It is clear that the lattice $\Lambda$ is greatly affected by the chosen (generalized) cohomology theory which topologically classifies the fields in the given gauge theory, and by the map realizing the free part of the relevant cohomology groups in real de Rham cohomology. In general, as in the above cases, the choice of the cohomology theory for the given gauge field is suggested by physical properties of the system. In principle, though, there is no argument to exclude a given generalized cohomology theory $\Gamma^{*}$. The above arguments suggest the following mathematical idea, which constitute the starting point in \cite{Hopkins2005} for the construction of generalized differential cohomology theories. Let $\Gamma^{*}$ be an arbitrary multiplicative\footnote{This condition can be relaxed.} generalized cohomology theory on the category of smooth manifolds. Denote with
\begin{displaymath}
\pi_{-*}\Gamma:=\Gamma^{*}(\rm pt)
\end{displaymath}  
the coefficient ring of the point. For example, for $\Gamma^{*}=\rm H^{*}$, we have
\begin{displaymath}
\pi_{-*}{\rm H}=\mathbb{Z}
\end{displaymath}
while for $\Gamma^{*}=\rm K^{*}$ we have
\begin{displaymath}
\pi_{-*}{\rm K}=\mathbb{Z}[[u^{-1},u]]
\end{displaymath}
where $u^{-1}$ is an element of degree -2, and corresponds to the Bott generator. For any generalized cohomology theory $\Gamma^{*}$ there exists a \emph{canonical} map 
\begin{displaymath}
\varphi:\Gamma^{*}(\topsp{X})\to{\rm H}(\topsp{X};\mathbb{R}\otimes\pi_{-*}\Gamma)^{*}
\end{displaymath}
which induces an isomorphism when tensored over the reals for any topological space $\topsp{X}$ \cite{Hopkins2005}. In the above expression, the grading on the cohomology groups is such that
\begin{displaymath}
{\rm H}(\topsp{X};\mathbb{R}\otimes\pi_{-*}\Gamma)^{n}:=\bigoplus_{p+q=n}{\rm H}^{p}(\topsp{X};\mathbb{R}\otimes\pi_{-q}\Gamma)
\end{displaymath}
In the case ${\Gamma^{*}=\rm H^{*}}$, the map $\varphi$ coincides with $i$, while for $\Gamma^{*}=\rm K^{*}$, the map $\varphi$ is given by the Chern character. Notice indeed that
\begin{displaymath}
{\rm H}(\topsp{X};\mathbb{R}\otimes\pi_{-*}{\rm K})^{0}\simeq{\rm H}^{ev}(\topsp{X};\mathbb{R}),\quad {\rm H}(\topsp{X};\mathbb{R}\otimes\pi_{-*}{\rm K})^{1}\simeq{\rm H}^{odd}(\topsp{X};\mathbb{R})
\end{displaymath}
We are then interested in assigning to any smooth manifold $\topsp{M}$ an abelian group which ``completes the square''
\begin{equation}\label{square}
\xymatrix{?\ar[r]\ar[d]&\Omega_{\rm cl}(\topsp{M};\mathbb{R}\otimes\pi_{-*}\Gamma)^{*}\ar[d]\\
\Gamma^{*}(\topsp{M})\ar[r]^{\hspace{-8mm}\varphi}&{\rm H}(\topsp{M};\mathbb{R}\otimes\pi_{-*}\Gamma)^{*}}
\end{equation}
where the map
\begin{displaymath}
\Omega_{\rm cl}(\topsp{M};\mathbb{R}\otimes\pi_{-*}\Gamma)^{*}\to{\rm H}(\topsp{M};\mathbb{R}\otimes\pi_{-*}\Gamma)^{*}
\end{displaymath}
is given by assigning to an element $\omega\in\Omega_{\rm cl}(\topsp{M};\mathbb{R}\otimes\pi_{-*}\Gamma)^{*}$ its de Rham class in\\ ${\rm H}(\topsp{M};\mathbb{R}\otimes\pi_{-*}\Gamma)^{*}$. A first guess to complete the square would be to consider the fibered product of $\Gamma^{*}(\topsp{M})$ and $\Omega_{\rm cl}(\topsp{M};\mathbb{R}\otimes\pi_{-*}\Gamma)^{*}$ over ${\rm H}(\topsp{M};\mathbb{R}\otimes\pi_{-*}\Gamma)^{*}$, i.e. the group  
\begin{displaymath}
{\rm A^{*}_{\Gamma}}:=\left\{(\omega,u)\in\Omega_{\rm cl}(\topsp{M};\mathbb{R}\otimes\pi_{-*}\Gamma)^{*}\times\Gamma^{*}(\topsp{M}):[\omega]=\varphi(u)\right\}
\end{displaymath}
However, we know that these groups, for each degree, are only a first approximation to our desired generalized differential cohomology theory: indeed, we know that for $\Gamma^{*}={\rm H}^{*}$, the Cheeger-Simons groups are an \emph{extension} of ${\rm A^{*}_{\Gamma}}$. In other words, by considering only the fibered product we are losing information about the group of flat and topologically trivial fields, which in general may be nonvanishing. Mathematically this can be understood as follows: for the two cohomology classes $[\omega]$ and $[v]$ to be equal, the cocyle representatives must satisfy the equation
\begin{displaymath}
\omega-\delta{h}=v
\end{displaymath}
for some cycle $h$. The cycle $h$ realizes the homotopy between the two representatives $\omega$ and $v$, and the information about it is lost if we only consider cohomology classes. In the case in which $\Gamma^{*}$ is obtained as the cohomology of a differential complex, the square (\ref{square}) could be completed by a \emph{homotopy refinement}. The difficulty in doing this is in the fact that the functor $\Omega_{\rm cl}(\cdot;\mathbb{R}\otimes\pi_{-*}\Gamma)^{*}$ is \emph{not} a cohomology functor on the category of manifolds. In the case in which $\Gamma^{*}=\rm H^{*}$  the problem can be solved by regarding for each $q\geq{0}$ the space $\Omega_{\rm cl}^{q}(\topsp{M};\mathbb{R})$ as the 0-th cohomology of the complex 
\begin{displaymath}
\Omega^{q}(\topsp{M};\mathbb{R})\xrightarrow{\rm d}\Omega^{q+1}(\topsp{M};\mathbb{R})\xrightarrow{\rm d}\cdots\xrightarrow{\rm d}\Omega^{n}(\topsp{M};\mathbb{R})\to{0}
\end{displaymath}
where $n={\rm dim}\:\topsp{M}$. By using standard results, then, one can define for a given $q$ the complex \cite{Hopkins2005}
\begin{equation}\label{hopkinsinger}
\check{\rm C}^{p}(q):=\left\{\begin{array}{c}
{\rm C}^{p}(\topsp{M};\mathbb{Z})\times{\rm C}^{p-1}(\topsp{M};\mathbb{R})\times\Omega^{p}(\topsp{M};\mathbb{R})\quad{n\geq{q}}\\
{\rm C}^{p}(\topsp{M};\mathbb{Z})\times{\rm C}^{p-1}(\topsp{M};\mathbb{R})\quad{n<q}
\end{array}\right .
\end{equation}
with differential 
\begin{displaymath}
{\rm d}(c,h,\omega):=(\delta{c},\omega-c-\delta{h},d\omega)
\end{displaymath}
and
\begin{displaymath}
{\rm d}(c,h):=\left\{\begin{array}{c}
(\delta{c},-c-\delta{h},0)\quad{n=q-1}\\
(\delta{c},-c-\delta{h})\quad{\text{otherwise}}
\end{array}\right .
\end{displaymath}
Notice that a triple $(c,h,\omega)$ is a cocycle if it satisfies the equations
\begin{equation}\label{cocyclecond}
\begin{array}{c}
\delta{c}=0\\
{\rm d}\omega=0\\
\omega-c-\delta{h}=0
\end{array}
\end{equation}
The last of the above equations implies that $[\omega]=i([c])$: this means that the cohomology groups 
\begin{displaymath}
\check{\rm H}(q)^{*}(\topsp{M}):={\rm H}^{*}(\check{\rm C}^{*}(q);{\rm d})
\end{displaymath}
 fit the square (\ref{square}). The groups $\check{\rm H}(q)^{*}(\topsp{M})$ are called the \emph{Cheeger-Simons cohomology groups} and as one can expect we have
\begin{displaymath}
\check{\rm H}(q)^{p}(\topsp{M})\simeq\mathbb{H}^{p}(\topsp{M};\mathcal{D}^{q})
\end{displaymath}
The fieldstrength map assigns to the class $[(c,h,\omega)]\in\check{\rm H}(q)^{q}(\topsp{M})$ the closed form $\omega$, while the characteristic class map assigns the class $[c]$. Morever, the groups $\check{\rm H}(q)^{q}(\topsp{M})$ satisfies the same exact sequences as the Cheeger-Simons groups.\\

Unfortunately, the approach followed in the above paragraph to construct the theory completing the square (\ref{square}) for ordinary cohomology cannot be used in general for generalized cohomology theories as K-theory, since these are not obtained as the cohomology of a certain differential complex. Neverthless, a homotopy refinement for a given theory $\Gamma$ can be obtained by substituting cocycles with maps to the classifying space $\rm B\Gamma$, and defining an analog of conditions (\ref{cocyclecond}). By using this strategy, Hopkins and Singer showed in \cite{Hopkins2005} that to \emph{any} arbitrary generalized cohomology theory one can assign a theory, that they call a \emph{generalized differential cohomology theory}, which fits the square (\ref{square}) and satisfies analogous properties to the Cheeger-Simons or Deligne cohomology groups. We will focus in the following on a particular definition of differential K-theory, which will be the basis in the following chapter for a suitable generalization to the equivariant setting; we refer the reader to \cite{Hopkins2005} for an extensive and detailed introduction to this beautiful and exciting part of modern mathematics.\\
Let $\rm Fred$ denote the space of Fredholm operators on an infinite dimensional Hilbert space: recall that $\Fred$ is a classifying space for the group ${\rm K}^{0}(\topsp{X})$ for a given topological space, i.e.
\begin{displaymath}
{\rm K^{0}}(\topsp{X})\simeq[\topsp{X},{\rm Fred}]
\end{displaymath}
where the isomorphism above is given by considering the index bundle 
\begin{equation}\label{indexbundle}
[{\rm Ker}\:{\rm F}_{f}]-[{\rm Coker}\:{\rm F}_{f}]
\end{equation}
where 
\begin{displaymath}
\begin{array}{c}
{\rm Ker}\:{\rm F}_{f}:=\bigcup_{x}{\rm Ker}\:f(x)\\
{\rm Coker}\:{\rm F}_{f}:=\bigcup_{x}{\rm Coker}\:f(x)
\end{array}
\end{displaymath}
for a given map $f:\topsp{X}\to{\rm Fred}$. Of course, the expression (\ref{indexbundle}) is naive, since the dimension of ${\rm Ker}\:f(x)$ and ${\rm Coker}\:f(x)$ can change while $x$ varies, and needs to be stabilized, as shown in \cite{Atiyah1967}, for instance.\\
Let $u$ denote a cocycle in 
\begin{displaymath}
\mathcal{Z}({\Fred};\mathbb{R}\otimes\pi_{-*}{\rm K})^{0}=\bigoplus_{n}\mathcal{Z}^{2n}({\Fred};\mathbb{R})
\end{displaymath}
representing the universal Chern character, i.e. such that if $f:\topsp{X}\to{\rm Fred}$ classifies a vector bundle $\rm E$, $f^{*}u$ represents ${\rm ch}({\rm E})$. For a manifold $\topsp{M}$, an element in the \emph{differential K-theory} group $\check{\rm K}^{0}(\topsp{M})$ is represented by a triple
\begin{displaymath}
(c,h,\omega)
\end{displaymath}
where $c:\topsp{M}\to{\rm Fred}$, $h\in\mathcal{C}^{ev-1}(\topsp{M};\mathbb{R})$, and $\omega\in\Omega^{ev-1}(\topsp{M};\mathbb{R})$ such that
\begin{equation}\label{tripledef}
\delta{h}=\omega-c^{*}u
\end{equation}
Moreover, the triples above defined must satisfy the following equivalence relation. Two triples $(c^{0},h^{0},\omega^{0})$ and $(c^{1},h^{1},\omega^{1})$ are said to be equivalent if there exists a triple $(c,h,\omega)$ on $\topsp{M}\times[0,1]$, with $\omega$ constant along $[0,1]$, such that
\begin{equation}\label{equivrel}
\begin{array}{c}
(c,h,\omega)|_{0}=(c^{0},h^{0},\omega^{0})\\
(c,h,\omega)|_{1}=(c^{1},h^{1},\omega^{1})
\end{array}
\end{equation}
Notice that equation (\ref{tripledef}) enforces the condition  
\begin{displaymath}
[\omega]={\rm ch}([c])
\end{displaymath}
The relations (\ref{equivrel}) can be rephrased \cite{Freed2000} in a way that will be useful later on. Indeed, two triple $(c^{0},h^{0},\omega^{0})$ and $(c^{1},h^{1},\omega^{1})$ are equivalent if there exists a map
\begin{displaymath}
f:\topsp{M}\times[0,1]\to{\rm Fred}
\end{displaymath}
and an element $\sigma\in\Omega^{ev-2}(\topsp{M};\mathbb{R})$ such that
\begin{displaymath}
\begin{array}{c}
f|_{0}=c^{0}\\
f|_{1}=c^{1}\\
\omega^{1}=\omega^{0}\\
h^{1}=h^{0}+\pi_{*}f^{*}u+{\rm d}\sigma
\end{array}
\end{displaymath}  
where $\pi:\topsp{M}\times[0,1]\to\topsp{M}$ is the projection onto the first factor, and $\pi_{*}$ denotes integration along the fibre, i.e. pairing with the fundamental class of $[0,1]$. The last of the above equations is obtained by imposing the condition (\ref{tripledef}).\\
To define the higher differential K-groups, recall that the iterated loop spaces $\Omega^{i}\topsp{Fred}$ classify the functors ${\rm K}^{-i}$. Let $u^{-i}$ be the cocycle 
\begin{displaymath}
u^{-i}\in\mathcal{Z}^{2n-i}(\Omega^{i}\topsp{Fred};\mathbb{R})
\end{displaymath}
defined as
\begin{displaymath}
u^{-i}:=\Pi_{*}ev^{*}u
\end{displaymath}
where    
\begin{displaymath}
ev:\topsp{S}^{i}\times\Omega^{i}{}\topsp{Fred}\to\topsp{Fred}
\end{displaymath}
denotes the evaluation map, and $\Pi_{*}$ denotes the integration over the fiber of the projection map $\Pi:\topsp{S}^{i}\times\Omega^{i}\topsp{Fred}\to\Omega^{i}\topsp{Fred}$. Hence, elements of the higher differential groups $\check{\rm K}^{-i}(\topsp{M})$ are represented by triples
\begin{displaymath}
(c,h,\omega)
\end{displaymath}
where $c:\topsp{M}\to{\Omega^{i}{\rm Fred}}$, $h\in\mathcal{C}^{ev-i}(\topsp{M};\mathbb{R})$, and $\omega\in\Omega^{ev-i}(\topsp{M};\mathbb{R})$ such that
\begin{displaymath}
\delta{h}=\omega-c^{*}u^{-i}
\end{displaymath}
and satisfying the equivalence relations (\ref{equivrel}).\\
The differential K-groups satisfy the following exact sequences
\begin{equation}\label{exactdiffK}
\begin{array}{c}
0\to{\rm K}^{-i-1}(\topsp{M};\mathbb{R}/\mathbb{Z})\to\check{\rm K}^{-i}(\topsp{M})\to\Omega_{\rm K}^{ev-i}(\topsp{M})\to{0}\\
0\to{\Omega}^{ev-i-1}(\topsp{\topsp{M}})/{\Omega}_{\rm K}^{ev-i-1}(\topsp{\topsp{M}})\to\check{\rm K}^{-i}(\topsp{M})\to{\rm K}^{-i}(\topsp{M})\to{0}\\
0\to{\rm H}^{ev-i-1}(\topsp{M})/{\rm ch}({\rm K}^{-i-1}(\topsp{M}))\to\check{\rm K}^{-i}(\topsp{M})\to{\rm A}_{\rm K}^{-i}(\topsp{M})\to{0}
\end{array}
\end{equation}
where $\Omega_{\rm K}^{ev-i}(\topsp{M})$ denotes the elements in $\Omega^{ev-i}(\topsp{M})$ whose cohomology class is in the image of the Chern character. We will give a proof of the existence of the above exact sequences in the context of equivariant K-theory in the next chapter. Finally, the homotopy equivalence $\Omega^{n}{\topsp{Fred}}\simeq\Omega^{n+2}{\topsp{Fred}}$ induces Bott periodity on the differential K-groups, and in particular it allows to define them for positive degrees as  
\begin{displaymath}
\check{\rm K}^{n}(\topsp{M}):=\check{\rm K}^{n-2N}(\topsp{M})\quad{n-2N<0}
\end{displaymath}
This definition of differential K-theory will be at the core of a generaliztion we will propose in the next chapter. However, we should point out that there are different models for differential K-theory, such as the one introduced in \cite{smooth}, where axioms for generalized differential cohomology theories are given. We have seen that ordinary differential cohomology admits different models, which lead to isomorphic theories, characterized by the fact that they satisfy the same exact sequences. As shown in \cite{axiomsdiff}, these exact sequences in a sense completely characterize ordinary differential cohomology. As far as the author knows, a similar result does not exist for generalized differential cohomology theories.
\chapter{Ramond-Ramond fields and Orbifold differential K-theory}
As we have seen in the previous chapter, Ramond-Ramond fields in type II String theory require the use of differential K-theory to be properly described. We have also explained that this is due to the fact that the Dirac quantization condition for Ramond-Ramond fieldstrength is dictated by K-theory, rather than integral cohomology. Moreover, the lattice in which charges take their values has been suggested by the form of the coupling between D-branes and Ramond-Ramond field.\\ 
In this chapter we will generalize the above arguments to the case of superstring theory defined on an orbifold. To be precise, we will only consider the case of presentable orbifolds of the form $[\topsp{M}/\topsp{G}]$, where $\topsp{M}$ is a spin manifold, and $\topsp{G}$ is a finite group acting by spin structure preserving diffeomorphisms. In this case, as proposed by Witten in \cite{Witten1998}, D-branes are classified by the equivariant K-theory of the ``covering'' space $\topsp{M}$. It is natural, then, to ask that the Dirac quantization condition for the total fieldstrength of Ramond-Ramond fields on the orbifold $[\topsp{M}/\topsp{G}]$ is given by a Chern character homomorphism on the equivariant K-theory ${\rm K}_{\rm G}^{*}(\topsp{M})$, which induces an isomorphism when tensored over $\mathbb{R}$. This is supported by the fact that we expect to obtain type II String theory features when ${\rm G}=\{e\}$. A Chern character with the above properties has been constructed in \cite{Luck1998}, and makes use of Bredon cohomology, an equivariant cohomology theory defined on the category of G-CW complexes. The important feature of this equivariant cohomology theory is that it naturally takes into account some important features of String theory on orbifolds, namely the presence of twisted sectors. We then pass to propose a definition of differential K-theory suitable for good orbifolds. Indeed, the existence theorem for generalized differential cohomology theories developed in \cite{Hopkins2005} cannot be applied to equivariant cohomology functors on the category of $\rm G$-manifolds. More precisely, we define abelian groups that behave as a natural generalization of the ordinary differential K-theory groups, in the sense that they agree in the case of a trivial group and they satisfy analogous exact sequences. We will ``test'' our definition in the case of linear abelian orbifolds, and give a proposal for the group of (equivalence classes of) flat fields. We will also make a brief digression to describe D-branes using geometric equivariant K-homology, showing that the use of K-cycle is well-suited to the description of \emph{fractional} D-branes and their topological charges computed using equivariant Dirac operator theory. By using the Chern character, we will define electric Ramond-Ramond couplings to D-branes on good orbifolds, and compare with previous examples in the literature. Finally, we will support our definition of orbifold differential K-theory, or to be precise a complex version of it, by generalizing Moore and Witten argument as introduced in the previous chapter to the equivariant setting, and by expressing the Ramond-Ramond equations of motion by writing an equivariant version of the Ramond-Ramond current in term of the equivariant Chern character and an equivariant version of the Riemann-Roch theorem.
\section{G-CW complexes and equivariant cohomology theories}
In this section we will recall some basic notions about (generalized)
equivariant cohomology theories. In the following, $\topsp{X}$ denotes a topological space and $\topsp{G}$ a
finite group, unless otherwise stated. In the following, a (left) action $\topsp{G}\times
\topsp{X}\to \topsp{X}$ of $\topsp{G}$ on $\topsp{X}$ will be denoted $(g,x)\mapsto g\cdot x$, and we will call $\topsp{X}$ a \emph{G-space}. The
stabilizer or isotropy group of a point $x\in \topsp{X}$ is denoted
$\topsp{G}_x=\{g\in \topsp{G}~|~g\cdot x=x\}$. Recall that a continuous map $f:\topsp{X}\to \topsp{Y}$
of $\topsp{G}$-spaces is a \emph{$\topsp{G}$-map} if $f(g\cdot x)=g\cdot f(x)$ for all
$g\in \topsp{G}$ and $x\in \topsp{X}$.
\begin{definition}\label{Gcomplex}
A \emph{$\topsp{G}$-equivariant CW-decomposition} of a $\topsp{G}$-space $\topsp{X}$ consists
of a filtration $\topsp{X}_{n}$, $n\in{\mathbb{N}_0}$ such that
\begin{displaymath}
\topsp{X}=\bigcup_{n\in\nat_0}\,{\topsp{X}_{n}}
\end{displaymath}
and $\topsp{X}_{n}$ is obtained from $\topsp{X}_{n-1}$ by ``attaching'' equivariant
cells via the following procedure. Define
\begin{displaymath}
\topsp{X}_{0}=\coprod_{j\in{J_0}}\,\topsp{G}/\topsp{K}_{j} \ ,
\end{displaymath}
with $\topsp{K}_{j}$ a collection of subgroups of ${\topsp{G}}$ and the standard
(left) $\topsp{G}$-action on any coset space $\topsp{G}/\topsp{K}_j$. For $n\geq1$ set
\beq
\topsp{X}_{n}=\Big(\topsp{X}_{n-1}\amalg\coprod_{j\in J_n}\,\big({\rm D}^{n}_{j}
\times{\topsp{G}/{\topsp{K}_{j}}}\big)\Big)\,\Big/\,\sim
\label{Xnattach}\eeq
where the equivalence relation $\sim$ is generated by $\topsp{G}$-equivariant
``attaching maps''
\beq
\phi^{n}_{j}\,:\,{\rm S}^{n-1}_{j}\times{\topsp{G}/\topsp{K}_{j}}
~\longrightarrow~ \topsp{X}_{n-1} \ .
\label{Gattach}\eeq
One requires that $\topsp{X}$ carries the colimit topology with
respect to ($\topsp{X}_{n}$), {i.e.}, $\topsp{B}\subset{\topsp{X}}$ is closed if and only
if $\topsp{B}\cap{\topsp{X}_{n}}$ is closed in $\topsp{X}_{n}$ for all $n\in\nat_0$. We call
the image of ${\rm D}^{n}_{j}\times{\topsp{G}/{\topsp{K}_{j}}}$
(resp.~$\mathring{\rm D}^{n}_{j}\times{\topsp{G}/{\topsp{K}_{j}}}$) a
\emph{closed} (resp.~\emph{open}) $n$-cell of orbit type $\topsp{G}/\topsp{K}_{j}$. As
usual, we call the subspace $\topsp{X}_{n}$ the {$n$-skeleton} of
$\topsp{X}$. If $\topsp{X}=\topsp{X}_{n}$ and $\topsp{X}\neq{\topsp{X}_{n-1}}$, then $n$ is called
the (\emph{cellular}) \emph{dimension} of $\topsp{X}$ and $\topsp{X}$ is said to be of
\emph{finite type}. A $\topsp{G}$-space with a $\topsp{G}$-equivariant
CW-decomposition is called a \emph{$\topsp{G}$-complex}.
\end{definition}
When $\topsp{G}=e$ is the trivial group, a $\topsp{G}$-complex is just an ordinary
CW-complex. In general, if $\topsp{X}$ is a $\topsp{G}$-complex then the orbit space
$\topsp{X}/\topsp{G}$ is an ordinary CW-complex. Conversely, there is an intimate
relation between $\topsp{G}$-complexes and ordinary CW-complexes whenever $\topsp{G}$
is a discrete group. Let $\topsp{X}$ be a $\topsp{G}$-space which is an ordinary
CW-complex. We say that \emph{$\topsp{G}$ acts cellularly} on $\topsp{X}$ if 
\begin{itemize}
\item[1)] For each $g\in{\topsp{G}}$ and each open cell $\topsp{E}$ of $\topsp{X}$, the left
  translation $g\cdot \topsp{E}$ is again an open cell of $\topsp{X}$; and
\item[2)] If $g\cdot \topsp{E}=\topsp{E}$, then the induced map $\topsp{E} \to \topsp{E}$, $x \mapsto
  g\cdot x$ is the identity.
\end{itemize}
Then we have the following \cite{Dieck1987}
\begin{proposition} Let $\topsp{X}$ be a CW-complex with a cellular action of
  a discrete group $\topsp{G}$. Then $\topsp{X}$ is a $\topsp{G}$-complex with $n$-skeleton
  $\topsp{X}_{n}$.
\end{proposition}
In the case that $\topsp{X}$ is a smooth manifold, we require the $\topsp{G}$-action
on $\topsp{X}$ to be smooth and there is an analogous result. Recall that the
applicability of algebraic topology to manifolds relies on the fact
that any manifold comes equiped with a canonical CW-decomposition. In
the case in which a group acts on the manifold one has the following
result due to Illman~\cite{illmann1,illmann1}.
\begin{theorem}
If $\topsp{G}$ is a compact Lie group or a finite group acting on a smooth
compact manifold $\topsp{X}$, then $\topsp{X}$ is triangulable as a finite
$\topsp{G}$-complex.
\end{theorem}
The collection of $\topsp{G}$-complexes with $\topsp{G}$-maps as morphisms form a
category. We are interested in equivariant cohomology theories defined
on this category (or on subcategories thereof).\\

We will now briefly spell out the main ingredients involved in
building an equivariant cohomology theory on the category of finite
$\topsp{G}$-complexes, leaving the details to the comprehensive treatments
of \cite{Dieck1987} and~\cite{Luck2006}, and focusing instead on some
explicit examples. Fix a group $\topsp{G}$ and a commutative
ring $\topsp{R}$. A \emph{$\topsp{G}$-cohomology theory $\E^*_{\topsp{G}}$ with values in
  $\topsp{R}$-modules} is a collection of contravariant functors $\E^{n}_{\topsp{G}}$
from the category of $\topsp{G}$-CW~pairs to the category of $\topsp{R}$-modules
indexed by $n\in\mathbb{Z}$ together with natural transformations
\begin{displaymath}
\delta^{n}_{\topsp{G}}(\topsp{X},\topsp{A})\,:\,\E^{n}_{\topsp{G}}(\topsp{X},\topsp{A}) ~\longrightarrow~
\E^{n+1}_{\topsp{G}}(\topsp{X}):=\E^{n+1}_{\topsp{G}}(\topsp{X},\emptyset)
\end{displaymath}
for all $n\in{\mathbb{Z}}$ satisfying the axioms of $\topsp{G}$-homotopy
invariance, long exact sequence of a pair, excision, and disjoint
union. The theory is called \emph{ordinary} if for any orbit $\topsp{G}/\topsp{H}$ one
has $\E_{\topsp{G}}^{q}(\topsp{G}/\topsp{H})=0$ for all $q\neq{0}$. These axioms are
formulated in an analogous way to that of ordinary cohomology. The new
ingredients in an equivariant cohomology theory (which we have not
yet defined) are the \emph{induction structures}, which we shall now
describe.

Let $\alpha:\topsp{H}\to{\topsp{G}}$ be a group homomorphism, and let $\topsp{X}$ be an
$\topsp{H}$-space. Define the \emph{induction of $\topsp{X}$ with respect to $\alpha$}
to be the $\topsp{G}$-space $\ind_{\alpha}\topsp{X}$ given by
\begin{displaymath}
\ind_{\alpha}\topsp{X}:=\topsp{G}\times_{\alpha}\topsp{X} \ .
\end{displaymath}
This is the quotient of the product $\topsp{G}\times \topsp{X}$ by the $\topsp{H}$-action
$h\cdot(g,x):=(g\,\alpha(h^{-1}),h\cdot x)$, with the $\topsp{G}$-action on
$\ind_{\alpha}\topsp{X}$ given by $g'\cdot[g,x]=[g'\,g,x]$. If $\topsp{H}<\topsp{G}$ is a subgroup, and
$\alpha$ is the subgroup inclusion, the induced $\topsp{G}$-space is denoted
$\topsp{G}\times_\topsp{H}\topsp{X}$.

An \emph{equivariant cohomology theory $\E_{(-)}^{*}$ with values in
  $\topsp{R}$-modules} consists of a collection of $\topsp{G}$-cohomology theories
$\E_{\topsp{G}}^{*}$ with values in $\topsp{R}$-modules for each group $\topsp{G}$ such
that for any group homomorphism $\alpha:\topsp{H} \to \topsp{G}$ and any $\topsp{H}$-CW~pair
$(\topsp{X},\topsp{A})$ with $\ker(\alpha)$ acting freely on $\topsp{X}$, there are for each
$n\in\mathbb{Z}$ natural isomorphisms
\beq
\ind_{\alpha}\,:\,\E^{n}_{\topsp{G}}\big(\ind_{\alpha}(\topsp{X},\topsp{A})\big)
~\xrightarrow{\approx}~\E^{n}_{\topsp{H}}(\topsp{X},\topsp{A})
\label{indgen}\eeq
satisfying
\begin{itemize}
\item[(a)] Compatibility with the coboundary homomorphisms:
\begin{displaymath}
\delta^{n}_{\topsp{H}}\circ{\ind_{\alpha}}=\ind_{\alpha}\circ{\delta^{n}_{\topsp{G}}}
\ ;
\end{displaymath}
\item[(b)] Functoriality: If $\beta:\topsp{G}\to{\topsp{K}}$ is another group
  homomorphism such that $\ker(\beta\circ\alpha)$ acts freely on $\topsp{X}$,
  then for every $n\in\mathbb{Z}$ one has
\begin{displaymath}
\ind_{\beta\circ\alpha}=\ind_{\alpha}\circ{\ind_{\beta}}\circ{}\E^{n}_{\topsp{K}}(f_{1})
\end{displaymath}
where
\begin{eqnarray*}
f_{1}\,:\,\ind_{\beta}\big(\ind_{\alpha}(\topsp{X},\topsp{A})\big)&\xrightarrow{\approx}&
\ind_{\beta\circ{\alpha}}(\topsp{X},\topsp{A}) \\
(k,g,x)&\longmapsto&\big(k\,\beta(g)\,,\,x\big)
\end{eqnarray*}
is a $\topsp{K}$-homeomorphism and $\E^{n}_{\topsp{K}}(f_{1})$ is the morphism on
$\topsp{K}$-cohomology induced by $f_1$; and
\item[(c)] Compatibility with conjugation: For $g,g'\in{\topsp{G}}$ define
  ${\rm Ad}_g(g'\,)=g\,g'\,g^{-1}$. Then the homomorphism $\ind_{{\rm Ad}_g}$
  coincides with $\E^{n}_{\topsp{G}}(f_{2})$, where
\begin{eqnarray*}
f_{2}\,:\,(\topsp{X},\topsp{A})&\xrightarrow{\approx}& \ind_{{\rm Ad}_g}(\topsp{X},\topsp{A})\\
x &\longmapsto& \big(e\,,\,g^{-1}\cdot x\big)
\end{eqnarray*}
is a $\topsp{G}$-homeomorphism, where throughout $e$ denotes the identity
element in the group~$\topsp{G}$.
\end{itemize}
Thus the induction structures connect the various $\topsp{G}$-cohomologies and
keep track of the equivariance. They will be very important in the
construction of the equivariant Chern character for equivariant
K-theory in a later section, even if we are only interested in a
fixed group $\topsp{G}$.
\begin{example}[\emph{Borel cohomology}]\label{Borelex}
Let $\H^{*}$ be a cohomology theory for CW-pairs (for example,
singular cohomology). Define
\begin{displaymath}
 \H^{n}_{\topsp{G}}(\topsp{X},A):=\H^{n}\big(\topsp{EG}\times_{\topsp{G}}(\topsp{X},A)\big)
\end{displaymath}
where $\topsp{EG}$ is the total space of the classifying principal
$\topsp{G}$-bundle $\topsp{EG}\to \topsp{BG}$ which is contractible and carries a free
$\topsp{G}$-action. This is called (\emph{equivariant}) \emph{Borel
  cohomology}, and is the most commonly used form of equivariant
cohomology in the physics literature. Note that $\H_\topsp{G}^*$ is
well-defined because the quotient $\topsp{EG}\times_\topsp{G}\topsp{X}$ is unique up to the
homotopy type of $\topsp{X}/\topsp{G}$. The ordinary $\topsp{G}$-cohomology structures on
$\H^{*}_{\topsp{G}}$ are inherited from the cohomology structures on
$\H^{*}$. The induction structures for $\H^{*}_{\topsp{G}}$ are
constructed as follows. Let $\alpha:\topsp{H} \to \topsp{G}$ be a group
homomorphism and $\topsp{X}$ an $\topsp{H}$-space. Define
\begin{eqnarray*}
b\,:\topsp{EH}\times_{\topsp{H}}\topsp{X} &\longrightarrow& \topsp{EG}\times_{\topsp{G}}\topsp{G}\times_{\alpha}\topsp{X}\\
(\varepsilon,x) &\longmapsto& \big(E\alpha(\varepsilon)\,,\,e\,,\,x
\big)
\end{eqnarray*}
where $\varepsilon\in{\topsp{EH}}$, $x\in{\topsp{X}}$ and $E\alpha:\topsp{EH} \to \topsp{EG}$ is the
$\alpha$-equivariant map induced by $\alpha$. The induction map
$\ind_{\alpha}$ is then given by pullback
\begin{displaymath}
\ind_{\alpha}:=b^{*}\,:\,\H^{n}_{\topsp{G}}(\ind_{\alpha}\topsp{X})=\H^{n}(\topsp{EG}\times_{\topsp{G}}\topsp{G}
\times_{\alpha}\topsp{X})
~\longrightarrow~\H^{n}(\topsp{EH}\times_{\topsp{H}}\topsp{X})={\H^{n}_{\topsp{H}}(\topsp{X})} \ .
\end{displaymath}
If $\ker(\alpha)$ acts freely on $\topsp{X}$, then the map $b$ is a homotopy
equivalence and hence the map $\ind_{\alpha}$ is an isomorphism.
\end{example}
\begin{example}[\emph{Equivariant K-theory}]\label{EqKex}
In \cite{Segal1968}, equivariant topological K-theory is defined for
any $\topsp{G}$-complex $\topsp{X}$ as the abelian group completion of the semigroup
$\Vect_\topsp{G}^\complex(\topsp{X})$ of complex $\topsp{G}$-vector bundles over $\topsp{X}$, with $\topsp{G}$ a compact Lie group. Recall that for a $\topsp{G}$-space $\topsp{X}$, a complex $\topsp{G}$-vector bundle is given by a $\topsp{G}$-space $\topsp{E}$, and a $\topsp{G}$-map $\pi:\topsp{E}\to{\topsp{X}}$ such that $\topsp{E}$ is the total space of a complex vector bundle, and such that for any $g\in{\topsp{G}}$ and any $x\in{\topsp{X}}$, the map $g:\topsp{E}_{x}\to{\topsp{E}_{gx}}$ is an homomorphism. The compactness property of $\topsp{G}$ assures that the Grothendieck functor $\rm{K}_{\topsp{G}}^{*}$ satisfies the $\topsp{G}$-homotopy invariance\footnote{This is due to the fact that pullbacks of a $\topsp{G}$-bundle via $\topsp{G}$-homotopic maps are isomorphic only if $\topsp{G}$ is compact.}. The higher groups are defined via iterated
suspension, similarly to ordinary K-theory, and Bott periodicity holds.\\
To define the induction structures, recall that if $\topsp{X}$ is
an $\topsp{H}$-space and $\alpha:\topsp{H} \to \topsp{G}$ is a group homomorphism, then the
map
\begin{eqnarray*}
\varphi\,:\,\topsp{X} &\longrightarrow& \topsp{G}\times_{\alpha}\topsp{X}\\
x &\longmapsto& (e,x)
\end{eqnarray*}
is an $\alpha$-equivariant map which embeds $\topsp{X}$ as the subspace
$\topsp{H}\times_{\alpha}\topsp{X}$ of $\topsp{G}\times_{\alpha}\topsp{X}$, and which induces via
pullback of vector bundles the homomorphism
\begin{displaymath}
\varphi^{*}\,:\,\K^{*}_{\topsp{G}}(\topsp{G}\times_{\alpha}\topsp{X}) ~\longrightarrow~
\K^{*}_{\topsp{H}}(\topsp{X}) \ .
\end{displaymath}
This map defines the induction structure. It is invertible when
$\ker(\alpha)$ acts freely on $\topsp{X}$, with inverse the ``extension'' map
$\topsp{E}\mapsto \topsp{G}\times_{\topsp{H}}\topsp{E}$ for any $\topsp{H}$-vector bundle $\topsp{E}$ over $\topsp{X}$. This can be proven by using the following \cite{Segal1968}
\begin{theorem}\label{free}
Let $\topsp{G}$ be a compact group, and let $\topsp{N}$ be a normal subgroup acting freely on $\topsp{X}$. Then
\beq
{\rm pr}^{*}:\K^{*}_{\topsp{G}/\topsp{N}}(\topsp{X}/\topsp{N})\xrightarrow{\simeq}\K^{*}_{\topsp{G}}(\topsp{X})
\label{eqexcision}\eeq
where ${\rm pr}:\topsp{X}\to{\topsp{X}/\topsp{N}}$ is the usual projection.
\end{theorem}
 By noticing that
\begin{displaymath}
\topsp{X}/\topsp{N}\simeq(\topsp{G}/\topsp{N})\times_{\topsp{G}}\topsp{X}
\end{displaymath}
and if we define denote with $\topsp{N}$ the kernel of $\alpha:\topsp{H}\to{\topsp{G}}$, then
\begin{displaymath}
\K^{*}_{\topsp{G}}\big(\topsp{G}\times_{\alpha}\topsp{X}\big)\simeq\K^{*}_{\topsp{H}/\topsp{N}}\big((\topsp{H}/\topsp{N})\times_{\alpha}\topsp{X}\big)\simeq{\K^{*}_{\topsp{H}/\topsp{N}}(\topsp{X}/\topsp{N})}\simeq{\K^{*}_{\topsp{H}}(\topsp{X})}
\end{displaymath}
since $\topsp{N}$ acts freely on $\topsp{X}$ by hypothesis.\\
In the case in which $\topsp{X}=\rm pt$, a $\topsp{G}$-vector bundle is just a $\topsp{G}$-module, hence we have
\begin{equation}\label{point}
{\rm K}^{0}_{\topsp{G}}(\rm pt)\simeq{\rm R(\topsp{G})},\quad {\rm K}_{\topsp{G}}^{-1}(pt)\simeq{0}
\end{equation}
where $\rm R(\topsp{G})$ is the \emph{respresentation ring} of $\topsp{G}$, i.e. the ring generated over $\mathbb{Z}$ by the irreducible representations of $\topsp{G}$.\\
We can use the above results to show that 
\begin{displaymath}
{\rm K}^{0}_{\topsp{G}}(\topsp{G}/\topsp{H})\simeq{\rm K}^{0}_{\topsp{G}}(\topsp{G}\times_{\topsp{H}}{\rm pt})\simeq{\rm K}^{0}_{\topsp{H}}(\rm pt)\simeq{\rm R(\topsp{H})},\quad{\rm K}^{-1}_{\topsp{G}}(\topsp{G}/\topsp{H})\simeq{0}
\end{displaymath}
for $\topsp{H}<\topsp{G}$.\\
In the case in which the group $\topsp{G}$ acts freely on the space $\topsp{X}$, theorem \ref{free} implies        
\begin{displaymath}
\K^{*}(\topsp{X}/\topsp{G}){\simeq}\K^{*}_{\topsp{G}}(\topsp{X})
\end{displaymath}
On the other extreme, when $\topsp{G}$ acts trivially on $\topsp{X}$, the following isomorphism holds \cite{Segal1968}
\begin{equation}
\K^{*}_{\topsp{G}}(\topsp{X})\simeq\K^{*}(\topsp{X})\otimes{{\rm R}(\topsp{G})}
\end{equation}
\end{example}
\section{The equivariant Chern character}\label{equivChernsec}
As we mentioned in the beginning of this chapter, we expect that the Dirac quantization condition for the total Ramond-Ramond fieldstrength on a good orbifold is dictated by some homomorphism on equivariant K-theory which is a generalization of the ordinary Chern character. One might naively think that the correct target
theory for the equivariant Chern character would naturally be
Borel cohomology, as defined in example \ref{Borelex}, with real coefficients. This is not the case, as emphasised in particular by a \emph{completion theorem} of Atiyah and Segal, which we briefly summarize. Any conjugacy classes of an element $\gamma\in{\topsp{G}}$ induces a homomorphism  
\begin{displaymath}
\nu_{\gamma}:{\rm R}(\topsp{G})\to\mathbb{C}
\end{displaymath}
given by $\nu_{\gamma}(\rho):=\chi_{\rho}(\gamma)$, where $\chi_{\rho}$ is the character associated to the representation $\rho$, which is constant on conjugacy classes. The kernel of such a homomorphism is a prime ideal\footnote{A \emph{prime ideal} $\rm P$ in a commutative ring $\rm R$ is an ideal such that whenever the product $ab$ of two elements in $\rm R$ lies in $\rm P$, then $a$ or $b$ lies in $\rm P$. } in the ring ${\rm R}(\topsp{G})$. Let us denote with $\rm I_{G}$ the prime ideal associated to element $e\in{\topsp{G}}$. Then we have \cite{Atiyah1969}
\begin{theorem}\label{completion}
For a G-space $\topsp{X}$, with $\topsp{G}$ a compact Lie group, the Borel cohomology ${\rm H}_{\topsp{G}}^{*}(\topsp{X};\mathbb{Q})$ is isomorphic to the completion of the ${\rm R}(\topsp{G})$-module ${\rm K}^{*}_{\topsp{G}}(\topsp{X})\otimes{\mathbb{Q}}$ with respect to the ideal $\rm I_{G}$.
\end{theorem}
The completion of ${\rm K}^{*}_{\topsp{G}}(\topsp{X};\mathbb{Q}):={\rm K}^{*}_{\topsp{G}}(\topsp{X})\otimes{\mathbb{Q}}$ is defined as the tensor product ${\rm K}^{*}_{\topsp{G}}(\topsp{X};\mathbb{Q})\otimes_{{\rm R}(\topsp{G})}\widehat{\rm R}(\topsp{G})$, where $\widehat{\rm R}(\topsp{G})$ is given by the limit of the quotients ${\rm R}(\topsp{G})/{\rm I}^{n}_{\rm G}\cdot{\rm R}(\topsp{G})$ for $n$ going to infinity.\\
The above theorem suggests then that Borel cohomology is not the correct target theory for a Chern character inducing an isomorphism over $\mathbb{R}$. If we think of $\rm R(\topsp{G})$ as the ring of functions over $\topsp{G}$, the prime ideal $\rm I_{G}$ corresponds to the unit element in $\topsp{G}$. Theorem \ref{completion} then states that Borel cohomology does not take into account ``contributions'' of the non-trivial elements in $\topsp{G}$, and hence, in a sense, it is \emph{localised} around the unit element.\\[2mm]
There are several approaches to the equivariant Chern character (see
refs.~\cite{Atiyah1989,Slominska1976,Block1994,Freed2002a,Adem2003},
for example) which strongly depend on the types of groups involved
(discrete, continuous, \emph{etc.}) and on the ring one tensors with
($\real$, $\complex$, \emph{etc.}). As we are interested in finite
groups and real coefficients,
we will use the Chern character constructed in \cite{Luck2006}
and~\cite{Luck1998}.\\
In the following section we will briefly recall the basic constructions in Bredon cohomology \cite{Bredon,Luck2006,mislin}, which will turn
out to be the best suited equivariant cohomology theory for all of our
purposes. We will refer to {\appCat} for some pertinent aspects of functor categories.
\subsection{Bredon cohomology}
In the following, $\topsp{G}$ will denote a discrete group. The \emph{orbit category}
$\ocat{\topsp{G}}$ of $\topsp{G}$ is defined as the category whose objects are
homogeneous spaces $\topsp{G}/\topsp{H}$, with $\topsp{H}<{\topsp{G}}$, and whose morphisms are
$\topsp{G}$-maps between them. From general considerations~\cite{Dieck1987} it
follows that a $\topsp{G}$-map between two homogeneous spaces $\topsp{G}/\topsp{H}$ and $\topsp{G}/\topsp{K}$
exists if and only if $\topsp{H}$ is conjugate to a subgroup of $\topsp{K}$, and hence
any such map is of the form
\beq
\big(g\,\topsp{H} ~\longmapsto~ g\,a\,\topsp{K}\big)
\label{MapGH}\eeq
for some $a\in{\topsp{G}}$ such that $a^{-1}\,\topsp{H}\,a<{\topsp{K}}$. If
$\mathfrak{F}$ is any family of subgroups of $\topsp{G}$ then there is a
subcategory $\ocat{\topsp{G},\mathfrak{F}}$ with objects $\topsp{G}/\topsp{H}$ for
$\topsp{H}\in\mathfrak{F}$. A simple example is provided by the cyclic groups
$\topsp{G}=\mathbb{Z}_{p}$ with $p$ prime, for which the orbit category has
just two objects, $\topsp{G}/e=\topsp{G}$ and $\topsp{G}/\topsp{G}=\pt$.

If $\cat{Ab}$ denotes the category of abelian groups, then a
\emph{coefficient system} is a functor
\begin{displaymath}
\underline{\topsp{F}}\,:\,\ocat{\topsp{G}}^{\text{op}}~\longrightarrow~\cat{Ab}
\end{displaymath}
where $\ocat{\topsp{G}}^{\rm op}$ denotes the opposite category to
$\ocat{\topsp{G}}$. With such a functor and any $\topsp{G}$-complex $\topsp{X}$,\footnote{When
  $\topsp{G}$ is an infinite discrete group, one should restrict to
  \emph{proper} $\topsp{G}$-complexes, {i.e.}, with finite stabilizer for
  any point of $\topsp{X}$. Some further minor assumptions are needed when $\topsp{G}$
  is a Lie group.} one can define for each $n\in\mathbb{Z}$ the
group
\beq
C^{n}_{\topsp{G}}(\topsp{X},\,\underline{\topsp{F}}\,):=\text{Hom}_{\ocat{\topsp{G}}}\big(\,
\underline{C}\,_{n}(\topsp{X})\,,\,\underline{\topsp{F}}\,\big)
\label{CGnXF}\eeq
where $\underline{C}\,_{n}(\topsp{X}):\ocat{\topsp{G}}^{\rm op}\to\cat{Ab}$ is the
projective functor defined by
$$\underline{C}\,_{n}(\topsp{X})(\topsp{G}/\topsp{H}):=C_{n}\big(\topsp{X}^{\topsp{H}}\big) \ , $$ the cellular
homology of the fixed point complex 
\beq
\topsp{X}^{\topsp{H}}:=\big\{x\in{\topsp{X}}~\big|~h\cdot{x}=x \quad 
\forall{h\in{\topsp{H}}}\big\} \ .
\label{XHdef}\eeq
In equation~(\ref{CGnXF}), $\Hom_{\ocat{\topsp{G}}}(-,-)$ denotes the group of
natural transformations between two contravariant functors, with the
group structure inherited by the images of the functors in
$\cat{Ab}$. The functoriality property of
$\underline{C}\,_{n}(\topsp{X})$ is the natural one induced by the
identification $\topsp{X}^{\topsp{H}}\simeq\text{Map}_{\topsp{G}}(\topsp{G}/\topsp{H},\topsp{X})$. Indeed, the two maps
\begin{eqnarray*}
\topsp{X}^{\topsp{H}}~\longrightarrow~\text{Map}_{\topsp{G}}(\topsp{G}/\topsp{H},\topsp{X}) & , & \qquad
x~\longmapsto~f_{x}\big([g\,\topsp{H}]\big)=g\cdot{x} \ , \\[4pt]
\text{Map}_{\topsp{G}}(\topsp{G}/\topsp{H},\topsp{X})~\longrightarrow~\topsp{X}^{\topsp{H}} & , & \qquad
f~\longmapsto~{f(\topsp{H})}
\end{eqnarray*}
are easily seen to be inverse to each other, and the desired
homeomorphism is obtained by giving the space
$\text{Map}_{\topsp{G}}(\topsp{G}/\topsp{H},\topsp{X})$ the \emph{compact-open} topology. In
particular, a $\topsp{G}$-map (\ref{MapGH}) induces a cellular map $\topsp{X}^\topsp{K}\to
\topsp{X}^\topsp{H}$, $x\mapsto a\cdot x$.

These groups can be expressed in terms of the $\topsp{G}$-complex structure of
$\topsp{X}$. If the $n$-skeleton $\topsp{X}_{n}$ is obtained by attaching equivariant
cells as in equation~(\ref{Xnattach}) with $\topsp{K}_j$ the stabilizer of an
$n$-cell of $\topsp{X}$, then the cellular chain complex $C_*(\topsp{X})$
consists of $\topsp{G}$-modules $C_n(\topsp{X})=\bigoplus_{j\in J_n}\,\zed[\topsp{G}/\topsp{K}_j]$ and
hence
\begin{displaymath}
\underline{C}^{}\,_{n}(\topsp{X})(\topsp{G}/\topsp{H})\simeq\bigoplus_{j\in J_n}\,\mathbb{Z}\big[
\text{Mor}_{\ocat{\topsp{G}}}(\topsp{G}/\topsp{H},\topsp{G}/\topsp{K}_j)\big] \ .
\end{displaymath}
For each $n\geq0$, the group $C_\topsp{G}^n(\topsp{X}\,,\,\underline{\topsp{F}}\,)$ is the
direct limit functor over all $n$-cells of orbit type $\topsp{G}/\topsp{K}_j$ in $\topsp{X}$
of the groups $\underline{\topsp{F}}\,(\topsp{G}/\topsp{K}_j)$. This follows by restricting
equation~(\ref{CGnXF}) to the full subcategory $\ocat{\topsp{G},\mathfrak{F}(\topsp{X})}$,
with $\mathfrak{F}(\topsp{X})$ the family of subgroups of $\topsp{G}$ which occur as
stabilizers of the $\topsp{G}$-action on $\topsp{X}$~\cite{mislin}.

The $\zed$-graded group
$C^{*}_{\topsp{G}}(\topsp{X},\,\underline{\topsp{F}}\,)=\bigoplus_{n\in\mathbb{Z}}\,
C^{n}_{\topsp{G}}(\topsp{X},\,\underline{\topsp{F}}\,)$ inherits a coboundary operator
$\delta$, and hence the structure of a cochain complex,
from the boundary operator on cellular chains. To a natural
transformation $f:\underline{C}\,_{n}(\topsp{X})\to{\underline{\topsp{F}}\,}$, one
associates the natural transformation $\delta{f}$ defined by
\begin{eqnarray*}
\delta{f}(\topsp{G}/\topsp{H})\,:\,C_{n}\big(\topsp{X}^{\topsp{H}}\big)&\longrightarrow&
{\underline{\topsp{F}}\,(\topsp{G}/\topsp{H})}\\
\sigma&\longmapsto&f(\topsp{G}/\topsp{H})(\partial\sigma)
\end{eqnarray*}
for $\sigma\in{C_{n-1}(\topsp{X}^{\topsp{H}})}$, with naturality induced from that of
the cellular boundary operator $\partial$. Then the \emph{Bredon
  cohomology} of $\topsp{X}$ with coefficient system $\underline{\topsp{F}}$ is
defined as
\begin{displaymath}
\H^{*}_{\topsp{G}}(\topsp{X};\,\underline{\topsp{F}}\,):=\text{H}\big(C^{*}_{\topsp{G}}(\topsp{X},\,
\underline{\topsp{F}}\,)\,,\,\delta\big) \ .
\end{displaymath}
This defines a $\topsp{G}$-cohomology theory. See \cite{Luck2002} for the
proof that $\H^{*}_{\topsp{G}}(\topsp{X};\,\underline{\topsp{F}}\,)$ is an equivariant
cohomology theory, {i.e.}, for the definition of the induction
structure. One can also define cohomology groups by restricting the
functors in equation~(\ref{CGnXF}) to a subcategory
$\ocat{\topsp{G},\mathfrak{F}}$. The definition of Bredon cohomology is
independent of $\mathfrak{F}$ as long as $\mathfrak{F}$ contains the
family $\mathfrak{F}(\topsp{X})$ of stabilizers~\cite{mislin}. This fact
is useful in explicit calculations. In particular, by taking
$\mathfrak{F}=\topsp{H}$ to consist of a single subgroup, one shows that the
Bredon cohomology of $\topsp{G}$-homogeneous spaces is given by
\beq
\H^*_{\topsp{G}}(\topsp{G}/\topsp{H};\,\underline{\topsp{F}}\,)\=
\H^0_{\topsp{G}}(\topsp{G}/\topsp{H};\,\underline{\topsp{F}}\,)\=\underline{\topsp{F}}\,(\topsp{G}/\topsp{H}) \ .
\label{BredonGH}\eeq
\begin{example}[\emph{Trivial group}] 
When $\topsp{G}={e}$ is the trivial group, {i.e.}, in the non-equivariant
case, the functors $\underline{C}\,_{n}(\topsp{X})$ and $\underline{\topsp{F}}$ can be
identified with the abelian groups
$C_{n}(\topsp{X})=\underline{C}\,_{n}(\topsp{X})(e)$ and $\topsp{F}=\underline{\topsp{F}}\,(e)$. Then
$$C^{n}_{e}(\topsp{X},\topsp{F})=C^{n}(\topsp{X},\topsp{F})$$ and one has
$\H^{n}_{e}(\topsp{X};\,\underline{\topsp{F}}\,)=\text{H}\left(C^{n}(\topsp{X},\topsp{F}),\delta\right)$,
{i.e.}, the ordinary $n$-th cohomology group of $\topsp{X}$ with
coefficients in $\topsp{F}$.
\label{Bredontrivgpex}\end{example}

\begin{example}[\emph{Free action}]
If the $\topsp{G}$-action on $\topsp{X}$ is \emph{free}, then all
stabilizers $\topsp{K}_j$ are trivial and $\topsp{X}^\topsp{H}=\emptyset$ for every $\topsp{H}\leq \topsp{G}$,
$\topsp{H}\neq e$. In this case one may take $\mathfrak{F}=e$ to compute the
cochain complex
$$
C_\topsp{G}^*(\topsp{X},\,\underline{\topsp{F}}\,)\simeq\Hom_\topsp{G}\big(C_*(\topsp{X})\,,\,
\underline{\topsp{F}}\,(\topsp{G}/e)\big)
$$
and so the Bredon cohomology $\H^{*}_{\topsp{G}}(\topsp{X};\,\underline{\topsp{F}}\,)$
coincides with the equivariant cohomology
$$\H_\topsp{G}^*\big(\topsp{X}\,;\,\underline{\topsp{F}}\,(\topsp{G}/e)\big)$$ of $\topsp{X}$ with
coefficients in the $\topsp{G}$-module
$\underline{\topsp{F}}\,(\topsp{G}/e)=\underline{\topsp{F}}\,(\topsp{G})$. In the case of the constant
functor $\underline{\topsp{F}}=\underline{\zed}$, with
$\underline{\zed}\,(\topsp{G}/\topsp{H})=\zed$ for every $\topsp{H}\leq \topsp{G}$ and the value on
morphisms in $\ocat{\topsp{G}}^{\rm op}$ given by the identity homomorphism of
$\zed$, this group reduces to the ordinary cohomology
$\H^*(\topsp{X}/\topsp{G};\zed)$.
\label{Bredonfreeactionex}\end{example}

\begin{example}[\emph{Trivial action}]
If the $\topsp{G}$-action on $\topsp{X}$ is \emph{trivial}, then the collection of
isotropy groups $\topsp{K}_j$ for the $\topsp{G}$-action is the set of all subgroups
of $\topsp{G}$ and $\topsp{X}^\topsp{H}=\topsp{X}$ for all $\topsp{H}\leq \topsp{G}$. In this case the functor
$\underline{C}\,_n(\topsp{X})$ can be decomposed into a sum over $n$-cells of
projective functors $\underline{P}\,_{\topsp{K}_j}$ with
$\topsp{K}_{j}=\topsp{G}$~\cite{mislin}, and so one has
$$
\Hom_{\ocat{\topsp{G}}}\big(\,\underline{C}\,_n(\topsp{X})\,,\,\underline{\topsp{F}}\,\big)\simeq
\Hom\big(C_n(\topsp{X})\,,\,\lim_{\longleftarrow}{}_{\ocat{\topsp{G}}^{\rm op}}\,
\underline{\topsp{F}}\,(\topsp{G}/\topsp{H})\big)
$$
where the inverse limit functor is taken over the opposite category
$\ocat{\topsp{G}}^{\rm op}$. It follows that the Bredon cohomology
$$
\H^*_\topsp{G}(\topsp{X};\,\underline{\topsp{F}}\,)=\H^*\big(\topsp{X}\,;\,
\underline{\topsp{F}}\,(\topsp{G}/\topsp{G})\big)
$$
is the ordinary cohomology of $\topsp{X}$ with coefficients in the abelian
group $\underline{\topsp{F}}\,(\topsp{G}/\topsp{G})=\underline{\topsp{F}}\,(\pt)$.
\label{Bredontrivactionex}\end{example}
We will now specialize the coefficient system for Bredon
cohomology to the
\emph{representation ring functor} $\underline{{\rm R}}(-)$
defined on the orbit category $\ocat{\topsp{G}}$ by sending the left coset
$\topsp{G}/\topsp{H}$ to ${\rm R}(\topsp{H})$, the representation ring of the group
$\topsp{H}$. A morphism (\ref{MapGH}) is sent to the
homomorphism ${\rm R}(\topsp{K}) \to {\rm R}(\topsp{H})$ given by first restricting the
representation from $\topsp{K}$ to the subgroup conjugate to $\topsp{H}$, and then
conjugating by $a$. Since $\underline{{\rm R}}(-)$ is a functor to rings,
the Bredon cohomology $\topsp{H}^{*}_{\topsp{G}}(\topsp{X};\underline{{\rm R}}(-))$ naturally
has a ring structure.\\
By equation (\ref{BredonGH}), we have
\begin{displaymath}
\H^*_{\topsp{G}}(\topsp{G}/\topsp{H};\,\underline{\topsp{R}}(-)\,)\={\underline{\rm R}(\topsp{G}/\topsp{H})}\={\rm R}(\topsp{H})\={\rm K}_{\topsp{G}}^{*}(\topsp{G}/\topsp{H})
\end{displaymath}
which is already an indication that Bredon cohomology is a better
relative of equivariant K-theory than Borel cohomology. Indeed, using
the induction structure of Example~\ref{Borelex} one shows that the
Borel cohomology
$$
\H_\topsp{G}^*(\topsp{G}/\topsp{H})=\H^*(\topsp{BH})
$$
coincides with the cohomology of the classifying space $\topsp{BH}=\topsp{EH}/\topsp{H}$,
which computes the group cohomology of $\topsp{H}$ and is typically
infinite-dimensional (even for finite groups $\topsp{H}$). 

In the construction of the equivariant Chern character in the next section, it will be important to represent
the rational Bredon cohomology
$\H^{*}_{\topsp{G}}(\topsp{X};\mathbb{Q}\otimes{\underline{{\rm R}}(-)})$ as a certain
group of homomorphisms of functors, similarly to the cochain groups
(\ref{CGnXF}). For this, we introduce another
category $\subc{\topsp{G}}$. The objects of $\subc{\topsp{G}}$ are the
subgroups of $\topsp{G}$,\footnote{If $\topsp{G}$ is infinite then one should restrict
  to finite subgroups of $\topsp{G}$.} and the morphisms are given by
\begin{displaymath}
\text{Mor}_{\subc{\topsp{H},\topsp{K}}}:=\left\{f:\topsp{H}\to{\topsp{K}}~\big|~\exists\:{g\in{\topsp{G}} \ ,
    \ g\,\topsp{H}\,g^{-1}\leq{\topsp{K}}}\ , \ f={\rm
    Ad}_g\right\}\,\big/\,\text{Inn}(\topsp{K}) \ .
\end{displaymath}
In particular, there is a functor $\ocat{\topsp{G}}\to{\subc{\topsp{G}}}$ which sends
the object $\topsp{G}/\topsp{H}$ to $\topsp{H}$ and the morphism (\ref{MapGH}) in $\ocat{\topsp{G}}$
to the homomorphism $(g\mapsto{a^{-1}\,g\,a})$ in $\subc{\topsp{G}}$. If
$a$ lies in the centralizer
\beq
{Z_{\topsp{G}}(\topsp{H})}:=\big\{g\in \topsp{G}~\big|~g^{-1}h\,g=h, \forall\:h\in{\topsp{H}}\big\}
\label{ZGH}\eeq
 of $\topsp{H}$ in $\topsp{G}$, then the morphism (\ref{MapGH}) is sent to the
identity map. Any functor
$\underline{\topsp{F}}:\subc{\topsp{G}}^{\text{op}}\to\cat{Ab}$ can be naturally
regarded as a functor on $\ocat{\topsp{G}}^{\text{op}}$.

Define the quotient functors
$\underline{C}\,_{*}^{\text{qt}}(\topsp{X})\,,\,
\underline{\H}\,_{*}^{\text{qt}}(\topsp{X}):\subc{\topsp{G}}^{\rm
  op}\to\cat{Ab}$ by
\begin{displaymath}
\underline{C}\,_{*}^{\text{qt}}(\topsp{X})(\topsp{H})~:=~
C_{*}\big(\topsp{X}^{\topsp{H}}/Z_{\topsp{G}}(\topsp{H})\big) \qquad\text{and}\qquad
\underline{\H}\,_{*}^{\text{qt}}(\topsp{X})(\topsp{H})~:=~\H_{*}
\big(\topsp{X}^{\topsp{H}}/Z_{\topsp{G}}(\topsp{H})\big) \ .
\end{displaymath}
For any functor $\underline{\topsp{F}}:\subc{\topsp{G}}^{\text{op}}\to\cat{Ab}$
one has
\begin{displaymath}
\text{Hom}\big(C_{*}(\topsp{X}^{\topsp{H}}/Z_{\topsp{G}}(\topsp{H}))\,,\,\underline{\topsp{F}}\,(\topsp{H})\big)
\simeq\text{Hom}_{Z_{\topsp{G}}(\topsp{H})}\big(C_{*}(\topsp{X}^{\topsp{H}})\,,\,\underline{\topsp{F}}\,(\topsp{H})
\big) \ .
\end{displaymath}
By observing that the centralizer (\ref{ZGH}) is precisely the group of
automorphisms of $\topsp{G}/\topsp{H}$ in the orbit category $\ocat{\topsp{G}}$ sent to the
identity map in the subgroup category $\subc{\topsp{G}}$, we finally
have
\beq
C_\topsp{G}^*(\topsp{X},\,\underline{\topsp{F}}\,)\=
\text{Hom}_{\ocat{\topsp{G}}}\big(\,\underline{C}\,_{*}(\topsp{X})\,,\,
\underline{\topsp{F}}\,\big)~\simeq~\text{Hom}_{\subc{\topsp{G}}}\big(\,
\underline{C}\,_{*}^{\text{qt}}(\topsp{X})\,,\,\underline{\topsp{F}}\,\big) \ .
\label{HomOrSubG}\eeq
At this point one can apply equation~(\ref{HomOrSubG}) to the rational
representation ring functor
$\underline{\topsp{F}}=\rat\otimes\underline{{\rm R}}(-)$, which by construction can
be regarded as an injective functor $\subc{\topsp{G}}^{\rm op}\to\cat{Ab}$, to
prove the
\begin{lemma}[\cite{Luck1998}]\label{subcat}
For any finite group $\topsp{G}$ and any $\topsp{G}$-complex $\topsp{X}$, there exists an
isomorphism of rings
\begin{displaymath}
\Phi_{\topsp{X}}\,:\,\H^{*}_{\topsp{G}}\big(\topsp{X}\,;\,\mathbb{Q}
\otimes{\underline{{\rm R}}(-)}\big)~\xrightarrow{\approx}~{\Hom}_{{\sf
    Sub}(\topsp{G})}\big(\,\underline{\H}\,_{*}^{\rm qt}(\topsp{X})\,,\,\mathbb{Q}
\otimes\underline{{\rm R}}(-)\big) \ .
\end{displaymath}
\end{lemma}
\subsection{Chern character in equivariant K-theory}\label{equivChern}
Before spelling out the definition of the equivariant Chern character given in \cite{Luck1998}, we recall some basic properties of the equivariant K-theory of a
$\topsp{G}$-complex $\topsp{X}$. Let $\topsp{H}$ be a subgroup of $\topsp{G}$, and consider the fixed
point subspace of $\topsp{X}$ defined in (\ref{XHdef}). The action of $\topsp{G}$
does not preserve $\topsp{X}^{\topsp{H}}$, but the action of the normalizer $\topsp{N}_\topsp{G}(\topsp{H})$ of
$\topsp{H}$ in $\topsp{G}$ does. If we denote with $i:\topsp{X}^{\topsp{H}}\hookrightarrow{\topsp{X}}$ the
inclusion of $\topsp{X}^{\topsp{H}}$ as a subspace of $\topsp{X}$, and with
$\alpha:\topsp{N}_\topsp{G}(\topsp{H})\hookrightarrow{\topsp{G}}$ the inclusion of $\topsp{N}_\topsp{G}(\topsp{H})$ as a
subgroup of $\topsp{G}$, then we naturally have the equality
\begin{displaymath}
i(n\cdot{x})=\alpha(n)\cdot{i(x)}
\end{displaymath}
for all ${n\in{\topsp{N}_\topsp{G}(\topsp{H})}}$ and $x\in \topsp{X}^\topsp{H}$. It follows that the induced
homomorphism on equivariant K-theory is a map~\cite{Segal1968}
\begin{displaymath}
i^{*}\,:\,\K^*_{\topsp{G}}(\topsp{X})~\longrightarrow~{\K^*_{\topsp{N}_\topsp{G}(\topsp{H})}\big(\topsp{X}^{\topsp{H}}
\big)}
\end{displaymath}
which is called a \emph{restriction morphism}.

We also  need a somewhat less known property~\cite{Luck1998}. Let
$\topsp{N}<\,{\topsp{G}}$ be a finite normal subgroup, and let $\text{Rep}(\topsp{N})$ be the
category of (isomorphism classes of) irreducible complex representations of
$\topsp{N}$. Let $\topsp{X}$ be a (proper) $\topsp{G}/\topsp{N}$-complex, and let $\topsp{G}$ act on $\topsp{X}$ via
the projection map $\topsp{G}\to{\topsp{G}/\topsp{N}}$. Then for any complex $\topsp{G}$-vector bundle
$\topsp{E}\to{\topsp{X}}$ and any representation $\topsp{V}\in{\text{Rep}(\topsp{N})}$, define
$\text{Hom}_{\topsp{N}}(\topsp{V},\topsp{E})$ as the vector bundle over $\topsp{X}$ with total space
\begin{displaymath}
\text{Hom}_{\topsp{N}}(\topsp{V},\topsp{E}):=\bigcup_{x\in{\topsp{X}}}\,\text{Hom}_{\topsp{N}}(\topsp{V},\topsp{E}_{x})
\end{displaymath}
where $\topsp{N}$ acts on the fibres of $\topsp{E}$ because of the action of $\topsp{G}$ via
the projection map. Now if $\topsp{H}\leq{\topsp{G}}$ is a subgroup which commutes
with $\topsp{N}$, $[\topsp{H},\topsp{N}]=e$, then one can induce an $\topsp{H}$-vector bundle from
$\text{Hom}_{\topsp{N}}(\topsp{V},\topsp{E})$ by defining $(h\cdot f)(v)=h\cdot{f(v)}$, $v\in
\topsp{V}$ for any $h\in{\topsp{H}}$ and any $f\in\text{Hom}_{\topsp{N}}(\topsp{V},\topsp{E})$ (remembering
that $\topsp{G}$ acts on $\topsp{E}$). Hence there is a homomorphism of rings
\begin{displaymath}
\Psi\,:\,\K^*_{\topsp{G}}(\topsp{X})~\longrightarrow~{\K^*_{\topsp{H}}(\topsp{X})\otimes{\topsp{R}(\topsp{N})}}
\end{displaymath}
defined on $\topsp{G}$-vector bundles by
\beq
\Psi\big([\topsp{E}]\big):=\sum_{\topsp{V}\in{{\rm Rep}(\topsp{N})}}\,\big[\text{Hom}_{\topsp{N}}(\topsp{V},\topsp{E})
\big]\otimes[\topsp{V}] \ .
\label{PsiEdef}\eeq
This homomorphism satisfies some naturality properties; see \cite{Luck1998}. Note that the sum
(\ref{PsiEdef}) is \emph{finite}, since $\topsp{N}$ is a finite subgroup.

We are now ready to construct the equivariant Chern character as a
homomorphism
\begin{displaymath}
\ch_{\topsp{X}}\,:\,\K^{0,1}_{\topsp{G}}(\topsp{X})~\longrightarrow~{\topsp{H}^{\rm even,odd}_{\topsp{G}}
\big(\topsp{X}\,;\,\mathbb{Q}\otimes\underline{\topsp{R}}(-)\big)}
\end{displaymath}
for any finite proper $\topsp{G}$-complex $\topsp{X}$. The strategy used
in~{}~\cite{Luck1998} is to construct $\zed_2$-graded
homomorphisms
\beq
\ch^{\topsp{H}}_{\topsp{X}}\,:\,\K^{*}_{\topsp{G}}(\topsp{X})~\longrightarrow~
{\text{Hom}\big(\H_{*}(\topsp{X}^{\topsp{H}}/Z_{\topsp{G}}(\topsp{H}))\,,\,\mathbb{Q}\otimes
{\topsp{R}}(\topsp{H})\big)}
\label{chXH}\eeq
for any finite subgroup $\topsp{H}$, and then \emph{glue} them together as $\topsp{H}$
varies through the finite subgroups of $\topsp{G}$. To define the homomorphism
(\ref{chXH}), we first compose the ring homomorphisms
\begin{displaymath}
\K^{*}_{\topsp{G}}(\topsp{X})~\xrightarrow{i^{*}}~\K_{\topsp{N}_\topsp{G}(\topsp{H})}^*\big(\topsp{X}^{\topsp{H}}
\big)~\xrightarrow{\Psi}~\K^{*}_{Z_{\topsp{G}}(\topsp{H})}\big(\topsp{X}^{\topsp{H}}\big)
\otimes{\topsp{R}(\topsp{H})}~\xrightarrow{\pi^{*}_{2}\otimes\Id}~\K^{*}_{Z_{\topsp{G}}(\topsp{H})}
\big(\topsp{EG}\times{\topsp{X}^{\topsp{H}}}\big)\otimes{\topsp{R}(\topsp{H})}
\end{displaymath}
where $\pi_2:\topsp{EG}\times \topsp{X}^\topsp{H}\to \topsp{X}^\topsp{H}$ is the projection onto the second
factor. By using the induction structure of Example~\ref{EqKex}, one
then has
\bea
\K^{*}_{Z_{\topsp{G}}(\topsp{H})}\big(\topsp{EG}\times{\topsp{X}^{\topsp{H}}}\big)\otimes{\topsp{R}(\topsp{H})}&
\xrightarrow{\approx}&\K^{*}\big(\topsp{EG}\times_{Z_{\topsp{G}}(\topsp{H})}\topsp{X}^{\topsp{H}}
\big)\otimes \topsp{R}(\topsp{H}) \nonumber\\ && \qquad ~
\xrightarrow{\ch\otimes\text{id}}~
\H^*\big(\topsp{EG}\times_{Z_{\topsp{G}}(\topsp{H})}\topsp{X}^{\topsp{H}}\,;\,\mathbb{Q}\big)
\otimes{\topsp{R}(\topsp{H})} \nonumber
\eea
where $\ch$ is the ordinary Chern character. One finally has
\bea
\H^{*}\big(\topsp{EG}\times_{Z_{\topsp{G}}(\topsp{H})}\topsp{X}^{\topsp{H}}\,;\,\mathbb{Q}\big)
\otimes{\topsp{R}(\topsp{H})}&\xrightarrow{\approx}&\H^{*}\big(\topsp{X}^{\topsp{H}}/Z_{\topsp{G}}(\topsp{H})\,;\,
\mathbb{Q}\big)\otimes{\topsp{R}(\topsp{H})} \nonumber\\ && \qquad ~\simeq~
\text{Hom}\big(\H_{*}(\topsp{X}^{\topsp{H}}/Z_{\topsp{G}}(\topsp{H}))\,,\,
\mathbb{Q}\otimes{\topsp{R}(\topsp{H})}\big) \ , \nonumber
\eea
where the first isomorphism follows from the Leray spectral sequence
by observing that the fibres of the projection
\begin{displaymath}
\topsp{EG}\times_{Z_{\topsp{G}}(\topsp{H})}\topsp{X}^{\topsp{H}}~\longrightarrow~{\topsp{X}^{\topsp{H}}\,\big/\,Z_{\topsp{G}}(\topsp{H})}
\end{displaymath}
are all classifying spaces of finite groups, having trival reduced
cohomology with $\rat$-coefficients and are therefore $\rat$-acyclic.

The equivariant Chern character is now defined as\footnote{If $\topsp{G}$ is
  infinite then the direct sum in equation~(\ref{chXprod}) is understood as
  the inverse limit functor over the dual subgroup category
  $\subc{\topsp{G}}^{\rm op}$.}
\beq
\ch_{\topsp{X}}=\bigoplus_{\topsp{H}\leq \topsp{G}}\,\ch_\topsp{X}^{\topsp{H}} \ .
\label{chXprod}\eeq
By using the various naturality properties of the
homomorphism~(\ref{PsiEdef})~\cite{Luck1998}, one sees that $\ch_{\topsp{X}}$
takes values in
$\text{Hom}_{\subc{\topsp{G}}}\big(\underline{\H}\,_{*}^{\text{qt}}(\topsp{X})\,,\,
\mathbb{Q}\otimes\underline{\topsp{R}}(-)\big)$, and by Lemma~\ref{subcat} it
is thus a $\zed_2$-graded map
\begin{displaymath}
\ch_{\topsp{X}}\,:\,\K^{*}_{\topsp{G}}(\topsp{X})~\longrightarrow~\textrm{Hom}_{\subc{\topsp{G}}}
\big(\,\underline{\H}\,_{*}^{\textrm{qt}}(\topsp{X})\,,\,\mathbb{Q}\otimes
\underline{\topsp{R}}(-)\big)\simeq{\H^{*}_{\topsp{G}}\big(\topsp{X}\,;\,
\mathbb{Q}\otimes{\underline{\topsp{R}}(-)}\big)} \ .
\end{displaymath}
This map is well-defined as a ring homomorphism because all maps
involved above are homomorphisms of rings. As with the definition of
Bredon cohomology, the sum (\ref{chXprod}) may be restricted to any
family of subgroups of $\topsp{G}$ containing the set of stabilizers
$\mathfrak{F}(\topsp{X})$.

To conclude, we have to prove that this map becomes an isomorphism
upon tensoring over $\rat$. For this, one proves that the morphism
$\ch_\topsp{X}$ in equation~(\ref{chXprod}) is an isomorphism on homogeneous spaces
$\topsp{G}/\topsp{H}$, with $\topsp{H}$ a finite subgroup of $\topsp{G}$, and then uses induction on
the number of orbit types of cells in $\topsp{X}$ along with the
Mayer-Vietoris sequences for the pushout squares induced by the
attaching $\topsp{G}$-maps (\ref{Gattach}). The isomorphism on $\topsp{G}/\topsp{H}$ is a
consequence of the isomorphisms (\ref{BredonGH}) and
(\ref{point}). The details may be found
in~{}~\cite{Luck1998}. Let $\underline{\pi_{-*}\K}\,_\topsp{G}(-)$ be
the functor on $\ocat{\topsp{G}}$ defined by
$\topsp{G}/\topsp{H}\mapsto\K^*_\topsp{G}(\topsp{G}/\topsp{H})$. Then one has the following
\begin{theorem}{\rm (\cite{Luck1998})}
For any finite proper $\topsp{G}$-complex $\topsp{X}$, the Chern character $\ch_\topsp{X}$
extends to a natural $\zed$-graded isomorphism of rings
\begin{displaymath}
{\ch_{\topsp{X}}}\otimes\mathbb{Q}\,:\,{\K^{*}_{\topsp{G}}(\topsp{X})}\otimes
\mathbb{Q}~\xrightarrow{\approx}~\H_{\topsp{G}}\big(\topsp{X}\,;\,\mathbb{Q}
\otimes{\underline{\pi_{-*}\K}\,_\topsp{G}(-)}\big)^{*} \ .
\end{displaymath}
\label{eqChernthm}\end{theorem}
\section{String theory on orbifolds}
The techniques of equivariant cohomology and K-theory illustrated in the previous section play an important role in understanding the behaviour of String theory defined on orbifolds. As mentioned in the introduction to this chapter, we will be interested in good orbifolds, which are obtained as the orbit space of the action of a finite group $\topsp{G}$ on a smooth manifold $\topsp{X}$. In particular, the action of $\topsp{G}$ will be isometric, proper, and cocompact\footnote{An action of a group $\topsp{G}$ on a space $\topsp{X}$ is said to be \emph{cocompact} if the orbit space $\topsp{X}/\topsp{G}$ is compact}. It is know that when $\topsp{G}$ acts on $\topsp{X}$ with nontrivial stabilizers, the orbit space $\topsp{X}/\topsp{G}$ cannot be given a differential structure such that the usual projection $\pi:\topsp{X}\to{\topsp{X}/\topsp{G}}$ is a smooth map. In the case in which all the stabilizers are trivial, i.e. the group $\topsp{G}$ acts freely, the orbit space naturally carries a manifold structure. We will not attempt to give a definition of nonglobal orbifolds, since we will work in the equivariant ``regime''. We direct instead the reader to the seminal paper \cite{satake} for a local description of orbifolds, and to \cite{lupercio} for a modern description in terms of groupoids.\\
The quantum behaviour of String theory on an orbifold $[\topsp{X}/\topsp{G}]$ is different from that of a quantum particle, as first realized in \cite{Dixon1985,Dixon1986}. Indeed, suppose that $\topsp{G}$ acts freely on $\topsp{X}$. To describe the quantum mechanics of a point particle propagating on the smooth manifold $\topsp{X}/\topsp{G}$, one could think of first contructing the Hilbert space of states for a particle on the manifold $\topsp{X}$, and then restrict to the Hilbert subspace of $\topsp{G}$ invariant states\footnote{A choice of a ``lift'' of the action of $\topsp{G}$ on the internal degrees of freedom, should be made, if possible. In other words, the vector bundle whose the wave function is a section of must be $\topsp{G}$-equivariant.}. Following the same logic, the first step to the quantum string propagation on $\topsp{X}/\topsp{G}$ consists in constructing the Hilbert space $\mathcal{H}_{0}$ for a string propagating on $\topsp{X}$, and restrict to the $\topsp{G}$-invariant states. In contrast to the particle case, this is not yet a complete Hilbert space of states. Indeed, $\mathcal{H}_{0}$ does not contain states of strings which are closed on the quotient manifold $\topsp{X}/\topsp{G}$, but are only modulo a $\topsp{G}$ transformation. More precisely, consider an embedding $f:[0,1]\times\mathbb{R}\to{\topsp{X}}$ of the worldsheet strip, with local coordinates $(\sigma,\tau)\in[0,1]\times\mathbb{R}$. The open strings obeying
\begin{equation}\label{twisted}
f(\sigma+2\pi,\tau)=h\cdot{f(\sigma,\tau)}
\end{equation}
for some $h\in{\topsp{G}}$ are closed on the quotient $\topsp{X}/\topsp{G}$, since the point $x$ and $h\cdot{x}$ are identified\footnote{In superstring theory $\topsp{X}$ is a $\rm \topsp{G}-spin^{c}$ manifold, and an analogous condition to (\ref{twisted}) should be imposed on the worldsheet fermion fields.}. Hence, the Hilbert space $\mathcal{H}$ for a quantum closed string propagating on $\topsp{X}/\topsp{G}$ is given by 
\begin{equation}\label{Hilbertsp}
\mathcal{H}:=\left(\bigoplus_{h\in{\topsp{G}}}\mathcal{H}_{h}\right)^{\topsp{G}}
\end{equation}
where the sector $\mathcal{H}_{h}$ is given by the space of states of an open string satisfying condition (\ref{twisted}).\\
At this point, notice that the action of $\topsp{G}$ permutes the sectors in the conjugacy classes of the associated element $h$. Indeed, for any $f$ satisfying condition (\ref{twisted}), for any $g\in{\topsp{G}}$ we have
\begin{equation}\label{twisting}
g\cdot{f(\sigma+2\pi,\tau)}=gh\cdot{f(\sigma,\tau)}=(ghg^{-1})g\cdot{f(\sigma,\tau)}
\end{equation}
We can then define
\begin{displaymath}
\mathcal{H}_{[h]}:=\bigoplus_{l\in[h]}\mathcal{H}_{l}=\bigoplus_{i=1}^{n_{h}}\mathcal{H}_{p_{i}hp_{i}^{-1}}
\end{displaymath}
where $n_{h}$ is the number of element in the conjugacy class $[h]$, and $\{p_{i}\}$ is an appropriate set of elements of $\topsp{G}$. Any element $\xi_{[h]}$ can be expressed as
\begin{displaymath}
\xi_{[h]}=(\xi_{h},\xi_{p_{1}hp_{1}^{-1}},\cdots)
\end{displaymath}
and clearly the action of $\topsp{G}$ preserves the vector space $\mathcal{H}_{[h]}$. Since we are interested in $\topsp{G}$-invariant states, we have only to consider the action of elements $g\in{Z_{\topsp{G}}(h)}$, since these are the only elements of $\topsp{G}$ which do not permute the sectors. More precisely, if we define the action of $g\in{Z_{\topsp{G}}(h)}$ on $\mathcal{H}_{[h]}$ as
\begin{displaymath}
g\cdot\xi_{[h]}:=(g\xi_{h},p_{1}gp_{1}^{-1}\xi_{p_{1}hp_{1}^{-1}},\cdots)
\end{displaymath}
we have 
\begin{displaymath}
\mathcal{H}\simeq{\bigoplus_{[h]}\mathcal{H}_{[h]}^{Z_{\topsp{G}}(h)}}
\end{displaymath}
The subspaces $\mathcal{H}_{[h]}$ associated to a nontrivial conjugacy class are called \emph{twisted sectors}.
The above construction of the Hilbert space of closed strings also ensures that the theory is modular invariant \cite{Dixon1985,Dixon1986}. One expects that these twistor sectors will appear also in the case of a G-action with fixed points, and that the Hilbert space of states can be constructed as above. Moreover, in contrast with ordinary quantum field theory, String theory is usually well defined on the singular orbifold points.\\ 
As it is expected, the twisted sectors will play a role in the behaviour of the low-energy limit of type II orbifold String theory and of D-branes. This will be illustrated in the following sections.\\
We conclude this section with a basic result that will be constantly used later on.
\begin{theorem}
Let $\topsp{G}$ be a finite group acting via isometries on a smooth Riemannian manifold $\topsp{X}$. Then the set
\begin{displaymath}
\topsp{X}^{g}:=\{x\in{\topsp{X}}:g\cdot{x}=x\}
\end{displaymath}
is naturally a (possibly disconnected) submanifold of $\topsp{X}$, for any $g\in{\topsp{G}}$.
\end{theorem}
\begin{proof}
Let $\bar{x}$ be a fixed point for the action of $g$. The pushforward $g_{*}$ acts linearly on the vector space ${\rm T}_{\bar{x}}\topsp{X}$: denote with $K\subset{\rm T}_{\bar{x}}\topsp{X}$ the space which is left fixed by $g_{*}$. Then the exponential map $\exp:{\rm T}_{\bar{x}}\topsp{X}\to{\topsp{X}}$ maps diffeomorphically the subspace $K$ on the fixed points of $g$, since $g$ acts by isometries, and hence we can use this coordinate system to define local charts for $\topsp{X}^{g}$. 
\end{proof}
\subsection{D-branes and equivariant K-cycles}\label{equivkhom}
In this section we will make some remarks concerning the topological
classification of D-branes and their charges on global orbifolds of
Type~II superstring theory with vanishing $H$-flux. As proposed by Witten in \cite{Witten1998} and emphasised in \cite{Olsen1999,garciacompean-1999-557}, Ramond-Ramond charges in type II String theory on a global orbifold $[\topsp{X}/\topsp{G}]$ are
classified by the equivariant K-theory $\K_\topsp{G}^*(\topsp{X})$ of spacetime. The arguments are essentially the same as those presented in chapter 3, hence we will avoid their restatement. Instead, we will show how equivariant K-homology $\K_*^\topsp{G}(\topsp{X})$, dual to equivariant K-theory, leads to a description of \emph{fractional D-branes} in terms of equivariant K-cycles. In the following we will refer to {\appequiv} for the definition of equivariant K-homology, both geometric and equivariant.\\
Similarly to K-homology, the cycles for equivariant K-homology, called $\topsp{G}$-equivariant K-cycles, live in an additive category
$\cat{D}^\topsp{G}(\topsp{X})$ whose objects are triples $(\topsp{W},\topsp{E},f)$ where $\topsp{W}$ is a
$\topsp{G}$-spin$^c$ manifold without boundary, $\topsp{E}$ is a $\topsp{G}$-vector bundle
over $\topsp{W}$, and
\beq\label{DGmap}
f\,:\,\topsp{W}~\longrightarrow~ \topsp{X}
\eeq
is a $\topsp{G}$-map. The group $\K_*^\topsp{G}(\topsp{X})$ is the quotient of this
category by the equivalence relation generated by bordism, direct sum,
and vector bundle modification, as detailed in Appendix~B. Note that
$\topsp{W}$ need not be a submanifold of spacetime. However, since $\topsp{X}$ is a
manifold, we can restrict the bordism equivalence relation to
\emph{differential bordism} and assume that the map
(\ref{DGmap}) is a differentiable $\topsp{G}$-map in equivariant K-cycles
$(\topsp{W},\topsp{E},f)\in\cat{D}^\topsp{G}(\topsp{X})$. In this way the category $\cat{D}^\topsp{G}(\topsp{X})$
extends the standard K-theory classification to include branes
supported on non-representable cycles in spacetime. This definition of
equivariant K-homology thus gives a concrete geometric model for the
topological classification of D-branes $(\topsp{W},\topsp{E},f)$ in a global orbifold
$[\topsp{X}/\topsp{G}]$ which captures the physical constructions of orbifold D-branes
as $\topsp{G}$-invariant states of branes on the covering space $\topsp{X}$. In the
subsequent sections we will study the pairing of Ramond-Ramond fields
with these D-branes.

Consider a D-brane localized on the submanifold $\topsp{X}^g$ of the covering
space $\topsp{X}$. Since the Chan-Paton bundle $\topsp{E}$ is $\topsp{G}$-equivariant, the fiber of the restriction $\topsp{E}$ to $\topsp{X}^{g}$ at each point carries a representation of the cyclic group $<g>$. In this case the D-brane is said to be
\emph{fractional}. Fractional D-branes are stuck at the fixed points and they couple to Ramond-Ramond fields coming for the twisted sector labelled by $[g]$. The term fractional is used since fractional D-branes, in the simple examples known, carry a \emph{fraction} of the corresponding Ramond-Ramond charge. 

We can use equivariant K-homology to geometrically describe a particular class of fractional D-branes. Indeed, let $\topsp{G}^\vee$ denote the set of conjugacy classes
$[g]$ of elements $g\in \topsp{G}$. There is a natural
subcategory $\cat{D}_{\rm frac}^\topsp{G}(\topsp{X})$ of $\cat{D}^\topsp{G}(\topsp{X})$ consisting of
triples $(\topsp{W},\topsp{E},f)$ for which $\topsp{W}$ is a $\topsp{G}$-fixed space, {i.e.}, for
which
\beq
\topsp{W}^{g}= \topsp{W}
\label{WgW}\eeq
for all $g\in \topsp{G}$. By $\topsp{G}$-equivariance this implies $f(\topsp{W})^g=f(\topsp{W})$ for
all $g\in \topsp{G}$, and so the image of the brane worldvolume lies in the
subspace
$$
f(\topsp{W})~\subset~\bigcap_{g\in \topsp{G}}\,\topsp{X}^g \ .
$$
This is the set of $\topsp{G}$-fixed points of $\topsp{X}$, and so the
objects $(\topsp{W},\topsp{E},f)$ of the category $\cat{D}_{\rm frac}^\topsp{G}(\topsp{X})$ can naturally be intepreted in terms of fractional branes. More precisely, we call $\cat{D}_{\rm
  frac}^\topsp{G}(\topsp{X})$ the category of ``maximally fractional D-branes''.

In this case, an application of Schur's lemma shows that the
Chan-Paton bundle admits an isotopical decomposition and there is a
canonical isomorphism of $\topsp{G}$-bundles
\beq
\topsp{E}~\simeq~\bigoplus_{[g]\in \topsp{G}^\vee}\,\topsp{E}_{[g]}\otimes\id_{[g]} \qquad
\mbox{with} \quad \topsp{E}_{[g]}\=\Hom_\topsp{G}\big(\id_{[g]}\,,\,\topsp{E}\big) \ ,
\label{Eisodecomp}\eeq
where $\topsp{E}_{[g]}$ is a complex vector bundle with trivial $\topsp{G}$-action and
$\id_{[g]}$ is the $\topsp{G}$-bundle $\topsp{W}\times \topsp{V}_{[g]}$ with
$\gamma:\topsp{G}\to\End(\topsp{V}_{[g]})$ the irreducible representation
corresponding to the conjugacy class $[g]\in \topsp{G}^\vee$.\\
From the
direct sum relation in equivariant K-homology it follows that a fractional D-brane, represented by a K-cycle $(\topsp{W},\topsp{E},f)$ in $\cat{D}_{\rm frac}^\topsp{G}(\topsp{X})$, splits  into a sum over irreducible fractional branes represented by the
K-cycles $(\topsp{W},\topsp{E}_{[g]}\otimes\id_{[g]},f)$, $[g]\in \topsp{G}^\vee$, which can then be considered stable.\\
We then propose that in the framework of equivariant K-homology, the topological charge of a fractional D-brane, in a given closed
string twisted sector of the orbifold String theory on a
$\topsp{G}$-spin$^c$ manifold $\topsp{X}$, can be
computed by using the equivariant Dirac operator theory introduced in {\appequiv}. The equivariant index of the $\topsp{G}$-invariant \spinc Dirac
operator $\Dirac_\topsp{E}^\topsp{X}$ coupled to a $\topsp{G}$-vector bundle $\topsp{E}\to \topsp{X}$ takes
values in $\K_\topsp{G}^*(\pt)\simeq \topsp{R}(\topsp{G})$. We can turn this into a
homomorphism on $\K_*^\topsp{G}(\topsp{X})$ with values in $\zed$ by composing
with the projection $\topsp{R}(\topsp{G})\to\zed$ defined by taking the multiplicity
of a given representation
\beq
\gamma\,:\,\topsp{G}~\longrightarrow~\End(\topsp{V}_\gamma)
\label{gammaunitary}\eeq
of $\topsp{G}$ on a finite-dimensional complex vector space
$\topsp{V}_\gamma$. There is a corresponding class in the KK-theory group
$$[\gamma]~\in~\KK_*\big(\complex[\topsp{G}]\,,\,\End(\topsp{V}_\gamma)\big)$$
which is represented by the Kasparov module
$(\topsp{V}_\gamma,\gamma,0)$ associated with the extension of the
representation~(\ref{gammaunitary}) to a complex representation of
group ring $\complex[\topsp{G}]$. By Morita invariance, the Kasparov product
with $[\gamma]$ is the homomorphism on K-theory
$$\K_0\big(\complex[\topsp{G}]\big)~\longrightarrow~
\K_0\big(\End(\topsp{V}_\gamma)\big)\simeq\K_0(\complex)\simeq\zed$$ induced by
$\gamma:\complex[\topsp{G}]\to\End(\topsp{V}_\gamma)$ \cite{Bunke2007}. We may then define a
homomorphism $$\mu_\gamma\,:\,\K_0^\topsp{G}(\topsp{X})~\longrightarrow~\zed$$ of
abelian groups by
\beq
\mu_\gamma\big([\topsp{W},\topsp{E}, f]\big)\=\Index_\gamma
\big( f_*[\Dirac_\topsp{E}^\topsp{W}]\big)~:=~
\ass\big( f_*[\Dirac_\topsp{E}^\topsp{W}]\big)\otimes_{\complex[\topsp{G}]}[\gamma]
\label{mugammadef}\eeq
on equivariant K-cycles $(\topsp{W},\topsp{E}, f)\in\cat{D}^\topsp{G}(\topsp{X})$ (and extended
linearly), where
$$\ass\,:\,\K_*^\topsp{G}(\topsp{X})~\longrightarrow~
\K_*\big(\complex[\topsp{G}]\big)$$ is the analytic assembly map mentioned in {\appequiv}. We then naturally interpret \ref{mugammadef} as the topological charge of the D-brane represented by $(\topsp{W},\topsp{E},f)$. Notice that for $\topsp{G}=e$, (\ref{mugammadef}) reduces to the ordinary expression for the charge of a D-brane.\\

We may now consider a simple class of examples. Let $\topsp{V}$ be a complex vector space
of dimension $\dim_\complex(\topsp{V})=d\geq1$, and let $\topsp{G}$ be a finite
subgroup of $\SL(\topsp{V})$. Our spacetime $\topsp{X}$ is the $\topsp{G}$-space identified
with the product
$$
\topsp{X}=\real^{p,1}\times \topsp{V} \ ,
$$
where $\topsp{G}$ acts trivially on the Minkowski space $\real^{p,1}$, and $p$ is odd. We will consider fractional D-branes with worldvolume $\mathbb{R}^{p,1}\hookrightarrow\mathbb{R}^{p,1}\times{v}$, where $v\in{\topsp{V}}$ is a fixed vector under the linear action of $\topsp{G}$. In analogy with the nonequivariant case, the group of charge of these fractional D-branes is given by the compact support equivariant K-theory of the normal bundle, which in this case is given by \cite{Atiyah1969}
\begin{displaymath}
\K^*_{\topsp{G},\cpt}(\topsp{V})~\simeq~\K^*_{\topsp{G}}(\pt)~\simeq~\topsp{R}(\topsp{G})\=
\zed^{|\topsp{G}^\vee|} \ .
\end{displaymath}
since $\topsp{V}$ is G-contractible.\\
The same result can be obtained by considering the equivariant K-theory of the world volume $\mathbb{R}^{p,1}$, which is G-contactible. 
 It follows that the fractional D-branes, as
defined by elements of equivariant K-theory, can be identified with
representations of the orbifold group
$$
\gamma=\bigoplus_{a=1}^{|\topsp{G}^\vee|}\,N_a\,\gamma_a
$$
consisting of $N_a\geq0$ copies of the $a$-th irreducible
representation $$\gamma_a\,:\,\topsp{G}~\longrightarrow~\End(\topsp{V}_a) \ , \quad
a\=1,\dots,\big|\topsp{G}^\vee\big| \ , $$
which defines the action of $\topsp{G}$ on the fibres of the Chan-Paton
bundle. More precisely, each irreducible fractional brane is
associated to the $\topsp{G}$-bundle $\mathbb{R}^{p,1}\times \topsp{V}_a$ over $\mathbb{R}^{p,1}$.\\
We can then consider the K-cycles $(\mathbb{R}^{p,1},\mathbb{R}^{p,1}\times \topsp{V}_a,i_{v})$, where $i_{v}:\mathbb{R}^{p,1}\to{\mathbb{R}^{p,1}\times{v}\subset{\topsp{X}}}$. In equivariant K-homology, these
cycles can be contracted to $[\pt,\topsp{V}_a,i]$, where $i$ is the inclusion
of a point $\pt\subset \topsp{V}$ whose induced homomorphism
$$i_*\,:\,\K_*^\topsp{G}(\pt)~\longrightarrow~\K_*^\topsp{G}(\topsp{X})$$ can be
taken to be the identity map $\topsp{R}(\topsp{G})\to \topsp{R}(\topsp{G})$. The $\topsp{G}$-invariant Dirac operator $\Dirac_{\topsp{V}_a}^\pt$ is
just Clifford multiplication twisted by the $\topsp{G}$-module $\topsp{V}_a$, and thus
the topological charges (\ref{mugammadef}) of the corresponding
fractional branes in the twisted sector labelled by $b$ are given by
$$
\mu_b\big([\pt,\topsp{V}_a,i]\big)\=\Index_{\gamma_b}
\big([\Dirac_{\topsp{V}_a}^\pt]\big)\=
\big[\topsp{V}_a\otimes(\Delta^+\oplus\Delta^-)\big]\otimes_{\complex[\topsp{G}]}
[\gamma_b] \ ,
$$
where $\Delta^\pm$ are the half-spin representations of $\SO(p+1)$ on
$\complex^{\frac{p+1}{2}}$. Acting on the
character ring the projection gives
$[\topsp{W}]\otimes_{\complex[\topsp{G}]}[\gamma_b]=\gamma_{*}([\topsp{W}])$, where $$\gamma_{*}:\K_0\big(\complex[\topsp{G}]\big)~\longrightarrow~
\K_0\big(\End(\topsp{V}_\gamma)\big)$$ is the map induced by $\gamma$.
\subsection{Delocalization and Ramond-Ramond fields}\label{delocKtheory}
As discussed in chapter 1, the gauge theory of Ramond-Ramond fields arises as a low-energy limit of type II superstring theory from the Hilbert space of states of closed superstrings. Intuitively, the low-energy limit is the limit in which the string becomes pointlike, i.e. the lenght of the string goes to zero. As we have seen in the previous section, the Hilbert space for type II superstring theory defined on the good orbifold $[\topsp{X}/\topsp{G}]$ is given by  
\begin{displaymath}
\mathcal{H}:=\left(\bigoplus_{h\in{\topsp{\topsp{G}}}}\mathcal{H}_{h}\right)^{\topsp{\topsp{G}}}
\end{displaymath}
where the subspaces $\mathcal{H}_{h}$ are spaces of states of open strings satisfying the boundary condition (\ref{twisting}); we have also noticed that such open strings ``look like'' closed strings on the submanifolds $\topsp{X}^{h}$. We then expect massless fields, in particular Ramond-Ramond fields, arising from each of these sectors, defined on $\topsp{X}^{h}$. This is due to the fact that the center of mass of an open string satisfying condition (\ref{twisting}) is constrained to be a point of  $\topsp{X}^{h}$. We can mathematically ``organize'' the information about these Ramond-Ramond fields in the following way. We can associate to each $\topsp{G}$-manifold the space  
\begin{equation}\label{brilisnki}
\hat{\topsp{X}}:=\coprod_{h\in{\topsp{G}}}\topsp{X}^{h}
\end{equation}
Notice that we have an action of $\topsp{G}$ on $\hat{\topsp{X}}$, with $g\in{\topsp{G}}$ inducing the diffeomorphism  $\topsp{X}^{h}\to \topsp{X}^{ghg^{-1}}$. We can then consider Ramond-Ramond fields as elements of the differential complex
\begin{equation}\label{equivforms}
\Omega^{*}_{\topsp{G}}(\topsp{X};\mathbb{R}):=\Omega^{*}(\hat{\topsp{X}};\mathbb{R})^{\topsp{G}}=\left(\bigoplus_{h\in{\topsp{G}}}\Omega^{*}(\topsp{X}^{h};\mathbb{R})\right)^{\topsp{G}}
\end{equation}
equiped with the differential
\begin{displaymath}
{\rm d}_{\topsp{G}}:=\bigoplus_{h\in{\topsp{G}}}{\rm d}_{h}
\end{displaymath}
where ${\rm d}_{h}:\Omega^{*}(\topsp{X}^{h};\mathbb{R})\to\Omega^{*}(\topsp{X}^{h};\mathbb{R})$ is the usual deRham exterior derivative. Since $\topsp{X}^{h}$ is diffeomorphic to $\topsp{X}^{ghg^{-1}}$ for any $g\in{\topsp{G}}$, by making a choice of submanifolds $\topsp{X}^{g}$ we have
\begin{equation}
\Omega^{*}_{\topsp{G}}(\topsp{X};\mathbb{R})\simeq{\bigoplus_{[h]\in{\topsp{G}^{\vee}}}}\Omega^{*}(\topsp{X}^{h};\mathbb{R})^{Z_{\topsp{G}}(h)}
\end{equation}
The cohomology of the complex (\ref{equivforms}) with respect to the differential ${\rm d}_{\topsp{G}}$ is given by 
\begin{displaymath}
{\rm H^{*}}\left(\Omega^{*}_{\topsp{G}}(\topsp{X};\mathbb{R});{\rm d}_{\topsp{G}}\right)=\left(\bigoplus_{h\in{\topsp{G}}}{\rm H}^{*}(\topsp{X}^{h};\mathbb{R})\right)^{\topsp{G}}\simeq\bigoplus_{[h]\in{\topsp{G}}^{\vee}}{\rm H}^{*}(\topsp{X}^{h};\mathbb{R})^{Z_{\topsp{G}}(h)}
\end{displaymath}
where we have used
\begin{displaymath}
{\rm H}^{*}(\topsp{X}^{h}/Z_{\topsp{G}}(h);\mathbb{R})\simeq{\rm H}^{*}(\topsp{X}^{h};\mathbb{R})^{Z_{\topsp{G}}(h)}
\end{displaymath}
The cohomology groups above correspond to the \emph{delocalized equivariant cohomology} theory defined by Baum and Connes \cite{baumconnes}. Notice that the group ${\rm H^{*}}\left(\Omega^{*}_{\topsp{G}}(\topsp{X};\mathbb{R});{\rm d}_{\topsp{G}}\right)$ is non-canonically isomorphic to ${\rm H}^{*}(\topsp{X};\mathbb{R})\otimes{\rm R}(\topsp{G})$ when the $\topsp{G}$-action on $\topsp{X}$ is trivial.\\
   
We will now show how Bredon cohomology can be used to compute
the cohomology of the complex (\ref{equivforms}) of orbifold
Ramond-Ramond fields by giving a delocalized description of Bredon
cohomology with \emph{real} coefficients, following \cite{mislin} and~\cite{Luck1998} where further details can
be found. \\
Denote with $\underline{\real}(-)$ the real representation ring
functor $\mathbb{R}\otimes{\underline{R}(-)}$ on the orbit category
$\ocat{\topsp{G}}$. Let $\langle \topsp{G}\rangle$ denote the set of conjugacy classes
$[\topsp{C}]$ of cyclic subgroups $\topsp{C}$ of $\topsp{G}$. Let
$\underline{\real}\,_{\topsp{C}}(-)$ be the contravariant functor on
$\ocat{\topsp{G}}$ defined by $\underline{\real}\,_\topsp{C}(\topsp{G}/H)=0$ if $[\topsp{C}]$
contains no representative $g\,\topsp{C}\,g^{-1}<H$, and otherwise
$\underline{\real}\,_\topsp{C}(\topsp{G}/H)$ is isomorphic to the cyclotomic field
$\mathbb{R}(\zeta_{|\topsp{C}|})$ over $\mathbb{R}$ generated by the primitive
root of unity $\zeta_{|\topsp{C}|}$ of order $|\topsp{C}|$. A standard result from the
representation theory of finite groups then gives a natural
splitting
\begin{displaymath}
\underline{\real}(-)=\bigoplus_{[\topsp{C}]\in\langle \topsp{G}\rangle}\,
\underline{\real}\,_\topsp{C}(-) \ .
\end{displaymath}
By definition, for any module $\underline{M}\,(-)$ over the orbit
category one has
\bea\nonumber
\text{Hom}_{\ocat{\topsp{G}}}\big(\,\underline{M}\,(-)\,,\,
\underline{\real}\,_{\topsp{C}}(-)\big)&\simeq&
\text{Hom}_{N_{\topsp{G}}(\topsp{C})}\big(\,\underline{M}\,(\topsp{G}/\topsp{C})\,,\,
\underline{\real}\,_{\topsp{C}}(\topsp{G}/\topsp{C})\big) \\[4pt] \nonumber &\simeq&
{\rm Hom}\left(\underline{M}\,(\topsp{G}/\topsp{C});\mathbb{Z}\right)\otimes_{N_{\topsp{G}}(\topsp{C})}\,\underline{\real}\,_{\topsp{C}}(\topsp{G}/\topsp{C})
\eea
where the normalizer subgroup $N_{\topsp{G}}(\topsp{C})$ acts on
$\underline{\real}\,_{\topsp{C}}(\topsp{G}/\topsp{C})\simeq\mathbb{R}(\zeta_{|\topsp{C}|})$ via
identification of a generator of $\topsp{C}$ with $\zeta_{|\topsp{C}|}$.

These facts together imply that the cochain groups (\ref{CGnXF}) with
$\underline{F}=\underline{\real}(-)$ admit a splitting given by
\begin{displaymath}
C^*_{\topsp{G}}\big(\topsp{X}\,,\,\underline{\real}(-)\big)\simeq
\bigoplus_{[\topsp{C}]\in\langle
  \topsp{G}\rangle}\,C^*\big(\topsp{X}^{\topsp{C}}\big)\otimes_{N_{\topsp{G}}(\topsp{C})}\,
\underline{\real}\,_{\topsp{C}}(\topsp{G}/\topsp{C}) \ .
\end{displaymath}
As the centralizer $Z_{\topsp{G}}(\topsp{C})$ acts properly on $\topsp{X}^{\topsp{C}}$, the natural
map
\begin{displaymath}
\bigoplus_{[\topsp{C}]\in\langle \topsp{G}\rangle}\,C^*\big(\topsp{X}^{\topsp{C}}\big)\otimes_{N_{\topsp{G}}(\topsp{C})}\,
\underline{\real}\,_{\topsp{C}}(\topsp{G}/\topsp{C}) ~\longrightarrow~ 
\bigoplus_{[\topsp{C}]\in\langle \topsp{G}\rangle}\,C^*\big(\topsp{X}^{\topsp{C}}/Z_\topsp{G}(\topsp{C})\big)
\otimes_{W_{\topsp{G}}(\topsp{C})}\,\underline{\real}\,_{\topsp{C}}(\topsp{G}/\topsp{C})
\end{displaymath}
is a cohomology isomorphism, where $W_{\topsp{G}}(\topsp{C}):=N_{\topsp{G}}(\topsp{C})/Z_{\topsp{G}}(\topsp{C})$ is
the Weyl group of $\topsp{C}<\topsp{G}$ which acts by translation on
$\topsp{X}^\topsp{C}/Z_\topsp{G}(\topsp{C})$. Since $\underline{\real}\,_{\topsp{C}}(\topsp{G}/\topsp{C})$ is a projective
$\mathbb{R}[W_{\topsp{G}}(\topsp{C})]$-module, it follows that for any proper
$\topsp{G}$-complex $\topsp{X}$ the Bredon cohomology of $\topsp{X}$ with coefficient system
$\mathbb{R}\otimes{\underline{\topsp{R}}(-)}$ has a splitting
\beq\label{Bredoncyclic}
\H_{\topsp{G}}^{*}\big(\topsp{X}\,;\,\mathbb{R}\otimes{\underline{\topsp{R}}(-)}\big)
\simeq\bigoplus_{[\topsp{C}]\in\langle
  \topsp{G}\rangle}\,\H^{*}\big(\topsp{X}^{\topsp{C}}/Z_{\topsp{G}}(\topsp{C})\,;\,
\mathbb{R}\big)\otimes_{W_{\topsp{G}}(\topsp{C})}\,\underline{\real}\,_{\topsp{C}}(\topsp{G}/\topsp{C}) \ .
\eeq

At this point, we note that the dimension of the $\mathbb{R}$-vector
space
\begin{displaymath}
\underline{\real}\,_\topsp{C}(\topsp{G}/\topsp{C})^{W_\topsp{G}(\topsp{C})}\simeq
\mathbb{R}\otimes_{W_{\topsp{G}}(\topsp{C})}\,\underline{\real}\,_{\topsp{C}}(\topsp{G}/\topsp{C})
\end{displaymath}
is equal to the number of $\topsp{G}$-conjugacy classes of generators for
$\topsp{C}$. We also use the fact that for a finite group $\topsp{G}$ a sum over
conjugacy classes of cyclic subgroups is equivalent to a sum
over conjugacy classes of elements in $\topsp{G}$, and that $\topsp{X}^{\langle
  g\rangle}=\topsp{X}^g$ and $Z_\topsp{G}(\langle g\rangle)=Z_\topsp{G}(g)$. One finally
obtains a splitting of real Bredon cohomology groups\footnote{This
  splitting in fact holds over $\rat$~\cite{mislin}.}
\beq\label{Bredonsplit}
\H_{\topsp{G}}^{*}\big(\topsp{X}\,;\,\mathbb{R}\otimes{\underline{\topsp{R}}(-)}\big)
\simeq\bigoplus_{[g]\in \topsp{G}^\vee}\,\H^{*}\big(\topsp{X}^{g}\,;\,
\mathbb{R}\big)^{Z_{\topsp{G}}(g)}
\eeq
which is the cohomology of the differential complex (\ref{equivforms}).\\
By using Theorem~\ref{eqChernthm}, one also has a decomposition for equivariant
K-theory with real coefficients given by
\begin{displaymath}
\K_{\topsp{G}}^{*}(\topsp{X})\otimes\mathbb{R}
\simeq\bigoplus_{[g]\in \topsp{G}^\vee}\,\big(\K^{*}(\topsp{X}^{g})\otimes
\mathbb{R}\big)^{Z_{\topsp{G}}(g)} \ .
\end{displaymath}
\subsection{Delocalization of the equivariant Chern character}
It is well known that the ordinary Chern character, when tensored over $\mathbb{R}$, admits a Chern-Weyl refinement, expressed in terms of the curvature of an arbitrary connection on the given vector bundle. This is not the case for the equivariant Chern character defined in section {\ref{equivChern}}. This is expected, since the equivariant Chern character was constructed in terms of homomorphisms between K-theory and cohomology, without any reference to any geometric description. However, when tensored over $\mathbb{C}$, the equivariant Chern character admits a more geometric description. We will now explain this construction, referring the reader to \cite{Bunke2007} for the technical details. Consider a complex $\topsp{G}$-bundle $\topsp{E}$ over $\topsp{X}$ equiped
with a $\topsp{G}$-invariant hermitean metric and a $\topsp{G}$-invariant metric
connection~$\nabla^{\topsp{E}}$. One can then define a closed $\topsp{G}$-invariant
differential form $$\ch(\topsp{E})~\in~\Omega^*(\topsp{X};\complex)^{\topsp{G}}$$ in the
usual way by the Chern-Weil construction
\begin{displaymath}
\ch(\topsp{E}):=\Tr\big(\exp(-\topsp{F}^\topsp{E}/2\pi{\ii})\big)
\end{displaymath}
where $\topsp{F}^\topsp{E}$ is the curvature of the connection $\nabla^\topsp{E}$. It
represents a cohomology class
$$\big[\ch(\topsp{E})\big]~\in~{\H^*(\topsp{X};\mathbb{C})^{\topsp{G}}}$$ in the fixed
point subring of the action of $\topsp{G}$ as automorphisms of
$\H^*(\topsp{X};\mathbb{C})$. By using the definition of the
homomorphisms (\ref{chXH}), with $\mathbb{Q}$ substituted by
$\mathbb{C}$ and $H=e$, one can establish the equality
\begin{displaymath}
\big[\ch(\topsp{E})\big]=\ch^{e}_{\topsp{X}}\big([\topsp{E}]\big) \ .
\end{displaymath}

Let $\topsp{C}<{\topsp{G}}$ be a cyclic subgroup, and define the
cohomology class $$\big[\ch(g,\topsp{E})\big]~\in~
{\H^*\big(\topsp{X}^{\topsp{C}}\,;\,\mathbb{\topsp{C}}\big)^{Z_{\topsp{G}}(C)}}
\simeq\H^*\big(\topsp{X}^{\topsp{C}}/Z_{\topsp{G}}(\topsp{C})\,;\,\mathbb{C}\big)\simeq
\H\big(\Omega^*(\topsp{X}^\topsp{C};\complex)^{Z_\topsp{G}(\topsp{C})}\,,\,\dd\big)$$
represented by
\begin{displaymath}
\ch(g,\topsp{E}):=\Tr\big(\gamma(g)\,
\exp(-\topsp{F}_\topsp{C}^\topsp{E}/2\pi{\ii})\big)
\end{displaymath}
where $g$ is a generator of $\topsp{C}$, $\topsp{F}_\topsp{C}^\topsp{E}$ is the restriction of the
invariant curvature two-form $\topsp{F}^\topsp{E}$ to the fixed point subspace
$\topsp{X}^{\topsp{C}}$, and $\gamma$ is a representation of $\topsp{C}$ on the fibres of the
restriction bundle $\topsp{E}|_{\topsp{X}^{\topsp{C}}}$ which is an $N_\topsp{G}(\topsp{C})$-bundle over
$\topsp{X}^\topsp{C}$. The character $\chi_\topsp{C}$ naturally identifies
$\topsp{R}(\topsp{C})\otimes\complex$ with the $\complex$-vector space of class
functions $\topsp{C}\to\complex$. By using the splitting (\ref{Bredoncyclic})
for complex Bredon cohomology, one can then show that
\begin{displaymath}
\ch^{\topsp{C}}_{\topsp{X}}\big([\topsp{E}]\big)(g)=\big[\ch(g,\topsp{E})\big]
\end{displaymath}
up to the restriction homomorphism
$\topsp{R}(\topsp{C})\otimes\complex\to\underline{\complex}\,_{\topsp{C}}(\topsp{G}/\topsp{C})$ of rings with
kernel the ideal of elements whose characters vanish on all generators
of $\topsp{C}$.

Using {}~(\ref{chXprod}) we can then define the map
\begin{displaymath}
\ch^{\mathbb{C}}\,:\,{\rm Vect}^{\mathbb{C}}_{\topsp{G}}(\topsp{X})~
\longrightarrow~\Omega_{\topsp{G}}^{\rm even}\big(\topsp{X}\,;\,
\mathbb{C}\big)
\end{displaymath}
from complex $\topsp{G}$-bundles $\topsp{E}\to \topsp{X}$ given by
\beq
\ch^{\mathbb{C}}(\topsp{E})=\bigoplus_{[g]\in \topsp{G}^\vee}\,\Tr\big(\gamma(g)\,\exp(-\topsp{F}_g^\topsp{E}/2\pi{\ii})\big) \ .
\label{chCEdef}\eeq
At the level of equivariant K-theory, from Theorem~\ref{eqChernthm} it
follows that this map induces an isomorphism
\beq
\ch^{\mathbb{C}}\,:\,\K^*_{\topsp{G}}(\topsp{X})\otimes\mathbb{C}~
\xrightarrow{\approx}~\H_{\topsp{G}}\big(\topsp{X}\,;\,\mathbb{C}
\otimes{\underline{\pi_{-*}\K}\,_\topsp{G}(-)}\big)^{*}
\label{chCiso}\eeq
where we have used the splitting (\ref{Bredonsplit}). The map
(\ref{chCEdef}) coincides with the equivariant Chern character defined
in \cite{Atiyah1989}.  
\subsection{Ramond-Ramond couplings with D-branes}
We now have all the necessary ingredients to define a coupling of the
Ramond-Ramond fields to a D-brane in the orbifold $[\topsp{X}/\topsp{G}]$. In this
section we will only consider Ramond-Ramond fields which are
topologically trivial, {i.e.}, elements of the differential
complex (\ref{equivforms}), and use the delocalized cohomology theory
above by working throughout with complex coefficients. Moreover, we will only consider electric couplings to D-branes, i.e. we will not impose selfduality. Under these
conditions we can straightforwardly make contact with existing
examples in the physics literature and write down their appropriate
generalizations.

To this aim, we introduce the bilinear product
\begin{displaymath}
\wedge_\topsp{G}\,:\,\Omega_{\topsp{G}}^*(\topsp{X};\real)\otimes\Omega_{\topsp{G}}^*(\topsp{X};
\real)~\longrightarrow{}~\Omega_{\topsp{G}}^*(\topsp{X};\real)
\end{displaymath}
defined on $\omega=\bigoplus_{g\in \topsp{G}}\,\omega_{g}$ and
$\eta=\bigoplus_{g\in \topsp{G}}\,\eta_{g}$ by
\beq
\omega\wedge_\topsp{G}{\eta}~:=~
\bigoplus_{g\in \topsp{G}}\,\omega_{g}\wedge_g\eta_{g}
\label{orbdiffprod}\eeq
where $\wedge_g=\wedge$ is the usual exterior product on
$\Omega^*(\topsp{X}^g;\real)$. There is also an integration
$$\int_\topsp{X}^\topsp{G}\,:\,\Omega^*_\topsp{G}(\topsp{X};\real)~\longrightarrow~\real \
. $$ If $\omega\in\Omega^*_{\topsp{G}}( \topsp{X};\mathbb{R})$ then we set
\begin{displaymath}
\int_{\topsp{X}}^\topsp{G}\,\omega~:=~\frac1{|\,\topsp{G}^\vee\,|}~\sum_{[g]\in \topsp{G}^\vee}~
\int_{\topsp{X}^{g}}\,\omega_{[g]} \ .
\end{displaymath}
where we have used in the above construction that for a G-manifold $\topsp{X}$ admitting a G-equivariant spin structure, the fixed point manifold $\topsp{X}^{g}$ is naturally oriented, for any $g\in{\topsp{G}}$ \cite{Berline}.\\ 
The normalization ensures that $\int_{\topsp{X}}^\topsp{G}\,\omega=\int_{\topsp{X}}\,\omega$ when
$\topsp{G}$ acts trivially on $\topsp{X}$, and $\omega$ is ``diagonal'' in $\Omega^{*}({\rm X})\otimes{\rm R}({\rm G})$.

Suppose now that $ f:\topsp{W}\rightarrow{\topsp{X}}$ is the smooth immersed
worldvolume cycle of a wrapped D-brane state $(\topsp{W},\topsp{E},f)\in\cat{D}^\topsp{G}(\topsp{X})$
in the orbifold $[\topsp{X}/\topsp{G}]$, {i.e.}, $ \topsp{W}$ is a $\topsp{G}$-\spinc manifold
equiped with a $\topsp{G}$-bundle $\topsp{E}\to \topsp{W}$ and an invariant connection
$\nabla^\topsp{E}$ on $\topsp{E}$. We define the \emph{Wess-Zumino pairing}
$$
\CS\,:\,\cat{D}^\topsp{G}(\topsp{X})\times\Omega_\topsp{G}^*(\topsp{X};\complex)~\longrightarrow~
\complex
$$
between such D-branes and Ramond-Ramond fields as
\begin{equation}\label{coupling2}
\CS\big((\topsp{W},\topsp{E},f)\,,\,C\big)=
\int_{ \topsp{W}}^\topsp{G}\,\tilde{C}\wedge_\topsp{G}{\,\ch^{\mathbb{C}}}(\topsp{E})\wedge_\topsp{G}
\mathcal{R}(\topsp{W},f) \ ,
\end{equation}
where $\tilde C=f^{*}C$ is the pullback along $f:\topsp{W}\to \topsp{X}$ of the
total Ramond-Ramond field $$C=\bigoplus_{[g]\in \topsp{G}^\vee}\,C_{[g]}$$ and
the equivariant Chern character is given by {}~(\ref{chCEdef}) with
$\gamma$ giving the action of $\topsp{G}$ on the Chan-Paton factors of the
D-brane. The closed worldvolume form
$\mathcal{R}(\topsp{W},f)\in\Omega^{\rm even}_{\topsp{G},{\rm cl}}(\topsp{W};\complex)$
represents a complex Bredon cohomology class which accounts for
gravitational corrections due to curvature in the spacetime $\topsp{X}$ and
depends only on the bordism class of $(\topsp{W},f)$. We refer the reader to \cite{mine}, where a construction of $\mathcal{R}(\topsp{W},f)$ can be found.

Modulo the curvature contribution $\mathcal{R}(\topsp{W},f)$, the very natural expression~(\ref{coupling2})
reduces to the usual Wess-Zumino coupling of topologically trivial
Ramond-Ramond fields to D-branes in the case $\topsp{G}=e$. But even if a
group $\topsp{G}\neq e$ acts trivially on the brane worldvolume $ \topsp{W}$ (or on
the spacetime $\topsp{X}$), there can still be additional contributions to the
usual Ramond-Ramond coupling if $\topsp{E}$ is a \emph{non-trivial
  $\topsp{G}$-bundle}. This is the situation, for instance, for fractional
D-branes $$(\topsp{W},\topsp{E},f)~\in~\cat{D}_{\rm frac}^\topsp{G}(\topsp{X})$$  placed at
orbifold singularities. In this case, we may use the isotopical
decomposition (\ref{Eisodecomp}) of the Chan-Paton bundle along with
{}~(\ref{WgW}). Then the Wess-Zumino pairing (\ref{coupling2})
descends to a pairing $$\CS_{\rm frac}\,:\,\cat{D}_{\rm
  frac}^\topsp{G}(\topsp{X})\times\Omega_\topsp{G}^*(\topsp{X};\complex)~\longrightarrow~
\complex$$ with the additive subcategory of fractional branes at
orbifold singularities. 
\begin{example}\label{Linorb2}
We will now ``test'' our definition (\ref{coupling2}) on the class of
examples considered in section~\ref{equivkhom}. These are flat
orbifolds for which there are no non-trivial curvature contributions,
{i.e.}, $\mathcal{R}(\topsp{W},f)=1$. Let us specialize to the case of
cyclic orbifolds having twist group $\topsp{G}=\zed_n$ with $n\geq d$. In this
case, as $\mathbb{Z}_{n}$ is an abelian group, one has
$\zed_n^\vee=\zed_n$ (setwise) and we can label the non-trivial
twisted sectors of the orbifold String theory on $\topsp{X}$ by
$k=1,\dots,n-1$. The untwisted sector is labelled by $k=0$. We take a
generator $g$ of $\zed_n$ whose action on $V\simeq\complex^d$ is given
by
$$
g\cdot\big(z^1\,,\,\dots\,,\,z^d\big):=\big(\zeta^{a_1}\,z^1\,,\,
\dots\,,\,\zeta^{a_d}\,z^d\big) \ ,
$$
where $\zeta=\exp(2\pi\ii/n)$ and $a_1,\dots,a_d$ are integers
satisfying $a_1+\cdots+a_d\equiv0~{\rm mod}~n$\footnote{Both the
  requirement that the representation $V$ be complex and the form of
  the $\topsp{G}$-action are physical inputs ensuring that the closed string
  background $\topsp{X}$ preserves a sufficient amount of supersymmetry after
  orbifolding.}. In this case the action of any element in $\zed_n$ has
only one fixed point, an orbifold singularity at the origin
$(0,\dots,0)$. Hence for any $g\neq e$ one has
\begin{displaymath}
\topsp{X}^{g}\simeq\mathbb{R}^{p,1}
\end{displaymath}
and the differential complex (\ref{equivforms}) of orbifold
Ramond-Ramond fields is given by
\begin{displaymath}
\Omega_{\zed_n}^*(\topsp{X};\real)=\Omega^*(\topsp{X};\real)\oplus\Big(\,
\bigoplus_{k=1}^{n-1}\,\Omega^*\big(\mathbb{R}^{p,1};\real\big)
\Big) \ .
\end{displaymath}

Consider now a fractional D-brane $(\topsp{W},\topsp{E},f)\in\cat{D}_{\rm frac}^{\zed_n}(\topsp{X})$ with
worldvolume cycle $f(\topsp{W})\subset{\mathbb{R}^{p,1}}$ placed at the
orbifold singularity, {i.e.},
$f:\topsp{W}\rightarrow{\mathbb{R}^{p,1}\times{(0,\dots,0)}}\subset{\topsp{X}}$. Let
the generator $g$ act on the fibres of the Chan-Paton bundle $\topsp{E}\to \topsp{W}$
in the $n$-dimensional regular representation
$\gamma(g)_{ij}=\zeta^i\,\delta_{ij}$. The action on worldvolume
fermion fields is determined by a lift $\hat\zed_n$ acting on the spinor
bundle $S\to \topsp{W}$. Then the pairing (\ref{coupling2}) contains the
following terms. First of all, we have the coupling of the untwisted
Ramond-Ramond fields to $ \topsp{W}$ given by
\begin{displaymath}
\dfrac{1}{n}\int_{ \topsp{W}}\,\tilde C\wedge\Tr\big(\exp(-\topsp{F}^\topsp{E}/2\pi{\ii})\big) \ ,
\end{displaymath}
which is just the usual Wess-Zumino coupling and hence the untwisted Ramond-Ramond charge of the brane is 1. Then there are the contributions from the twisted sectors,
which by recalling {}~(\ref{WgW}) are given by the expression
\begin{displaymath}
\dfrac{1}{n}\int_{ \topsp{W}}\,~\sum_{k=1}^{n-1}\,\tilde
C_{k}\wedge\Tr\big(\gamma(g^{k})\,
\exp(-\topsp{F}^\topsp{E}/2\pi{\ii})\big)
\end{displaymath}
where $g^{k}$ is an element of $\zed_n$ of order $k$. Since
$\gamma(g^k)_{ii}=\zeta^{ik}$, the brane associated with the $i$-th irreducible
representation of $\zed_n$ has charge $\zeta^{ik}/n$ with respect to the top
Ramond-Ramond field in the $k$-th twisted sector. For $d=2$ and $d=3$,
the form of these couplings agrees with those computed in \cite{douglas1996}.
\end{example}
\section{An equivariant Riemann-Roch formula}\label{Gravcoupl}
Let $\topsp{X},\topsp{W}$ be smooth compact $\topsp{G}$-manifolds, and $f:
\topsp{W}\rightarrow{\topsp{X}}$ a smooth proper $\topsp{G}$-map. If we want to make sense of
the equations of motion for the Ramond-Ramond field $C$, which is a
quantity defined on the spacetime $\topsp{X}$, then we need to pushforward
classes defined on the brane worldvolume $\topsp{W}$ to classes defined on the
spacetime. This will enable the construction of Ramond-Ramond currents
in a later section induced by the background and
D-branes which appear as source terms in the Ramond-Ramond field
equations.\\
As we have seen in chapter 3, given a smooth embedding $f:\topsp{W}\to{\topsp{X}}$ with normal bundle $\nu\to{\topsp{W}}$ equiped with a $\rm spin^{c}$ structure, we have
\beq
\ch\big(f_!^{\K}(\xi)\big)=f^{\H}_{!}\big(\ch(\xi)\,\cup\,\Todd(\nu)^{-1}
\big)
\label{ordGRRthm}\eeq
for any class $\xi\in{\rm K}^{*}(\topsp{W})$, where $f_{!}$ denote the Gysin homomorphism defined in chapter 3, and $\Todd(\topsp{E})\in\Omega_{\rm
  cl}^{\rm even}(\topsp{W};\complex)$ denotes the Todd genus characteristic
class of a hermitean vector bundle $\topsp{E}$ over $\topsp{W}$, whose Chern-Weil
representative is
$$
\Todd(\topsp{E})=\sqrt{\det\Big(\,\frac{F^\topsp{E}/2\pi\ii}{\tanh\big(F^\topsp{E}/2\pi\ii
\big)}\,\Big)}
$$
where $F^\topsp{E}$ is the curvature of a hermitean connection $\nabla^\topsp{E}$ on
$\topsp{E}$. The important aspect is that the Chern character
does not commute with the Gysin pushforward maps, and the defect in
the commutation relation is precisely the Todd genus of the bundle
$\nu$. This ``twisting'' by the bundle $\nu$ over the D-brane
contributes in a crucial way to the Ramond-Ramond current in the
non-equivariant case~\cite{currents,Minasian1997,Olsen1999}.\\
We will now attempt to find an equivariant version of the Riemann-Roch
theorem. We will consider G-equivariant embeddings $f:W\rightarrow \topsp{X}$: in this case
normal bundle $\nu$ is itself a $\topsp{G}$-bundle. We assume that $\nu$ is $\K_\topsp{G}$-oriented. This requirement
is just the Freed-Witten anomaly cancellation
condition~\cite{Freed1999} in this case, generalized to global
worldsheet anomalies for D-branes represented by generic
$\topsp{G}$-equivariant K-cycles. It enables, analogously to the
non-equivariant case, the construction of an equivariant Gysin
homomorphism
\beq\label{eqGysin}
f_!^{\K_\topsp{G}}\,:\,\K^*_\topsp{G}(W)~\longrightarrow~\K^*_\topsp{G}(\topsp{X}) \ .
\eeq

We will show that, under some very
special conditions, one can construct a complex Bredon cohomology
class which is analogous to the Todd genus and which plays the role of
the equivariant commutativity defect as above.\\
For this we will need the following\footnote{We are grateful to U.Bunke for suggesting this result to us.}
\begin{lemma}
Let $\pi:\topsp{E}\to{\topsp{X}}$ be an equivariant $\topsp{G}$-vector bundle, where $\topsp{G}$ is a finite group acting properly. Then for any $g\in{\topsp{G}}$ the map $\pi|_{\topsp{E}^{g}}:\topsp{E}^{g}\to{\topsp{X}}^{g}$ is vector bundle projection. 
\end{lemma}
\begin{proof}
Consider a point $w\in{\topsp{E}}$. Then $d\pi:\topsp{T}_{w}\topsp{E}\to\topsp{T}_{\pi(w)}\topsp{X}$ is surjective, since it is the projection map of a fiber bundle. If $w\in{\topsp{E}^{g}}$, then $\pi(w)$ belongs to $\topsp{X}^{g}$, and the element $g$ acts linearly on $\topsp{T}_{\pi(w)}\topsp{X}$. If $v\in\topsp{T}_{\pi(w)}\topsp{X}$ is a $g$-fixed vector, one can then find a $g$-fixed preimage in $\topsp{T}_{w}\topsp{E}$ by taking any preimage and averaging with respect to the action of $g$. By using that $(\topsp{T}\topsp{X})^{g}=\topsp{T}\topsp{X}^{g}$, we have that the map $\topsp{T}_{w}\topsp{E}^{g}\to\topsp{T}_{\pi(w)}\topsp{X}^{g}$ is surjective, hence the restricted map $\pi_{\topsp{E}^{g}}$ is a proper surjection. Any proper submersion is automatically a fibre bundle, and the linear structure on the fibers of $\topsp{E}^{g}\to{\topsp{X}^{g}}$ is inherited by that on the fibers of $\topsp{E}\to{\topsp{X}}$.
\end{proof}
Notice that the above lemma strongly use the properties that the group $\topsp{G}$ is finite and acting properly. To see what can happen when these conditions are not satisfied\footnote{We are grateful to J.Figueroa-O'Farrill for suggesting this example to us}, let $\topsp{X}=\mathbb{R}$ and $\topsp{G}=\mathbb{R}^{+}$ the group of positive reals
under multiplication. Consider the $\topsp{G}$-bundle $\topsp{X}\times{\topsp{V}}\to{\topsp{X}}$ given
by projection onto the first factor, where $\topsp{V}$ is a finite-dimensional
real vector space and the $\topsp{G}$-action is
\begin{displaymath}
g\cdot(x,v)=\big(x\,,\,g^{x}\,v\big)
\end{displaymath}
for all $g\in \topsp{G}$. For any $g\neq1$, $(\topsp{X}\times{\topsp{V}})^{g}$ is not a fibre
bundle over $\topsp{X}^g=\topsp{X}$, as the $\topsp{G}$-invariant fibre space over $x=0$ is
$\topsp{V}$ while it is the null vector over any other point. In particular, $(\topsp{X}\times{\topsp{V}})^{g}$ is not even a manifold.\\
Consider now two G-manifolds $\topsp{W}$ and $\topsp{X}$ equiped with G-invariant $\rm spin^{c}$ structure, and a proper G-equivariant embedding $f:\topsp{W}\to{\topsp{X}}$. For any $g\in{\topsp{G}}$, the vector bundle $\nu^{g}=\nu(\topsp{W};f)^{g}\to{\topsp{W}}^{g}$ is the normal bundle $\nu^g=f^*|_{\topsp{W}^g}(\topsp{T}_{\topsp{X}^g})\oplus \topsp{T}_{\topsp{W}^g}$ over the immersion
$f|_{\topsp{W}^g}:\topsp{W}^{g}\rightarrow{\topsp{X}^{g}}$. Recall that this is a $Z_{\topsp{G}}(g)$-bundle. We will suppose that $\nu^{g}$ is equiped with a $Z_{\topsp{G}}(g)$-invariant $\rm spin^{c}$ structure.\\
As we have seen in section \ref{delocKtheory}, equivariant K-theory ``delocalizes'' when tensored over $\mathbb{R}$, thanks to the equivariant Chern character. We will show in the following that when working over $\mathbb{C}$ the delocalization of equivariant K-theory is compatible, in a certain sense, with the complex equivariant Chern character defined in (\ref{chCEdef}).\\
Let $\topsp{E}$ be a $\topsp{G}$-vector bundle over a $\topsp{X}$. For any $g\in{\topsp{G}}$, the restriction $\topsp{E}_{g}$ of $\topsp{E}$ to the fixed point subspace $\topsp{X}^{g}$ gives a $Z_{\topsp{G}}(g)$-invariant vector bundle carrying a representation $\gamma$ of $<g>$ on the fibers, where $<g>$ is the cyclic group generated by $g$. The decomposition of the representation $\gamma$ in terms of irreducible representations $\gamma_{\alpha}$ of $<g>$ induces the homomorphism 
\begin{displaymath}
{\rm K}^{*}_{\topsp{G}}(\topsp{X})\to{\rm K}^{*}(\topsp{X}^{g})\otimes{\topsp{R}(<g>)}
\end{displaymath}
given on vector bundles by 
\begin{displaymath}
\topsp{E}\to{\topsp{E}_{g}}\simeq{\bigoplus_{\gamma_{\alpha}}\topsp{E}_{\alpha}\otimes{\gamma_{\alpha}}}
\end{displaymath}
where $\topsp{E}_{\alpha}:={\rm Hom_{<g>}(\topsp{E},\id_{\alpha})}$, where $\id_{\alpha}=\topsp{X}^{g}\times{\topsp{V}_{\alpha}}$, with $\topsp{V}_{\alpha}$ carrying the irreducible representation $\gamma_{\alpha}$. We can then define the morphism
\begin{displaymath}
{\rm K}^{*}(\topsp{X}^{g})\otimes{\topsp{R}(<g>)}\to{\rm K}^{*}(\topsp{X}^{g})\otimes\mathbb{C}
\end{displaymath}
by tracing over the second factor, i.e.
\begin{displaymath}
[\topsp{E}_{g}]\to\sum_{\alpha}[\topsp{E}_{\alpha}]{\rm Tr}(\gamma_{\alpha}(g))
\end{displaymath}
Denoting $\phi_{g}$ the composition of the above morphism, we can define 
\begin{displaymath}
\phi:=\oplus_{[g]\in{\topsp{G}^{\vee}}}\phi_{g}
\end{displaymath}
Recalling that $\topsp{E}_{g}$ is a $Z_{\topsp{G}}(g)$-invariant bundle, we have the isomorphism \cite{Atiyah1989}
\begin{equation}\label{Kdecomp}
\phi:{\rm K}^{*}_{\topsp{G}}(\topsp{X})\otimes{\mathbb{C}}\simeq\bigoplus_{[g]\in{\topsp{G}^{\vee}}}\left({\rm K}^{*}(\topsp{X}^{g})\otimes\mathbb{C}\right)^{Z_{\topsp{G}}(g)}
\end{equation}
By using that \cite{Bunke2007}
\begin{displaymath}
\Tr\big(\gamma(g)\,\exp(-\topsp{F}_g^\topsp{E}/2\pi{\ii})\big)=\sum_{\alpha}\Tr\big(\exp(-\topsp{F}_\alpha^\topsp{E}/2\pi{\ii})\big)\Tr(\gamma_{\alpha}(g))
\end{displaymath}
where $F_\alpha^{\topsp{E}}$ is the restriction of the curvature form $F_{g}^{\topsp{E}}$ to the subbundle $\topsp{E}_{\alpha}$, we have that the complex equivariant Chern character (\ref{chCEdef}) coincides on the components of the decomposition (\ref{Kdecomp}) with the ordinary Chern character.\\

Suppose now that the equivariant Thom class
$\Thom_\topsp{G}(\nu)\in\K^*_{\topsp{G},\cpt}(\nu)$ can be decomposed according
to the splitting (\ref{Kdecomp}) in such a way that the component
in any subgroup $$\Thom\big(\nu^g\big)~\in~
\big(\K_\cpt^*(\nu^{g})\otimes\complex\big)^{Z_\topsp{G}({g})}$$
coincides with the (ordinary) Thom class of the vector bundle
$\nu^{g}\to{\topsp{W}^{g}}$. Under these conditions, the equivariant Gysin
homomorphism (\ref{eqGysin}) decomposes
according to the splitting
\begin{displaymath}
f^{\K_{\topsp{G}}}_{!}=\bigoplus_{[g]\in \topsp{G}^\vee}\,f^{\K}_{g}
\end{displaymath}
where $f^{\K}_{g}$ is the K-theoretic Gysin homomorphism associated
to the smooth map
$$f\big|_{\topsp{W}^g}\,:\,\topsp{W}^{g}~\longrightarrow~{\topsp{X}^{g}} \ . $$

Define the characteristic class $\Todd_{\topsp{G}}$ by %
\beq\label{ToddGdef}
\Todd_{\topsp{G}}(\nu):=\bigoplus_{[g]\in \topsp{G}^\vee}\,\Todd\big(\nu^{g}\big)
\eeq
This class defines an
element of the even degree complex Bredon cohomology of the brane
worldvolume $\topsp{W}$. Under the conditions spelled out above, we can now
use the equivariant Chern character (\ref{chCiso}) and the usual
Riemann-Roch theorem for each pair $(\topsp{W}^{g},\topsp{X}^{g})$ to prove the
identity
\beq
f^{\H_{\topsp{G}}}_{!}\big(\ch^{\mathbb{C}}(\xi)\,\cup_\topsp{G}\,
\Todd_{\topsp{G}}(\nu)^{-1}\big)=
\ch^{\mathbb{C}}\big(f^{\K_{\topsp{G}}}_{!}(\xi)\big)
\label{GRRthm}\eeq
for any class $\xi\in{\K_{\topsp{G}}^*( \topsp{W})}\otimes{\mathbb{C}}$, as all
quantities involved in the expression (\ref{GRRthm}) are compatible with
the $\topsp{G}$-equivariant decompositions given above.\\
This equivariant Riemann-Roch formula will be important when we will argue the flux quantization for Ramond-Ramond fields on orbifolds.
\section{Orbifold differential K-theory and flux quantization}
In this section we will propose an extension of
differential K-theory as defined in~section \ref{diffKHS} to
incorporate the case of a $\topsp{G}$-manifold. These are the groups needed to
extend the analysis of the previous section to topologically
non-trivial, real-valued Ramond-Ramond fields. While we do not have formal arguments that this is a proper definition of an equivariant
differential cohomology theory, we will see that it matches exactly
with expectations from String theory on orbifolds and also has the
correct limiting properties. For this reason we dub the theory that we
define `orbifold' differential K-theory, defering the terminology
`equivariant' to a more thorough treatment of our model. We will use the Riemann-Roch formula developed in the last section to argue that the Dirac quantization condition for Ramond-Ramond fieldstrengths on a good orbifold is dictated by equivariant K-theory via the equivariant Chern character, suggesting that our orbifold differential K-theory is the correct guess for the space gauge equivalent classes of Ramond-Ramond fields. Finally, we will study the classification of Ramond-Ramond fields on the linear orbifolds considered in section \ref{equivkhom}, and give a definition for the group of flat Ramond-Ramond fields on more general good orbifolds.
\subsection{Orbifold differential K-groups}
As mentioned above, we will generalize the definition of differential K-theory given in section \ref{diffKHS} to accomodate the action of a finite group.\\
First, let us recall some further basic
facts about equivariant K-theory. Similarly to ordinary K-theory, a
model for the classifying space of the functor $\K_{\topsp{G}}^{0}$ is given
by the $\topsp{G}$-algebra of Fredholm operators $\Fred_{\topsp{G}}$ acting on a
separable Hilbert space which is a representation space for $\topsp{G}$ in
which each irreducible representation occurs with infinite
multiplicity~\cite{Atiyah1967}. Then there is an isomorphism
\begin{displaymath}
\K^{0}_{\topsp{G}}(\topsp{X})\simeq\left[\topsp{X},\Fred_{\topsp{G}}\right]_{\topsp{G}}
\end{displaymath}
where $[-,-]_{\topsp{G}}$ denotes the set of equivalence classes of
$\topsp{G}$-homotopic maps, and the isomorphism is given by taking the
index bundle.

There is also a universal space $\text{Vect}^{n}_{\topsp{G}}$, equiped with a
universal $\topsp{G}$-bundle $\widetilde{\topsp{E}}^{n}_{\topsp{G}}$, such that
$\left[\topsp{X},\text{Vect}^{n}_{\topsp{G}}\right]_{\topsp{G}}$ corresponds to the set of
isomorphism classes of $n$-dimensional $\topsp{G}$-vector bundles over
$\topsp{X}$~\cite{Luck1998}. These spaces are constructed as follows. Let
$\cat{E}\topsp{G}$ be the category whose objects are the elements of $\topsp{G}$ and
with exactly one morphism between each pair of objects. The geometric
realization (or nerve) of the set of isomorphism classes in $\cat{E}\topsp{G}$
is, as a simplicial space, the total space of the classifying principal
$\topsp{G}$-bundle $\topsp{E}\topsp{G}$. With $\Vect^n(\pt)$ the category of $n$-dimensional
complex vector spaces $\topsp{V}\simeq\complex^n$, the universal space
$\text{Vect}^{n}_{\topsp{G}}$ is defined to be the geometric realization of
the functor category $[\cat{E}\topsp{G},\Vect^n(\pt)]$. The
universal $n$-dimensional $\topsp{G}$-vector bundle $\widetilde{\topsp{E}}_\topsp{G}^n$ is
then defined as
\beq
\widetilde{\topsp{E}}_\topsp{G}^n=\widetilde{\Vect}{}_\topsp{G}^n
\times_{{\rm GL}(n,\complex)}
\,\complex^n~\longrightarrow~\Vect_\topsp{G}^n \ ,
\label{univGbun}\eeq
where $\widetilde{\Vect}{}_\topsp{G}^n$ is the geometric realization of the
functor category defined as above but with $\Vect^n(\pt)$ replaced
with the category consisting of objects $\topsp{V}$ in $\Vect^n(\pt)$ together
with an oriented basis of $\topsp{V}$.

We assume sufficient regularity conditions on the
infinite-dimensional spaces $\Fred_{\topsp{G}}$ and
$\widetilde{\topsp{E}}^{n}_{\topsp{G}}$. Since $\Fred_\topsp{G}$ and the group completion
$\Vect_\topsp{G}$ are both classifying spaces for equivariant
K-theory, they are $\topsp{G}$-homotopic and we can thereby choose a cocycle
$$u_{\topsp{G}}~\in~{Z_{\topsp{G}}^{\rm even}(\Fred_{\topsp{G}};\real)}$$ representing the
equivariant Chern character of the universal $\topsp{G}$-bundle
(\ref{univGbun}). Generally, the group $Z_{\topsp{G}}^{\rm even}(\topsp{X};\real)$ is
the subgroup of closed cocycles in the complex
\beq
C_{\topsp{G}}^{{\rm even}}(\topsp{X};\mathbb{R}):=\bigoplus_{[g]\in \topsp{G}^\vee}\,
C^{\rm even}\big(\topsp{X}^{g}\,;\,
\mathbb{R}\big)^{Z_{\topsp{G}}(g)}
\label{CGevenXR}\eeq
which, by the results of section \ref{delocKtheory}, is a cochain model
for the Bredon cohomology group $\H_{\topsp{G}}^{{\rm even}-1}(\topsp{X};\mathbb{R}
\otimes{\underline{\topsp{R}}}(-))$. The equivariant Chern character is
understood to be composed with the delocalizing isomorphism of
section~\ref{delocKtheory}. Since it is a natural homomorphism, for any
$\topsp{G}$-bundle $\topsp{E}\to \topsp{X}$ classified by a $\topsp{G}$-map $f:\topsp{X}\to\Fred_\topsp{G}$ one has
\begin{displaymath}
\ch_{\topsp{X}}\big([\topsp{E}]\big)=\big[f^{*}u_{\topsp{G}}\big] \ .
\end{displaymath}

\begin{definition}
The \emph{orbifold differential K-theory}
$\check{\K}^{0}_{\topsp{G}}(\topsp{X})$ of the (global) orbifold $[\topsp{X}/\topsp{G}]$ is the group
of triples $(c,h,\omega)$, where $c:\topsp{X}\to\Fred_{\topsp{G}}$ is a $\topsp{G}$-map,
$\omega$ is an element in $\Omega_{\topsp{G},{\rm cl}}^{\rm
  even}(\topsp{X};\mathbb{R})$, and $h$ is an element
in $C_{\topsp{G}}^{{\rm even}-1}(\topsp{X};\mathbb{R})$
satisfying
\beq\label{Eqdeltah}
\delta{h}=\omega-c^{*}u_{\topsp{G}} \ .
\eeq
Two triples $(c_{0},h_{0},\omega_{0})$ and $(c_{1},h_{1},\omega_{1})$
are said to be \emph{equivalent} if there exists a triple
$(c,h,\omega)$ on $\topsp{X}\times[0,1]$, with the group $\topsp{G}$ acting trivially
on the interval $[0,1]$ and with $\omega$ constant along $[0,1]$, such
that
$$
(c,h,\omega)\big|_{t=0}\=(c_{0},h_{0},\omega_{0}) \qquad \mbox{and}
\qquad (c,h,\omega)\big|_{t=1}\=(c_{1},h_{1},\omega_{1}) \ .
$$
\label{orbdiffKdef}\end{definition}
In {}~(\ref{Eqdeltah}) the closed orbifold differential form $\omega$
is regarded as an orbifold cochain in the complex (\ref{CGevenXR}) by
applying the de~Rham map componentwise on the fixed point submanifolds
$\topsp{X}^g$, $g\in \topsp{G}$. The higher orbifold differential K-theory groups
$\check{\K}^{-n}_\topsp{G}(\topsp{X})$ are defined analogously to those of
section~\ref{diffKHS}. To confirm that this is a suitable
extension of the ordinary differential K-theory of $\topsp{X}$, we should show
that the orbifold differential K-theory groups fit into exact
sequences which reduce to those given by (\ref{exactdiffK}) when $\topsp{G}$ is taken to be the trivial
group. For this, we define the group
\begin{displaymath}
A_{\K_\topsp{G}}^{0}(\topsp{X}):=\left\{(\xi,\omega)\in \K_\topsp{G}^{0}(\topsp{X})\times\Omega_{\topsp{G},{\rm
    cl}}^{\rm even}(\topsp{X};\mathbb{R}
)~\big|~ \ch_\topsp{X}(\xi)=[\omega]_{\topsp{G}-{\rm dR}}\right\} \ .
\end{displaymath}

\begin{theorem}
The orbifold differential K-theory group $\check{\K}^{0}_{\topsp{G}}(\topsp{X})$
satisfies the exact sequence
\bea
&&0~\longrightarrow~{\frac{\H_{\topsp{G}}^{{\rm even}-1}\big(\topsp{X}\,;\,\mathbb{R}
\otimes{\underline{\topsp{R}}}(-)\big)}{\ch_\topsp{X}\big({\K_{\topsp{G}}^{-1}(\topsp{X})}\big)}}~
\longrightarrow~{\check{\K}_{\topsp{G}}^{0}(\topsp{X})}~\longrightarrow~
{A^{0}_{\K_{\topsp{G}}}(\topsp{X})}~\longrightarrow~{0} \ 
\label{exseqeqgen}\eea
\label{exseqthm}\end{theorem}

\begin{proof}
Consider the subgroup of $\H_\topsp{G}^{{\rm
    even}-1}(\topsp{X};\mathbb{R}\otimes\underline{\topsp{R}}(-))$ defined as the
image of the equivariant K-theory group $\K_\topsp{G}^{-1}(\topsp{X})$ under the Chern
character $\ch_\topsp{X}$. It consists of Bredon cohomology classes of the
form $[\tilde{c}^{*}u_\topsp{G}^{-1}]$, where
$\tilde{c}:\topsp{X}\to{\Omega\Fred_\topsp{G}}$. There is a surjective map
\begin{eqnarray*}
f\,:\,\check{\K}_\topsp{G}^{0}(\topsp{X}) &\longrightarrow& A_{\K_\topsp{G}}^{0}(\topsp{X}) \\
\big[(c,h,\omega)\big] &\longmapsto& \big([c]\,,\,\omega\big)
\end{eqnarray*}
which is a well-defined homomorphism, {i.e.}, it does not depend
on the chosen representative of the orbifold differential K-theory
class. By definition, the kernel of $f$ consists of triples of the form
$(\,\underline{c}\,,h,0)$, where $\underline{c}$ is G-homotopic to the constant (identity) map. We also define the map
\begin{eqnarray*}
g\,:\,\H^{{\rm even}-1}_\topsp{G}\big(\topsp{X}\,;\,\mathbb{R}\otimes\underline{\topsp{R}}(-)
\big)&\longrightarrow& {\check{\K}_\topsp{G}^{0}(\topsp{X})}\\
\left[h\right] &\longmapsto& \big[(\,\underline{c}\,,h,0)\big] \ ,
\end{eqnarray*}
which is a well-defined homomorphism because the class
$\left[(\,\underline{c}\,,h,0)\right]$ depends only on the Bredon
cohomology class $[h]\in{\H_\topsp{G}^{{\rm
      even}-1}(\topsp{X};\mathbb{R}\otimes\underline{\topsp{R}}(-))}$.
Then by construction one has $\im({g})=\ker({f})$.

The homomorphism $g$ is not injective. To determine the kernel of $g$, we
use the fact that the zero element in $\check{\K}^{0}_\topsp{G}(\topsp{X})$ can be
represented as
\begin{displaymath}
\big[(\,\underline{c}\,,0,0)\big]=
\big[(\,\underline{c}\,,\pi_{*}\,F^{*}u_\topsp{G}+\dd_{\topsp{G}}\sigma,0)\big]
\end{displaymath}
with $F:\topsp{X}\times{\topsp{S}^1}\to\Fred_\topsp{G}$ and $\sigma\in\Omega^{{\rm
    even}-2}_\topsp{G}(\topsp{X};\real)$. To the map $F$ we can associate a map
$\tilde{c}:\topsp{X}\to\Omega\Fred_\topsp{G}$ such that
$F=\ev\circ(\tilde{c}\times{\Id_{\topsp{S}^1}})$. This follows from the
isomorphism
\begin{displaymath}
\K_\topsp{G}^{-1}(\topsp{X})\simeq\ker\left({i^{*}}:\K_\topsp{G}^{0}(\topsp{X}\times{\topsp{S}^1})\to{\K_\topsp{G}^{0}(\topsp{X})}
\right)
\end{displaymath}
where $i$ is the inclusion
$i:\topsp{X}\hookrightarrow{\topsp{X}\times{\pt}}\subset{\topsp{X}\times{\topsp{S}^1}}$. Now use the
fact that at the level of (real) Bredon cohomology one has an
equality
\begin{displaymath}
\pi_{*}\,\big(\tilde{c}\times{\Id_{\topsp{S}^1}}\big)^{*}=\tilde{c}{}^{*}\,\Pi_{*}
\end{displaymath}
since the projection homomorphisms $\pi_{*}$ and $\Pi_{*}$ both
correspond to integration (slant product) along the $\topsp{S}^1$ fibre. Then
one has the identity
\begin{displaymath}
\big[\pi_{*}\,F^{*}u_\topsp{G}\big]\=
\big[\pi_{*}\,(\tilde{c}\times{\Id_{\topsp{S}^1}})^{*}~
\ev^{*}u_\topsp{G}\big]\=\big[\tilde{c}^{*}\,\Pi_{*}~\ev^{*}u_\topsp{G}\big]\=
\big[\tilde{c}^{*}u_\topsp{G}^{-1}\big] \ .
\end{displaymath}
It follows that $\ker({g})$ is exactly the group
$\ch_\topsp{X}({\K_\topsp{G}^{-1}(\topsp{X})})$, and putting everything together we arrive at (\ref{exseqeqgen}).
\end{proof}
The torus $\H_\topsp{G}^{{\rm
    even}-1}\big(\topsp{X};\mathbb{R}\otimes\underline{\topsp{R}}(-)\big)/
\ch_\topsp{X}\big(\K_\topsp{G}^{-1}(\topsp{X})\big)\simeq \K_\topsp{G}^{-1}(\topsp{X})\otimes\real/\zed$
is called the group of \emph{topologically trivial flat fields}.\\
Consider now the characteristic class map
\begin{eqnarray*}
f_{\rm cc}\,:\,\check{\K}_\topsp{G}^{0}(\topsp{X})&\longrightarrow&{\K_\topsp{G}^{0}(\topsp{X})}\\
\big[(c,h,\omega)\big]&\longmapsto&\left[{c}\right]
 \end{eqnarray*}
and the map
\begin{eqnarray*}
g_{\rm cc}\,:\,\Omega_\topsp{G}^{{\rm even}-1}(\topsp{X};\mathbb{R})&\longrightarrow&
\check{\K}_\topsp{G}^{0}(\topsp{X})\\
h&\longmapsto&\big[(\,\underline{c}\,,h,\dd_{\topsp{G}}h)\big] \ .
\end{eqnarray*}
Let $\Omega_{\K_\topsp{G}}^{{\rm
    even}-1}(\topsp{X};\mathbb{R})$ be the
subgroup of elements in $\Omega_{\topsp{G},{\rm cl}}^{{\rm
    even}-1}(\topsp{X};\mathbb{R})$ whose Bredon
cohomology class lies in $\ch_\topsp{X}({\K_\topsp{G}^{-1}(\topsp{X})})$. Then by using
arguments similar to those used in arriving at the sequence
(\ref{exseqeqgen}), one finds the
\begin{cor}[{\bf Characteristic class exact sequence}]
In analogy with the ordinary case, the orbifold differential K-theory group $\check{\K}^{0}_{\topsp{G}}(\topsp{X})$
satisfies the exact sequence
\beq
0~\longrightarrow~\dfrac{\Omega_{\topsp{G}}^{{\rm even}-1}(\topsp{X};
\mathbb{R})}{\Omega_{\K_{\topsp{G}}}^{{\rm
    even}-1}(\topsp{X};\mathbb{R})}~
\longrightarrow~{\check{\K}_{\topsp{G}}^{0}(\topsp{X})}~\longrightarrow~
{\K_{\topsp{G}}^{0}(\topsp{X})}~\longrightarrow~{0} \ .
\label{charclasseq}\eeq
\label{charclasscorr}\end{cor}

The quotient space of orbifold differential forms in the exact
sequence (\ref{charclasseq}) is called the group of
topologically trivial fields.\\
Finally, consider the field strength map
\bea \nonumber
f_{\rm fs}\,:\,\check{\K}_\topsp{G}^{0}(\topsp{X})&\longrightarrow&
{\Omega_{\topsp{G},{\rm cl}}^{\rm
    even}(\topsp{X};\mathbb{R})}\\ 
\big[(c,h,\omega)\big]&\longmapsto&\omega \ .
\label{fieldstrmap}\eea
The kernel of the homomorphism $f_{\rm fs}$ is the group which
classifies the flat fields (which are not necessarily
topologically trivial) and is denoted
$\K_{\topsp{G}}^{-1}(\topsp{X};\mathbb{R}/\mathbb{Z})$. This group will be described
in more detail in the next section, where we shall also conjecture an
essentially purely algebraic definition of
$\K_{\topsp{G}}^{-1}(\topsp{X};\mathbb{R}/\mathbb{Z})$ which explains the notation. In
any case, we have the

\begin{cor}[{\bf Field strength exact sequence}]
The orbifold differential K-theory group $\check{\K}^{0}_{\topsp{G}}(\topsp{X})$
satisfies the exact sequence
\beq
0~\longrightarrow~{\K_{\topsp{G}}^{-1}(\topsp{X};\mathbb{R}/\mathbb{Z})}~
\longrightarrow~{\check{\K}_{\topsp{G}}^{0}(\topsp{X})}~\longrightarrow~
{\Omega_{\K_\topsp{G}}^{\rm even}(\topsp{X};\mathbb{R})}~\longrightarrow~{0} \ 
\label{diffKexseqGfs}\eeq
\label{fieldstrenghtcorr}\end{cor}
Higher orbifold differential K-theory groups satisfy analogous exact
sequences, with the appropriate degree shifts throughout. It is clear
from our definition that one recovers the ordinary differential
K-theory groups in the case of the trivial group $\topsp{G}={e}$, and in this
sense our orbifold differential K-theory may be regarded as its
equivariant generalization. At this point we hasten to add that,
although our groups are well-defined and satisfy desired properties
which are useful for physical applications such as
the various exact sequences above, we have not showed that our
orbifold theory generalizes \emph{all} the properties of an ordinary differential cohomology theory. For instance, it would be
interesting to define a ring structure and an integration on
$\check{\K}^*_{\topsp{G}}(\topsp{X})$. We have not developed these constructions in this thesis.
\subsection{Flux quantization of orbifold Ramond-Ramond fields}
In this section we will argue the flux quantization condition for Ramond-Ramond fieldstrengths on orbifolds of Type~II superstring theory with vanishing $H$-flux, showing that it is dictated by equivariant K-theory via the complex equivariant Chern character. Essentially, we will closely follow the Moore-Witten argument discussed in section \ref{moorewitten}.\\
Suppose
that our spacetime $\topsp{X}$ is a non-compact $\topsp{G}$-manifold. Suppose further
that there are D-branes present in Type~II superstring theory on
$\topsp{X}/\topsp{G}$. Their Ramond-Ramond charges are classified by the equivariant
K-theory $\K^i_{\topsp{G},\cpt}(\topsp{X})$ with compact support, where $i=0$ in
Type~IIB theory and $i=-1$ in Type~IIA theory.

We require that the brane be a source for the equation of motion for
the total Ramond-Ramond field strength $\omega$. This means that it
creates a Ramond-Ramond current $J$. If we require that the
worldvolume $\topsp{W}$ be compact in equivariant K-cycles
$(\topsp{W},\topsp{E},f)\in\cat{D}^\topsp{G}(\topsp{X})$, then $J$ is supported in the interior
$\mathring{\topsp{X}}$ of $\topsp{X}$. Let $\topsp{X}_\infty$ be the ``boundary of $\topsp{X}$ at
infinity'', which we assume is preserved by the action of $\topsp{G}$. Then
$\K^*_{\topsp{G},\cpt}(\topsp{X})\simeq{\K^*_{\topsp{G}}(\topsp{X},\topsp{X}_\infty)}$. Since $J$ is
trivialized by $\omega$ in $\mathring{\topsp{X}}$, the D-brane charge lives in
the kernel of the natural forgetful homomorphism
\beq\label{forgetmap}
\mathfrak{f}^*\,:\,\K^*_{\topsp{G},\cpt}(\topsp{X})~\longrightarrow~
{\K^*_{\topsp{G}}(\topsp{X})}
\eeq
induced by the inclusion
$(\topsp{X},\emptyset)\hookrightarrow(\topsp{X},\topsp{X}_\infty)$. We denote by
$i:\topsp{X}_\infty\hookrightarrow \topsp{X}$ the canonical inclusion.

The long exact sequence for the pair $(\topsp{X},\topsp{X}_\infty)$ in equivariant
K-theory truncates, by Bott periodicity, to the six-term exact
sequence
\begin{displaymath}
\xymatrix{\K_{\topsp{G}}^{-1}(\topsp{X}_\infty)~\ar[r]&~{\K_{\topsp{G}}^{0}(\topsp{X},\topsp{X}_\infty)}~
\ar[r]^{~~~~~~ \ \ \mathfrak{f}^0}& ~\K_{\topsp{G}}^{0}(\topsp{X})\ar[d]^{i^*}\\
\K_{\topsp{G}}^{-1}(\topsp{X})~\ar[u]^{i^*}&\ar[l]^{\!\!\mathfrak{f}^{-1}}~\K_\topsp{G}^{-1}(\topsp{X},\topsp{X}_\infty)~&
\ar[l]~\K_{\topsp{G}}^{0}(\topsp{X}_\infty) \ . }
\end{displaymath}
It follows that the charge groups are given by
\begin{displaymath}
\ker\big(\mathfrak{f}^0\big)~\simeq~\frac{\K_{\topsp{G}}^{-1}(\topsp{X}_\infty)}
{i^*\big(\K_{\topsp{G}}^{-1}(\topsp{X})\big)} \qquad \mbox{and} \qquad
\ker\big(\mathfrak{f}^{-1}\big)~\simeq~\frac{\K_{\topsp{G}}^0(\topsp{X}_\infty)}
{i^*\big(\K_{\topsp{G}}^0(\topsp{X})\big)} \ .
\end{displaymath}
This formula means that the group of Type~IIB (resp.~Type~IIA) brane
charges is measured by the group $\K_\topsp{G}^{-1}(\topsp{X}_\infty)$
(resp.~$\K_\topsp{G}^0(\topsp{X}_\infty)$) of ``orbifold Ramond-Ramond fluxes at
infinity'' which cannot be extended to all of spacetime $\topsp{X}$. We may
then interpret, for arbitrary spacetimes $\topsp{X}$, the
group $\K_{\topsp{G}}^{-1}(\topsp{X})$ (resp.~$\K_{\topsp{G}}^{0}(\topsp{X})$) as the group classifying
Ramond-Ramond fields in the orbifold $\topsp{X}/\topsp{G}$ in absence of branes in Type~IIB (resp.~Type~IIA) String theory.

The Ramond-Ramond current can be described explicitly by using the delocalization of Bredon cohomology over $\mathbb{C}$. The Wess-Zumino
pairing (\ref{coupling2}) between a topologically trivial, complex
Ramond-Ramond potential and a D-brane represented by an equivariant
K-cycle $(\topsp{W},\topsp{E},f)\in\cat{D}^\topsp{G}(\topsp{X})$ contributes a source term to the
Ramond-Ramond equations of motion, which is the class
$$\big[Q(\topsp{W},\topsp{E},f)\big]~\in~\H_\topsp{G}^{\rm
  even}\big(\topsp{X}\,;\,\complex\otimes\underline{\topsp{R}}(-)\big)$$ represented
by the pushforward
$$
Q(\topsp{W},\topsp{E},f)=f_!^{\H_\topsp{G}}\big(\ch^\complex(\topsp{E})\wedge_\topsp{G}\,\mathcal{R}(\topsp{W},f)\big)
\ .
$$
We now use the Riemann-Roch formula (\ref{GRRthm}) and the fact that $f_!^{\H_\topsp{G}}\circ
f^*=\Id_{\H_\topsp{G}^*(\topsp{X};\,\underline{\complex}(-))}$. Using the explicit
expression for the curvature form and the definition for $\Lambda_\topsp{G}(\topsp{X})$ given in \cite{mine}, we can
then rewrite this class as
\beq
Q(\topsp{W},\topsp{E},f)=\ch^\complex\big(f_!^{\K_\topsp{G}}(\topsp{E})\big)
\wedge_\topsp{G}\,\sqrt{\Todd_\topsp{G}(\topsp{T}_\topsp{X})\wedge_\topsp{G}\,\Lambda_\topsp{G}(\topsp{X})} \ .
\label{QWEfclass}\eeq
This is the complex Bredon cohomology class of the Ramond-Ramond
current $J$ created by the D-brane $(\topsp{W},\topsp{E},f)$. In the case $\topsp{G}=e$, the
expression (\ref{QWEfclass}) reduces to the standard class of the
current for Ramond-Ramond fields in Type~II superstring theory on
$\topsp{X}$~\cite{currents,Minasian1997,Moore2000,Olsen1999}. We can then formally conclude, in analogy with the non-equivariant
case, that the complex Bredon cohomology class associated to a class
$\xi\in{\K^*_{\topsp{G}}(\topsp{X})\otimes\complex}$ representing a
Ramond-Ramond field is assigned by the equivariant Chern character, and that the total Ramond-Ramond fieldstrength $\omega$ associated to $\xi$ satisfies
\beq\label{FCxich}
\frac{[\omega(\xi)]}{2\pi\,[\sqrt{\Todd_\topsp{G}(\topsp{T}_\topsp{X})\wedge_\topsp{G}\,\Lambda_\topsp{G}(\topsp{X})}]}=
[\ch^\complex(\xi)] \ .
\eeq
The above expression the suggests that the orbifold differential K-theory developed in the previous section is a suitable candidate to describe Ramond-Ramond fields with Dirac quantization condition, at least in the complex case. Moreover, we should stress that this analysis of the delocalized theory assumes
the strong conditions spelled out in section~\ref{Gravcoupl}, which
require a deep geometrical compatibility of the equivariant K-cycle
$(\topsp{W},\topsp{E},f)$ with the orbifold structure of $[\topsp{X}/\topsp{G}]$. The example of the linear orbifolds considered in  section ~\ref{equivkhom} is simple enough to satisfy these
conditions. It would be very interesting to find a geometrically
non-trivial explicit example to test these requirements on. However, the linear orbifolds case is very useful to understand certain aspects of the orbifold differential K-theory groups.\\
Since the $\complex$-linear $\topsp{G}$-module $\topsp{X}$ is equivariantly contractible, one
has 
\begin{displaymath}
\H_\topsp{G}^{\rm odd}(\topsp{X};\real\otimes\underline{\topsp{R}}(-))=0
\end{displaymath}
 and
$\K_\topsp{G}^0(\topsp{X})=\topsp{R}(\topsp{G})$. From Theorem~\ref{exseqthm} it then follows that
$$
\check{\K}_\topsp{G}^0(\topsp{X})~\simeq~A_{\K_\topsp{G}}^0(\topsp{X})~
\simeq~\big\{(\gamma,\omega)\in \topsp{R}(\topsp{G})\times
\Omega_{\topsp{G},{\rm cl}}^{\rm even}(\topsp{X};\real)~
\big|~\ch_{\rm pt}(\gamma)=[\omega]_{\topsp{G}-{\rm dR}}\big\} \ .
$$
Since the equivariant Chern character $\ch_{\topsp{G}/H}:\topsp{R(H)}\to \topsp{R(H)}$ for
$\topsp{H}\leq \topsp{G}$ is the identity map, the setwise fibre product truncates to
the lattice of quantized orbifold differential forms and one has
\beq
\check{\K}_\topsp{G}^0(\topsp{V})=\Omega_{\K_\topsp{G}}^{\rm even}(\topsp{V};\real) \ .
\label{checkKG0V}\eeq
This is the group of Type~IIA Ramond-Ramond fieldstrengths on $\topsp{X}$. It
naturally contains those fields which trivialize the Ramond-Ramond
currents sourced by the stable fractional D0-branes of the Type~IIA
theory, corresponding to characteristic classes $[c]$ in the
representation ring $\topsp{R}(\topsp{G})$.

This can be explicitly described as an extension of the
equivariant K-theory of $\topsp{X}$ by the group of
topologically trivial Ramond-Ramond fields $C$ of odd degree, as implied by
Corollary~\ref{charclasscorr}. Since $\topsp{X}$ is connected and
$\topsp{G}$-contractible, one has $\Omega_{\topsp{G},{\rm cl}}^0(\topsp{X};\real)=\real\otimes
\topsp{R}(\topsp{G})$ and the group (\ref{checkKG0V}) has a natural splitting
\beq
\check{\K}_\topsp{G}^0(\topsp{X})=\topsp{R}(\topsp{G})\oplus\Big(\,\bigoplus_{k=1}^d\,
\Omega_{\topsp{G},{\rm cl}}^{2k}(\topsp{X};\real)\Big) \ .
\label{checkKG0Vsplit}\eeq
Any closed orbifold form $\omega$ on $\topsp{X}$ of positive degree
is exact, $\omega=\dd_{\topsp{G}}C$, with the gauge invariance $C\mapsto
C+\dd_{\topsp{G}}\xi$. It follows that there is a natural map
$$
\bigoplus_{k=1}^d\,\Omega_{\topsp{G},{\rm cl}}^{2k}(\topsp{X};\real)~\longrightarrow~
\frac{\Omega_\topsp{G}^{{\rm even}-1}(\topsp{X};\real)}{\Omega_{\K_\topsp{G}}^{{\rm even}-1}(\topsp{X};\real)}
$$
which associates to the field strength $\omega$ the corresponding
globally well-defined Ramond-Ramond potential $C$.

On the other hand, the orbifold differential K-theory group
$\check{\K}_\topsp{G}^{-1}(\topsp{X})$ of Type~IIB Ramond-Ramond fields on $\topsp{X}$ can be
computed by using the characteristic class exact sequence
(\ref{charclasseq}) with degree shifted by $-1$. Using
$\K_\topsp{G}^{-1}(\topsp{X})=0$, one finds
\beq
{\check{\K}_{\topsp{G}}^{-1}(\topsp{X})}=\dfrac{\Omega_{\topsp{G}}^{\rm even}(\topsp{X};
\mathbb{R})}{\Omega_{\K_{\topsp{G}}}^{{\rm even}}(\topsp{X};\mathbb{R})} \ .
\label{checkKG1V}\eeq
This result reflects the fact that the Type~IIB theory has no stable
fractional D0-branes. Hence there is no extension and the
Ramond-Ramond fields are induced solely by the closed string
background. Their field strengths $\omega=\dd_{\topsp{G}}C$ are
determined entirely by the potentials $C$, which are globally defined
differential forms of even degree.\\
Setting $\topsp{X}={\rm pt}$, we obtain the orbifold differential K-theory of the point, which are given by
\begin{displaymath}
\check{\K}_\topsp{G}^0({\rm pt})~\simeq~\topsp{R}(\topsp{G})
\end{displaymath}
and
\begin{displaymath}
\check{\K}_\topsp{G}^{-1}({\rm pt})~\simeq~\topsp{R}(\topsp{G})\otimes{\mathbb{R}}/\mathbb{Z}
\end{displaymath}
For $\topsp{G}=e$, these groups reduce to the usual differential K-teory groups of the point.\\

In the previous section we defined the group of flat Ramond-Ramond fields on a general orbifold $[\topsp{X}/\topsp{G}]$ as the subgroup of orbifold differential K-theory with vanishing
curvature. In the following we will conjecture a very natural
algebraic definition of these groups which ties them somewhat more
directly to equivariant K-theory groups.\\
To motivate this conjecture, we first compute the groups
$\K_\topsp{G}^*(\topsp{X};\real/\zed)$ for the linear orbifolds above, wherein the associated differential
K-theory groups were determined explicitly. Using the field strength
exact sequence (\ref{diffKexseqGfs}), by definition one has
$$
\K_\topsp{G}^{-1}(\topsp{X};\real/\zed)\simeq\ker\big(f_{\rm fs}:\check{\K}_\topsp{G}^0(\topsp{X})\to
\Omega_{\K_\topsp{G}}^{\rm even}(\topsp{X};\real)\big)
$$
which from the natural isomorphism (\ref{checkKG0V}) trivially gives
\beq
\K_\topsp{G}^{-1}(\topsp{X};\real/\zed)=0 \ .
\label{KGflat1V}\eeq
Similarly, using $\K_\topsp{G}^{-1}(\topsp{X})=0$ one has
$$
\K_\topsp{G}^0(\topsp{X};\real/\zed)\simeq\ker\big(f_{\rm fs}:\check{\K}_\topsp{G}^{-1}(\topsp{X})\to
\Omega_{\topsp{G},{\rm cl}}^{\rm odd}(\topsp{X};\real)\big) \ .
$$
Using the natural isomorphism (\ref{checkKG1V}), the field strength
map is $$f_{\rm fs}\big([C]\big)\=\dd_{\topsp{G}}C \qquad \mbox{for}
\quad C\in\Omega_{\topsp{G}}^{\rm
  even}(\topsp{X};\real) \ , $$ giving
$$
\K_\topsp{G}^0(\topsp{X};\real/\zed)\simeq\frac{\Omega_{\topsp{G},{\rm cl}}^{\rm even}(\topsp{X};
\mathbb{R})}{\Omega_{\K_{\topsp{G}}}^{{\rm even}}(\topsp{X};\mathbb{R})} \ .
$$
Similarly to {}~(\ref{checkKG0Vsplit}), there is a natural splitting
of the vector space of closed orbifold differential forms given by
$$
\Omega_{\topsp{G},{\rm cl}}^{\rm even}(\topsp{X};\mathbb{R})=
\big(\topsp{R}(\topsp{G})\otimes\real\big)~\oplus~\Big(\,\bigoplus_{k=1}^d\,
\Omega_{\topsp{G},{\rm cl}}^{2k}(\topsp{X};\real)\Big)
$$
and we arrive finally at
\beq
\K_\topsp{G}^0(\topsp{X};\real/\zed)=\topsp{R}(\topsp{G})\otimes\real/\zed \ .
\label{KGflat0V}\eeq

These results of course simply follow from the fact that $\topsp{X}$ is
$\topsp{G}$-contractible, so that every $\dd_{\topsp{G}}$-closed Ramond-Ramond
field is trivial, except in degree zero where the gauge equivalence
classes are naturally parametrized by the twisted sectors of the
String theory in~(\ref{KGflat0V}). Note that both groups of flat
fields (\ref{KGflat1V}) and (\ref{KGflat0V}) are unchanged by
(equivariant) contraction of the $\topsp{G}$-module $\topsp{X}$ to a point, as an
analogous (but simpler) calculation shows. This suggests that the
groups $\K_\topsp{G}^*(\topsp{X};\real/\zed)$ have at least some $\topsp{G}$-homotopy
invariance properties, unlike the differential $\K_\topsp{G}$-theory
groups. This motivates the following conjectural algebraic framework
for describing these groups.

We will propose that the group $\K_\topsp{G}^*(\topsp{X};\real/\zed)$ is an
\emph{extension} of the torus of topologically trivial flat orbifold
Ramond-Ramond fields by the torsion elements in $\K_\topsp{G}^{*+1}(\topsp{X})$,
as they have vanishing image under the equivariant Chern character
$\ch_\topsp{X}$. The resulting group may be called the ``equivariant K-theory
with coefficients in $\real/\zed$''. The short exact sequence of
coefficient groups
$$
0~\longrightarrow~\zed~\longrightarrow~\real~\longrightarrow~
\real/\zed~\longrightarrow~0
$$
induces a long exact sequence of equivariant K-theory groups which, by
Bott periodicity, truncates to the six-term exact sequence
\beq\label{Bocksteinseq}
\xymatrix{\K_{\topsp{G}}^{0}(\topsp{X})~\ar[r]&~{\K_{\topsp{G}}^{0}(\topsp{X};\real)}~
\ar[r]& ~\K_{\topsp{G}}^{0}(\topsp{X};\real/\zed)\ar[d]^{\beta}\\
\K_{\topsp{G}}^{-1}(\topsp{X};\real/\zed)~\ar[u]^{\beta}&\ar[l]~\K_\topsp{G}^{-1}(\topsp{X};\real)~&
\ar[l]~\K_{\topsp{G}}^{-1}(\topsp{X}) \ . }
\eeq
The connecting homomorphism $\beta$ is a suitable variant of the usual
Bockstein homomorphism. We assume that the equivariant K-theory with
real coefficients is defined simply by the $\zed_2$-graded ring
$$\K_\topsp{G}^*(\topsp{X};\real)\=\K_\topsp{G}^*(\topsp{X})\otimes\real~\simeq~
\H_\topsp{G}^*\big(\topsp{X}\,;\,\real\otimes\underline{\topsp{R}}(-)\big) \ , $$ where
we have used Theorem~\ref{eqChernthm}. The maps to real K-theory in
{}~(\ref{Bocksteinseq}) may then be identified with the equivariant
Chern character $\ch_\topsp{X}$, whose image is a full lattice in the Bredon
cohomology group $\H_\topsp{G}^*(\topsp{X};\real\otimes\underline{\topsp{R}}(-))$. Then
the abelian group $\K_\topsp{G}^*(\topsp{X};\mathbb{R}/\mathbb{Z})$ sits in the
exact sequence
\beq
0~\longrightarrow~\K_\topsp{G}^*(\topsp{X})\otimes\real/\zed~\longrightarrow~
\K_\topsp{G}^*(\topsp{X};\real/\zed)~\xrightarrow{\beta}~\Tor
\big(\K_\topsp{G}^{*+1}(\topsp{X})\big)~\longrightarrow~0 \ .
\label{RZcoeffsdef}\eeq

When $\topsp{G}=e$, (\ref{Bocksteinseq}) is the usual Bockstein exact
sequence for K-theory. In this case, an explicit geometric
realization of the groups $\K^*(\topsp{X};\real/\zed)$ in terms of
bundles with connection has been given by Lott~\cite{Lott1994}. 
Moreover, in ~\cite{Hopkins2005} a geometric construction of
the map ${\K^{-1}(\topsp{X};\mathbb{R}/\mathbb{Z})}\to{\check{\K}^{0}(\topsp{X})}$ in
the field strength exact sequence is given. Unfortunately, no such
geometrical description is immediately available for our equivariant
differential K-theory, due to the lack of a Chern-Weil theory for the
homotopy theoretic equivariant Chern character. Our conjectural definition (\ref{RZcoeffsdef})
is satisfied by the linear orbifold groups (\ref{KGflat1V})
and~(\ref{KGflat0V}).

In~\cite{deBoer2001} a very different definition of the groups
$\K_\topsp{G}^*(\topsp{X};\real/\zed)$ is given, by defining both equivariant
K-theory and cohomology using the Borel construction of
Example~\ref{Borelex}. Then the Bockstein exact sequence
(\ref{Bocksteinseq}) is written for the ordinary K-theory groups of
the homotopy quotient $\topsp{X}_\topsp{G}=\topsp{E}\topsp{G}\times_\topsp{G}\topsp{X}$. While these groups reduce,
like ours, to the usual K-theory groups of flat fields when $\topsp{G}=e$,
they do not obey the exact sequence (\ref{RZcoeffsdef}). The reason is
that the equivariant Chern character used is \emph{not} an isomorphism
over the reals, as explained in section~\ref{equivChernsec}. Moreover, an associated differential K-theory construction
would directly involve differential forms on the infinite-dimensional
space $\topsp{X}_\topsp{G}$ which is only homotopic to the finite-dimensional
CW-complex $\topsp{X}/\topsp{G}$. The physical interpretation of such fields is not
clear. Even in the simple case of the linear orbifolds $\topsp{X}$ studied
above, this description predicts an infinite set of equivariant fluxes
of arbitrarily high dimension on the infinite-dimensional classifying
space $B\topsp{G}$, and one must perform some non-canonical quotients in order
to try to isolate the physical fluxes. In contrast, with our
constructions the relation between orbifold flux groups and Bredon
cohomology is much more natural, and it involves 
Ramond-Ramond fields defined on submanifolds of the covering space $\topsp{X}$.   

\section{K-homology and flat fields in Type II String theory}
In this final section, we will briefly comment on how the flat Ramond-Ramond fields can in principle couple to the wrapped D-branes defined in section \ref{wrappedbranes}. The analysis will be restricted to type IIA String theory.\\
As we have seen in the previous sections, the Ramond-Ramond fields in ordinary type IIA String theory are described by differential $\K$ theory. Hence, the group of flat Ramond-Ramond fields is given by ${\rm K}^{-1}(\topsp{X};\mathbb{R}/\mathbb{Z})$. Consider the short exact
sequence of coefficient groups given by
$$
1~\longrightarrow~\zed~\longrightarrow~\real~\longrightarrow~
\real/\zed~\longrightarrow~1 \ ,
$$
It induces the 
corresponding long exact sequence
$$
\cdots~\longrightarrow~\K^i(\topsp{X})~\longrightarrow~\K^i(\topsp{X})\otimes\real~
\longrightarrow~\K^i(\topsp{X};\mathbb{R}/\mathbb{Z})~\longrightarrow~\K^{i+1}(\topsp{X})~
\longrightarrow~\cdots \ .
$$
By truncating the above long exact sequence, we have 
\beq
0~\to{\K}^{-1}(\topsp{X})\otimes\mathbb{R}/\mathbb{Z}~\longrightarrow~
\K^{-1}(\topsp{X};\mathbb{R}/\mathbb{Z})~
\stackrel{\beta}{\longrightarrow}~{\rm Tor}({\K^0(\topsp{X})})~\longrightarrow~0
\label{ChernRRphases}\eeq
where $\beta$ is the Bockstein homomorphism. Thus the identity
component of the circle coefficient K-theory group is the torus
$\K^{-1}(\topsp{X};\mathbb{R}/\mathbb{Z})$. Suppose now that $\K^{-1}(\topsp{X})$ is pure torsion. In this case, $\K^{-1}(\topsp{X};\mathbb{R}/\mathbb{Z})\simeq{\rm Tor}({\K^0(\topsp{X})})$ , and the corresponding flat Ramond-Ramond fields can be
represented by virtual vector bundles over $\topsp{X}$. 
A torsion Ramond-Ramond flat field $\xi\in\K^0(\topsp{X})$ gives an additional phase factor to
a D-brane in the String theory path integral \cite{deBoer2001}. Generally, the origin of
these phases can be understood from the topological classification of
the physical coupling between Type~I D-branes and Ramond-Ramond fields.\\
Let $\topsp{W}$ be a compact ${\rm spin}^{c}$ submanifold of $\topsp{X}$ of dimension $p+1$, and let the spacetime manifold $\topsp{X}$ be a spherical ${\rm spin}^{c}$ fibration $\pi:\topsp{X}\to{\topsp{W}}$ such that $\topsp{X}/\topsp{W}\simeq\topsp{S}^{9-p}$.\\
The group of flat Ramond-Ramond fields is given by
$$
\K^{-1}\big(\topsp{X};\mathbb{R}/\mathbb{Z}\big)=\Hom\big(\K^{\rm t}_{-1}(\topsp{X})\,,\,
\real/\zed\big)\simeq\Hom\big(\K_{p-10}^{\rm t}(\topsp{W})\,,\,\real/\zed\big)
$$
where we have use the following exact sequence \cite{yosimura,Reis2006}
\begin{displaymath}
0~\to~\Ext\big(\K_{i-1}^{\rm t}(\topsp{X})\,,\,\topsp{G}\big)~
\to~\K^{i}\big(\topsp{X};\topsp{G}\big)~\to~
\Hom\big(\K_i^{\rm t}(\topsp{X})\,,\topsp{G}\big)~\to~0
\end{displaymath}
for $\topsp{G}=\mathbb{R}/\mathbb{Z}$, and the Thom isomorphism.\\ 
Using  Bott periodicity, we have finally
\beq
\K^{-1}\big(\topsp{X};\mathbb{R}/\mathbb{Z}\big)\simeq\Hom\big(\K_{p+2}^{\rm t}
(\topsp{W})\,,\,\real/\zed\big) \ .
\label{RRfluxKOtoprel}\eeq
The K-homology group $\K_{p+2}^{\rm t}(\topsp{W})$ consists of wrapped D-branes $[M,\topsp{E},\phi]$ with the properties $\dim M=p+2$ and
$\phi(M)\subset \topsp{W}$. The dimension shift is related to the topological
anomaly in the worldvolume fermion path integral~\cite{Moore2000}, as the following argument seems to suggest.\\
Consider a one-parameter family of $p+1$-dimensional brane
worldvolumes specified by a circle bundle $U\to \topsp{W}$ whose total space
$U$ is a $p+2$-dimensional
submanifold of spacetime $\topsp{X}$ with the topology of $\topsp{W}\times\topsp{S}^1$. Complex vector bundles $\topsp{E}_g$ of rank $n$ over generic fibres $U/\topsp{W}\simeq\topsp{S}^1$
are determined by elements $g\in\U(n)$ by the clutching construction. Thus the family of self-adjoint Atiyah-Singer operators $\,\nslash{\mathfrak{D}}_{\topsp{E}_g}^{\topsp{S}^1}$
determined by (\ref{ASopS1}) is parametrized by the group $\U(n)$. The
anomaly~\cite{Moore2000} arises as the determinant line bundle of
this family, which is essentially defined as the highest exterior
power of the kernel of the family. This defines a non-trivial line bundle on the group $\U(n)$ called the \emph{Pfaffian
line bundle}, which has the property that its lift to $\Spin(n)$ is the
trivial line bundle. One can also construct a connection and
holonomy of the Pfaffian line bundle~\cite{Freed1999}. The manifold $U$ is wrapped by D-branes in $\K_{p+2}^{\rm
  t}(U)$. One can now restrict to the subgroup $\K_{p+2}^{\rm
  t}(\topsp{W})\subset\K_{p+2}^{\rm t}(U)$ by keeping only those D-branes
which are wrapped on the embedding $\topsp{W}\hookrightarrow U$ by the zero
section of $U\to \topsp{W}$. The isomorphism (\ref{RRfluxKOtoprel}) reflects
the fact that the topological anomaly could in principle be cancelled by coupling
D-branes to the Ramond-Ramond fields through a phase factor.\\

\setcounter{section}{0} 

\renewcommand{\thesection}{\Alph{section}}
\newpage
\pagestyle{plain}
\begin{center}\appendix{}\end{center}
\begin{center}
{\Large \textbf{Linear algebra in Functor categories}}
\end{center}
In this appendix we will summarize some notions about algebra in
functor categories that were used in the main text of the paper. They
generalize the more commonly used concepts for modules over a
ring. For further details see \cite{Dieck1987}.

Let $\topsp{R}$ be a commutative ring, and denote the category of (left)
$\topsp{R}$-modules by $\lmod{\topsp{R}}$. Let $\Gamma$ be a \emph{small} category,
\emph{i.e.}, its class of objects ${\rm Obj}(\Gamma)$ is a set. If
$\mathcal{C}$ is another category, then one denotes by
\begin{displaymath}
\left[\Gamma,\mathcal{C}\right]
\end{displaymath}
the \emph{functor category} of (covariant) functors
$\Gamma\to{\mathcal{C}}$. The objects of
$\left[\Gamma,\mathcal{C}\right]$ are (covariant) functors
$\phi:\Gamma\to\mathcal{C}$ and a morphism from $\phi_{1}$ to
$\phi_{2}$ is a natural transformation $\alpha:\phi_{1}\to\phi_{2}$
between functors.

In particular, in the main text we used the functor category
\begin{displaymath}
\lmod{\topsp{R}\Gamma}:=\left[\Gamma,\lmod{\topsp{R}}\right]
\end{displaymath}
whose objects are called \emph{left $\topsp{R}\Gamma$-modules}. If one denotes
with $\Gamma^{\text{op}}$ the dual category to $\Gamma$, then there is
also the functor category
\begin{displaymath}
\rmod{\topsp{R}\Gamma}:=\left[\Gamma^{\text{op}},\lmod{\topsp{R}}\right]
\end{displaymath}
of contravariant functors $\Gamma\to\lmod{\topsp{R}}$, whose objects are
called \emph{right $\topsp{R}\Gamma$-modules}. As an example, let $\topsp{G}$ be a
discrete group regarded as a category with a single object and a
morphism for each element of $\topsp{G}$. A covariant functor $\topsp{G}\to\lmod{\topsp{R}}$
is then the same thing as a left module over the group ring $\topsp{R}[\topsp{G}]$ of
$\topsp{G}$ over $\topsp{R}$.

As the name itself suggests, all standard definitions from the linear
algebra of modules have extensions to this more general setting. For
instance, the notions of \emph{submodule, kernel, cokernel, direct
  sum, coproduct, etc.} can be naturally defined objectwise. If $\topsp{M}$
and $\topsp{N}$ are $\topsp{R}\Gamma$-modules, then $\text{Hom}_{\topsp{R}\Gamma}(\topsp{M},\topsp{N})$ is the
$\topsp{R}$-module of all natural transformations $\topsp{M}\to{\topsp{N}}$. This notation
should not be confused with the one used for the set of all morphisms
between two objects in $\Gamma$, and usually it is clear from the
context.

If $\topsp{M}$ is a right ${\topsp{R}\Gamma}$-module and $\topsp{N}$ is a left
${\topsp{R}\Gamma}$-module, then one can define their categorical \emph{tensor
  product}
\begin{displaymath}
\topsp{M}\otimes_{\topsp{R}\Gamma}\topsp{N}
\end{displaymath}
in the following way. It is the $\topsp{R}$-module given by first forming the
direct sum
\begin{displaymath}
\topsp{F}=\bigoplus_{\lambda\in\text{Obj}(\Gamma)}\,
\topsp{M}(\lambda)\otimes_{\topsp{R}}\topsp{N}(\lambda)
\end{displaymath}
and then quotienting $\topsp{F}$ by the $\topsp{R}$-submodule generated by all
relations of the form
\begin{displaymath}
f^*(m)\otimes{n}-m\otimes{f_*(n)}=0 \ ,
\end{displaymath}
where $(f:\lambda\to{\rho})\in\text{Mor}(\Gamma)$,
$m\in{\topsp{M}(\rho)},\:n\in{\topsp{N}(\lambda)}$ and $f^*(m)=\topsp{M}(f)(m),\:f_*
(n)=\topsp{N}(f)(n)$. This tensor product commutes with coproducts.
If $\topsp{M}$ and $\topsp{N}$ are functors from $\Gamma$ to the category of vector
spaces over a field $\mathbb{K}$, then their tensor product is
naturally equiped with the structure of a vector space over
$\mathbb{K}$. When $\Gamma$ is the orbit category $\ocat{\topsp{G}}$ and
$\topsp{R}=\zed$, the tensor product has precise limiting cases. For an
arbitrary contravariant module $\topsp{M}$ and the constant covariant module
$\topsp{N}$, the categorical product $\topsp{M}\otimes_{\zed\ocat{\topsp{G}}}\topsp{N}$ is the tensor
product of the right $\zed[\topsp{G}]$-module $\topsp{M}(\topsp{G}/e)$ with the constant left
$\zed[\topsp{G}]$-module $\topsp{N}(\topsp{G}/e)$, $\topsp{M}(\topsp{G}/e)\otimes_{\zed[\topsp{G}]}\topsp{N}(\topsp{G}/e)$. On the
other hand, if the contravariant module $\topsp{M}$ is constant and the
covariant module $\topsp{N}$ is arbitrary, then $\topsp{M}\otimes_{\zed\ocat{\topsp{G}}}\topsp{N}$ is
just $\topsp{N}(\topsp{G}/\topsp{G})$.

\newpage
\begin{center}\appendix{}\end{center}
\begin{center}
{\Large \textbf{Clifford algebras and Spin manifolds}}
\end{center}
In this appendix we will briefly recall some basic properties of Clifford algebras and Spin manifolds used throughout this thesis. We direct the reader to \cite{ABS,LM} for an extensive treatment of these topics.\\
Let $\rm V$ be a real vector space, and let $q:{\rm V}\to\mathbb{R}$ be a quadratic form. Consider the tensor algebra of ${\rm V}$
\begin{displaymath}
\mathcal{F}({\rm V}):=\sum_{r=0}^{\infty}\bigotimes^{r}{\rm V}
\end{displaymath}
and denote with ${\rm I}_{q}(\rm V)$ the two-sided ideal in $\mathcal{F}(\rm V)$ generated by all elements of the form \\$v\otimes{v}+q(v)1$ for $v\in{V}$. Then the \emph{Clifford algebra} $\Cl({\rm V};q)$ associated to $V$ and $q$ is the associative algebra with unit defined as
\begin{displaymath}
\Cl({\rm V};q):=\mathcal{F}(\rm V)/{\rm I}_{q}(\rm V)
\end{displaymath}
The algebra $\Cl({\rm V};q)$ is generated by the vector space ${\rm V}\subset\Cl({\rm V};q)$ subject to the relations
\begin{displaymath}
v\cdot{v}=-q(v)1
\end{displaymath}
which give a ``universal'' characterization of the algebra.\\
Given two vector spaces ${\rm V},{\rm V'}$, equiped with quadratic forms $q,q'$, respectively, any linear map $f:{\rm V}\to{\rm V'}$ preserving the quadratic forms induces a homomorphism $\tilde{f}:\Cl({\rm V};q)\to\Cl({\rm V'};q')$. An important example is given by the homomorphism $\tilde{\alpha}$ induced by the map $\alpha(v)=-v$ on $\rm V$. Since $\alpha^{2}={\rm id}$, there is a decomposition
\begin{displaymath}
\Cl({\rm V};q)=\Cl^{0}({\rm V};q)\oplus\Cl^{1}({\rm V};q)
\end{displaymath}
where $\Cl^{i}({\rm V};q):=\left\{\varphi\in\Cl({\rm V};q):\alpha(\varphi)=(-1)^{i}\varphi\right\}$. The vector space $\Cl^{0}({\rm V};q)$ is a subalgebra of $\Cl({\rm V};q)$, and it is called the \emph{even part}, while the subspace $\Cl^{1}({\rm V};q)$ is called the \emph{odd part} of $\Cl({\rm V};q)$.\\
Given a Clifford algebra $\Cl({\rm V};q)$, we can define the \emph{multiplicative group of units} as
\begin{displaymath}
\Cl^{*}({\rm V};q):=\left\{\varphi\in\Cl({\rm V};q):\exists\:\varphi^{-1}\in\Cl({\rm V};q)\:\text{such that}\:\varphi^{-1}\cdot\varphi=\varphi\cdot\varphi^{-1}=1\right\}
\end{displaymath}
When ${\rm dimV}=n<\infty$, $\Cl^{*}({\rm V};q)$ is a Lie group of dimension $2^{n}$. Consider the subgroup ${\rm P}({\rm V};q)\subset\Cl^{*}({\rm V};q)$ generated by the elements $v\in{\rm V}$ with $q(v)\neq{0}$. Then the \emph{Pin group} associated to the pair $({\rm V};q)$ is the subgroup ${\rm Pin}({\rm V};q)$ of ${\rm P}({\rm V};q)$ generated by elements $v\in{V}$ with $q(v)=\pm{1}$. The associated \emph{Spin group} of $({\rm V};q)$ is defined as 
\begin{displaymath}
{\rm Spin}({\rm V};q):={\rm Pin}({\rm V};q)\cap\Cl^{0}({\rm V};q)
\end{displaymath}
We will now restrict to the case in which ${\rm V}=\mathbb{R}^{n}$, and 
\begin{displaymath}
q(x):=x^{2}_{1}+x^{2}_{2}+\dots+x^{2}_{n}
\end{displaymath}
is the quadratic form induced by the usual scalar product. Denote with $\Cl_{n}$ the Clifford algebra $\Cl({\mathbb{R}^{n}};q)$. If $e_{1},e_{2},\cdots,e_{n}$ is any orthonormal basis of $\mathbb{R}^{n}$, then $\Cl_{n}$ is generated as an algebra by $e_{1},e_{2},\cdots,e_{n}$ and $1$ subject to the relations
\begin{displaymath}
e_{i}\cdot{e}_{j}+e_{j}\cdot{e}_{j}=-2\delta_{ij}\:1
\end{displaymath}
It is straightforward to see that $\Cl_{1}\simeq{\mathbb{C}}$, and $\Cl_{2}\simeq{\mathbb{H}}$. Moreover, one can prove that ${\rm dim}_{\mathbb{R}}(\Cl_{n})=2^{n}$. The decomposition of $\Cl_{n}$ in even and odd part induces the following isomorphism
\begin{displaymath}
\Cl_{n}\simeq\Cl^{0}_{n+1}
\end{displaymath}
An important result concerning the Clifford algebra $\Cl_{n}$ is the following. There is a canonical isomorphism $\Cl_{n}\simeq\Lambda^{*}\mathbb{R}^{n}$, according to which the Clifford multiplication between $v\in\mathbb{R}^{n}$ and any $\varphi\in\Cl_{n}$ can be written as
\begin{displaymath}
v\cdot\varphi\simeq{v\wedge\varphi}-v\lfloor{\varphi}
\end{displaymath}
where we have identified $\mathbb{R}^{n}$ with its dual via the scalar product, and $\lfloor$ denotes the contraction of the vector $v$ with an element of $\Lambda^{*}\mathbb{R}^{n}$. The Clifford algebras can be described as matrix algebras over $\mathbb{R},\mathbb{C}$, or $\mathbb{H}$. Indeed, denote with $\mathbb{C}\ell_{n}$ the complexification $\Cl_{n}\otimes_{\mathbb{R}}\mathbb{C}$ of $\Cl_{n}$. Then for all $n\leq{0}$ we have the \emph{periodicity isomorphisms}
\begin{displaymath}
\begin{array}{c}
\Cl_{n+8}\simeq\Cl_{n}\otimes\Cl_{8}\\
\mathbb{C}\ell_{n+2}\simeq\mathbb{C}\ell_{n}\otimes_{\mathbb{C}}\mathbb{C}\ell_{2}
\end{array}
\end{displaymath}
The above isomorphisms allow to deduce all the algebras $\Cl_{n}$ and $\mathbb{C}\ell_{n}$ from the following table 
\begin{displaymath}
\begin{array}{c|cccccccc}
&1&2&3&4&5&6&7&8\\
\hline
\Cl_{n}&\mathbb{C}&\mathbb{H}&\mathbb{H}\oplus\mathbb{H}&\mathbb{H}(2)&\mathbb{C}(4)&\mathbb{R}(8)&\mathbb{R}(8)\oplus\mathbb{R}(8)&\mathbb{R}(16)\\
\hline
\mathbb{C}\ell_{n}&\mathbb{C}\oplus\mathbb{C}&\mathbb{C}(2)&\mathbb{C}(2)\oplus\mathbb{C}(2)&\mathbb{C}(4)&\mathbb{C}(4)\oplus\mathbb{C}(4)&\mathbb{C}(8)&\mathbb{C}(8)\oplus\mathbb{C}(8)&\mathbb{C}(16)
\end{array}
\end{displaymath}\\
where ${\rm K}(n)$ denotes the algebra of $n\times{n}$ matrices over the field ${\rm K}$, with ${\rm K}=\mathbb{R},\mathbb{C},\mathbb{H}$.\\
The above table also dictates the theory of representations of Clifford algebras. Indeed, if we see ${\rm K}(n)$ as an algebra over $\mathbb{R}$, the natural representation of ${\rm K}(n)$ on the vector space ${\rm K}^{n}$ is the only irreducible real representation of ${\rm K}(n)$ up to isomorphism, while the algebras ${\rm K}(n)\oplus{\rm K}(n)$ have exactly two equivalence classes of irreducible real representations \cite{LM}.\\
Given a module   
An important role is played by the Spin groups associated to $\mathbb{R}^{n}$ with the Euclidean quadratic form. Indeed, if we denote with ${\rm Spin}_{n}$ the group ${\rm Spin}(\mathbb{R}^{n};q)$, we have the following exact sequence
\begin{displaymath}
0\to\mathbb{Z}_{2}\to{\rm Spin}_{n}\xrightarrow{\xi_{0}}{\rm SO}_{n}\to{0}
\end{displaymath}
for all $n\geq{3}$. In particular, the map $\xi_{0}$ denotes the universal covering homomorphism of ${\rm SO}_{n}$. The representation theory of Clifford algebras can be used to construct representations of ${\rm Spin}_{n}\subset\Cl^{0}_{n}$ which are not trivial on the element -1, hence they do not arise from representation of the orthogonal group ${\rm SO}_{n}$. Indeed, we can define the \emph{real spinor representation} of ${\rm Spin}_{n}$ as the homomorphism
\begin{displaymath}
\Delta_{n}:{\rm Spin}_{n}\to{\rm GL}(\rm S)
\end{displaymath}
induced by retricting an irreducible real representation $\Cl_{n}\to\rm Hom_{\mathbb{R}}(S,S)$ to ${\rm Spin}_{n}$.\\
We can also define the \emph{complex spinor representation} of ${\rm Spin}_{n}$ as the homomorphism
\begin{displaymath}
\Delta^{\mathbb{C}}_{n}:{\rm Spin}_{n}\to{\rm GL}_{\mathbb{C}}({\rm S})
\end{displaymath}
induced by restricing an irreducible complex representation $\mathbb{C}\ell_{n}\to{\rm Hom}_{\mathbb{C}}({\rm S},{\rm S})$ to ${\rm Spin}_{n}\subset\Cl^{0}_{n}\subset\mathbb{C}\ell_{n}$. In particular, for $n$ odd, the complex spinor representation $\Delta^{\mathbb{C}}_{n}$ is irreducible, and independent of which irreducible represention of $\mathbb{C}\ell_{n}$ is used. When $n=2m$ , there is a decomposition
\begin{displaymath}
\Delta^{\mathbb{C}}_{2m}=\Delta^{\mathbb{C}+}_{2m}\oplus\Delta^{\mathbb{C}-}_{2m}
\end{displaymath}
into a direct sum of irreducible complex representations of ${\rm Spin}_{n}$. The representation $\Delta^{\mathbb{C}\pm}_{2m}$ is given by the composition of $\Delta^{\mathbb{C}}_{2m}$ with the projection $1\pm\omega_{\mathbb{C}}$, where $$\omega_{\mathbb{C}}=i^{m}e_{1}\cdot{e_{2}}\cdots{e_{2m}}$$ is the complex volume element. Notice that $\Delta^{\mathbb{C}\pm}_{2m}$ is \emph{not} a representation of $\mathbb{C}\ell_{n}$. Similar results hold for the real spinor representation.\\
Finally, a \emph{$\mathbb{Z}_{2}$-graded module} for $\Cl_{n}$ is a module $\rm W$ with a decomposition $\rm W=W^{0}\oplus W^{1}$ such that
\begin{displaymath}
\Cl^{i}_{n}{\rm W}^{j}\subseteq{\rm W}^{(i+j)(\rm mod 2)}
\end{displaymath}
Given a $\mathbb{Z}_{2}$-graded module $\rm W$ over $\Cl_{n}$, the even part ${\rm W}^{0}$ is a module over $\Cl^{0}_{n}$. Conversely, given a module ${\rm W}^{0}$ over $\Cl^{0}_{n}$, we can form the $\mathbb{Z}_{2}$-graded module
\begin{displaymath}
{\rm W}:=\Cl_{n}\otimes_{\Cl^{0}_{n}}{\rm W}^{0}
\end{displaymath}
Hence there is an equivalence between the category of $\mathbb{Z}_{2}$-graded modules over $\Cl_{n}$ and the category of ungraded modules over $\Cl_{n-1}$.\\

Let $\rm E$ be an oriented $n$-dimensional Riemannian vector bundle over a manifold $\rm X$, and suppose $n\geq{3}$. Let $P_{\rm SO}(\topsp{E})$ denote the orthonormal frame bundle of $\rm E$. Then a \emph{spin structure} on $E$ is a principal ${\rm Spin}_{n}$-bundle $P_{\rm Spin}(\rm E)$ together with a double covering
\begin{displaymath}
\xi:P_{\rm Spin}({\rm E})\to{P_{\rm SO}({\rm E})}
\end{displaymath}
such that $\xi(p\cdot{g})=\xi(p)\cdot\xi_{0}(g)$ for all $p\in{P_{\rm Spin}(\topsp{E})}$ and all $g\in{\rm Spin}_{n}$.\\
For $n=2$, a Spin structure on $\rm E$ is defined analogously, with $\rm Spin_{2}$ replaced by $\rm SO_{2}$ and $\xi_{0}:{\rm SO}_{2}\to{\rm SO}_{2}$ the connected double covering.\\
For $n=1$, $P_{\rm SO}(\rm E)\simeq{X}$ and a spin structure is simply defined to be a double covering of $\rm X$.\\
The existence of a spin structure on a vector bundle $E$ is dictated by the vanishing of the first and second Stiefel-Whitney class $w_{1}(\rm E)$ and $w_{2}(\rm E)$, respectively.\\
A \emph{spin manifold} is an oriented Riemannian manifold $\rm X$ with a spin structure on its tangent bundle. Moreover, the inequivalent spin structures on $\rm X$ are in one-to-one correspondence with elements of ${\rm H}^{1}({\rm X};\mathbb{Z}_{2})$.
Given an oriented Riemannian vector bundle, we can construct the associated \emph{Clifford bundle} defined as
\begin{displaymath}
\Cl({\rm E}):=P_{\rm SO}(\rm E)\times_{\rm c\ell(\rho)}\Cl_{n}
\end{displaymath}
where ${\rm c\ell}(\rho):{\rm SO}_{n}\to{\rm Aut}(\Cl_n)$ is induced by lifting orthogonal transformations of $\mathbb{R}^{n}$ to the Clifford algebra $\Cl_{n}$. The Clifford bundle $\Cl(\rm E)$ is a bundle of Clifford algebras over $\rm X$, and the fibrewise multiplication in $\Cl(\rm E)$ gives an algebra structure to the space of sections of $\Cl(\rm E)$. Hence, all the notions regarding Clifford algebras carry over to Clifford bundles. We can then look for bundles of irreducible modules over the Clifford bundle $\Cl(\rm E)$. These bundles can be constructed if $\rm E$ has a spin structure. Indeed, a \emph{real spinor bundle} of $\rm E$ is a bundle of the form  
\begin{displaymath}
{\rm S(E)}:=P_{\rm Spin}(\rm E)\times_{\mu}{\rm W}
\end{displaymath} 
where $\rm W$ is a left module for $\Cl_{n}$, and where ${\rm Spin}_{n}$ acts on $\rm W$ by left multiplication by elements of ${\rm Spin}_{n}\subset\Cl_{n}$. Analogously, a \emph{complex spinor bundle} of $\rm E$ is given by
\begin{displaymath}
{\rm S_{\mathbb{C}}(E)}:=P_{\rm Spin}(\rm E)\times_{\mu}{\rm W^{\mathbb{C}}}
\end{displaymath}
where $\rm W^{\mathbb{C}}$ is a complex left module for $\mathbb{C}\ell_{n}$. One can easily prove that the sections of the spinor bundle are a module over the sections of the Clifford bundle. A spinor bundle ${\rm S(E)}$ is called \emph{irreducible} if the left module $\rm W$(${\rm W}^{\mathbb{C}}$) is the irreducible real(complex) spinor representation of ${\rm Spin}_{n}$. In particular, when the vector bundle $\rm E$ has rank $n=2m$, the complex spinor
 bundle $S_{\mathbb{C}}(\rm E)$ associated to the irreducible representation of $\mathbb{C}\ell_{2m}$ decomposes as
\begin{displaymath}
{\rm S}_{\mathbb{C}}(\rm E)={\rm S}^{+}_{\mathbb{C}}(\rm E)\oplus{\rm S}^{-}_{\mathbb{C}}(\rm E)
\end{displaymath}
where ${\rm S}^{\pm}_{\mathbb{C}}(\rm E)$ is the $\pm{1}$ eigenbundle for Clifford multiplication by the complex volume element $\omega_{\mathbb{C}}$, and can be represented as
\begin{displaymath}
{\rm S}^{+}_{\mathbb{C}}({\rm E})\simeq{P_{\rm Spin}({\rm E})\times_{\Delta^{\mathbb{C}\pm}_{2m}}\mathbb{C}^{2m-1}}
\end{displaymath}
Let $\rm E$ be now equiped with a covariant derivative
\begin{displaymath}
\nabla:\Gamma(\rm E)\to\Gamma(\rm E\otimes\topsp{T}^{*}\rm X)
\end{displaymath}
induced by a connection $\tau$ on $P_{\rm SO}(\rm E)$. Since $\Cl(\rm E)$ is an associated vector bundle to $P_{\rm SO}(\rm E)$, $\tau$ induces a covariant derivative $\nabla_{\Cl}$ on $\Cl(\rm E)$ with the property that
\begin{displaymath}
\nabla_{\Cl}(\varphi\cdot\psi)=(\nabla_{\Cl}\varphi)\cdot\psi+\varphi\cdot(\nabla_{\Cl}\psi)
\end{displaymath} 
Moreover, the covariant derivative $\nabla_{\Cl}$ preserves the subbundles $\Cl^{0}(\rm E)$ and $\Cl^{1}(\rm E)$, and the volume element $\omega={e_{1}e_{2}\cdots{e_{n}}}$ is globally parallel, i.e. $\nabla_{\Cl}\omega=0$.\\ 
Suppose $\rm E$ is also equiped with a spin structure $\xi:P_{\rm Spin}({\rm E})\to{P_{\rm SO}({\rm E})}$. We can then lift the connection $\tau$ on $P_{\rm SO}(\rm E)$ to a connection $\tilde{\tau}$ on $P_{\rm Spin}(\rm E)$. Since a spinor bundle $\rm S(E)$ is an associated vector bundle to $P_{\rm SO}(\rm E)$, the connection $\tilde{\tau}$ induces a covariant derivative $\nabla_{\rm S}$ on $\rm S(E)$. In particular, the covariant derivative $\nabla_{\rm S}$ is compatible with $\nabla_{\Cl}$, in the sense that
\begin{displaymath}
\nabla_{\rm S}(\varphi\cdot\sigma)=(\nabla_{\Cl}\varphi)\cdot\sigma + \varphi\cdot(\nabla_{\rm S}\sigma)
\end{displaymath}
for any $\varphi\in\Gamma(\Cl(\rm E))$ and any $\sigma\in\Gamma(\rm S(E))$.\\
Given an $n$-dimensional oriented Riemannian manifold $\rm X$, we denote with $\Cl(\rm X)$ the Clifford bundle associated to $\rm TX$. The Clifford bundle carries a canonical covariant derivative, which is induced by the Levi-Civita connection on $P_{\rm SO}(\rm TX)$. Consider now any bundle $\rm S$ of left modules over $\Cl(X)$, not necessarily a spinor bundle, and suppose $\rm S$ is Riemannian and equiped with a Riemannian connection. We can then define a first-order elliptic differential operator $\rm D:\Gamma(\rm S)\to\Gamma(\rm S)$ called the \emph{Dirac operator} of $\rm S$, defined as
\begin{displaymath}
{\rm D}\sigma|_{x}:=\sum_{j=1}^{n}e_{j}\cdot(\nabla_{e_{j}}\sigma)|_{x}
\end{displaymath}
at $x\in{\rm X}$, where $e_{1},e_{2},\ldots,e_{n}$ is an orthonormal basis of $\rm T_{\it x}X$.\\
If we let ${\rm S}=\Cl(\rm X)$ with its canonical Riemannian connection, and view $\Cl(\rm X)$ as a bundle of left modules over itself by Clifford multiplication, then the Dirac operator is a square root of the classical Hodge laplacian.\\
Another case of major importance is the following. Let $\rm X$ be a spin manifold, and let $\rm S$ be any spinor bundle. The vector bundle $\rm S$ is Riemannian and carries a canonical covariant derivative $\nabla_{\rm Spin}$, called the \emph{spin connection}, induced by the lift of the Levi-Civita connection on $P_{\rm SO}({\rm TX})$. The Dirac operator in this case is called the \emph{Atiyah-Singer operator}, and plays a fundamental role in the Index theorem.
Finally, a Dirac operator associated to a bundle $\rm S$ of left modules over $\Cl(\rm X)$ can be \emph{twisted} by a Riemannian vector bundle with connection $\rm E$ by considering the Clifford multiplication on $\Gamma(\rm S\otimes{E})$ induced by  
\begin{displaymath}
\varphi\cdot(\sigma\otimes{e}):=(\varphi\cdot\sigma)\otimes{e}
\end{displaymath}
and equiping $\rm S\otimes{E}$ with the tensor product connection $\nabla$ defined on sections of the form $\sigma\otimes{e}$ by 
\begin{displaymath}
\nabla(\sigma\otimes{e}):=(\nabla^{\rm S}\sigma)\otimes{e}+\sigma\otimes(\nabla^{\rm E}e)
\end{displaymath}
where $\nabla^{\rm S}$ and $\nabla^{\rm E}$ are the covariant derivatives on $\rm S$ and $\rm E$, respectively.\\ 

Finally, we conclude this appendix by giving the definition of a $\rm Spin^{c}$ manifold. This involves first constructing the group $\rm Spin^{c}$. Consider the complex spinor representation 
\begin{displaymath}
\Delta_{\mathbb{C}}:{\rm Spin}_{n}\to{\rm GL}_{\mathbb{C}}(\rm W^{\mathbb{C}})
\end{displaymath}
and let $z:{\rm U(1)}\to{\rm GL}_{\mathbb{C}}(\rm W^{\mathbb{C}})$ denote the multiplication by scalar. We then get the homomorphism $\Delta_{\mathbb{C}}\times{z}:{\rm Spin}_{n}\times{\rm U(1)}\to{\Hom}_{\mathbb{C}}(\rm W^{\mathbb{C}},W^{\mathbb{C}})$, which has the element $(-1,-1)$ in its kernel. Dividing by this element gives the group
\begin{displaymath}
{\rm Spin^{c}}_{n}:={\rm Spin}_{n}\times_{\mathbb{Z}_{2}}{\rm U(1)}
\end{displaymath}
which satisfies the short exact sequence
\begin{displaymath}
0\to\mathbb{Z}_{2}\to{\rm Spin^{c}}_{n}\xrightarrow{\xi_{0}}{{\rm SO}_{n}\times{\rm U(1)}}\to{1}
\end{displaymath}
Let $\rm E$ be an oriented Riemannian vector bundle of rank $n$ over a manifold $X$. A \emph{$Spin^{c}$-structure} on $\rm E$ consists of a principal ${\rm Spin^{c}}_{n}$-bundle $P_{{\rm Spin^{c}}_{n}}(\rm E)$, and also a principal U(1)-bundle $P_{\rm U(1)}(\rm E)$ over $\rm X$ with a bundle map
\begin{displaymath}
P_{{\rm Spin^{c}}_{n}}({\rm E})\xrightarrow{\xi}P_{\rm SO_{n}}({\rm E})\times{P_{\rm U(1)}({\rm E})}
\end{displaymath}
such that $\xi$ satisfies $\xi(p\cdot{g})=\xi(p)\xi_{0}(g)$ for all $p\in P_{{\rm Spin_{\it n}^{c}}}(\rm E)$ and all $g\in{\rm Spin^{c}}_{n}$. The first Chern class $d(E)$ of the $\rm U(1)$-bundle $P_{\rm U(1)}(\rm E)$ is called the \emph{canonical class} of the $\rm Spin^{c}$-structure.\\
A \emph{$Spin^{c}$-manifold} is an oriented Riemannian manifold $\rm X$ with a $\rm Spin^{c}$-structure on the tangent bundle $\rm TX$. 
\thispagestyle{plain}
\newpage
\begin{center}\appendix{}\end{center}
\begin{center}
{\Large \textbf{Characteristic classes for vector bundles}}
\end{center}
In this appendix we will briefly recall some basic facts from the theory of characteristic classes, as developed in \cite{Milnorstash}.\\
All the spaces are assumed to be of the homotopy type of countable CW-complexes.
Let $\rm G$ be a Lie group. A \emph{classifying space} for $\rm G$ is a connected topological space $\rm BG$, together with a principal $\rm G$-bundle $\rm EG\to{BG}$ such that for any compact Hausdorff space $\rm X$ the set of homotopy classes of maps from $\rm X$ to $\rm BG$ is in bijective correspondence with the set of equivalence classes of principal G-bundles over $\rm X$. In particular, the above correspondence is induced by associating to each map $f:{\rm X\to{BG}}$ the pullback bundle $f^{*}{\rm EG}$ over $\rm X$.\\
The principal bundle $\rm EG\to{BG}$ is called the \emph{universal principal G-bundle}. It can be proven by direct construction that for any Lie group a classifying space does exist. Moreover, by the very properties of a classifying space, it is unique up to homotopy type.\\
Consider the singular cohomology ${\rm H}^*(\rm BG;\Lambda)$ with coefficients in a ring $\Lambda$. Each non-zero class in ${\rm H}^*(\rm BG;\Lambda)$ is a \emph{universal characteristic class} for principal $\rm G$-bundles. Fix a class $c\in{\rm H}^*(\rm BG;\Lambda)$. For each principal G-bundle $\rm P\to X$ we define the \emph{c-characteristic class} $c(\rm P)\in{\rm H}^*(\rm X;\Lambda)$ as 
\begin{displaymath}
c(\rm P):=f_{\rm P}^{*}(c)
\end{displaymath}
where $f_{\rm P}:{\rm X\to P}$ is a classifying map for $\rm P$, and is well defined, since $f_{\rm P}$ is uniquely defined up to homotopy. Moreover, any such characteristic class is ``natural'', in the sense that given a principal G-bundle P over X and a continuous map $\varphi:{\rm Y\to X}$ we have
\begin{displaymath}
c(\varphi^{*}\rm P)=\varphi^{*}c(\rm P)
\end{displaymath}
We will now specialize to the case $\rm G=BO_{\it n},BU_{\it n}$. In these particular cases, the universal bundle $\rm EG$ can be obtained as the appropriate bundle of frames of a \emph{universal real (or complex) vector bundle} $\mathbb{E}_{n}$ over $\rm BO_{\it n}$ (or $\rm BU_{\it n}$), which classifies real (or complex) vector bundles. Moreover, despite still difficult to compute, the cohomology rings of the classifying spaces $\rm BO_{\it n}$ and $\rm BU_{\it n}$ are quite manageable. The cohomology ring ${\rm H}^{*}({\rm BO}_{n};\mathbb{Z}_{2})$ is a $\mathbb{Z}_{2}$-polynomial ring
\begin{displaymath}
\mathbb{Z}_{2}[w_{1},w_{2},\ldots,w_{n}]
\end{displaymath}  
where $w_{k}$ is a canonical generators of ${\rm H}^{k}({\rm BO}_{n};\mathbb{Z}_{2})$, and  is called the \emph{universal k-th Stiefel-Whitney class}. To any $n$-dimensional real vector bundle $\rm E\to X$ classified by a map $f_{\rm E}:\rm X\to BO_{\it n}$, we can associate the $k$-th Stiefel-Whitney class of $\rm E$, $w_{k}(\rm E):=f^{*}_{\rm E}(w_{k})$. In particular, the \emph{total Stiefel-Whitney class} $w:=1+w_{1}+\cdots+w_{n}$ satisfies
\begin{displaymath}
w(\rm E\oplus{E'})=w(\rm E)\cup{w(\rm E')}
\end{displaymath}
If $\rm X$ is a smooth manifold, we can define the $k$-th Stiefel-Whitney class $w_{k}(X)$ of $X$ as $w_{k}(\rm TX)$. An important property of the Stiefel-Whitney classes is that for a compact smooth manifold they are invariants of the homotopy type of the manifold.\\

Given an $n$-dimensional real vector bundle $\rm E$, the first and second Stiefel-Whitney vanish exactly when $\rm E$ is orientable and admits a Spin structure, respectively. This can be seen by using the following equivalent definition of $w_{1}(\rm E)$ and $w_{2}(\rm E)$.\\
Recall that any isomorphism class of  principal G-bundles on a space $\rm X$ can be represented as a class in ${\rm H}^{1}(\rm X;G)$, the cohomology of $\rm X$ with coefficients in the sheaf of $\rm G$-valued functions, via the transition functions. When G is not abelian, ${\rm H}^{1}(\rm X;G)$ is not a group, but rather a set with a distinguished element given by the trivial $\rm G$-bundle. However, it still preserves some cohomological properties. Indeed, consider the short exact sequence
\begin{displaymath}
0\to{\rm SO}_{n}\xrightarrow{i}{\rm O}_{n}\xrightarrow{\rho}\mathbb{Z}_{2}\to{0},
\end{displaymath}
which induces the exact sequence 
\begin{displaymath}
{\rm H}^{1}({\rm X};{\rm SO}_{n})\xrightarrow{i_{*}}{\rm H}^{1}({\rm X};{\rm O}_{n})\xrightarrow{\rho_{*}}{\rm H}^{1}(\rm X;\mathbb{Z}_{2})
\end{displaymath}  
Given a rank $n$ vector bundle, we can define $w_{1}(\rm E)=\rho_{*}([P_{\rm O}(\rm E)])$, where $[P_{\rm O}(E)]$ denotes the class of the orthonormal frame bundle of $\rm E$. Hence, when $w_{1}(\rm E)=0$ we have that the class $[P_{\rm O}(E)]$ is the image of the class of an ${\rm SO}_{n}$-principal bundle, which is possible if and only if $\rm E$ is orientabale.\\
Similarly, the short exact sequence
\begin{displaymath}
0\to\mathbb{Z}_{2}\to{\rm Spin}_{n}\xrightarrow{\xi_{0}}{\rm SO}_{n}\to{0}
\end{displaymath}
induces the exact sequence
\begin{displaymath}
{\rm H}^{0}({\rm X};{\rm SO}_{n})\xrightarrow{\delta^{0}}{\rm H}^{1}({\rm X};\mathbb{Z}_{2})\to{\rm H}^{1}({\rm X};{\rm Spin}_{n})\xrightarrow{\xi_{0*}}{\rm H}^{1}({\rm X};{\rm SO}_{n})\xrightarrow{\delta}{\rm H}^{2}({\rm X};\mathbb{Z}_{2})
\end{displaymath} 
We can define the second Stiefel-Whitney class by $w_{2}(\rm E)=\delta([P_{\rm SO}(\rm E)])$. Hence, $w_{2}(\rm E)=0$ if and only if $P_{\rm SO}(\rm E)$ is equivalent to the $\mathbb{Z}_{2}$-quotient of a principal ${\rm Spin}_{n}$-bundle on $\rm X$.\\

For the classifying space ${\rm BU}_{n}$, we have that the cohomology ring ${\rm H}^{*}({\rm BU}_{n};\mathbb{Z})$ is a $\mathbb{Z}$-polynomial ring
\begin{displaymath}
\mathbb{Z}[c_{1},c_{2},\cdots,c_{n}]
\end{displaymath}
where $c_{k}\in{\rm H}^{k}({\rm BU}_{n};\mathbb{Z})$ is a canonical generator, and is called the \emph{universal k-th Chern class}. Thus, to any $n$-dimensional complex vector $\rm E\to{X}$ classified by a map $f_{\rm E}:{\rm X\to{BU}_{\it n}}$  we can associated the k-the Chern class $c_{k}(\rm E)=f_{\rm E}^{*}(c_{k})$. We can define the \emph{total Chern class} $c=1+c_{1}+\cdots +c_{n}$, which satisfies
\begin{equation}\label{split}
c(\rm E\oplus{E}')=c(\rm E)\cup{c(\rm E')}
\end{equation}  
For a given complex $n$-dimensional vector bundle $\rm E\to X$, we can compute the Chern classes which are \emph{nontorsion} in the following geometric way. Let $\rm F$ be the curvature of an arbitrary covariant derivative on $\rm E$. Recall that $\rm F$ is a differential form valued in the adjoint representation of $\rm U(n)$, hence in $n\times{n}$ antihermitian matrices. Then the $k$-th Chern class of $\rm E$ is given by the deRham class of $\alpha_{k}(\rm E)$, where
\begin{displaymath}
{\rm det}({\rm I}+\dfrac{{\rm F}t}{2\pi})=\sum_{k=1}^{n}\alpha_{k}(\rm E)t^{k}
\end{displaymath} 
The Chern classes defined above can be used as the basic ingredient to define other important characteristic class of vector bundle over a manifold. This is due to the following\\[2mm]
\textbf{Theorem C.1}\emph{
Let $\rm E$ be a $n$-dimensional complex vector bundle over a manifold $\rm X$. Then there exists a manifold $\mathcal{M}_{\rm E}$ and a smooth and proper fibration $\pi:\mathcal{M}_{\rm E}\to{\rm X}$ such that
\begin{itemize}
\item[i)] the homomorphism $\pi^{*}:{\rm H}^{*}(\rm X)\to{\rm H}^{*}(\mathcal{M}_{\rm E})$ is injective
\item[ii)] the bundle $\pi^{*}{\rm E}$ splits into the direct sum of complex line bundles
\begin{displaymath}
\pi^{*}{\rm E}\simeq\mathcal{L}_{1}\oplus\mathcal{L}_{2}\oplus\cdots\oplus\mathcal{L}_{n}
\end{displaymath}
\end{itemize}}
The above theorem ``induces'' the following \emph{splitting principle}: all polynomial identities in the Chern classes of complex vector bundles can be proven under the assumption that all vector bundles are direct sums of line bundles.\\
For a real vector bundle $\rm E$ of dimension $2n$ one can prove that the complexification $\rm E\otimes\mathbb{C}$ splits on $\mathcal{M}_{\rm E}$ as
\begin{displaymath}
\pi^{*}(\rm E\otimes\mathbb{C})\simeq\mathcal{L}_{1}\oplus\overline{\mathcal{L}}_{1}\oplus\cdots\oplus\mathcal{L}_{n}\oplus\overline{\mathcal{L}}_{n}
\end{displaymath}
where $\overline{\mathcal{L}}_{1}$ denotes the complex conjugate of the $\mathcal{L}_{1}$.\\

  Notice that by naturality property, any c-characteristic class of $\pi^{*}{\rm E}$ is in the image of the homomorphism $\pi^{*}$. We can then construct characteristic classes as the unique preimage via the splitting homomorphism $\pi^{*}$ of functions of the Chern classes of the splitting line bundles. For instance, by using the property (\ref{split}) of the total Chern class, we have
\begin{displaymath}
c({\rm E})=\prod_{k=1}^{n}(1+x_{k})
\end{displaymath}  
where $x_{k}=c_{1}(\mathcal{L}_{1})$.\\
In this way, we can associate to formal power series rational characteristic classes. We can define the \emph{total Todd class} of a complex vector bundle $\rm E$ by
\begin{displaymath}
{\rm Td}({\rm E}):=\prod_{k=1}^{n}\dfrac{x_{k}}{1-{e}^{-x_{k}}}
\end{displaymath} 
Given a real vector bundle $\rm E$ of dimension $2n$, the \emph{{total $\hat{A}$-class}} can be defined by
\begin{displaymath}
\hat{\rm A}({\rm E}):=\prod_{k=1}^{n}\dfrac{x_{k}/2}{{\rm sinh}(x_{k}/2)}
\end{displaymath} 
Finally, we can define the \emph{total Chern character} by
\begin{displaymath}
{\rm ch}({\rm E}):=\sum_{k=1}^{n}e^{x_{k}}
\end{displaymath}
The Chern charater can be represented in the deRham cohomology of $\rm X$ as
\begin{displaymath}
{\rm ch}({\rm E})=[{\rm Tr}(e^{{\rm F}/2\pi})]\in{\rm H}^{ev}({\rm X};\mathbb{R})
\end{displaymath}
where $\rm F$ is the curvature of a covariant derivative on $\rm E$, and the trace $\rm Tr$ is in the adjoint representation of $\rm U(n)$.
\thispagestyle{plain}
\newpage
\begin{center}\appendix{}\end{center}
\begin{center}
{\Large \textbf{Equivariant K-homology}}
\end{center}
\pagestyle{plain}
\subsection*{Spectral definition\label{EqKhom}}

A natural way to define the equivariant homology theory $\K_\bullet^\topsp{G}$
is by means of a \emph{spectrum} for equivariant topological K-theory
$\K^\bullet_\topsp{G}$, which  is a
particular covariant functor $\underline{\Vect}^\topsp{G}(-)$ from the orbit
category $\ocat{\topsp{G}}$ to the tensor category $\cat{Spec}$ of
spectra~\cite{Davis1998,mislin}. Given any $\topsp{G}$-complex $\topsp{X}$, the corresponding
pointed $\topsp{G}$-space is $\topsp{X}_+=\topsp{X}\amalg\pt$ and one defines the loop
spectrum $\topsp{X}_+\otimes_\topsp{G}\underline{\Vect}^\topsp{G}(-)$ by
\beq
\topsp{X}_+\otimes_\topsp{G}\underline{\Vect}^\topsp{G}(-)=\coprod_{\topsp{G}/\topsp{H}\in\ocat{\topsp{G}}}\,
\big(\topsp{X}_+^\topsp{H}\wedge\underline{\Vect}^\topsp{G}(\topsp{G}/\topsp{H})\big)\,\big/\,\sim \ ,
\label{loopspec}\eeq
where the equivalence relation $\sim$ is generated by the
identifications $f^*(x)\wedge s\sim x\wedge f_*(s)$ with $(f:\topsp{G}/{\rm K}\to
\topsp{G}/\topsp{H})\in{\rm Mor}(\ocat{\topsp{G}})$, $x\in \topsp{X}_+^\topsp{H}$,
and $s\in\underline{\Vect}^\topsp{G}(\topsp{G}/{\rm K})_{*}$. One then puts
\beq
\K_{*}^\topsp{G}(\topsp{X}):=\pi_{*}\big(\topsp{X}_+\otimes_\topsp{G}\underline{\Vect}^\topsp{G}(-)
\big) \ .
\label{KiGXspec}\eeq

By using various $\topsp{G}$-homotopy equivalences of the loop spectra
(\ref{loopspec}), one shows that this definition of equivariant
K-homology comes with a natural induction structure. For the trivial group it reduces to the
ordinary K-homology $\K_*^e=\K_*$ given by the Bott
spectrum ${BU}$. If $\topsp{G}$ is a finite group, any finite-dimensional
representation of $\topsp{G}$ naturally extends to a complex representation of
the group ring $\complex[\topsp{G}]$. Then there is an analytic assembly map
$$
\ass\,:\,\K_*^\topsp{G}(\topsp{X})~\longrightarrow~\K_*\big(\complex[\topsp{G}]\big)
$$
to the K-theory of the ring $\complex[\topsp{G}]$, induced by the collapsing
map $\topsp{X}\to\pt$ and the isomorphisms
$$\K_*\big(\complex[\topsp{H}]\big)~\cong~
\pi_*\big(\,\underline{\Vect}^\topsp{G}(\topsp{G}/\topsp{H})\big)~\cong~
\K_*^\topsp{G}(\topsp{G}/\topsp{H})~\cong~ {\rm R}(\topsp{H})$$ for any subgroup
$\topsp{H}\leq \topsp{G}$ \cite{mislin}. In the following we will give two concrete realizations of
the homotopy groups~(\ref{KiGXspec}).

\subsection*{Analytic definition\label{Andef}}

The simplest realization of the equivariant K-homology
group $\K_*^\topsp{G}(\topsp{X})$ is within the framework of an equivariant
version of Kasparov's KK-theory $\KK^\topsp{G}_*$. Let $\alg$ be a
$\topsp{G}$-algebra, \emph{i.e.}, a $C^*$-algebra $\alg$ together
with a group homomorphism $$\lambda\,:\,\topsp{G}~\longrightarrow~{\rm
  Aut}(\alg) \ . $$ By a Hilbert
$(\topsp{G},\alg)$-module we mean a Hilbert $\alg$-module $\bun$ together with
a $\topsp{G}$-action given by a homomorphism $\Lambda:\topsp{G}\to{\rm GL}(\bun)$ such
that
\beq
\Lambda_g(\varepsilon\cdot a)=\Lambda_g(\varepsilon)\cdot
\lambda_g(a)
\label{covrep}\eeq
for all $g\in \topsp{G}$, $\varepsilon\in\bun$ and $a\in\alg$. Let
$\lin(\bun)$ denote the $*$-algebra of $\alg$-linear maps
$\topsp{T}:\bun\to\bun$ admitting an adjoint with respect to the $\alg$-valued
inner product on $\bun$. The induced $\topsp{G}$-action on $\lin(\bun)$
is given by $g\cdot \topsp{T}:=\Lambda_g\circ \topsp{T}\circ\Lambda_{g^{-1}}$. Let
$\comp(\bun)$ be the subalgebra of $\lin(\bun)$ consisting of
generalized compact operators.

Given a pair $(\alg,\balg)$ of $\topsp{G}$-algebras, let
$\cat{D}^\topsp{G}(\alg,\balg)$ be the set of triples $(\bun,\phi,\topsp{T})$ where
$\bun$ is a countably generated Hilbert $(\topsp{G},\balg)$-module,
$\phi:\alg\to\lin(\bun)$ is a $*$-homomorphism which commutes with the
$\topsp{G}$-action,
\beq
\phi\big(\lambda_g(a)\big)=\Lambda_g\circ\phi(a)\circ\Lambda_{g^{-1}}
\label{covrep2}\eeq
for all $g\in \topsp{G}$ and $a\in\alg$, and $\topsp{T}\in\lin(\bun)$ such that
\begin{itemize}
\item[1)] $[\topsp{T},\phi(a)]\in\comp(\bun)$ for all $a\in\alg$; and
\item[2)] $\phi(a)\,(\topsp{T}-\topsp{T}^*)$, $\phi(a)\,(\topsp{T}^2-1)$, $\phi(a)\,(g\cdot
  \topsp{T}-\topsp{T})\in\comp(\bun)$ for all $a\in\alg$ and $g\in \topsp{G}$.
\end{itemize}
The standard equivalence relations of KK-theory are now analogously
defined. The set of equivalence classes in $\cat{D}^\topsp{G}(\alg,\balg)$
defines the equivariant KK-theory groups $\KK_*^\topsp{G}(\alg,\balg)$.

If $\topsp{X}$ is a smooth proper $\topsp{G}$-manifold without boundary, and $\topsp{G}$ acts
on $\topsp{X}$ by diffeomorphisms, then the algebra
$\alg=\C_0(\topsp{X})$ of continuous functions on $\topsp{X}$ vanishing at infinity is
a $\topsp{G}$-algebra with automorphism $\lambda_g$ on $\alg$ given by
$$\lambda_g(f)(x)~:=~\big(g^*f\big)(x)\=f\big(g^{-1}\cdot x\big) \ , $$
where $g^*$ denotes the pullback of the $\topsp{G}$-action on $\topsp{X}$ by left
translation by $g^{-1}\in \topsp{G}$. We define
\beq\label{KKGandef}
\K_*^\topsp{G}(\topsp{X}):=\KK_*^\topsp{G}\big(\C_0(\topsp{X})\,,\,\complex\big)
\eeq
with $\topsp{G}$ acting trivially on $\complex$. The conditions (\ref{covrep})
and (\ref{covrep2}) naturally capture the physical requirements that
physical orbifold string states are $\topsp{G}$-invariant and also that the
worldvolume fields on a fractional D-brane carry a ``covariant
representation'' of the orbifold group~\cite{douglas1996}.

\subsection*{The equivariant Dirac class\label{Diracclass}}

We can determine a canonical class in the abelian group
(\ref{KKGandef}) as follows. We refer to {\appclifford} for basic notions of Clifford algebras.\\
Let $\dim(\topsp{X})=2n$, and let $\topsp{G}$ be a finite
subgroup of the rotation group $\SO(2n)$.\footnote{Throughout the
  extension to $\K^\topsp{G}_1$ or $\K^{-1}_\topsp{G}$ and $\dim(\topsp{X})$ odd can be
  described in the same way as in degree zero by replacing $\topsp{X}$ with
  $\topsp{X}\times\S^1$.} 
A choice of a complete $\topsp{G}$-invariant riemannian
metric on $\topsp{X}$ naturally lieft the Clifford bundle $\Cl(\topsp{X})$. 

The $\topsp{G}$-manifold $\topsp{X}$ is said to have a \emph{$G$-\spinc structure} or
to be \emph{$\K_G$-oriented} if there is an extension of the orthonormal frame
bundle $P_{{\rm SO}_{2n}}(\topsp{X})$ to a principal $\Spin^c(2n)$-bundle $P_{{\rm Spin^{c}}_{2n}}(\topsp{X})$ over $\topsp{X}$ which is
compatible with the $\topsp{G}$-action. The extension $P_{{\rm Spin^{c}}_{2n}}(\topsp{X})$ may be
regarded as a principal circle bundle over $P_{{\rm SO}_{2n}}(\topsp{X})$,
$$
\xymatrix{ & & \U(1)\ar[ld]\ar[d] & \\ 
\hat{\topsp{G}}~\ar[r]\ar[d] & ~\Spin^c(2n)~\ar[r]\ar[d] & ~
P_{{\rm Spin^{c}}_{2n}}(\topsp{X})\ar[r]\ar[d]&~ \topsp{X} \ , \\
\topsp{G}~\ar[r] & {\rm SO}(2n)\ar[r] & P_{{\rm SO}_{2n}}(\topsp{X})\ar[ru] & }
$$
where the pullback square on the bottom left defines the required
covering of the orbifold group $\topsp{G}<\SO(2n)$ by a subgroup of the \spinc
group $\hat{\topsp{G}}<\Spin^c(2n)$. The
kernel of the homomorphism $\hat{\topsp{G}}\to\topsp{G}$ is identified with the circle
group $\U(1)<\Spin^c(2n)$. We fix
a choice of lift and hence assume that $\topsp{G}$ is a discrete subgroup of
the \spinc group. $\zed_2$-graded Clifford modules are likewise
extended to representations of $\complex[\topsp{G}]\otimes\Cliff(2n)$, with
$\complex[\topsp{G}]$ the group ring of $\topsp{G}$, called $\topsp{G}$-Clifford modules.\\ 
Since $\topsp{G}$ lifts to $\hat{\topsp{G}}$ in the \spinc group, the
 the spinor bundles bundles ${\rm S}^{\pm}$ associated to $P_{{\rm Spin^{c}}_{2n}}(\topsp{X})$ are naturally $\topsp{G}$-bundles. The
$\topsp{G}$-invariant Levi-Civita connection determines a connection
on $P_{{\rm SO}_{2n}}(\topsp{X})$, and together with a choice of $\topsp{G}$-invariant connection form on the principal $\U(1)$-bundle $P_{{\rm Spin^{c}}_{2n}}(\topsp{X})\to P_{{\rm SO}_{2n}}(\topsp{X})$,
they determine a connection one-form on the principal bundle $P_{{\rm Spin^{c}}_{2n}}(\topsp{X})\to \topsp{X}$ which is $\topsp{G}$-invariant. This
determines an invariant covariant derivative
$$
\nabla^{{\rm S}\otimes
  \topsp{E}}:\,
\Gamma\big({\rm S}^+\otimes \topsp{E}\big)~\to~
\Gamma\big(\,\topsp{T}^*_\topsp{X}\otimes{\rm S}^+\otimes \topsp{E}\big)
$$
where $\nabla^\topsp{E}$ is a $\topsp{G}$-invariant connection on a $\topsp{G}$-bundle $\topsp{E}\to
\topsp{X}$. The contraction given by Clifford multiplication defines a map
$$
\Cl\,:\,\Gamma\big(\,\topsp{T}^*_\topsp{X}\otimes{\rm S}^+\otimes \topsp{E}\big)~
\to~\Gamma\big({\rm S}^-\otimes \topsp{E}\big)
$$
which graded commutes with the $\topsp{G}$-action, and the $\topsp{G}$-invariant
\spinc Dirac operator on $\topsp{X}$ with coefficients in $\topsp{E}$ is defined as
the composition
\beq
\Dirac_\topsp{E}^{\topsp{X}}=\Cl\circ\nabla^{{\rm S}\otimes \topsp{E}} \ .
\label{DiracEX}\eeq

We will view the operator (\ref{DiracEX}) as an operator on
$\Ltwo$-spaces $$\Dirac_\topsp{E}^{\topsp{X}}\,:\,\Ltwo\big(\topsp{X}\,,\,{\rm S}^+\otimes
\topsp{E}\big)~\longrightarrow~\Ltwo\big(\topsp{X}\,,\,{\rm S}^-\otimes \topsp{E}\big) \ . $$
It induces a class
$\big[\Dirac_\topsp{E}^\topsp{X}\big]\in\K_0^\topsp{G}(\topsp{X})$ as follows. The $\topsp{G}$-algebra
$\C_0(\topsp{X})$ acts on the $\zed_2$-graded $\topsp{G}$-Hilbert space
$\bun:=\Ltwo(\topsp{X},{\rm S}\otimes \topsp{E})$ by multiplication. Define the
bounded $\topsp{G}$-invariant operator
$\topsp{T}:=\Dirac_\topsp{E}^\topsp{X}\,\big((\Dirac_\topsp{E}^\topsp{X})^2+1\big)^{-1/2}\in\Fred_\topsp{G}$. Then
$\big[\Dirac_\topsp{E}^\topsp{X}\big]$ is represented by the $\topsp{G}$-equivariant Fredholm
module $(\bun,\topsp{T})$.
\thispagestyle{plain}
\subsection*{Geometric definition\label{Topdef}}

Geometric equivariant K-homology can be defined for an arbitrary
discrete, countable group $\topsp{G}$ on the category of proper, finite
$\topsp{G}$-complexes $\topsp{X}$ and proven to be isomorphic to analytic equivariant
K-homology~\cite{baum-2007-3}. Recall that the topological equivariant
K-theory $\K_\topsp{G}^*(\topsp{X})$ is defined by applying the Grothendieck
functor $\K^*$ to the additive category $\Vect^\complex_\topsp{G}(\topsp{X})$
whose objects are complex $\topsp{G}$-vector bundles over $\topsp{X}$, \emph{i.e.},
$\K_\topsp{G}^*(\topsp{X}):=\K^*\big(\Vect^\complex_\topsp{G}(\topsp{X})\big)$.
In the homological setting, the relevant category
is instead the additive category of \emph{$G$-equivariant K-cycles}
$\cat{D}^\topsp{G}(\topsp{X})$, whose objects are triples $(\topsp{W},\topsp{E},f)$ where
\begin{itemize}
\item[(a)] $\topsp{W}$ is a manifold without boundary with a smooth proper
  cocompact $\topsp{G}$-action and $\topsp{G}$-\spinc structure;
\item[(b)] $\topsp{E}$ is an object in $\Vect^\complex_\topsp{G}(\topsp{W})$; and
\item[(c)] $ f:\topsp{W}\to \topsp{X}$ is a $\topsp{G}$-map.
\end{itemize}
Two $\topsp{G}$-equivariant K-cycles $(\topsp{W},\topsp{E}, f)$ and $(\topsp{W}',\topsp{E}', f'\,)$ are
said to be \emph{isomorphic} if there is a $\topsp{G}$-equivariant
diffeomorphism $h:\topsp{W}\to \topsp{W}'$ preserving the $\topsp{G}$-\spinc structures on
$\topsp{W},\topsp{W}'$ such that $h^*(\topsp{E}'\,)\cong \topsp{E}$ and $ f'\circ h= f$.

Define an equivalence relation $\sim$ on the category $\cat{D}^\topsp{G}(\topsp{X})$
generated by the operations of
\begin{itemize}
\item[i)] Bordism: $(\topsp{W}_i,\topsp{E}_i, f_i)\in\cat{D}^\topsp{G}(\topsp{X})$,
  $i=0,1$ are \emph{bordant} if there is a triple $(\topsp{M},\topsp{E}, f)$ where
  $\topsp{M}$ is a manifold with boundary $\partial \topsp{M}$, with a smooth proper
  cocompact $\topsp{G}$-action and $\topsp{G}$-\spinc structure, $\topsp{E}\to \topsp{M}$ is a complex
  $\topsp{G}$-vector bundle, and $ f:\topsp{M}\to \topsp{X}$ is a $\topsp{G}$-map such that
  $(\partial \topsp{M},\topsp{E}|_{\partial \topsp{M}}, f|_{\partial
    \topsp{M}})\cong(\topsp{W}_0,\topsp{E}_0, f_0)\amalg(-\topsp{W}_1,\topsp{E}_1, f_1)$. Here $-\topsp{W}_1$
  denotes $\topsp{W}_1$ with the reversed $\topsp{G}$-\spinc structure;
\item[ii)] Direct sum: If $(\topsp{W},\topsp{E}, f)\in\cat{D}^\topsp{G}(\topsp{X})$ and
  $\topsp{E}=\topsp{E}_0\oplus \topsp{E}_1$, then
  $$(\topsp{W},\topsp{E}, f)\cong(\topsp{W},\topsp{E}_0, f)\amalg(\topsp{W},\topsp{E}_1, f) \ ; $$ and
\item[iii)] Vector bundle modification: Let
  $(\topsp{W},\topsp{E}, f)\in\cat{D}^\topsp{G}(\topsp{X})$ and $\topsp{H}$ an even-dimensional $\topsp{G}$-\spinc
  vector bundle over $\topsp{W}$. Let $\widehat{\topsp{W}}=\S(\topsp{H}\oplus\id)$ denote the
  sphere bundle of $\topsp{H}\oplus\id$, which is canonically a $\topsp{G}$-\spinc
  manifold, with $\topsp{G}$-bundle projection $\pi:\widehat{\topsp{W}}\to \topsp{W}$. Let
  $${\rm S}(\topsp{H})={\rm S}(\topsp{H})^+\oplus{\rm S}(\topsp{H})^-$$ denote the
  $\zed_2$-graded $\topsp{G}$-bundle over $\topsp{W}$ of spinors on $\topsp{H}$. Set
  $\widehat{\topsp{E}}=\pi^*\big(({\rm S}(\topsp{H})^+)^\vee\otimes \topsp{E}\big)$ and
  $\widehat{ f}= f\circ\pi$. Then
  $\big(\,\widehat{\topsp{W}}\,,\,\widehat{\topsp{E}}\,,\,\widehat{
    f}~\big)\in\cat{D}^\topsp{G}(\topsp{X})$ is the \emph{vector bundle modification}
  of $(\topsp{W},\topsp{E}, f)$ by $\topsp{H}$.
\end{itemize}
We set
$$
\K^\topsp{G}_{0,1}(\topsp{X})=\cat{D}_{{\rm even},{\rm odd}}^\topsp{G}(\topsp{X})\,\big/\,\sim
$$
where the parity refers to the dimension of the K-cycle, which is
preserved by $\sim$.

Using the equivariant Dirac class, one can construct a
homomorphism from the geometric to the analytic K-homology group. On
K-cycles we define $(\topsp{W},\topsp{E},f)\mapsto f_*\big[\Dirac_\topsp{E}^\topsp{W}\big]$ and extend
linearly. This map can be used to express $\topsp{G}$-index theorems within
this homological framework and it extends to give an isomorphism
between the two equivariant K-homology groups~\cite{baum-2007-3}  
\newpage
\pagestyle{plain}
\addcontentsline{toc}{section}{Bibliography}
\bibliography{biblio}
\bibliographystyle{abbrv}
\newpage
\thispagestyle{empty}
\quad
\vspace{19.5cm}
\begin{center}
This document has been entirely\\
\includegraphics[scale=0.7]{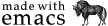}\\
using {La\TeX} on Fedora 7 (Moonshine) Linux distribution.
\end{center}
\end{document}